\documentclass[11pt,letterpaper]{amsart}

\RequirePackage[OT1]{fontenc}
\RequirePackage{amsthm,amsmath}
\RequirePackage[numbers]{natbib}
\RequirePackage[colorlinks,citecolor=blue,urlcolor=blue]{hyperref}
\RequirePackage{cleveref}

\usepackage{amsfonts}

\usepackage{dsfont}
\usepackage{graphicx}
\usepackage{xcolor}
\usepackage{tikz}
\usepackage{enumitem}
\usepackage[colorinlistoftodos]{todonotes}
\usetikzlibrary{decorations.markings}
\usepackage{mathtools}

\usepackage{graphicx}
\newtheorem{theorem}{Theorem}
\newtheorem{lemma}[theorem]{Lemma}
\newtheorem*{lemma*}{Lemma}
\newtheorem{proposition}[theorem]{Proposition}
\newtheorem{corollary}[theorem]{Corollary}

\newtheorem{remark}[theorem]{Remark}

\newtheorem*{fact*}{Fact}

\usepackage{upgreek}

\usepackage[utf8]{inputenc}
\usepackage[english]{babel}
\usepackage{latexsym}
\usepackage{amssymb}
\usepackage{amsmath}

\usepackage{pgfplots}
\usepackage{tikz}
\usetikzlibrary{decorations.markings}

\pgfplotsset{every axis/.append style={
                    axis x line=middle,    
                    axis y line=middle,    
                    axis line style={->}, 
                    ymajorticks=false,
                    xtick={-1,1},                     
                    }}
\tikzset{>=stealth}

\newcommand{\N}{\mathbb{N}}

\newcommand{\Z}{\mathbb{Z}}
\newcommand{\R}{\mathbb{R}}
\newcommand{\C}{\mathbb{C}}
\newcommand{\E}{\mathbb{E}}

\newcommand{\bigO}{\mathcal{O}}

\newcommand{\Tr}{\mathrm{Tr}\,}
\renewcommand{\d}{\, \mathrm{d}}
\renewcommand{\Re}{\mathrm{Re} \,}
\renewcommand{\Im}{\mathrm{Im} \,}
\newcommand{\1}{\mathbf{1}}
\newcommand{\ualpha}{\boldsymbol{\upalpha}}
\newcommand{\ubeta}{\boldsymbol{\upbeta}}

\newcommand{\p}{\mathsf p}

\numberwithin{equation}{section}
\numberwithin{theorem}{section}

\newcommand{\ds}{\displaystyle}

\makeatletter
\def\@tocline#1#2#3#4#5#6#7{\relax
  \ifnum #1>\c@tocdepth 
  \else
    \par \addpenalty\@secpenalty\addvspace{#2}%
    \begingroup \hyphenpenalty\@M
    \@ifempty{#4}{%
      \@tempdima\csname r@tocindent\number#1\endcsname\relax
    }{%
      \@tempdima#4\relax
    }%
    \parindent\z@ \leftskip#3\relax \advance\leftskip\@tempdima\relax
    \rightskip\@pnumwidth plus4em \parfillskip-\@pnumwidth
    #5\leavevmode\hskip-\@tempdima
      \ifcase #1
       \or\or \hskip 1em \or \hskip 2em \else \hskip 3em \fi
      #6\nobreak\relax
    \dotfill\hbox to\@pnumwidth{\@tocpagenum{#7}}\par
    \nobreak
    \endgroup
  \fi}
\makeatother

\title[Hankel determinants with a multi-cut potential]{
Asymptotics of Hankel determinants with a multi-cut regular potential and Fisher-Hartwig singularities}

\author[C. Charlier, B. Fahs, C. Webb, and M.D. Wong]{Christophe Charlier$^1$ \and Benjamin Fahs$^2$ \and Christian Webb$^3$ \and Mo Dick Wong$^4$}
\date{ \today \\
    $^1$University of Copenhagen, Department of Mathematical Sciences, \texttt{charlier@math.ku.dk}\\%
    $^2$KTH Royal Institute of Technology, Department of Mathematics, \texttt{bfahs@kth.se}\\%
    $^3$University of Helsinki, Department of Mathematics and Statistics, \texttt{christian.webb@helsinki.fi}\\%
    $^4$Durham University, Department of Mathematical Sciences, \texttt{ mo-dick.wong@durham.ac.uk}%
}

\begin{document}

\begin{abstract}
We obtain large $N$ asymptotics for $N \times N$ Hankel determinants corresponding to non-negative symbols with Fisher-Hartwig (FH) singularities in the multi-cut regime. Our result includes the explicit computation of the multiplicative constant.

More precisely, we consider symbols of the form $\omega e^{f-NV}$, where $V$ is a real-analytic potential whose equilibrium measure $\mu_V$ is supported on several intervals, $f$ is analytic in a neighborhood of  $\textrm{supp}(\mu_V)$, and $\omega$ is a function with any number of jump- and root-type singularities in the interior of $\textrm{supp}(\mu_V)$. 

 While the special cases  $\omega\equiv1$ and $\omega e^f\equiv1$ have been considered previously in the literature, we also prove new results for these special cases. No prior asymptotics were available in the literature for symbols with FH singularities in the multi-cut setting.

As an application of our results, we discuss a connection between the spectral fluctuations of random Hermitian matrices in the multi-cut regime and the Gaussian free field on the Riemann surface associated to $\mu_V$. As a second application, we obtain new rigidity estimates for random Hermitian matrices in the multi-cut regime.

\end{abstract}

\maketitle
{\hypersetup{linkcolor=black}
	\tableofcontents}

\section{Introduction and main results}\label{sec: intro} 
In this article, we study the asymptotics of Hankel determinants $H_N(\nu)$ as $N\to \infty$, where
\begin{equation} \label{eq:HNnu}
H_N(\nu)=\det\left(\int_{\R}x^{i+j}\nu(x)dx\right)_{i,j=0}^{N-1},
\end{equation}
for certain integrable functions $\nu=\nu_N:\mathbb R\to [0,\infty)$ which we allow to vary with $N$. An important reason to study these objects is their connection to statistical mechanics and matrix integrals. 

For instance,  the classical fact (see e.g. \cite[Chapter 3]{Deift}) that $H_N$ has the multiple integral representation
\begin{equation}\label{eq:lgmi1}H_N(\nu)=\frac{1}{N!}\int_{\R^N}\prod_{1\leq i<j\leq N}(x_i-x_j)^2 \prod_{j=1}^N \nu(x_j)dx_j, \end{equation}
 provides a connection to the statistical mechanics of a gas of particles interacting through a logarithmic repulsion. 

The viewpoint of matrix integrals stems from the fact (see e.g. \cite[Chapter 5]{Deift}) that if $dM$ denotes the Lebesgue measure on the space of $N\times N$ Hermitian matrices and $V:\R\to \R$ is a smooth function with sufficient growth at $\pm \infty$, then
\begin{multline}\label{eq:lgmi2}
H_N(e^{-NV})=\frac{c_N}{N!} \int e^{-N\Tr V(M)}dM,\\  dM=\prod_{j=1}^N dM_{j,j} \prod_{1\leq i < j\leq N}d\Re M_{i,j}\, d\Im M_{i,j},
\end{multline}
where   $c_N=\pi^{-N(N-1)/2}\prod_{j=1}^Nj!$. The trace in \eqref{eq:lgmi2} is interpreted as $\Tr V(M)$ $=\sum_{j=1}^N V(\lambda_j)$, where $(\lambda_1,...,\lambda_N)\in \R^N$ denote the eigenvalues of the matrix $M$.

Such integrals play a role in a number of different fields: for example, they can be used to enumerate certain graphs \cite{BIZ,EML}, they are connected to orthogonal polynomials, and they are important in random matrix theory \cite{Deift}. The last two subjects feature prominently in this paper. We now discuss the role Hankel determinants play in the theory of random matrices, while we return to the connections to orthogonal polynomials in detail in Section \ref{eq:ortho}.

 Consider  the probability measure whose density is proportional to 
\begin{equation}\label{distributionM}e^{-N\Tr V(M)}dM \end{equation} on the space of Hermitian $N\times N$ matrices, and denote the eigenvalues of $M$ by $\lambda_1,\dots, \lambda_N$. The marginal density of $\lambda_1,\dots, \lambda_N$ is proportional to $\prod_{1\leq i<j\leq N} (\lambda_i-\lambda_j)^2\prod_{j=1}^Ne^{-NV(\lambda_j)}$, and it follows from \eqref{eq:lgmi1} that
\begin{equation}\label{eq:RMTrat}
\E e^{\Tr g(M)}=\frac{H_N(e^{g-NV})}{H_N(e^{-NV})},
\end{equation}
for suitable functions $g$, where $\E$ is the expectation over the random matrix ensemble, and the trace is interpreted as in \eqref{eq:lgmi2}.
Thus asymptotics of (the ratio of) Hankel determinants translate into characterizing the asymptotic behavior of the random variable $\Tr g(M)$ as $N\to\infty$, which in turn is of interest because it yields information about the spectrum of the underlying random matrix. Expectations of the form \eqref{eq:RMTrat} and their generalizations\footnote{Typical generalizations are e.g. $\beta$-matrix models, where in \eqref{eq:lgmi1} the term $(x_i-x_j)^2$ is replaced by $|x_i-x_j|^\beta$. A number of the references below also cover this more general setting, e.g. \cite{BG1,BG2,Eynard,   Johansson, SS, Shcherbina}.} have been studied extensively. It is a classical fact (see e.g. \cite[Section 6]{Deift}) for rather general classes of functions $g$, that the leading order behavior of $\Tr g(M)$ is governed asymptotically by a deterministic quantity which can be written as $N\int g(x)d\mu_V(x)$, where $\mu_V$ is known as the equilibrium measure associated to the potential $V$. We return to the equilibrium measure in more detail in Section \ref{sec:intro_em} below, for now we simply mention the probabilistic interpretation which is that it describes the asymptotic density of the eigenvalues of $M$ as $N\to \infty$, and in this limit all eigenvalues of $M$ will lie in any fixed neighbourhood of the support of $\mu_V$ with probability tending to $1$.

Our goal in this paper is to obtain asymptotics for \eqref{eq:lgmi2} and \eqref{eq:RMTrat} for certain potentials $V$ and functions $g$, and before moving on to the technical details we give a brief overview of the situations in which we are interested. We require $V$ to be real analytic, satisfying $V(x)/|\log x|\to +\infty$ as $x\to \pm \infty$, and that the equilibrium measure of $V$ is $k$-cut regular, $k=1,2,\dots$, which means that $d\mu_V$ is supported on $k$ intervals and additionally satisfies certain technical conditions which we define in Section \ref{sec:intro_em}. Under these conditions, we are interested in the following problems:
\begin{itemize}
\item[(a)] Computing asymptotics for the partition function $H_N(e^{-NV})$ up to and including the multiplicative constant. In the one-cut case ($k=1$) complete asymptotics are by now well-known;  they have been studied  by Ercolani and McLaughlin \cite{EML}, Bleher and Its \cite{BI}, Borot and Guionnet \cite{BG1}, Berestycki, Webb and Wong \cite{BWW}, and Charlier \cite{Charlier}. Obtaining asymptotics in the multi-cut case is a technical challenge, and has been studied both in the physics community by among others Eynard  (see e.g. \cite{Eynard}), and in the mathematics community by Borot and Guionnet \cite{BG1,BG2}, Claeys, Grava and McLaughlin \cite{CGML}, Sandier and Serfaty \cite{SS}, and Shcherbina \cite{Shcherbina}. 

While all of the above mentioned papers study the same problem, the authors have  focused on different goals.
The work by Sandier and Serfaty holds in an impressive level of generality, and covers not just potentials which are regular but a wide class of potentials, up to a multiplicative error term of $\mathcal O(N)$. The works of Eynard and of Borot and Guionnet look at the structure of the full expansion in the multi-cut regular case. The work of Claeys, Grava and McLaughlin \cite{CGML} is restricted to $2$-cut regular potentials, and in this case a particularly elegant result is obtained, describing the asymptotics up to and including the multiplicative constant, in terms of the equilibrium measure $\mu_V$ and special functions such as Jacobi's $\theta$-function and elliptic integrals. 

One of the main goals of the current paper is to obtain an analogue of the results in \cite{CGML} in the $k$-cut regular situation for $k=3,4,\dots$, and in  Theorem \ref{th:pfasy} below we present asymptotics for $H_N(e^{-NV})$ valid up to and including the multiplicative constant. Our formula is described in terms of the equilibrium measure and special functions such as Riemann's $\theta$-function. Our formula is also valid for $k=2$, in which case it agrees with the results of \cite{CGML}. In the special case where $V(x)=\frac{2\nu}{k}\Pi_k(x)^2$, and $\Pi_k$ is a monic polynomial of degree $k$ with $k$ real and distinct zeros and $\nu>0$ is fixed and sufficiently large, then  the equilibrium measure is $k$-cut regular, and  in Remark \ref{RemarkPoly1} we provide a  simplification of the asymptotics of Theorem \ref{th:pfasy}.

\item[(b)] Computing the asymptotics of $\E e^{\Tr f(M)}$ for $f$ which is real on $\R$ and analytic in a neighbourhood of the support  $\mu_V$. In the one-cut case this was solved   by Johansson \cite{Johansson} for certain classes of polynomial potentials and has later been approached for one-cut regular potentials through a variety of methods, see e.g. \cite{BG1,BWW}.  It follows from these works that as $N\to \infty$, $\Tr f(M)-N\int fd\mu_V$ converges in distribution to a normally distributed random variable.

In the $k$-cut regular setting, this no longer holds in general, as first noticed by Pastur in \cite{Pastur}, and further developed by Shcherbina in \cite{Shcherbina} and Borot and Guionnet in \cite{BG2}. Borot and Guionnet provided a formula for $\E e^{\Tr f(M)-N\int fd\mu_V}$ in terms of Riemann theta functions, and made the remarkable observation that the fluctuations of $\Tr f(M)-N\int fd\mu_V$ can be understood as being asymptotically a sum of two independent contributions: a normally distributed random variable, and a ``discrete Gaussian" random variable whose distribution depends on $N$. Both of the works \cite{Shcherbina,BG2} as well as more recent work of Bekerman, Leblé, and Serfaty \cite{BLS} give conditions on $V$ and $f$ under which one does have convergence to a normally distributed random variable in the $k$-cut case.

We derive asymptotics    for $\E e^{\Tr f(M)}$  by giving a proof in terms of Riemann-Hilbert techniques, which are similar to the asymptotics obtained in \cite{BG2}. This is the simplest aspect of our paper, after setting up the Riemann-Hilbert problem the proof is only four pages.  We give a new interpretation of the result by identifying the ``normally distributed part" of $\Tr f(M)-N\int fd\mu_V$ as coming from a variant of the Gaussian free field on a torus with $k-1$ holes -- this is due to the covariance kernel being the bipolar Green's function on this surface.

\item[(c)] Obtaining asymptotics of Hankel determinants with Fisher-Hartwig singularities in the support of the equilibrium measure, of both root-type and jump-type.  Namely, in the most general setting of this paper, we study the asymptotics of $H_N(\nu_N)$ with $\nu_N(x)=e^{f(x)} \omega(x) e^{-NV(x)}$, where $f$ is real and analytic in a neighbourhood of $\textrm{supp}(\mu_V)$, and  $\omega(x)=\prod_{j=1}^p \omega_{\alpha_j}(x)\omega_{\beta_j}(x)$ where
\begin{equation} \nonumber \omega_{\alpha_j}(x)=|x-t_j|^{\alpha_j},\qquad \omega_{\beta_j}(x)=\begin{cases}
e^{\pi i \beta_j}, & x<t_j\\
e^{-\pi i \beta_j}, & x\geq t_j
\end{cases}
\end{equation}
with $\alpha_j>-1$ and $\Re \beta_j =0$ for all $j$, with $t_1,\dots, t_p$ in the interior of the support of $\mu_V$ (i.e. in the bulk of the spectrum of the underlying random matrix ensemble).   In the one-cut case this was solved in a series of papers by Krasovsky \cite{Krasovsky}, Garoni \cite{Garoni}, Its and Krasovsky \cite{IK}, Berestycki, Webb and Wong \cite{BWW}, and Charlier \cite{Charlier}, but no progress has been made in the multi-cut case.  In Theorem \ref{th:FHasy} below, we provide complete asymptotics up to and including the multiplicative constant in the multi-cut case. 

In Section \ref{Sec:rig} we  apply these results to obtain some rigidity estimates for the eigenvalues in the bulk of the spectrum -- estimates which provide upper bounds on how much particles in the bulk can fluctuate from their expected locations. 
\end{itemize}

For each of the three problems, our proof relies on the asymptotic analysis of orthogonal polynomials with exponential weights, through the analysis of Riemann-Hilbert problems. The asymptotics of such polynomials have been thoroughly analyzed in the multi-cut setting in \cite{DKMVZ,DKMVZfirst,BI2}. Additionally, double scaling limits moving between one-cut support and two-cut support are relevant and have been well studied, see e.g. \cite{CKV,BI3,CKuij,BertolaLee,Claeys,Mo}, although these particular transitions do not feature in the current paper.

Before stating our main results in Sections \ref{sec:Main1}-\ref{sec:Main3}, we discuss some background concerning equilibrium measures and Riemann theta-functions in Section \ref{Prelim}.

\subsection{Preliminary setup}\label{Prelim}
We describe all of the mathematical objects in terms of which our results are presented, namely equilibrium measures, $\theta$-functions, and certain Riemann-surfaces.
\subsubsection{Equilibrium measures}\label{sec:intro_em}

A key concept for us will be the equilibrium measure associated to the confining potential $V$. We will assume in what follows that $V:\R\to \R$ is a real analytic function satisfying 
\begin{equation}\label{eq:Vgrowth}
\lim_{|x|\to \infty}\frac{V(x)}{\log |x|}=\infty.
\end{equation}
The associated equilibrium measure $\mu_V$ is the unique minimizer of the energy functional
\begin{equation}\label{eq:energyf}
I_V(\mu)= \iint_{\R\times \R} \log |x-y|^{-1}d\mu(x)d\mu(y)+\int_\R V(x)d\mu(x)
\end{equation}
defined on the space of Borel probability measures on $\R$. 

The equilibrium measure plays a crucial role in our analysis for the following heuristic reason. From \eqref{eq:lgmi1} with $\nu=e^{-NV}$, we see that the main contribution to $H_{N}(e^{-NV})$ comes from the $N$-tuples which minimize
\begin{align}\label{lol1}
\frac{1}{N^2}\left(\sum_{1\leq i < j \leq N}\log|x_{j}-x_{i}|^{-2}+N\sum_{j=1}^{N}V(x_{j})\right).
\end{align}
The $N=+\infty$ version of this minimization problem is precisely the problem of finding a probability measure minimizing \eqref{eq:energyf}. Hence, for large $N$, one expects the density of points of the $N$-tuples minimizing \eqref{lol1} to be well approximated by the equilibrium measure $\mu_{V}$.  As an introduction to equilibrium measures in the setting of random matrices we suggest \cite[Chapters 6.6 and 6.7]{Deift} and \cite[Chapter 11.2]{PS}.

The equilibrium measure $\mu_V$ may alternatively be characterised through the Euler-Lagrange equations \cite[Theorem 1.3, Chapter I.1]{Safftot} -- namely $\mu_V$ is the unique Borel probability measure with compact support and finite logarithmic energy  $I_V(\mu_V)<\infty$,  such that there exists a real constant $\ell$ for which
\begin{align}
2\int\log |x-y|d\mu_V(y)&=V(x)- \ell \qquad \text{for} \qquad x\in \mathrm{supp}(\mu_V), \label{eq:EL1}\\
2\int\log |x-y|d\mu_V(y)&\leq V(x)- \ell \qquad \text{for} \qquad x\notin \mathrm{supp}(\mu_V). \label{eq:EL2}
\end{align}

\begin{figure}\begin{center}
		\begin{tikzpicture}
		\draw [dashed]  (6,0)--(7,0);
		\node [above] at (0.3,-0.1){$a_1$};
		\node [above] at (1.5,-0.1){$b_1$};
		\node [above] at (2.5,-0.1) {$a_2$};
		\node [above] at (4,-0.1) {$b_2$};
		\node [above] at (9,-0.1) {$a_{k}$};
		\node [above] at (11,-0.1) {$b_{k}$};
		
		\draw[black,fill=black]  (0.3,0) circle [radius=0.04];
		\draw[black,fill=black]  (1.5,0) circle [radius=0.04];
		\draw[black,fill=black]  (2.5,0) circle [radius=0.04];
		\draw[black,fill=black]  (4,0) circle [radius=0.04];
		\draw[black,fill=black]  (9,0) circle [radius=0.04];
		\draw[black,fill=black]  (11,0) circle [radius=0.04];
		
		\draw (0.3,0)--(1.5,0);
		\draw  (2.5,0)--(4,0);
		\draw  (9,0)--(11,0);
		
		\end{tikzpicture} 
		\caption{The support $J=\textrm{supp}( \mu_V)$.}\label{SuppmuV}\end{center}
\end{figure}
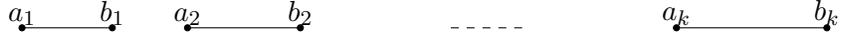

It was proven in \cite{DKM} that when $V$ is real analytic satisfying \eqref{eq:Vgrowth}, then the support of  $\mu_V$ consists of a finite union of closed intervals, which we denote by
\begin{equation}\label{eq:supp}
J=\mathrm{supp}(\mu_V)=\bigcup_{j=1}^k[a_j,b_j]
\end{equation}
with $-\infty<a_1<b_1<a_2<...<b_k<\infty$. Moreover, it was proven  that for $x\in J$, the density of $\mu_V$ has the form
\begin{equation}\label{eq:density}\begin{aligned}
d\mu_V(x)&=\psi_V(x) dx\\ \psi_V(x)&=\frac{1}{\pi i}\mathcal R^{1/2}_+(x)h_V(x),
\end{aligned}
\end{equation}
where $h_V(x)$ is a real analytic function, and where $\mathcal R_+^{1/2}(x)$ denotes the boundary values from the upper half plane of the function
\begin{equation}\label{eq:R}
\mathcal R^{1/2}(z)=\prod_{j=1}^k \left((z-a_j)(z-b_j)\right)^{1/2},
\end{equation}
with branches chosen such that $\mathcal R^{1/2}$ is analytic in $\C\setminus J$ and as $z\to \infty$, $\mathcal R^{1/2}(z)\sim z^k$.  

We will throughout the paper assume that the inequality \eqref{eq:EL2} is strict and that $h_V$ does not vanish on $J$. If $V$ satisfies these conditions, then we say that $V$ is $k$-cut regular.

Let $U_{V}$ be the domain of analyticity of $V$ -- in particular, as $V$ is real analytic, $\mathbb{R}\subset U_{V}$. Then $\psi_V$ admits an analytic continuation to $U_V\setminus J$ (see \cite{DKM}):
\begin{equation}\label{eq:psi}
\psi_V(z)= \begin{cases}
\frac{1}{\pi i}\mathcal R^{1/2}(z)h_V(z), & \mbox{if } z \in U_{V} \setminus J, \\
\frac{1}{\pi i}\mathcal R^{1/2}_{+}(z)h_V(z), & \mbox{if } z \in J.
\end{cases}
\end{equation}
With this notation $\psi_V(z)$ is well-defined for all $z \in U_{V}$, but analytic only in $z \in U_{V} \setminus J$.

By deforming the contours of integration in a suitable manner and taking the derivative of  \eqref{eq:EL1} with respect to $x$, one obtains the following useful formula 
\begin{equation} \label{ddzV}
V'(x)=-\oint_\Gamma \frac{\psi_{V}(w)}{x-w}dw,
\end{equation}
where $\Gamma$ is a closed curve oriented counter-clockwise in $U_V$ containing $x$ and $J$. 
By relying on \eqref{ddzV} it is  simple to verify that
\begin{equation} \label{eq:h}
h_V(x)=\frac{1}{4\pi i}\oint_\Gamma \frac{V'(z)}{z-x}\frac{dz}{\mathcal R^{1/2}(z)}. 
\end{equation}

Additionally, $\psi_V$ satisfies
\begin{equation} \label{zerointpsiV}
\int_{b_j}^{a_{j+1}} \psi_V(x)dx=0\end{equation}
for $j=1,2,\dots,k-1$.  Formula \eqref{zerointpsiV} is verified by representing $V(a_{j+1})-V(b_j)$ in terms of the left-hand side of \eqref{eq:EL1}, and comparing to the integral of the right-hand side of \eqref{ddzV} from $b_j$ to $a_{j+1}$.

We will also find it convenient to introduce the notation
\begin{equation}\label{eq:Omega}
\Omega_j=\int_{a_{j+1}}^{b_k}d\mu_V=\sum_{l=j+1}^k \mu_V([a_l,b_l])
\end{equation}
for $j=1,...,k-1$.

\subsubsection{Theta functions}\label{sec:intro_theta}
 Note that although \eqref{eq:energyf} is uniquely minimized, \eqref{lol1} is not necessarily so when $k\geq 2$. This, in turn, will produce some oscillations in the large $N$ asymptotics of $H_{N}(e^{-NV})$.

It turns out that these oscillations are described in terms of the Riemann $\theta$ function, which we now introduce. For a $(k-1)\times (k-1)$ symmetric complex matrix $\tau$ with positive definite imaginary part, one defines $\theta(\cdot|\tau): \C^{k-1}\to \C$, 
\begin{equation}\label{eq:thetadef}
\theta(\xi)=\theta(\xi|\tau)=\theta(\xi_1,...,\xi_{k-1}|\tau)=\sum_{n\in \Z^{k-1}} e^{2\pi i n^{T}\xi+\pi i n^{T}\tau n}.
\end{equation} 
They appear in the study of Riemann surfaces and we will return to their properties in greater detail in Section \ref{sec:thetaids}.

We  define theta functions with characteristics $\ualpha,\ubeta \in \mathbb C^{k-1}$  by 
\begin{equation}\label{eq:thetachar}\begin{aligned}
\theta\left[{\substack{\ualpha \\ \ubeta}}\right](\xi)&=e^{\pi i \ualpha^{T}\tau \ualpha+2\pi i \ualpha^{T}(\xi+\ubeta)}\theta(\xi+\tau \ualpha+\ubeta)\\ &=\sum_{n\in \Z^{k-1}}e^{2\pi i (n+\ualpha)^{T}(\xi+\ubeta)+\pi i (n+\ualpha)^{T}\tau (n+\ualpha)},
\end{aligned}
\end{equation}
for $\xi\in \C^{k-1}$.
Since $\theta=\theta\left[{\substack{0 \\ 0}}\right]$, all identities for $\theta\left[{\substack{\ualpha \\ \ubeta}}\right]$ are also valid for $\theta$ by setting $\alpha,\beta=0$. Note that $\xi \mapsto \theta\left[{\substack{\ualpha \\ \ubeta}}\right](\xi)$ is an entire function of $\xi\in \C^{k-1}$.

 Now assume that $2\ualpha,2\ubeta \in \mathbb Z^{k-1}$. If $4\ualpha^T\ubeta$ is even, then $\theta\left[{\substack{\ualpha \\ \ubeta}}\right]$ is an even function, and if $4\ualpha^T\ubeta$ is odd, then $\theta\left[{\substack{\ualpha \\ \ubeta}}\right]$ is an odd function (and in particular $\theta\left[{\substack{\ualpha \\ \ubeta}}\right](0)=0$). A fundamental property (see e.g. \cite[Chapter VI]{FK}) of the $\theta$-function is the quasi-periodicity property: if we write $e_j$ for the $j$th standard unit vector of $\C^{k-1}$ and $\tau_j$ for the $j$th column vector of $\tau$, then for all $\xi\in \C^{k-1}$
\begin{equation}\label{eq:quasi}\begin{aligned}
\theta\left[{\substack{\ualpha \\ \ubeta}}\right](\xi+e_j)&=(-1)^{2\ualpha_j}\theta\left[{\substack{\ualpha \\ \ubeta}}\right](\xi), \\  \theta\left[{\substack{\ualpha \\ \ubeta}}\right](\xi+\tau_j)&=(-1)^{2\ubeta_j}e^{-\pi i \tau_{j,j}-2\pi i \xi_j}\theta\left[{\substack{\ualpha \\ \ubeta}}\right](\xi). \end{aligned}
\end{equation}

\subsubsection{Riemann Surfaces} \label{Sec:calS}
Let $k\geq 2$, and let $\mathcal S$ be the two sheeted Riemann surface associated with $\mathcal R^{1/2}(z)$, constructed by gluing two copies of the Riemann sphere $\mathbb C \cup \{\infty\}$ along $J$, in such a manner that $\mathcal R^{1/2}(z)$ is meromorphic on $\mathcal S$, with a pole at $\infty$ on both sheets. Then $J_\pm$ on the first sheet is identified with $J_\mp$ on the second sheet, where $J_{\pm}=\left\{\lim_{\epsilon \downarrow 0}x\pm i\epsilon:\, x\in J\right\}$. As $z\to \infty$,  $z^{-k}\mathcal R^{1/2}(z)\to 1$ on the first sheet, while on the second sheet $z^{-k}\mathcal R^{1/2}(z)\to -1$. $\mathcal S$ is topologically equivalent to a torus with $k-1$ holes, and is thus of genus $k-1$.

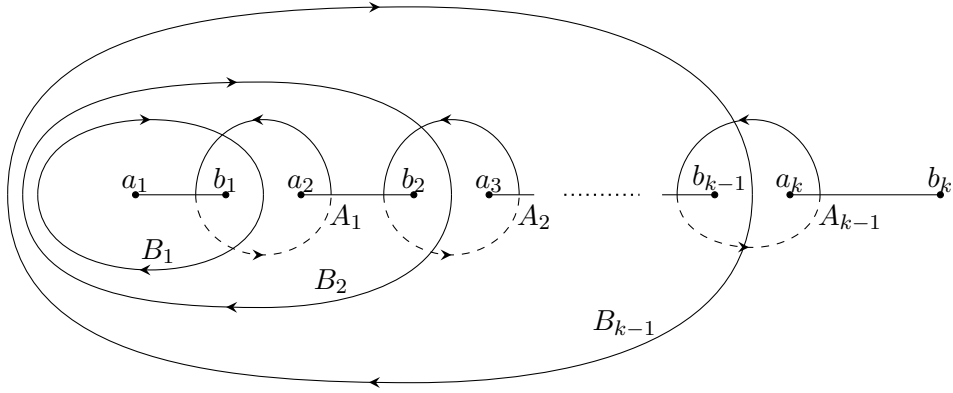
\begin{figure}[]
	\begin{center}
		\begin{tikzpicture}
		\draw [dotted,thick]  (6,0)--(7,0);
		\node [above] at (0.3,-0.1){$a_1$};
		\node [above] at (1.5,-0.1){$b_1$};
		\node [above] at (2.5,-0.1) {$a_2$};
		\node [above] at (4,-0.1) {$b_2$};
		\node [above] at (5,-0.1) {$a_3$};
		\node [above] at (8,-0.1) {$\,\, b_{k-1}$};
		\node [above] at (9,-0.1) {$a_{k}$};
		\node [above] at (11,-0.1) {$b_{k}$};
		
		\draw[black,fill=black]  (0.3,0) circle [radius=0.04];
		\draw[black,fill=black]  (1.5,0) circle [radius=0.04];
		\draw[black,fill=black]  (2.5,0) circle [radius=0.04];
		\draw[black,fill=black]  (4,0) circle [radius=0.04];
		\draw[black,fill=black]  (5,0) circle [radius=0.04];
		\draw[black,fill=black]  (8,0) circle [radius=0.04];
		\draw[black,fill=black]  (9,0) circle [radius=0.04];
		\draw[black,fill=black]  (11,0) circle [radius=0.04];
		
		\node [below] at (3.1,0) {$A_1$};
		\node [below] at (5.6,0) {$A_2$};
		\node [below] at (9.8,0) {$A_{k-1}$};
		\node [below] at  (0.6,-0.45){$B_1$};
		\node [below] at (2.9,-0.85){$B_2$};
		\node [below] at (6.8,-1.4) {$B_{k-1}$};

		\draw  (2.5,0)--(4,0);
		\draw  (0.3,0)--(1.5,0);
		\draw  (5,0)--(5.6,0);
		\draw  (7.3,0)--(8,0);
		\draw  (9,0)--(11,0);
		
		\draw[decoration={markings, mark=at position 0.5 with {\arrow[thick]{<}}},
		postaction={decorate}] (3.6,0) to [out=90,in=180] (4.5,1) to [out=0,in=90]  (5.4,0);
		\draw[dashed,decoration={markings, mark=at position 0.5 with {\arrow[thick]{>}}},
		postaction={decorate}] (7.5,0) to [out=270,in=180] (8.5,-0.7) to [out=0,in=270]  (9.4,0); 
		\draw[decoration={markings, mark=at position 0.5 with {\arrow[thick]{<}}},
		postaction={decorate}] (7.5,0) to [out=90,in=180] (8.5,1) to [out=0,in=90]  (9.4,0);
		\draw[dashed,decoration={markings, mark=at position 0.5 with {\arrow[thick]{>}}},
		postaction={decorate}] (3.6,0) to [out=270,in=180] (4.5,-0.8) to [out=0,in=270]  (5.4,0);        
		
		\draw[decoration={markings, mark=at position 0.5 with {\arrow[thick]{<}}},
		postaction={decorate}] (1.1,0) to [out=90,in=180] (2,1) to [out=0,in=90]  (2.9,0);
		\draw[dashed,decoration={markings, mark=at position 0.5 with {\arrow[thick]{>}}},
		postaction={decorate}] (1.1,0) to [out=270,in=180] (2,-0.8) to [out=0,in=270]  (2.9,0);

		\draw[decoration={markings, mark=at position 0.5 with {\arrow[thick]{>}}},
		postaction={decorate}] (-1.2,0) to [out=90,in=180] (2,1.5) to [out=0,in=90]  (4.5,0);
		\draw[decoration={markings, mark=at position 0.5 with {\arrow[thick]{<}}},
		postaction={decorate}] (-1.2,0) to [out=270,in=180] (2,-1.5) to [out=0,in=270]  (4.5,0);     
		
		\draw[decoration={markings, mark=at position 0.5 with {\arrow[thick]{>}}},
		postaction={decorate}] (-1,0) to [out=90,in=180] (0.5,1) to [out=0,in=90]  (2,0);
		\draw[decoration={markings, mark=at position 0.5 with {\arrow[thick]{<}}},
		postaction={decorate}] (-1,0) to [out=270,in=180] (0.5,-1) to [out=0,in=270]  (2,0);  
		
		\draw[decoration={markings, mark=at position 0.5 with {\arrow[thick]{>}}},
		postaction={decorate}] (-1.4,0) to [out=90,in=180] (4,2.5) to [out=0,in=90]  (8.5,0);
		\draw[decoration={markings, mark=at position 0.5 with {\arrow[thick]{<}}},
		postaction={decorate}] (-1.4,0) to [out=270,in=180] (4,-2.5) to [out=0,in=270]  (8.5,0);   
		
		\end{tikzpicture} 
		\caption{\label{ContHom}The canonical homology basis of $\mathcal S$ we consider. The solid parts are on the first sheet and the dashed parts are on the second sheet.}\end{center}
\end{figure}

Consider the cycles $A_1,\dots, A_{k-1}$ and $B_1,\dots,B_{k-1}$ in Figure \ref{ContHom}. Each cycle $B_j$ lives on the first sheet and encloses $\cup_{i=1}^j[a_i,b_i]$, while each cycle $A_j$ lives on both sheets, entering  the first sheet in the interval $(a_{j+1},b_{j+1})$ and passing to the second sheet through the interval $(a_j,b_j)$.

By standard theory, see e.g. \cite[Chapter III.2.8]{FK}, there is a unique basis of holomorphic 1-forms $\boldsymbol{\omega_1},\dots, \boldsymbol{\omega_{k-1}}$ on $\mathcal S$ satisfying
\begin{equation}\label{eq:Aint}
\oint_{A_i}\boldsymbol{\omega_j}=\delta_{i,j},
\end{equation}
for $i,j=1,2,\dots,k-1$, and  each $\boldsymbol{\omega_j}$ is of the form 
\begin{equation} \label{formomega} \boldsymbol{\omega_j}(z)=\frac{\mathsf Q_j(z)}{\mathcal R^{1/2}(z)}dz,
\end{equation}
where $\mathsf Q_j(z) $ is a polynomial of degree at most $k-2$ (see e.g. \cite[Chapter III.7]{FK}).\footnote{By $\mathsf Q_j(z)$, we mean $\mathsf Q_j({\bf P}(z))$, where ${\bf P}$ is the projection from $\mathcal S$ onto the complex plane.}

Note that by our definition of $A_i$ and $\boldsymbol{\omega_j}$, we have by contour deformation that
\begin{equation} \label{eq:Aint2}
\oint_{A_i}\boldsymbol{\omega_j}=-2\int_{b_i}^{a_{i+1}}\frac{\mathsf Q_j(x)}{\mathcal R^{1/2}(x)}dx,
\end{equation}
where the integral on the right-hand side is taken on the first sheet,
from which one can deduce (since $\mathcal R^{1/2}$ is real on $(b_i,a_{i+1})$) that the coefficients of $\mathsf Q_j$ are real valued. 

\medskip

An important quantity constructed from $\boldsymbol{\omega_j}$ is the period matrix $\tau=(\tau_{i,j})_{i,j=1}^{k-1}$ defined by 
\begin{equation}\label{eq:period}
\tau_{i,j}=\oint_{B_i}\boldsymbol{\omega_j}.
\end{equation}
We recall  from  \cite[Chapter III.3]{FK} (see also \cite{Schrader}) that $\tau$ is symmetric, that the entries of $\tau$ are purely imaginary, and that $-i\tau$ is positive definite.

\subsubsection*{The Abel map}
Define the Abel map $u(z)=(u_1(z),\dots, u_{k-1}(z))^T$ by
\begin{equation}\label{eq:Abel}
u_j(z)=\int_{b_k}^z\frac{\mathsf Q_j(s)}{\mathcal R^{1/2}(s)} ds, \qquad z\in (\C \cup \{\infty\})\setminus (-\infty,b_k],
\end{equation}
where the integration contour does not cross $(-\infty,b_k]$. 
For $k\geq 2$, we let $\ualpha=\frac{1}{2}e_1$ and let $\ubeta=\frac{1}{2}\sum_{j=1}^{k-1}e_j$ throughout the paper. Then $4\ualpha^T\ubeta$ is odd. Furthermore  $\left[{\substack{{ \ualpha} \\ { \ubeta}}}\right]$ is a non-singular characteristic, meaning that $\theta\left[{\substack{{ \ualpha} \\ { \ubeta}}}\right](u(z))$ is not identically zero, see Lemma \ref{Prime} below. It is a well known fact, for which  we also provide a proof in Lemma \ref{Prime}, that for fixed $\lambda \in \mathbb C\setminus J$,
\begin{equation}\label{Theta}
\Theta(z,\lambda)=
\frac{\theta\left[{\substack{{ \ualpha} \\ { \ubeta}}}\right](u(z)-u(\lambda))}{\theta\left[{\substack{\ualpha \\ \ubeta}}\right](u(z)+u(\lambda))} \end{equation}
is  analytic on $\mathbb C\setminus J$ as a function of $z$, with a single zero at $z=\lambda$. Here and below, unless otherwise stated, $\theta\left[{\substack{{ \ualpha} \\ { \ubeta}}}\right]=\theta\left[{\substack{{ \ualpha} \\ { \ubeta}}}\right](\cdot|\tau)$ with $\tau$ as in \eqref{eq:period}. When $k=1$, \eqref{Theta} is not well defined. In this case, we define
\begin{equation}\label{Thetaonecut} \begin{aligned}\Theta(z,\lambda)&=\frac{b_1-a_1}{2}\frac{z-\lambda}{\mathcal R^{1/2}(z)\mathcal R^{1/2}(\lambda)+z\lambda+a_1b_1-\frac{a_1+b_1}{2}(z+\lambda)}\\
&= \frac{2}{b_1-a_1}\frac{-\mathcal R^{1/2}(z)\mathcal R^{1/2}(\lambda)+z\lambda+a_1b_1-\frac{a_1+b_1}{2}(z+\lambda)}{z-\lambda}. \end{aligned}\end{equation}

Furthermore, provided $\lambda \in \mathbb C\setminus J$ is fixed, $\Theta(z,\lambda)$ remains bounded in a neighbourhood of $J$. Since $\theta\left[{\substack{{ \ualpha} \\ { \ubeta}}}\right]$ is odd, $\Theta(z,\lambda)=-\Theta(\lambda,z)$. Define
\begin{equation}\label{def:W} W(z,\lambda)=W(\lambda,z)=\frac{\partial }{\partial z} \frac{\partial}{\partial \lambda} \log \Theta(z,\lambda). \end{equation}
Then $W$ is meromorphic on $\mathbb C\setminus J$, and has a double pole at $z=\lambda$. In Proposition \ref{PropForM} below, we prove that

\begin{equation}\label{id:W}
W(z,\lambda)=\frac{1}{2(z-\lambda)^2}\left[\left(\frac{\gamma(z)}{\gamma(\lambda)}\right)^2+\left(\frac{\gamma(\lambda)}{\gamma(z)}\right)^2\right]-
2\sum_{i,j=1}^{k-1}(\partial_i \partial_j \log \theta)(0)u_i'(z)u_j'(\lambda),
\end{equation}
where $\gamma$ is defined by
\begin{equation}\label{eq:gamma}
\gamma(z)=\prod_{j=1}^k\Bigg(\frac{z-b_j}{z-a_j}\Bigg)^{1/4}, \qquad z\in \C\setminus J,
\end{equation}
and we choose the branch of the root such that $\gamma$ is analytic in $\C\setminus J$ and $\gamma(z)\to 1$ as $z\to \infty$.

Finally, define 
\begin{equation}\label{defwlambda}
w_z(\lambda)=\frac{\partial}{\partial\lambda}\log \Theta(z,\lambda)=\frac{\partial}{\partial \lambda}\log \Theta(\lambda,z).
\end{equation}
In \eqref{2ndrepwlambda} we provide the alternative representation
 \begin{equation}\label{repwlambda}
w_z(\lambda)=\frac{\mathcal R^{1/2}(z)}{\mathcal R^{1/2}(\lambda)(\lambda-z)}+2\mathcal R^{1/2}(z) \sum_{j=1}^{k-1} u_j'(\lambda)\int_{b_j}^{a_{j+1}}\frac{dx}{\mathcal R^{1/2}(x)(x-z)}.
\end{equation}

We are now in a position to state our main results.

\subsection{Main results: asymptotics of $H_N(e^{-NV})$}\label{sec:Main1} The simplest example is the setting of the Gaussian Unitary Ensemble. It is a classical result (see e.g. \cite[equation (3.3.10)]{Mehta}) that the asymptotics of the Hankel determinant with quadratic potential are given by
\begin{multline}\label{asymGUE}
\log H_N\left(e^{-2\sigma Nx^2}\right) \\
=-N^2\left(\frac{3}{4}-\frac{1}{2}\log \frac{1}{4\sigma}\right)+N\log (2\pi)-\frac{1}{12} \log N+\zeta'(-1)+\mathcal O(N^{-1}),
\end{multline}
uniformly for $\sigma$ in compact subsets of $(0,+\infty)$ as $N\to \infty$, where $\zeta'(-1)$ is the derivative of the Riemann-zeta function at $-1$.

More generally, for one-cut regular potentials $V$, the problem has been studied in \cite{EML,BI,BG1,BWW,Charlier}, and by \cite[Proposition 5.5]{BWW} or  \cite[Theorem 1.1]{Charlier} we have
\begin{multline} \log H_N(e^{-NV})=-N^2 I_V(\mu_V)+N\log 2\pi-\frac{1}{12} \log N\\ +\zeta'(-1)-\frac{1}{24}\log \left(\frac{\tilde \psi(a_1)\tilde \psi(b_1)|b_1-a_1|^3}{2^6}\right)+\mathcal O(N^{-1}),
\end{multline}
as $N\to \infty$, where $I_V$ was defined in \eqref{eq:energyf} and
\begin{equation} \label{tildepsi}
\tilde  \psi(q):=\lim_{\lambda\to q} \pi \left| \frac{\psi_{V}(\lambda)}{\left(\lambda-q\right)^{1/2}}\right|,
\end{equation} 
for $q\in \{a_1,b_1\}$. 

When $V$ is two-cut regular, Claeys, Grava and McLaughlin \cite[formula (1.9)]{CGML} obtained the following asymptotics
\begin{multline}\label{formCGML}
\log H_N(e^{-NV})=-N^2 I_V(\mu_V)+N\log (2\pi)-\frac{1}{6}\log N+ \log \theta(N\Omega)\\+2 \zeta'(-1) -\frac{1}{2}\log \frac{K(\mathrm k)}{\pi} -\frac{1}{24}\sum_{q\in\{a_j,b_j\}_{j=1}^2} \log \tilde \psi(q)\\+\frac{1}{8}\log (b_2-b_1)(a_2-a_1)
 -\frac{1}{8}\sum_{l,j=1}^2\log|b_j-a_l|+\mathcal O(N^{-1}),
\end{multline}
as $N\to \infty$, where $\mathrm k =\sqrt{\frac{(a_2-b_1)(b_2-a_1)}{(b_2-b_1)(a_2-a_1)}}$ and $K(\mathrm k)=\int_0^1 \frac{dx}{\sqrt{(1-x^2)(1-\mathrm k^2x^2)}}$ is the complete elliptic integral of the first kind, and where $\tilde \psi$ is given by \eqref{tildepsi} for $q\in \{a_j,b_j\}_{j=1}^2$. In \cite{CGML} the constant $\tau$ in $\theta(\cdot |\tau)$ had the alternative representation $iK(\mathrm k')/K(\mathrm k)$ where $\mathrm k'=\sqrt{1-\mathrm k^2}$; one may verify that this indeed matches our $\tau$ by standard formulas in \cite[3.147]{GR}.  

In the $k$-cut regular case an asymptotic formula was obtained by Borot and Guionnet \cite{BG2}, who found that as $N\to \infty$,
\begin{equation}\label{formBG} \log H_N(e^{-NV})=-N^2I_V(\mu_V) +N\log (2\pi) -\frac{k}{12}\log N + \log \theta(N\Omega)+\widehat C+\mathcal O(N^{-1}).
\end{equation} 
The formula for $\widehat C$  is rather involved (see Theorem 1.5 and Proposition 7.5 in version 5 of \cite{BG2}), and we have not found any direct comparisons between \eqref{formBG} and \eqref{formCGML} in the literature.
 This left the open problem of whether in fact, for $k=3,4,\dots$, there is a simple expression for $\widehat C$ such as the one appearing in \eqref{formCGML} in the two-cut case, which we address in the following theorem.

\begin{theorem}\label{th:pfasy}
	Let $V:\mathbb R \to \mathbb R$ be real-analytic, satisfying \eqref{eq:Vgrowth}, and assume that $V$ is $k$-cut regular for some $k\geq 2$. Then 
	\begin{multline} \label{MainFormula}\log H_N(e^{-NV})=-N^2I_V(\mu_V) +N\log (2\pi) -\frac{k}{12}\log N + \log \frac{\theta(N\Omega)}{\theta(0)}\\ +\frac{k}{4} \log 2+k\zeta'(-1)
-\frac{1}{24}\sum_{q\in\{a_j,b_j\}_{j=1}^k} \log \tilde \psi(q)\\ +\frac{1}{8}\Bigg(\sum_{1\leq l<j \leq k} [\log(a_j-a_l)+\log(b_j-b_l)]-\sum_{l,j=1}^k\log|b_j-a_l|\Bigg)
+o(1), 
\end{multline}
as $N\to \infty$, where $I_V$ was defined in \eqref{eq:energyf} and where $\tilde \psi$ is given by \eqref{tildepsi} for $q\in \{a_j,b_j\}_{j=1}^k$.
\end{theorem}

We provide an outline of the proof of Theorem \ref{th:pfasy} in Section \ref{sec:pfasy}. The main technical ingredient is the Deift-Zhou steepest descent analysis of a Riemann-Hilbert problem associated to a system of orthogonal polynomials, and in essence the proof is divided into two steps:
\begin{itemize}
\item[Step 1.] Find an example of a $k$-cut regular potential $V_0$ for which we can obtain asymptotics of $\log H_N(e^{-NV_0})$.
\begin{itemize}
\item When $k=1$ the obvious choice is the Gaussian Unitary Ensemble, because the corresponding determinant is closely related to a Selberg integral for which asymptotics are well-known. This is the approach of e.g. \cite{BWW, Charlier}.
\item When $k=2$, the authors of \cite{CGML} viewed (roughly speaking) a certain symmetric two-cut potential as a one-cut potential reflected through the origin to obtain a reduction to the known one-cut asymptotics.
\item In the case of general $k$, we also make a reduction to one-cut asymptotics, but this is not possible through a reflection through the origin. We opt instead to build a potential $V_0$ in terms of the Chebyshev polynomials, which we call the Chebyshev potential, and rely on the properties of Chebyshev polynomials  to make a comparison with a one-cut potential. The inspiration for such a potential comes from the study of Toeplitz determinants - it was observed in \cite{Baik,Bthesis} that certain determinants with rotationally invariant symbols could be simplified, and the Chebyshev potential is our attempt at creating an analogue of such a rotational symmetry on the real line. This does in fact not work quite as smoothly as in the Toeplitz case, we need to take a limit where $V_0$ remains a $k$-cut potential but comes closer and closer to being a one-cut potential to complete Step 1. However, the ordering of limits is important, and in Section \ref{sec:special} we take care to avoid complicated double scaling limits (we always remain in the $k$-cut situation).
\end{itemize}
\item[Step 2.] Take a continuous deformation of potentials $V_s$ for a parameter $s\in[0,1]$ such that $V=V_1$ and such that we can compute the asymptotics of $\frac{\partial}{\partial s}\log H_N(e^{-NV_s})$. When $k\geq 2$ these deformations involve combinations of $\theta$-functions associated to a Riemann surface of genus $k-1$. These combinations  need to be integrated in terms of the parameter of deformation $s$, and the challenge is to  discern which combinations may be discarded as error terms, and to extract and simplify the remaining terms. \newline
For the reader unfamiliar with $\theta$-functions it might be a good idea to think of $k=2$ where many serious simplifications occur, for example: (i) $\theta$ is a complex function of a single variable, (ii) if $f$ is a non-trivial function on $\mathcal S$ then $\theta\circ  f$ is non-trivial (this is not necessarily true if $k\geq 3$), and (iii) $\theta(x\Omega)$ is periodic in $x\in \mathbb R$ (again this is not necessarily true if $k\geq 3$).
\end{itemize}

To make a  comparison between Theorem \ref{th:pfasy} and the asymptotics of Claeys, Grava and McLaughlin presented in \eqref{formCGML}, we rely on the identity $K(\mathrm k)=\frac{\pi}{2}\theta(0)^2$, and we find that the formulas match. We have not been able to make a comparison with the results of Borot and Guionnet \cite{BG2}.

The function $N \mapsto \theta(N\Omega)$ is in general quasi-periodic, however for certain potentials $V$ it is periodic, e.g. if $\mu_V ([a_j,b_j])=1/k$ for $j=1,\dots,k$. We now provide details for a class of polynomial potentials where this holds, and furthermore the asymptotics of Theorem \ref{th:pfasy} simplify. Let $\Pi_k$ be a monic polynomial of degree $k$ with $k$ distinct real roots, and let $1/\nu^*$ be the smallest local maximum of $\Pi_k(x)^2$. Suppose that $\nu>\nu^*$, so that $\Pi_k(x)^2-1/\nu$ has $2k$ zeros, which we denote by $a_1<b_1<\dots<a_k<b_k$. Then it is well-known that (see e.g. \cite[Theorem 11.2.7]{PS}) the potential $V(x)=\frac{2\nu}{k}\Pi_k(x)^2$ is $k$-cut regular, $J:= \mbox{supp} ( \mu_{V}) = \cup_{j=1}^{k}[a_{j},b_{j}]$, and
\begin{equation} \label{eqmeasPoly} d\mu_V(x)=\frac{2\nu}{\pi k} |\Pi_k'(x)|\sqrt{1/\nu-\Pi_k^2(x)}dx, \qquad \mbox{for } x\in J.
\end{equation} 
We provide a proof of \eqref{eqmeasPoly} in Section \ref{ProofPik} for the reader's convenience, where we also prove the following corollary.
\begin{corollary}\label{RemarkPoly1}
Let $V(x)=\frac{2\nu}{k}\Pi_k(x)^2$ and let $\nu>\nu^*$. Then, as $N\to \infty$, 
\begin{multline*} 
\log H_N(e^{-NV}) \\
=-\frac{N^2}{2k}\left(3/2+\log \nu+2\log 2\right) +N\log (2\pi)-\frac{k}{12}\log N + \log \frac{\theta(r\Omega)}{\theta(0)}
\\+\frac{k}{8} \log 2 -\frac{k}{16}\log \nu+\frac{k}{12}\log k+k\zeta'(-1)-\frac{1}{16}\sum_{q\in\{a_j,b_j\}_{j=1}^k} \log \left|\Pi_k'(q)\right|\\+\frac{1}{8}\Bigg(\sum_{1\leq l<j\leq k} [\log(a_j-a_l)+\log(b_j-b_l)]-\sum_{l,j=1}^k\log|b_j-a_l|\Bigg)
+o(1),
\end{multline*}
where $r\equiv N\, \textrm{mod}\,\,  k$.
\end{corollary}

\subsection{Main results: ratio asymptotics $H_N(Fe^{-NV})/H_N(e^{-NV})$} 
We consider $k$-cut regular potentials $V$, and consider the ratio asymptotics $\frac{H_N(Fe^{-NV})}{H_N(e^{-NV})}$ under the following assumptions on $F$:
\subsubsection*{Assumptions on $F$}
\begin{itemize}
\item[(a)] $F$ is  non-negative on $\mathbb R$. On any compact subinterval of $\mathbb R$, $F$ is bounded and integrable.
\item[(b)] $f(x)=\log F(x)$ is real analytic in a neighbourhood of $J$. 
\item[(c)] There exists $c>0$ such that $\frac{F(x)}{e^{cV(x)}}\to 0$ as $x\to \pm \infty$.
\end{itemize}

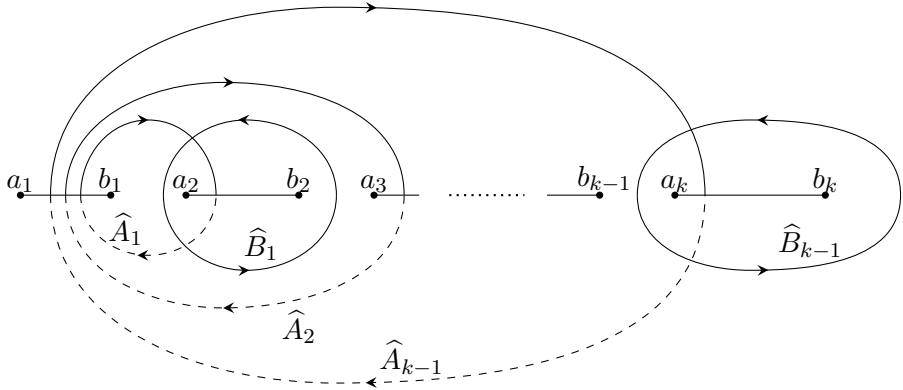
\begin{figure}[]
	\begin{center}
		\begin{tikzpicture}
		\draw [dotted,thick]  (6,0)--(7,0);
		\node [above] at (0.3,-0.1){$a_1$};
		\node [above] at (1.5,-0.1){$b_1$};
		\node [above] at (2.5,-0.1) {$a_2$};
		\node [above] at (4,-0.1) {$b_2$};
		\node [above] at (5,-0.1) {$a_3$};
		\node [above] at (8,-0.1) {$\,\, b_{k-1}$};
		\node [above] at (9,-0.1) {$a_{k}$};
		\node [above] at (11,-0.1) {$b_{k}$};
		
		\draw[black,fill=black]  (0.3,0) circle [radius=0.04];
		\draw[black,fill=black]  (1.5,0) circle [radius=0.04];
		\draw[black,fill=black]  (2.5,0) circle [radius=0.04];
		\draw[black,fill=black]  (4,0) circle [radius=0.04];
		\draw[black,fill=black]  (5,0) circle [radius=0.04];
		\draw[black,fill=black]  (8,0) circle [radius=0.04];
		\draw[black,fill=black]  (9,0) circle [radius=0.04];
		\draw[black,fill=black]  (11,0) circle [radius=0.04];
		
		\node [below] at (1.7,-0.1) {$\widehat A_1$};
		\node [below] at (4,-1.4) {$\widehat A_2$};
		\node [below] at (5.5,-1.8) {$\widehat A_{k-1}$};
		\node [below] at (3.5,-0.3){$\widehat B_1$};
		\node [below] at (10.8,-0.25) {$\widehat B_{k-1}$};

		\draw  (2.5,0)--(4,0);
		\draw  (0.3,0)--(1.5,0);
		\draw  (5,0)--(5.6,0);
		\draw  (7.3,0)--(8,0);
		\draw  (9,0)--(11,0);
		
		\draw[decoration={markings, mark=at position 0.5 with {\arrow[thick]{>}}},
		postaction={decorate}] (0.9,0) to [out=90,in=180] (3,1.5) to [out=0,in=90]   (5.4,0);
		\draw[dashed,decoration={markings, mark=at position 0.5 with {\arrow[thick]{<}}},
		postaction={decorate}] (0.7,0) to [out=270,in=180] (4.5,-2.5) to [out=0,in=270]  (9.4,0); 
		
		\draw[decoration={markings, mark=at position 0.5 with {\arrow[thick]{>}}},
		postaction={decorate}] (0.7,0) to [out=90,in=180] (4.5,2.5) to [out=0,in=90]  (9.4,0);
		\draw[dashed,decoration={markings, mark=at position 0.5 with {\arrow[thick]{<}}},
		postaction={decorate}] (0.9,0) to [out=270,in=180] (3,-1.5) to [out=0,in=270]  (5.4,0);        
		
		\draw[decoration={markings, mark=at position 0.5 with {\arrow[thick]{>}}},
		postaction={decorate}] (1.1,0) to [out=90,in=180] (2,1) to [out=0,in=90]  (2.9,0);
		\draw[dashed,decoration={markings, mark=at position 0.5 with {\arrow[thick]{<}}},
		postaction={decorate}] (1.1,0) to [out=270,in=180] (2,-0.8) to [out=0,in=270]  (2.9,0);

		\draw[decoration={markings, mark=at position 0.5 with {\arrow[thick]{<}}},
		postaction={decorate}] (2.2,0) to [out=90,in=180] (3.3,1) to [out=0,in=90]  (4.5,0);
		\draw[decoration={markings, mark=at position 0.5 with {\arrow[thick]{>}}},
		postaction={decorate}] (2.2,0) to [out=270,in=180] (3.3,-1) to [out=0,in=270]  (4.5,0);

		\draw[decoration={markings, mark=at position 0.5 with {\arrow[thick]{<}}},
		postaction={decorate}] (8.5,0) to [out=90,in=180] (10,1) to [out=0,in=90]  (12,0);
		\draw[decoration={markings, mark=at position 0.5 with {\arrow[thick]{>}}},
		postaction={decorate}] (8.5,0) to [out=270,in=180] (10,-1) to [out=0,in=270]  (12,0);   
		
		\end{tikzpicture} 
		\caption{\label{ContHom2}The second canonical homology basis of $\mathcal S$ we consider. The solid parts are on the first sheet and the dashed parts are on the second sheet.}\end{center}
\end{figure}

  When $V$ is one-cut regular, the problem has been well-studied (see e.g. \cite{Johansson, BG1,BWW}), and the following asymptotics hold as $N\to \infty$:
\begin{multline} \label{AsymJohan}\frac{H_N(Fe^{-NV})}{H_N(e^{-NV})}=e^{N\int_Jf(x)d\mu_V(x)}\\ \times \exp\left(\frac{1}{2}\oint_{\Gamma}\frac{\mathcal R^{1/2}(z)}{2\pi i }\int_J\frac{f(x)}{\mathcal R^{1/2}_+(x)(x-z)}dx f'(z)\frac{dz}{2\pi i}\right)(1+\mathcal O(N^{-1})),
\end{multline}
where $\Gamma$ is a closed loop oriented clockwise containing $[a_1,b_1]$.

When $V$ is $k$-cut regular, the problem has been studied in \cite{Shcherbina, BG2, BLS}, who all derive different representations for the asymptotics. The results in \cite{BG2} are similar to ours, but are given in the homology basis in Figure \ref{ContHom2}, see the discussion following Theorem \ref{th:smoothasy}. 

We start by discussing a condition under which $\Tr f(M)-N\int fd\mu_V$ converges to a normally distributed random variable. It was observed by Borot and Guionnet in \cite{BG2} that this occurs when $V$ is $k$-cut regular and
\begin{equation} \int_J \frac{f(x)x^jdx}{\mathcal R_+^{1/2}(x)}=0 \label{condition0}\end{equation}
for $j=0,1,\dots,k-2$. We verify this, and additionally we find that assuming \eqref{condition0}, the ratio asymptotics of $H_N(Fe^{-NV})/H_N(e^{-NV})$ are given precisely by \eqref{AsymJohan} also in the $k$-cut case,  now with $\Gamma=\cup_{i=1}^k \Gamma_i$ and where $\Gamma_i$ is a closed loop oriented counterclockwise containing $[a_i,b_i]$, giving a nice parallel to the one-cut case.

More generally, when \eqref{condition0} does not hold, the asymptotics of $\frac{H_N(Fe^{-NV})}{H_N(e^{-NV})}$ are governed by the following theorem.

\begin{theorem}\label{th:smoothasy} Let $V:\mathbb R \to \mathbb R$ be real-analytic, satisfying \eqref{eq:Vgrowth}, and assume that $V$ is $k$-cut regular for some $k\geq 2$.  Let $F$ satisfy the assumptions on $F$ above. Then as $N\to\infty$
	\begin{multline}\label{th:smoothasy1}
	\frac{H_N(Fe^{-NV})}{H_N(e^{-NV})}=e^{N\int_Jf(x)d\mu_V(x)}\frac{\theta(N\Omega+\Upsilon(f)|\tau)}{\theta(N\Omega|\tau)}\\ \times \exp \left[ \frac{1}{4}\oint_{\Gamma}\oint_{\widetilde \Gamma}W(z,\lambda)f(z)f(\lambda) \frac{dz}{2\pi i}\frac{d\lambda}{2\pi i}\right](1+\mathcal O(N^{-1}))
	\end{multline}
	where $\mu_V$, $\Omega$, $\mathcal R^{1/2}$, and $\{a_j,b_j\}_{j=1}^k$ are as in Section \ref{sec:intro_em}. We denote $\Gamma=\cup_{j=1}^k\Gamma_j$ and $\widetilde \Gamma=\cup_{j=1}^k\widetilde \Gamma_j$, where for each $j$, the smooth contours  $\Gamma_j$ and $\widetilde \Gamma_j$  surround $[a_j,b_j]$,  are oriented in a counter-clockwise manner, are contained in the domain of analyticity of $f$, and are defined such that $\Gamma$ and $\widetilde \Gamma$ don't intersect. The constant $\Upsilon=\Upsilon(f)\in \R^{k-1}$ is given by 
\begin{equation}\label{def:Omegahat}
	\Upsilon_m(f)=-\int_{J}\frac{\mathsf Q_m(z)}{\mathcal R_+^{1/2}(z)}f(z)\frac{dz}{\pi i},
\end{equation}
for $m=1,\dots,k-1$, where $\mathsf Q_m$ is defined by \eqref{eq:Aint}-\eqref{formomega}.

\end{theorem}

For definiteness, one may take $\Gamma$ to surround $\widetilde \Gamma$ in Theorem \ref{th:smoothasy} above, however the result remains unchanged if $\widetilde \Gamma$ instead surrounds $\Gamma$.

\begin{remark}
With the random matrix interpretation, there is a fluctuation in terms of how many eigenvalues tend to fall in each interval of $\textrm{supp}\left(\mu_V\right)$. These fluctuations are described in terms of $\theta$-functions. The parameter $\tau$ determines how freely the eigenvalues can jump between intervals. For example, if $k=2$ and $-i\tau$ is small then the number of eigenvalues in each interval is nearly deterministic: indeed, if $V$ is an even function and $N$ is even, then most likely there are $N/2$ eigenvalues in each interval. (There are some exceptions to this statement, for example if $k=2$, if $V$ is an even function, if $-i\tau$ is small and if $N$ is odd, then naturally there is a symmetry and the eigenvalues go left or right with equal probability, and so with high probability one of the intervals will have $N/2+1/2$ eigenvalues and the other will have $N/2-1/2$.) On the other hand, when $-i\tau$ is large then there is a greater variance and the number of eigenvalues in each interval vary to a larger extent. When $k\geq 3$ the situation can be more nuanced and one could for example have a deterministic number of eigenvalues in on interval and fluctuations between the remaining intervals.  We provide more details in Corollary \ref{Corrflucint} below.
\end{remark}

We will prove Theorem \ref{th:smoothasy} for functions $F$ which satisfy conditions (a)-(c), and which in addition are H\"older continuous. Then it follows that the theorem also holds for functions $F$ satisfying conditions (a)-(c) but which are not H\"older continuous. To verify this one merely takes $F_1$ and $F_2$ to be H\"older continuous functions satisfying the conditions for $F$,  with $F_1=F_2=F$ on a neighbourhood of $J$, and otherwise $F_1\leq F\leq F_2$. Then it follows that 
\begin{equation} H_N(F_1e^{-NV}) \leq H_N(Fe^{-NV})\leq  H_N(F_2e^{-NV}), \end{equation}
and since \eqref{th:smoothasy1} holds for both $F_1$ and $F_2$ the theorem follows in full generality.

The proof of Theorem \ref{th:smoothasy} for H\"older continuous $F$ is based on the steepest descent analysis of Riemann-Hilbert problems, we give an overview in Section \ref{sec:smasy} below. Here the overarching method consists of (i) obtaining asymptotics for $\frac{\partial}{\partial t} \log H_N\left(e^{-NV}F_t\right)$  where $F_{t=0}=1$ and $F_{t=1}=F$, and (ii) integrating the asymptotics. This integration in $t$ is not completely straightforward since $\theta$-functions enter the picture, and we mention that the main parametrix has a different form to the known form developed in \cite{KuijVanlessen}.

 It may be extended to complex valued $f$ in a straightforward manner under the assumption that $\theta(N\Omega+t\Upsilon(f)|\tau)\neq 0$ for all $t\in (0,1]$ and $N$ sufficiently large.

Integrating by parts in the second term on the right-hand side of \eqref{th:smoothasy1} (see Section \ref{SecProofSmasy2}), we obtain  a second representation for the double integral on the right-hand side of \eqref{th:smoothasy1}:
\begin{equation} \label{smoothasy2}
\begin{aligned}
& \frac{1}{4} \oint_{\Gamma}\oint_{\widetilde \Gamma} W(z,\lambda)f(z)f(\lambda) \frac{dz}{2\pi i}\frac{d\lambda}{2\pi i} \\
& =\frac{1}{2}\oint_{\Gamma}\frac{\mathcal R^{1/2}(z)}{2\pi i }\int_J\frac{f(x)}{\mathcal R^{1/2}_+(x)(x-z)}dx f'(z)\frac{dz}{2\pi i}\\&
 +\int_J \mathcal R_+^{1/2}(z)\sum_{j=1}^{k-1}  \Upsilon_j(f)\int_{b_j}^{a_{j+1}}\frac{dx}{\mathcal R^{1/2}(x)(x-z)}f'(z)\frac{dz}{2\pi i}\\
 & +\frac{1}{2} \sum_{j=1}^{k-1}(f(a_{j+1})-f(b_j))\Upsilon_j(f) .
\end{aligned} 
\end{equation}

If \eqref{condition0} holds, then $\Upsilon(f)=0$, and so it follows by \eqref{smoothasy2} that \eqref{AsymJohan} holds.

To make the probabilistic interpretation of the non-Gaussian part which was observed in \cite{BG2}, we would like to represent the second term $\frac{\theta(N\Omega+\Upsilon(f)|\tau)}{\theta(N\Omega|\tau)}$ appearing in \eqref{th:smoothasy1} as $\mathbb E[e^{Y_N}]$ for some real random variable $Y_N$. However it is not clear from \eqref{eq:thetadef} that this holds. To do this it is convenient to consider a different cycle structure on the Riemann surface $\mathcal S$, which is the basis in Figure \ref{ContHom2}.  Let $\boldsymbol{\widehat \omega_j}$ be the basis of holomorphic one-forms satisfying 
\begin{equation} \int_{\widehat B_j}\boldsymbol{\widehat \omega_i}=\delta_{i,j}, \end{equation}
where $\widehat B_j$ is as in Figure \ref{ContHom2}, and denote
\begin{equation} \widehat u_j(z)=\int_{b_k}^z \boldsymbol{\widehat \omega_j}(\xi) . \end{equation}
Let $\widehat \tau$ be the period matrix
\begin{equation} \widehat \tau_{i,j}=\int_{\widehat A_j}\boldsymbol{\widehat \omega_i}, \end{equation}
and denote by $\widehat \theta(x)=\theta(x|\widehat \tau)$.
 It is easy to move between bases: if  $C$ is the $(k-1)\times (k-1)$ lower-triangular matrix with entries $C_{i,j}= 1$ for $j\leq i$ and $C_{i,j}=0$ otherwise, then 
\begin{equation} \boldsymbol{\widehat \omega }=C\tau^{-1}\boldsymbol \omega,  \qquad \widehat \tau=-C\tau^{-1}C^T. \end{equation}
Let
\begin{equation}
\widehat \Theta(z,\lambda)=
\frac{\widehat \theta\left[{\substack{{ \widehat{ \ubeta}} \\ { \widehat {\ualpha}}}}\right]\left(\widehat u(z)-\widehat u(\lambda)\right)}{\widehat \theta\left[{\substack{{ \widehat{ \ubeta}} \\ { \widehat {\ualpha}}}}\right]\left(\widehat u(z)+\widehat u(\lambda)\right)} , \qquad \widehat W(z,\lambda)= \frac{\partial}{\partial z} \frac{\partial}{\partial \lambda}\log \widehat \Theta(z,\lambda),\end{equation}
with $\widehat{ \ubeta}=-C^{-T}\ubeta$ and $\widehat {\ualpha}=C \ualpha$ (where $C^{-T}$ is the inverse transpose of $C$).
By Lemma \ref{le:jacobi} below, we have the identities
\begin{align}
\label{idchange1}\Theta(z,\lambda)&=\widehat \Theta(z,\lambda)\exp \left(4\pi i u(z)^T\tau^{-1}u(\lambda)\right),\\
\label{idchange2}\frac{\theta(N\Omega+\Upsilon(f))}{\theta(N\Omega)}&=\exp\left(-\pi i \Upsilon(f)^T\tau^{-1}\Upsilon(f) \right)\frac{\widehat \theta\left[\substack{-N\widehat \Omega\\0}\right]\left(\widehat \Upsilon(f)\right)}{\widehat \theta\left[\substack{-N\widehat \Omega\\0}\right](0)},
\end{align}
where $\widehat \Omega_j=\mu_V([a_{j+1},b_{j+1}])$, and  $\widehat \Upsilon(f)=\left(\widehat \Upsilon_j(f)\right)_{j=1}^{k-1}$ with
\begin{equation} \label{defhatUps}\widehat \Upsilon_j(f)=-\int_J \widehat u_{j,+}'(z)f(z)\frac{dz}{\pi i} .\end{equation}
Observe that $\widehat \Upsilon$ are purely imaginary.
By Theorem \ref{th:smoothasy} and relying on \eqref{idchange1}-\eqref{idchange2}, we find that  as $N\to \infty$,
\begin{multline}\label{smasybasis3}
	\frac{H_N(Fe^{-NV})}{H_N(e^{-NV})}=e^{N\int_Jf(x)d\mu_V(x)}\frac{\widehat \theta\left[\substack{-N\widehat \Omega\\0}\right]\left(\widehat \Upsilon(f)\right)}{\widehat \theta\left[\substack{-N\widehat \Omega\\0}\right](0)}\\ \times \exp \left[ \frac{1}{4}\oint_{\Gamma}\oint_{\widetilde \Gamma}\widehat W(z,\lambda)f(z)f(\lambda) \frac{dz}{2\pi i}\frac{d\lambda}{2\pi i}\right](1+\mathcal O(N^{-1})).
	\end{multline}
Observe that
\begin{align}\label{lol8}
\frac{\widehat \theta\left[\substack{-N\widehat \Omega\\0}\right]\left(\widehat \Upsilon(f)\right)}{\widehat \theta\left[\substack{-N\widehat \Omega\\0}\right](0)},
\end{align}
which appears on the right-hand side of \eqref{smasybasis3}, can be written as $\E e^{v_N(f)}$ (in other words, \eqref{lol8} is the Laplace transform of $v_N(f)$) with 
\[
v_N(f):= 2\pi i\widehat \Upsilon(f)^T \left(\mathcal{X}-\langle N\widehat \Omega\rangle \right)
\]  
where $\langle N\widehat \Omega\rangle \equiv N\widehat \Omega \mod  1$ and  $(\mathcal{X}_1,...,\mathcal{X}_{k-1})$ is a $\Z^{k-1}$-valued random variable with probability mass function given by
\begin{multline}\label{DiscreteProb}
\mathbb{P}(\mathcal{X}_{1}=x_{1},\ldots,\mathcal{X}_{k-1}=x_{k-1}) = \hat{c} e^{\pi i\left(x-\langle N\widehat \Omega\rangle\right)^{T}\widehat \tau \left(x-\langle N\widehat \Omega\rangle\right)}, \\ x=(x_{1},\ldots,x_{k-1}) \in \mathbb{Z}^{k-1},
\end{multline} 
where $\hat{c}$ is the normalization constant.

Integrating by parts (see Section \ref{SecProofeqlaplace} for details) we find the following equivalent expression for the asymptotics of Theorem \ref{th:smoothasy}:
\begin{equation}\label{eqlaplace}
\frac{H_N(Fe^{-NV})}{H_N(e^{-NV})}=e^{N\int_Jf(x)d\mu_V(x)}e^{\frac{1}{2}\mathcal L(f)}\frac{\widehat \theta\left[\substack{-N\widehat \Omega\\0}\right]\left(\widehat \Upsilon(f)\right)}{\widehat \theta\left[\substack{-N\widehat \Omega\\0}\right](0)}(1+\mathcal O(N^{-1})),
\end{equation}
as $N\to \infty$,
where
\begin{equation} \label{defLG}\begin{aligned}
\mathcal{L}(f)&=\iint_{J\times J}\mathcal G(z_+,\lambda_+)f'(\lambda)f'(z)d\lambda dz,\\
\mathcal G(z,\lambda)&=\frac{1}{2\pi^2}\left(\log \frac{1}{|\Theta(\lambda,z)|}-4\pi\sum_{l,j=1}^{k}\mathrm{Im}(u_{j}(z))\mathrm{Im}(u_{l}(\lambda))(\mathrm{Im}\tau)^{-1}_{j,l}\right),
\end{aligned}
\end{equation}
and $\mathcal G(z_+,\lambda_+)=\lim_{\epsilon \downarrow 0}\mathcal G(z+i\epsilon,\lambda+i\epsilon)$.
The representation in \eqref{eqlaplace} has a number of useful properties. 

For example, it gives the asymptotic distribution of $\#=(\#_j)_{j=1}^{k-1}$, where 
\begin{equation} \#_j=\textrm{The number of eigenvalues in $(a_{j+1}-\epsilon,b_{j+1}+\epsilon)$}, \end{equation}
for fixed and sufficiently small $\epsilon>0$:
\begin{corollary} \label{Corrflucint}
As $N\to \infty$,
\begin{equation} \mathbb P\left(\#=N\widehat \Omega-\langle N\widehat \Omega\rangle+x\right)= \mathbb{P}(\mathcal{X}_{1}=x_{1},\ldots,\mathcal{X}_{k-1}=x_{k-1})(1+o(1)). \end{equation} 
\end{corollary}
\begin{proof}
We make a comparison of Laplace transforms. Let $f(\lambda)=s_j$ for $\lambda \in [a_{j+1}-\epsilon,b_{j+1}+\epsilon]$ for $j=1,\dots,k-1$,  and zero otherwise. Then
\begin{equation} \mathbb E \left[ e^{s_1 \#_1+\dots +s_{k-1} \#_{k-1}}\right]= \frac{H_N(e^{f-NV})}{H_N(e^{-NV})} . \end{equation}
Then $\mathcal L(f)=0$ and by the definition of $\widehat \Upsilon$ in \eqref{defhatUps} we have $\widehat \Upsilon(f)=\frac{1}{2\pi i}s$ where $s=(s_j)_{j=1}^{k-1}$. Also $\int_J f(\lambda) d\mu_V(\lambda)= \sum_{j=1}^{k-1} s_j \widehat \Omega_j$. By \eqref{eqlaplace}, 
\begin{equation} \label{Laplaces1sk} \mathbb E \left[ e^{s_1 \#_1+\dots +s_{k-1} \#_{k-1}}\right]=\exp\left(N\sum_{j=1}^{k-1}s_j\widehat \Omega_j\right)\frac{\widehat \theta\left[\substack{-N\widehat \Omega\\0}\right]\left(\frac{1}{2\pi i} s\right)}{\widehat \theta\left[\substack{-N\widehat \Omega\\0}\right](0)}(1+\mathcal O(N^{-1})), \end{equation}
as $N\to \infty$.

We now proceed by contradiction. Assume the corollary did not hold true. Then there is a sequence $N^{(k)}\in \mathbb Z$ with $k=1,2,\dots$ such that $\mathbb P\left(\#=N^{(k)}\widehat \Omega-\langle N^{(k)}\widehat \Omega\rangle+x\right)$ remains bounded away from $\mathbb{P}(\mathcal{X}_{1}=x_{1},\ldots,$ $\mathcal{X}_{k-1}=x_{k-1})$. Since $\langle N^{(k)}\widehat \Omega \rangle $ is in a compact subset of $\mathbb R^{k-1}$, the sequence $N^{(k)}$ in turn has a subsequence $N^{(k_j)}$ such that $\langle N^{(k_j)}\widehat \Omega \rangle $ is convergent, denote the limit by $y_*$. It follows that the Laplace transform of $\#-N\widehat \Omega $  converges, along the subsequence $N=N^{(k_j)}$,  to the Laplace transform  of $\mathcal{X}-y_{*}$ where the law of $\mathcal{X}$ is defined by \eqref{DiscreteProb} but with $\langle N \widehat \Omega \rangle $ replaced with $y_*$. Since the convergence of the Laplace transform implies that the corollary holds for $N^{(k_j)}$ (see e.g. \cite[page 390]{Bill}), we have obtained our contradiction.
\end{proof}

A second useful aspect of \eqref{eqlaplace} is that the expression for $\mathcal G$ is particularly natural because it is independent of basis - the equality in the second line of \eqref{defLG} also holds if we replace $\Theta$, $ u$ and $\tau$ by $\widehat \Theta$, $\widehat u$, and $\widehat \tau$. In Section \ref{Applications:connection} below, following the work of Kang and Makarov in \cite{KM}, we give a brief description of the role that $\mathcal G$ plays as the correlation kernel of the two dimensional Gaussian free field on $\mathcal S$ restricted to $J$.

We can compare \eqref{smasybasis3} to   \cite[(8.20)]{BG2}. Observe that \cite{BG2} deals with both real and imaginary $\phi$, so  \cite[(8.20)]{BG2} holds when we set $s=1$ and $\phi(x)=-if(x)$. Then the two formulas match if we assume that $\varpi_0=0$ in   \cite[(8.21)]{BG2}, up to a minor discrepancy of the multiplication of a unit of $i$ in the second line of  in \cite[(8.21)]{BG2} (which we believe is simply a typo).

\subsection{Main results: Fisher-Hartwig singularities}\label{sec:Main3}
Finally we describe ratio asymptotics in the situation where we have Fisher-Hartwig singularities. We introduce some further notation for this. Let $p\in \N_+:=\{1,2,...\}$, and
\begin{equation}\label{eq:omega}
\omega(x)=\prod_{j=1}^p \omega_{\alpha_j}(x)\omega_{\beta_j}(x), \quad \omega_{\alpha_j}(x)=|x-t_j|^{\alpha_j},\quad \omega_{\beta_j}(x)=\begin{cases}
e^{\pi i \beta_j}, & x<t_j\\
e^{-\pi i \beta_j}, & x\geq t_j
\end{cases}
\end{equation}
with $t_j\in \cup_{l=1}^k (a_l,b_l)$ and
with $\alpha_j>-1$ and $\Re \beta_j=0$ for all $j=1,\dots,p$.  A complete description of the asymptotics of $\frac{H_N(\omega F e^{-NV})}{H_N(e^{NV})}$ has been obtained in the one-cut case through the works of  \cite{Krasovsky, Garoni, IK, BWW, Charlier}, see \cite{Charlier} for the most general asymptotics available. We give a brief overview of the literature on determinants with FH singularities in Section \ref{LitOverFH}. Despite being well studied in the one-cut case, no progress has been made on this problem in the multi-cut case. We derive complete asymptotics in the multi-cut case which we describe in Theorem \ref{th:FHasy} below, but for the sake of exposition we start by considering pure root singularities (the case where $F=1, \beta=0$) and pure jump singularities ($F=1,\alpha=0$).

Let $\omega_\alpha(x)=\prod_{j=1}^p \omega_{\alpha_j}(x)$. By \eqref{eq:RMTrat}, if $\boldsymbol P$ is the characteristic polynomial of the random matrix  $M$ distributed according to \eqref{distributionM}, i.e. $\boldsymbol P(x)=\prod_{j=1}^N(x-\lambda_j)$  where $\lambda_j$ are the eigenvalues of $M$, then
\begin{equation} \frac{H_N(\omega_\alpha  e^{-NV})}{H_N(e^{-NV})}=\mathbb E\left(\prod_{j=1}^p\left| \boldsymbol P(t_j)\right|^{\alpha_j}\right). \end{equation}
When $V$ is one-cut regular, we have by \cite{Krasovsky,BWW}, 
\begin{multline}\label{onecutalpha}\mathbb E\left(\prod_{j=1}^p\left| \boldsymbol P(t_j)\right|^{\alpha_j}\right)=
e^{N \int  \log \omega_\alpha(x)d\mu_V(x)} \prod_{j=1}^p N^{\frac{\alpha_j^2}{4}} \frac{G\left(1+\frac{\alpha_j}{2}\right)^2}{G(1+\alpha_j)}\left(2\pi \psi_V(t_j)\right)^{\frac{\alpha_j^2}{4}}\\ \times 
\left(\frac{4}{b_1-a_1}\right)^{-\frac{\mathcal A^2}{4}}
 \prod_{1\leq j<l \leq p}|t_j-t_l|^{-\frac{\alpha_j\alpha_l}{2}}(1+\mathcal O((\log N)N^{-1})),
\end{multline}
as $N\to \infty$, where $\mathcal A=\sum_{j=1}^p \alpha_j$ and where $G$ is the Barnes G-function.  When $V$ is $k$-cut regular, we obtain in Theorem \ref{th:FHasy} below that
\begin{multline} \label{kcutalpha}\mathbb E\left(\prod_{j=1}^p\left| \boldsymbol P(t_j)\right|^{\alpha_j}\right)=e^{N \int  \log \omega_\alpha(x)d\mu_V(x)}\frac{\theta(N\Omega+\Upsilon(\log \omega_\alpha))}{\theta(N\Omega)}\exp \left(-\frac{\mathcal{A}^2}{4}\mathcal C_{\mathcal S} \right) \\ \prod_{j=1}^p N^{\frac{\alpha_j^2}{4} }\frac{G\left(1+\frac{\alpha_j}{2}\right)^2}{G(1+\alpha_j)}\left(2\pi \psi_V(t_j)\right)^{\frac{\alpha_j^2}{4}}
 \prod_{1\leq j<l \leq p}|t_j-t_l|^{-\frac{\alpha_j\alpha_l}{2}}(1+o(1)),
\end{multline}
as $N\to \infty$,
where  $\Upsilon(\log \omega_\alpha)$ is defined as in Theorem \ref{th:smoothasy} with $f$ replaced by $\log \omega_\alpha$, and $\mathcal C_{\mathcal S}$ is a constant depending only on $\mathcal S$ given by
\begin{equation}\label{defCS} \mathcal C_{\mathcal S}=\lim_{R\to +\infty}\left[ \int_{-R}^{a_1}\left| \frac{\widetilde{Q}(x)}{\mathcal{R}^{1/2}(x)} \right|dx-\log R\right],\end{equation}  
where $\widetilde{Q}$ is the unique monic polynomial of degree $k-1$ satisfying
\begin{equation} \label{intQhat}
\int_{b_j}^{a_{j+1}}\frac{\widetilde{Q}(x)}{\mathcal R^{1/2}(x)}dx=0
\end{equation}
for all $j=1,\dots,k-1$. By a contour deformation argument, one finds that if $t \in (a_j,b_j)$, then
\begin{equation} \Upsilon(\log|\cdot-t|)= \int_{-\infty}^{a_1}u'(x)dx-\frac{1}{2}\sum_{l=1}^{j-1}e_l.\end{equation}
In particular, although $\Upsilon(\log \omega_\alpha)$ depends on $t_j$, it only depends on which interval of  $\textrm{supp} ( \mu_V)$ the point $t_j$ lies in, not where in the interval it lies. We find it curious that the dependence on the location of $t_j$ is so simple -- if $t_j'$ is in the same interval as $t_j$ for $j=1,\dots,p$, then
\begin{multline}\frac{\mathbb E\left(\prod_{j=1}^p\left| \boldsymbol P(t_j)\right|^{\alpha_j}\right)}{\mathbb E\left(\prod_{j=1}^p\left| \boldsymbol P(t_j')\right|^{\alpha_j}\right)}
=e^{N\int \sum_{j=1}^p\alpha_j\log \left| \frac{x-t_j}{x-t_j'}\right| d\mu_V(x)}\\ \prod_{j=1}^p\left(\frac{\psi_V(t_j)}{\psi_V(t_j')}\right)^{\frac{\alpha_j^2}{4}}\prod_{1\leq j<l\leq p}\left|\frac{t_j-t_l}{t_j'-t_l'}\right|^{-\frac{\alpha_j\alpha_l}{2}}(1+o(1)),
\end{multline}
as $N\to \infty$.

The subleading terms are useful in studying ratios, for example relying on \eqref{idchange2} to change basis, we have
\begin{multline} \frac{\mathbb E\left(\left| \boldsymbol P(t_1)\right|^{\alpha_1}\left| \boldsymbol P(t_2)\right|^{\alpha_2}\right)}{\mathbb E\left(\left| \boldsymbol P(t_1)\right|^{\alpha_1}\right)\mathbb E\left(\left| \boldsymbol P(t_2)\right|^{\alpha_2}\right)}  \\
=|t_1-t_2|^{-\frac{\alpha_1\alpha_2}{2}}e^{-2\pi i\alpha_1\alpha_2\Upsilon(\log |\cdot -t_1|)^T\tau^{-1}\Upsilon(\log |\cdot -t_2|)} \exp\left(-\frac{\alpha_1\alpha_2}{2}\mathcal C_{\mathcal S}\right)\\ \times  \frac{\widehat \theta\left[\substack{-N\widehat \Omega\\0}\right]\left(\alpha_1\widehat \Upsilon(\log|\cdot-t_1|)+\alpha_2\widehat \Upsilon( \log |\cdot -t_2|)\right)}{\widehat \theta\left[\substack{-N\widehat \Omega\\0}\right]\left(\alpha_1\widehat \Upsilon(\log |\cdot -t_1|)\right)\widehat \theta\left[\substack{-N\widehat \Omega\\0}\right]\left(\alpha_2\widehat \Upsilon(\log |\cdot -t_2|)\right)} (1+o(1)),
\end{multline}
as $N\to \infty$, where for $t\in (a_j,b_j)$,
\begin{equation} \widehat \Upsilon(\log|\cdot -t|)= \int_{-\infty}^{a_1}\widehat u'(x)dx+\frac{1}{2}\widehat \tau_{j-1}. \end{equation}

We now consider pure jump singularities. 
Let $\omega_\beta(x)=\prod_{j=1}^p \omega_{\beta_j}(x)$. The Hankel determinant with pure jump singularities is the Laplace transform of the eigenvalue counting function as follows
\begin{align*}\mathbb E\left[e^{\sum_{j=1}^p2\pi  v_j\mathcal{H}_N(t_j) }\right]
&=\frac{H_N(\omega_\beta e^{-NV})}{H_N(e^{-NV})}e^{-N \int  \log \omega_\beta (x)d\mu_V(x)}, \qquad v_j = i \beta_j, \\
\mathcal{H}_N(t) &=   \sum_{j=1}^N \mathbf{1}\{\lambda_j \le t\} - N \mu_V((-\infty, t]), \qquad t \in \mathbb{R},
\end{align*}
where $\lambda_1,\dots, \lambda_N$ are the eigenvalues of the random matrix $M$.

 When $V$ is one-cut regular, we have \cite{IK, Charlier}
\begin{multline}\label{onecutbeta}
\mathbb E\left[e^{\sum_{j=1}^p2\pi  v_j\mathcal{H}_N(t_j) }\right]= \prod_{j=1}^p N^{v_j^2} 
 G(1+iv_j)G(1-iv_j)(2\pi \psi_V(t_j))^{v_j^2}\left|\widetilde \Theta(t_{j,+},t_{j,+})\right|^{-v_j^{2}}\\
\times \prod_{1\leq j<l \leq p} |t_j-t_l|^{-2v_jv_l}
 \left|\widetilde \Theta(t_{l,+},t_{j,+})\right|^{-2v_jv_l}(1+\mathcal O((\log N)N^{-1})),
\end{multline}
as $N\to \infty$, with 
\begin{equation} \label{thetatilde} \widetilde \Theta(z,w)=\frac{\Theta(z,w)}{w-z}, \end{equation}
and $\Theta$ given by \eqref{Thetaonecut}. When $V$ is $k$-cut regular, we obtain in  Theorem \ref{th:FHasy} below that
\begin{multline}\label{kcutbeta}
\mathbb E\left[e^{\sum_{j=1}^p2\pi  v_j\mathcal{H}_N(t_j) }\right] \\
= \prod_{j=1}^p N^{v_j^2} 
G(1+iv_j)G(1-iv_j)(2\pi \psi_V(t_j))^{v_j^2}\left|\widetilde \Theta(t_{j,+},t_{j,+})\right|^{-v_j^{2}}\\
\times \frac{\theta \left(N\Omega+2\pi \sum_{j=1}^pv_j\Upsilon(\mathbf{1}_{t_j}(x))\right)}{\theta(N\Omega)}  \\
 \times \prod_{1\leq j<l \leq p} |t_j-t_l|^{-2v_jv_l}
 \left|\widetilde \Theta(t_{l,+},t_{j,+})\right|^{-2v_jv_l} (1+o(1)),
\end{multline}
as $N\to \infty$,  where $\mathbf{1}_t(x)=\mathbf{1}\{t\le x\}$, and with $\widetilde \Theta$ given by \eqref{thetatilde} but where $\Theta$ is given by \eqref{Theta} instead. Relying on \eqref{idchange2} to change basis, we have
\begin{multline}
\frac{\mathbb E \left[ e^{v_1\mathcal{H}_N(t_1)+v_2\mathcal{H}_N(t_2)}\right]}{\mathbb E \left[ e^{v_1\mathcal{H}_N(t_1)}\right]\mathbb E \left[ e^{v_2\mathcal{H}_N(t_2)}\right]}= \frac{\widehat \theta\left[\substack{-N\widehat \Omega\\0}\right]\left(\widehat \Upsilon\left(v_1\mathbf{1}_{t_1}+v_2\mathbf{1}_{t_2}\right)\right)}{
\widehat \theta\left[\substack{-N\widehat \Omega\\0}\right]\left(\widehat \Upsilon\left(v_1\mathbf{1}_{t_1}\right)\right)
\widehat \theta\left[\substack{-N\widehat \Omega\\0}\right]\left(\widehat \Upsilon\left(v_2\mathbf{1}_{t_2}\right)\right)
}\\ \times \exp\left[ v_1v_2 \mathcal G(t_{1,+},t_{2,+})\right](1+o(1)),
\end{multline}
for fixed $v_1,v_2$ as $N\to \infty$,  and where $\mathcal G$ was defined in \eqref{defLG}.

The ratio asymptotics in the general situation for the Fisher-Hartwig case are provided by the following theorem, which is proven in Section \ref{sec:FHproof} based on results in Section \ref{sec:FH}. The main idea of the proof is to use a deformation from a smooth weight (for which the corresponding asymptotics are available from Theorem \ref{th:smoothasy}) to a Fisher-Hartwig weight by bringing singularities from the complex plane to the real line. We rely here on results due to Claeys, Its, and Krasovsky \cite{CIK} regarding Painlev\'e V and the associated RH problem, and it is the first time this type of transition has been relied on to give asymptotics for determinants with fixed Fisher-Hartwig singularities. Our conclusion after experimenting with this technique is that we obtained our results with fewer and simpler calculations than previous methods employed (and when $\theta$-functions enter the picture there is a very substantial benefit), at the expense of a more advanced theory (and a proof which is not self-contained, relying on \cite{CIK}).

\begin{theorem}[Ratio asymptotics for a FH symbol]\label{th:FHasy} Let $V:\mathbb R \to \mathbb R$ be real-analytic, satisfying \eqref{eq:Vgrowth}, and assume that $V$ is $k$-cut regular for some $k\geq 2$.   Let $F$ satisfy the conditions of Theorem \ref{th:smoothasy}, and denote $f(x)=\log F(x)$ for $x$ in  a neighbourhood of $ J$.
Let $\alpha_j>-1$ and $i\beta_j\in \R$. Then, as $N\to \infty$, 
\begin{equation}\label{asymptotics FH in main thm}
\begin{aligned}
\frac{H_N(\omega F e^{-NV})}{H_N(e^{-NV})}&=
e^{ N\int \left(f(x)+ \log \omega(x)\right)d\mu_V(x)}\frac{\theta(N\Omega+\Upsilon(f+\log \omega))}{\theta(N\Omega)}\prod_{j=1}^{p}N^{\frac{\alpha_j^2}{4}-\beta_j^2}  \\ & \times \exp\left(\frac{1}{4}\oint_{\Gamma}\oint_{\widetilde \Gamma}W(z,\lambda)f(z)f(\lambda)\frac{dz}{2\pi i}\frac{d\lambda}{2\pi i}\right)\\
&\times \exp\left[- \mathcal{A}\left(\int_J \frac{f(x)\widetilde{Q}(x)}{\mathcal R^{1/2}_+(x)}\frac{dx}{2\pi i}\right)-\frac{\mathcal A^2}{4}\mathcal C_{\mathcal S}\right]
\\
& \times\prod_{j=1}^{p}\exp \left(-\frac{\alpha_j}{2}f(t_j)+\frac{\beta_j}{\pi i}\mathcal{P.V.}\int_J w_{t_j,+}(\lambda_+)f(\lambda)d\lambda\right)\\ 
&\times \prod_{j=1}^p\frac{G(1+\frac{\alpha_j}{2}+\beta_j)G(1+\frac{\alpha_j}{2}-\beta_j)}{G(1+\alpha_j)}(2\pi \psi_V(t_j))^{\frac{\alpha_j^2}{4}-\beta_j^2}\\
& \times
\prod_{j=1}^p \left|\widetilde \Theta(t_{j,+},t_{j,+})\right|^{\beta_j^{2}}\exp\left(\mathcal{A}\beta_j\left[\int_{b_k}^\infty w_{t_j,+}(\lambda)d\lambda -\frac{\pi i}{2} \right]
 \right)  \\
 &\times 
\prod_{1\leq j<l \leq p}e^{\frac{\pi i}{2}(\alpha_l\beta_j-\alpha_j\beta_l)} |t_j-t_l|^{2\beta_j\beta_l-\frac{\alpha_j\alpha_l}{2}}\left|\widetilde \Theta(t_{l,+},t_{j,+})\right|^{2\beta_j\beta_l}
 \\  &
 \times (1+o(1)),
\end{aligned}
\end{equation}
where $\mathcal{A}=\sum_{j=1}^{p}\alpha_{j}$, $G$ is the Barnes G-function; $\Upsilon(f+\log \omega)$ is defined as in Theorem \ref{th:smoothasy} but with $f$ replaced by $f+\log \omega$; $W$ was defined in \eqref{def:W}; $w_{t_j,+}(\lambda_+)=\lim_{\epsilon \downarrow 0}w_{t_j+i\epsilon}(\lambda+i\epsilon)$ where  $w_z(\lambda)$ was defined in \eqref{defwlambda}; $w_{t_j,+}(\lambda)=\lim_{\epsilon \downarrow 0}w_{t_j+i\epsilon}(\lambda)$; $\widetilde{Q}$ is the unique monic polynomial of degree $k-1$ satisfying \eqref{intQhat}
for all $j=1,\dots,k-1$; $\mathcal C_{\mathcal S}$ was defined in \eqref{defCS}; where
\[
\widetilde \Theta(z,w)=\frac{\Theta(z,w)}{w-z},
\]
and  $\Theta$ was defined in \eqref{Theta}.
\end{theorem}

By combining Theorems \ref{th:pfasy} and \ref{th:FHasy}, we obtain the full asymptotics for $H_N(\omega F e^{-NV})$.
\begin{remark}
\label{Remarkonecut}
For $k=1$, we can verify that Theorem \ref{th:FHasy} matches with the asymptotics for the one-cut case given by \cite[Theorem 1.1]{Charlier}. Let us assume without loss of generality that $a_{1}=-1$ and $b_{1}=1$.  
Observe that for $k=1$ we have $\widetilde{Q}(x)=1$, set $\theta \equiv 1$, and recall that $\Theta$ is given by \eqref{Thetaonecut}. Hence,
\begin{equation} \label{wzonecut}
w_z(\lambda)=\frac{\mathcal R^{1/2}(z)}{\mathcal R^{1/2}(\lambda)(\lambda-z)}, \qquad W(z,\lambda)=\frac{z\lambda-1}{(z-\lambda)^2\sqrt{(z^2-1)(\lambda^2-1)}}.
\end{equation}
Using \eqref{Thetaonecut} and \eqref{wzonecut}, we verify that
\begin{equation}\label{lol5}
\int_{b_k}^\infty w_{t_j,+}(\lambda)d\lambda -\frac{\pi i}{2} = i\arcsin t_j. 
\end{equation}
With $\widetilde Q=1$ and by the definition of $\mathcal C_{\mathcal S}$ in \eqref{defCS}, we have $\mathcal C_{\mathcal S}=\log 2$ for $k=1$.
By substituting \eqref{wzonecut} and \eqref{lol5} in \eqref{asymptotics FH in main thm}, we can now verify that Theorem \ref{th:FHasy} with $k=1$ indeed matches with \cite[Theorem 1.1]{Charlier}.
\end{remark}
\begin{remark}\label{explicittau}
The quantities $\boldsymbol{\omega_j}, \, \tau,\, $ and $\Upsilon(f)$ were defined in \eqref{eq:Aint}, \eqref{formomega}, \eqref{eq:period} and \eqref{def:Omegahat}. Equivalently, each of these objects can also be written explicitly in terms of certain matrices as follows. Let $\boldsymbol Q:=(q_{j,r})_{j,r=1}^{k-1}$ and $\boldsymbol A:=(\boldsymbol A_{r,i})_{r,i=1}^{k-1}$, where 
\begin{align*}
\mathsf Q_j(x)=\sum_{r=1}^{k-1}q_{j,r}x^{r-1}, \qquad \mbox{ and } \qquad \boldsymbol A_{r,i}=\int_{b_i}^{a_{i+1}}\frac{x^{r-1}dx}{\mathcal R^{1/2}(x)}.
\end{align*}
Using the product formula for the  Vandermonde determinant, we infer that 
\begin{equation}\label{eq:detA}
\det \boldsymbol A = \int_{b_1}^{a_2} dx_1 \dots \int_{b_{k-1}}^{a_k} dx_{k-1} \prod_{i < j} (x_j - x_i) \prod_{j=1}^{k-1} \frac{1}{\mathcal{R}^{1/2}(x_j)}.
\end{equation} 
Since $x_i<x_j$ for $i<j$, and since $\mathcal R^{1/2}$ is real and does not change its sign on $(b_i,a_{i+1})$, the right-hand side of \eqref{eq:detA} is clearly non-zero and therefore $\boldsymbol A$ is invertible. By \eqref{eq:Aint} - \eqref{eq:Aint2}, we have
$-2\boldsymbol Q \boldsymbol A=I, $ and so
\begin{equation*}
\boldsymbol Q=-\frac{1}{2}\boldsymbol A^{-1}.
\end{equation*}
Similarly, let $\boldsymbol B_{j,r}:=\sum_{i=1}^j\int_{a_i}^{b_i} \frac{x^{r-1}dx}{\mathcal R_+^{1/2}(x)}$ and let $\boldsymbol B:=\left(\boldsymbol B_{j,r}\right)_{j,r=1}^{k-1}$. By \eqref{eq:period},
\begin{equation*} 
\tau=-\boldsymbol B \boldsymbol A^{-T}, 
\end{equation*}
where $\boldsymbol A^{-T}$ denotes the inverse transpose of $\boldsymbol A$. Finally, if we define
\begin{align*}
c_{r}(f) := \frac{1}{2\pi i} \int_{J} \frac{f(x)  x^{r-1}dx}{\mathcal{R}^{1/2}_+(x)}, \qquad \mbox{ and } \qquad \boldsymbol c(f):=\left(c_r(f)\right)_{r=1}^{k-1}, 
\end{align*}
then by \eqref{def:Omegahat} we have
\begin{equation*} 
\Upsilon(f)= \boldsymbol A^{-1}\boldsymbol c(f).
\end{equation*}
(Above, $-2\boldsymbol A$ is the matrix of  $A$-periods in Figure \ref{ContHom} and $2\boldsymbol B$ is the matrix of  $B$-periods).
\end{remark}

\subsection{Applications: connection to the Gaussian free field on a Riemann surface}\label{Applications:connection}

An active research topic in random matrix theory in the past years has been the connection between random matrices and the theory of logarithmically correlated random fields (see e.g. \cite{HKO,FyKe,CFLW} and references therein) -- that is random (generalized) functions $X:\R^d\to \R$ with a covariance of the form $\mathrm{Cov}(X(x),X(y))=\log |x-y|^{-1}+g(x,y)$ for some continuous function $g$. A prime example of such an object is the two-dimensional Gaussian free field, which is a Gaussian process on a subset of $\R^2$ whose covariance is given by the Green's function of the domain (with suitable boundary conditions).  See e.g. \cite{Sheffield} for a review on the Gaussian free field. As will be important for us shortly, similar objects can be defined on more general manifolds -- see e.g. \cite[Section 2]{KM} for a discussion of the Gaussian free field on compact Riemann surfaces.

Let us illustrate the connection between the Gaussian free field and random matrix theory with two examples taken from the one-cut setting. For this purpose, let $M$ be an $N\times N$ GUE random matrix normalized so that the equilibrium measure $d\mu$ has support $[-1,1]$, and define
\begin{align*}
& X_{1,N}(x) = \sqrt{2}\bigg(\log |\det(M-x)|-N\int \log|x-t|d\mu(t)\bigg), \\
& X_{2,N}(x) = \sqrt{2}\pi\bigg(\sum_{j=1}^N\mathbf 1\{\lambda_j\leq x\}-N\int_{-1}^xd\mu(y) \bigg),
\end{align*}
where the $\lambda_{j}$'s are the eigenvalues of $M$. $X_{1,N}$ and $X_{2,N}$ encode different types of information about the fluctuations of the spectrum of $M$ around the equilibrium measure. In the large $N$ limit, both $X_{1,N}$ and $X_{2,N}$ exhibit a covariance structure of logarithmically correlated random fields. Indeed, it is known (see e.g. \cite{FKS, CFLW}) that 
\begin{align}
& \lim_{N\to + \infty} \mathrm{Cov}(X_{1,N}(x),X_{1,N}(y)) = \log |x-y|^{-1} - \log 2, \nonumber \\
& \lim_{N\to + \infty} \mathrm{Cov}(X_{2,N}(x),X_{2,N}(y)) = \log |x-y|^{-1} \nonumber \\
& \hspace{5.065cm} + \log(1-xy+\sqrt{1-x^2}\sqrt{1-y^2}). \label{eq:cfcov}
\end{align}
The right-hand side of \eqref{eq:cfcov} can also be seen as the covariance of a variant of the Gaussian free field -- see e.g. \cite[Appendix A]{CFLW}. This shows that $X_{1,N}$ and $X_{2,N}$, which are objects that arise naturally in random matrix theory, can be understood, in the large $N$ limit, as being restrictions of (variants of) the two-dimensional Gaussian free field to certain one-dimensional subsets.

\medskip Assume from now on that $M$ is a $N\times N$ Hermitian matrix drawn from a multi-cut random unitary invariant ensemble. In this setting, the connection between spectral fluctuations of $M$ and the Gaussian free field is much more subtle. As was already pointed out in \cite{BG2}, in this case the fluctuations consist of two independent contributions: one that is connected to some $N$-dependent theta-functions and which does not necessarily have a meaningful limit as $N\to\infty$, and another one that does have a limit that is described by the restriction of a Gaussian free field on a Riemann surface to a one-dimensional subset. We now describe how Theorem \ref{th:smoothasy} can be used to obtain a new perspective on this connection. First, we discuss how Theorem \ref{th:smoothasy} is related to the spectral fluctuations of $M$, then we describe the ``theta-function part" of these fluctuations, and finally the ``Gaussian part". As before, we denote by $\lambda_{1},\ldots,\lambda_{N}$ the eigenvalues of $M$.

\medskip 

\underline{Theorem \ref{th:smoothasy} and spectral fluctuations:} Let $F$ be as in Theorem \ref{th:smoothasy}, and assume furthermore (for simplicity) that $\log F$ is compactly supported. Then the ratio of Hankel determinants appearing in Theorem \ref{th:smoothasy} can be understood as 
\begin{align}
\frac{H_{N}(Fe^{-NV})}{H_{N}(e^{-NV})}e^{-N\int \log F(x)d\mu_{V}(x)} & = \E e^{\sum_{j=1}^N \log F(\lambda_j)-N\int\log F(x)d\mu_V(x)}, \label{lol7} \\
& = \E e^{-\int_{\R} \left(\sum_{j=1}^N\mathbf 1\{\lambda_j\leq x\}-N\int_{-\infty}^x d\mu_V\right)h(x)dx}, \nonumber
\end{align}
with $h(x):=\frac{d}{dx}\log F(x)$. This last expression can be rewritten more compactly as $\mathbb{E}e^{-X_{N}(h)}$, where
\begin{align*}
X_{N}:C_c^\infty(\R)\to \R \; ; \; h \mapsto \int_{\R} \bigg(\sum_{j=1}^N\mathbf 1\{\lambda_j\leq x\}-N\int_{-\infty}^x d\mu_V\bigg)h(x)dx,
\end{align*}
and therefore the left-hand side of \eqref{lol7} can be viewed as the Laplace transform of $X_{N}(h)$. More generally, when studying random generalized functions, say $X: C_c^\infty(\R^d)\to \R$, one refers to $\mathbb{E}e^{-X(h)}$ as the Laplace transform of $X(h)$ and to $h\mapsto \E e^{-X(h)}$ as the Laplace functional of $X$. The above considerations show that the ratio of the Hankel determinants in Theorem \ref{th:smoothasy} encodes the Laplace functional of the eigenvalue counting function $X_{N}$ (which in turn, encodes the probability distribution of $X_{N}$).

\medskip 

\underline{The ``theta-function part" and discrete Gaussian random variables:} Let us first point out that the quantity
\begin{align}\label{lol8 2}
\frac{\widehat \theta\left[\substack{-N\widehat \Omega\\0}\right]\left(\widehat \Upsilon(f)\right)}{\widehat \theta\left[\substack{-N\widehat \Omega\\0}\right](0)},
\end{align}
which appears on the right-hand side of \eqref{eqlaplace}, can be written as $\E e^{-v_N(f)}$ (in other words, \eqref{lol8} is the Laplace transform of $v_N(f)$) with 
\[
v_N(f):= -2\pi i\sum_{j=1}^{k-1}\widehat \Upsilon_j(f)(\mathcal{X}_j-N\widehat \Omega_j)
\]  
where $(\mathcal{X}_1,...,\mathcal{X}_{k-1})$ is a $\Z^{k-1}$-valued random variable with probability mass function given by
\begin{multline*}
\mathbb{P}(\mathcal{X}_{1}=x_{1},\ldots,\mathcal{X}_{k-1}=x_{k-1}) = \hat{c} e^{\pi i(x-N\widehat \Omega)^{T}\widehat \tau(x-N\widehat \Omega)}, \\  \mbox{for } x=(x_{1},\ldots,x_{k-1}) \in \mathbb{Z}^{k-1},
\end{multline*} 
where $\hat{c}$ is the normalization constant. The Laplace functional $f \mapsto \E e^{-v_N(f)}$ can therefore be understood as that of a suitable random field $\hat{v}_{N}(x)$ whose values are discrete Gaussian random variables. Indeed, by the definition of $\widehat \Upsilon$ from \eqref{defhatUps}, we can write 
\[
v_N(f)=\int_{J}\hat{v}_N(x)f(x)dx:=\int_J 2\sum_{j=1}^{k-1}\widehat u_{j,+}'(x)(\mathcal{X}_j-N\widehat \Omega_j)f(x)dx.
\]
Let us now turn to the ``Gaussian part" of the fluctuations.

\medskip

\underline{Gaussian fluctuations and the Gaussian free field:} We now analyze the other part of the right-hand side of \eqref{eqlaplace}, namely $e^{\frac{1}{2}\mathcal L(f)}$. Since $f \mapsto \mathcal L(f)$ is quadratic in $f$, if we can show that $\mathcal L(f) \geq 0$ for all $f$, then this means that $f \mapsto e^{\frac{1}{2}\mathcal L(f)}$ is the Laplace functional of a Gaussian random field. Note that $\mathcal L(f)$ can be rewritten as $\mathcal L(f) = \iint_{J\times J} \mathcal{G}(x,y)f'(x)f'(y)dxdy$, where
\begin{align}\label{def of C}
\mathcal{G}(x,y):=\frac{1}{2\pi^2}\left(\log \frac{1}{|\Theta(x_+,
y_+)|}-4\pi\sum_{l,j=1}^{k-1}\mathrm{Im}(u_{j,+}(x))\mathrm{Im}(u_{l,+}(y))(\mathrm{Im}\tau)^{-1}_{j,l}\right).
\end{align}
In other words, showing $\mathcal{L}(f) \geq 0$ for all $f$ is equivalent to proving that $\mathcal{G}(x,y)$ is a covariance kernel. 

$\mathcal{G}(x,y)$ is in fact closely related to the (bipolar) Green's function, which is an object of central importance in the study of Riemann surfaces. We refer to e.g. \cite[Sections 2 and 3]{KM} for an in depth discussion of the role of the bipolar Green's function in the study of the Gaussian free field on a compact Riemann surface. For the convenience of the reader, we also briefly review here the concepts that are relevant for us.

We begin by recalling the definition and properties of the Laplace(-Beltrami) operator on a general Riemannian manifold (of which Riemann surfaces are a special case).  We then quote some results from \cite{KM} about how these general results translate to Riemann surfaces and what is known about the Laplacian Green's  function on Riemann surfaces. We refer the interested reader to e.g. \cite[Chapter 7]{Buser} for more on the Laplace(-Beltrami) operator.

First of all, we recall that for a $d$-dimensional smooth Riemannian manifold $\mathcal M$ with metric $g=(g_{i,j})_{i,j=1}^d$, the Laplace(-Beltrami) operator is defined in some given coordinate chart as 
\[
\Delta u=\frac{1}{\sqrt{\det(g_{p,q})_{p,q=1}^d}}\sum_{i,j=1}^d\partial_i\left(g^{i,j}\sqrt{\det(g_{p,q})_{p,q=1}^d}\partial_ju\right),
\]
for $u\in C^\infty(\mathcal M)$. Here we have written (as is common in Riemannian geometry), $g^{i,j}$ for the $i,j$ entry of the inverse of the matrix $(g_{i,j})_{i,j=1}^d$. Note that in the case where $\mathcal M=\R^d$ and where the metric tensor $g$ is just the constant identity matrix, this is the familiar Laplacian operator. 

A basic fact about the Laplace operator in this generality is that if the manifold is compact and does not have a boundary, then a spectral theorem holds. More precisely (see \cite[Theorem 7.2.6]{Buser}), there exists a sequence of real-valued $C^\infty(\mathcal M)$ functions $(\varphi_n)_{n=0}^\infty$ which form an orthonormal basis of $L^2(\mathcal M)$ (where the inner product is $\int_{\mathcal M}f(x)g(x)\sqrt{\det(g_{i,j}(x))_{i,j=1}^d}dx$) and they are eigenfunctions of the Laplacian operator: $\Delta \varphi_j=-\lambda_j\varphi_j$ with $0=\lambda_0<\lambda_1\leq \lambda_2\leq ...$. Furthermore, the only eigenfunctions with $\lambda_0=0$ are constant functions.

We can view a Riemann surface $\hat{\mathcal S}$ as a special case of a Riemannian manifold -- in this case, our metric is a conformal one: there exists a positive function $\rho:\hat{\mathcal S}\to \R$ such that $g_{i,j}=\rho \delta_{i,j}$ \cite[Lemma 2.3.3]{Jost}. In this case, the Laplacian operator becomes just $\rho^{-1}4\partial\bar \partial$, where $\partial = \frac{1}{2}\partial_{1}-\frac{i}{2}\partial_{2}$ and $\bar\partial = \frac{1}{2}\partial_{1}+\frac{i}{2}\partial_{2}$.

Just as in the study of the Laplacian on $\R^d$, a central importance in the study of the Laplacian on a Riemann surface is played by the notion of the Green's function. For this purpose, we introduce the resolvent kernel:
\[
R_\rho(z,w)=\sum_{n=1}^\infty \frac{1}{\lambda_n}\varphi_n(z)\varphi_n(w).
\]
It follows from the spectral theorem that on the subspace of $L^2(\mathcal M)$ which is the orthogonal complement of the constant functions, the integral operator with kernel $R_\rho$ is the inverse of the Laplacian, i.e. $\Delta^{-1}h(z)=\int_{\hat{\mathcal S}} R_{\rho}(z,w)h(w)\rho(w) dw$ for all $h$ such that $\int_{\hat{\mathcal S}} h(w)\rho(w) dw=0$. 

\medskip We will now specialize the above general theory to the Riemann surface $\mathcal{S}$ defined in Section \ref{Sec:calS}. Then we will verify that indeed $\mathcal L(f)\geq 0$. For the remainder of our discussion, we will think of $\rho$ as being given, and we will suppress it in our notation (for example, we will write $R$ instead of $R_{\rho}$).

\medskip It is known, see \cite[Corollary 3.2]{KM}, that there exists some constant $c$ only depending on $\mathcal S$ so that
\begin{align}\label{eq:thetabip}
&\frac{1}{\pi}\left(\log \left|\frac{\theta\left[{\substack{{ \ualpha} \\ { \ubeta}}}\right](u(z)-u(q))}{\theta\left[{\substack{{ \ualpha} \\ { \ubeta}}}\right](u(z)-u(p))}\right|-2\pi\sum_{l,j=1}^{k-1}\Im(u_l(p)-u_l(q))(\Im \tau)^{-1}_{j,l}\Im u_j(z)\right)\\
&=R(z,p)-R(z,q)+c,\notag
\end{align}
where $\left[{\substack{{ \ualpha} \\ { \ubeta}}}\right]$ is the odd non-singular characteristic defined by $\ualpha=\frac{1}{2}e_1$ and $\ubeta=\frac{1}{2}\sum_{j=1}^{k-1}e_j$. We mention here in passing that the function $R(z,p)-R(z,q)$ is the bipolar Green's function on $\mathcal{S}$ ($p$ and $q$ being the ``poles"). Using \eqref{eq:quasi}, \eqref{Theta}, \eqref{def of C} and \eqref{eq:ujump2}, one sees that if we let $z\to x_+$, $p\to y_+$, and $q\to y_-$, then the left-hand side tends to $2\pi \mathcal{G}(x,y)$, namely that 
\begin{equation}\begin{aligned}\label{C plus}
\mathcal{G}(x,y)&=\frac{1}{2\pi} (R(x_+,y_+)-R(x_+,y_-))+c\\ &=\sum_{n=1}^\infty\frac{1}{2\pi \lambda_n}\varphi_n(x_+)(\varphi_n(y_+)-\varphi_n(y_-))+c.
\end{aligned}\end{equation}
On the other hand, letting now $z\to x_-$, $p\to y_-$, and $q\to y_+$, and using again \eqref{eq:thetabip} (along with \eqref{eq:ujump2} and \eqref{eq:quasi}), one can check that 
\begin{equation}\begin{aligned}\label{C minus}
\mathcal{G}(x,y)&=\frac{1}{2\pi}(R(x_-,y_-)-R(x_-,y_+))+c \\ &= \sum_{n=1}^\infty\frac{-1}{2\pi \lambda_n}\varphi_n(x_-)(\varphi_n(y_+)-\varphi_n(y_-))+c.
\end{aligned}\end{equation}
Hence, combining \eqref{C plus} with \eqref{C minus},  we see that in fact
\begin{align}\label{C plus minus}
\mathcal{G}(x,y)=\sum_{n=1}^\infty \frac{1}{4\pi \lambda_n}(\varphi_n(x_+)-\varphi_n(x_-))(\varphi_n(y_+)-\varphi_n(y_-))+c.
\end{align}
Now, if we let $x\to b_k$ in the above expression, then $\varphi_n(x_+)-\varphi_n(x_-) \to 0$, while by \eqref{def of C} the left-hand side vanishes since $|\Theta(b_k,y_+)|=1$ and $\Im(u_{j,+}(b_k))=0$. We conclude that in fact $c=0$. 

Hence, it follows from $\mathcal{L}(f)= \iint_{J\times J} \mathcal{G}(x,y)f'(x)f'(y)dxdy$, \eqref{C plus minus} and $c=0$ that
\[
\mathcal{L}(f) =\sum_{n=1}^\infty \frac{1}{4\pi \lambda_n}\left(\int_J f'(x)(\varphi_n(x_+)-\varphi_n(x_-))dx\right)^2\geq 0.
\]
Since $f$ was arbitrary, this proves that $f \mapsto e^{\frac{1}{2}\mathcal L(f)}$ is the Laplace functional of a Gaussian random field $X$, such that $X(f)$ is a Gaussian random variable with mean zero and variance $\mathcal L(f)$.

It remains to address the claim that this Gaussian random field $X$ is related to the Gaussian free field on the Riemann surface $\mathcal{S}$. For this purpose, one first needs to define what is meant by the Gaussian free field on $\mathcal{S}$. Naively, one would like the Gaussian free field to be the centered (generalized) Gaussian process with covariance given by the Green's function of the surface, and this Green's function in turn one would like (by the spectral theorem) to define as $\sum_{n=0}^{+\infty} \lambda_n^{-1}\varphi_n(z)\varphi_n(w)$. In other words, one would like to define the free field $\mathcal{Y}$ as $\mathcal{Y}(f) = \int \hat{\mathcal{Y}}(z)f(z)dz$, with $\hat{\mathcal{Y}}(z)=\sum_{n=0}^{+\infty} \mathcal{Y}_n \lambda_n^{-1/2}\varphi_n(z)$, where $\{\mathcal{Y}_n\}_{n=0}^{+\infty}$ are independent and identically distributed standard real Gaussian variables. There is, however, a caveat: the issue here is the ``zero mode", namely that $\lambda_0=0$. For this reason, the free field $\mathcal{Y}$ is only defined ``up to constants" (because $z\mapsto \varphi_{0}(z)$ is a constant), so one can make sense of the random variable $\mathcal{Y}(f)=\int_{\mathcal S} \hat{\mathcal{Y}}(z)f(z)dz$ only for functions with $\int_{\mathcal S}f(z)dz=0$.

There are various ways of constructing random fields that circumvent this issue. In \cite{KM}, the approach taken by the authors is to look at  differences of the field. They introduce a variant of the free field, denoted $\Phi(p,q)$ and which can be thought of as $\Phi(p,q) = ``\hat{\mathcal{Y}}(p)-\hat{\mathcal{Y}}(q)"$ (since $\hat{\mathcal{Y}}(p)$ is not well-defined, one should interpret  this difference as $\sum_{n=1}^{+\infty} \mathcal{Y}_n \lambda_n^{-1/2}(\varphi_n(p)-\varphi_{n}(q))$, which removes the ambiguity about $\lambda_{0}$). More precisely, $\Phi(p,q)$ is defined with the following covariance: 
\begin{align*}
\E\Phi(p,q)\Phi(\tilde p,\tilde q)&=2\log\left|\frac{\theta\left[{\substack{{ \ualpha} \\ { \ubeta}}}\right](u(\tilde p)-u(q))\theta\left[{\substack{{ \ualpha} \\ { \ubeta}}}\right](u(\tilde q)-u(p))}{\theta\left[{\substack{{ \ualpha} \\ { \ubeta}}}\right](u(\tilde p)-u(p))\theta\left[{\substack{{ \ualpha} \\ { \ubeta}}}\right](u(\tilde q)-u(q))}\right|^2\\
&\quad -4\pi \sum_{j,l=1}^{k-1}\Im(u_j(p)-u_j(q))(\Im \tau^{-1})_{j,l}\Im(u_l(\tilde p)-u_l(\tilde q)).
\end{align*}
As pointed out in \cite{KM}, the above right-hand side is a difference of two bipolar Green's functions. Taking $p \to x_{+}$, $q \to x_{-}$, $\tilde{p} \to y_{+}$ and $\tilde{q} \to y_{-}$, we obtain that $\E\Phi(p,q)\Phi(\tilde p,\tilde q) \to 8\pi^{2} \mathcal{G}(x,y)$. Hence, the field $\Phi(p,q)$, when $(p,q)$ is restricted to $(p,q)=(x_{+},x_{-})$ with $x \in J$, describes the ``Gaussian part" of the fluctuations.

\medskip

\underline{Summary:} In this section, we have shown that Theorem \ref{th:smoothasy} can be interpreted as a statement about the eigenvalue counting function $x\mapsto \hat{X}_{N}(x)=\sum_{j=1}^N \mathbf 1\{\lambda_j\leq x\}-N\int_{a_1}^x d\mu_V(t)$. More precisely, since the Laplace functional of $f \mapsto X_{N}(f):=\int_{\mathbb{R}}\hat{X}_{N}(x)f(x)dx$ is asymptotically a product of two Laplace functionals, Theorem \ref{th:smoothasy} should be understood as saying that asymptotically, as $N\to \infty$, $X_{N}$ is distributed as a sum of two independent terms $v_N$ and $X$:
\begin{itemize}
\item $v_N$ is a random function of the form 
\[
2\sum_{j=1}^{k-1}\widehat{u}_{j,+}(x)(\mathcal{X}_j-N\widehat \Omega_j)
\]
where $\mathcal{X}_1,...,\mathcal{X}_{k-1}\in \Z$ are random variables with a probability mass function given by
\begin{align}\label{discrete Gaussian}
\mathbb{P}(\mathcal{X}_{j} = x_{j}, \; j=1,\ldots,k-1) = \hat{c} \;  e^{\pi i (x-N\widehat \Omega)^{T}\widehat \tau(x-N\widehat \Omega)}, \qquad x_{1},\ldots,x_{k-1}\in \mathbb{Z},
\end{align}
where $\hat{c}$ is the normalization constant. This field $v_N$ is the one that is connected to the $\theta$-functions in Theorem \ref{th:smoothasy}, and because of \eqref{discrete Gaussian}, the random vector $(\mathcal{X}_1,...,\mathcal{X}_{k-1})\in \Z^{k-1}$ is said to be a discrete Gaussian random variable.
\item $X$ is the restriction of a (variant of) the Gaussian free field on the Riemann surface to a one-dimensional set -- the covariance of this Gaussian process is 
\[
\mathcal{G}(x,y)=\frac{1}{2\pi^2}\left(\log \frac{1}{|\Theta(x_+,
y_+)|}-4\pi\sum_{l,j=1}^{k-1}\mathrm{Im}(u_{j,+}(x))\mathrm{Im}(u_{l,+}(y))(\mathrm{Im}\tau)^{-1}_{j,l}\right).
\]
\end{itemize}

This discussion elaborates on remarks from \cite{Shcherbina,BG2,BLS} where similar facts were pointed out, though the connection to the Gaussian free field was not touched on. We also emphasize that in Theorem \ref{th:smoothasy}, the asymptotic fluctuations depend on the potential only through what the associated Riemann surface $\mathcal{S}$ is, i.e. only through the endpoints of the support of $\mu_{V}$ -- if two potentials are associated to the same surface, the probability distribution of the asymptotic fluctuations are the same. This is an instance of universality that has not been emphasized in previous studies.

\subsection{Applications: eigenvalue rigidity}\label{Sec:rig}
Recall that $M = M_N$ is an $N \times N$ random Hermitian matrix sampled from the $k$-cut regular ensemble $d\mathbb{P}(M) \propto e^{-N \Tr V(M)}dM$, and let $\lambda_{(1)} \le \lambda_{(2)} \le \dots \le \lambda_{(N)}$ be its ordered eigenvalues.\footnote{Not to confuse with $(\lambda_1, \dots, \lambda_N) \in \mathbb{R}^N$, which are the unordered eigenvalues of $M$.} A very natural and important question in random matrix theory is to understand how much the eigenvalues $\lambda_{(n)}$ may deviate from their respective classical locations, defined by
\begin{align*}
\rho_n := \inf \left\{x\in \mathbb{R}: \mu_V((-\infty, x]) = \frac{n}{N}\right\}, \qquad n = 1, \dots, N
\end{align*}

\noindent where $\mu_V$ is the equilibrium measure \eqref{eq:density}. This problem is often known as eigenvalue rigidity, and has been studied under various settings in the literature \cite{BEY2012, BEY2014a, BEY2014b, EYY, Li2017, CFLW}.

Unlike the one-cut regular scenario which is addressed in \cite{CFLW}, in the multi-cut setting the eigenvalues near the edge may `jump between consecutive intervals'. Such a phenomenon is hinted at by the appearance of theta functions in the asymptotics for smooth linear statistics in Theorem \ref{th:smoothasy}, and was already observed in the work of Bekerman \cite{Bek2018}. Therefore, the best we can hope for is a rigidity result in the bulk of the spectrum, and this is precisely what is established in the following corollary.
\begin{corollary} The following are true.
\begin{itemize}[leftmargin=0.75cm]
\item[(i)] For each $0 < \delta < \frac{1}{2} \min_{i \in \{1, \dots, k\}} \mu_V([a_i, b_i])$, let
\begin{multline*}
I_N(\delta)=\Bigg\{n: \sum_{i=1}^{l-1}\mu_V([a_i, b_i])+\delta \leq \frac{n}{N}\leq \sum_{i=1}^l\mu_V([a_i, b_i])-\delta, \\
 \text{ for some } l=1,...,k\Bigg\}.
\end{multline*}

\noindent Then for any $\epsilon > 0$,
\begin{align*}
\lim_{N \to \infty} \mathbb{P} \left(|\lambda_{(n)}-\rho_n|\leq \frac{1+\epsilon}{\pi\psi_V(\rho_n)}\frac{\log N}{N}  \text{ for all } n \in I_N(\delta) \right) = 1
\end{align*}

\noindent where $\psi_V$ is the density \eqref{eq:density} of the equilibrium measure $\mu_V$.

\item[(ii)] For each $0<\eta < \tfrac{1}{2}\min_{i \in \{1, \dots, k\}} (b_i - a_i)$, let $J(\eta)= \bigcup_{i=1}^k [a_i + \eta, b_i - \eta]$, and consider the centred eigenvalue counting function
\begin{align*}
\mathcal{H}_N(u) &=  \left[ \sum_{i=1}^N \mathbf{1}\{\lambda_j \le u\} - N \mu_V((-\infty, u])\right], \qquad u \in \mathbb{R}.
\end{align*}

\noindent Then for any $\epsilon > 0$, 
\begin{align*}
\lim_{N \to \infty} \mathbb{P}\left(\sup_{x \in J(\eta)} |\mathcal{H}_N(x)| \le \frac{1}{\pi}(1+\epsilon) \log N\right)  = 1.
\end{align*}
\end{itemize}
\end{corollary}

\begin{remark}
Our result is also valid for $k=1$, but does not cover the edge where the rigidity estimate remains valid. For further details, we refer the readers to \cite[Theorem 1.2-1.3]{CFLW} where matching lower bounds were also established by techniques of Gaussian multiplicative chaos. Based on the one-cut result, we believe that our rigidity upper bounds for the multi-cut case are optimal, i.e. we expect
\begin{align*}
&\lim_{N \to \infty} \mathbb{P} \left(|\lambda_{(n)}-\rho_n|\ge \frac{1-\epsilon}{\pi\psi_V(\rho_n)}\frac{\log N}{N}  \text{ for all } n \in I_N(\delta) \right) = 1 \\
\text{and} \qquad &\lim_{N \to \infty} \mathbb{P}\left(\sup_{x \in J(\eta)} |\mathcal{H}_N(x)| \ge\frac{1}{\pi} (1-\epsilon) \log N\right) = 1
\end{align*}

\noindent to hold for any $\epsilon > 0$, and $\delta, \eta > 0$ sufficiently small. This would, however, require finer asymptotics for Hankel determinants with merging singularities and further probabilistic analysis which is beyond the scope of this paper.
\end{remark}

\begin{proof}
We begin by observing that one can choose $\eta =\eta(\delta)$ sufficiently small (and independent of $N$) such that $\rho_n \in J(2\eta)$ for all $n \in I_N(\delta)$. Indeed, if we define, for each $j = 1, \dots, k$,
\begin{align*}
\tilde{a}_j(\delta) &:= \inf\left\{x >a_j: \mu_V([a_j,x])\geq \delta\right\}\\
\text{and} \qquad \tilde{b}_j(\delta) &:= \sup\left\{x<b_j: \mu_V([x,b_j])\geq \delta\right\},
\end{align*}

\noindent then we may choose 
\begin{align*}
\eta := \tfrac{1}{2}\min_{i \in \{1, \dots, k\}} \left\{\min\left(\tilde{a}_i(\delta) - a_i, b_i - \tilde{b}_i(\delta)\right) \right\}.
\end{align*}

We will now prove both rigidity estimates in three steps.
\paragraph{Step 1: Chernoff bound.} By considering a Fisher-Hartwig singularity \eqref{eq:omega} with $p=1$, $\alpha_1 = 0$, $2\pi i\beta_1 = {\gamma} \in \mathbb{R}$ and $t_1 = u$, we obtain from Theorem \ref{th:FHasy} that
\begin{equation} \label{lbegammaH}
\mathbb{E}e^{\gamma \mathcal{H}_N(u)} \le C_{\gamma}(\eta) N^{\frac{\gamma^2}{4\pi^2}}
\end{equation}

\noindent for some constant $C_{\gamma}(\eta) > 0$ uniformly in $u \in J(\eta)$ and $N$ sufficiently large. We observe, using a union bound, that for any fixed $\epsilon > 0$ and finite subset $S \subset J(\eta)$,
\begin{align}
\label{eq:max_bound1} & \mathbb{P}\left(\sup_{x \in S} |\mathcal{H}_N(x)| > \frac{1}{\pi}(1+\epsilon/2) \log N\right)\\
\nonumber &  \le |S| \sup_{x\in S} \Bigg[\left|\mathbb{P}\left(\mathcal{H}_N(x) >\frac{1}{\pi} (1+\epsilon/2) \log N\right)\right|
\\
& \hspace{1.5cm} + \left|\mathbb{P}\left(-\mathcal{H}_N(x)> \frac{1}{\pi}(1+\epsilon/2) \log N\right)\right|\Bigg], \nonumber
\end{align}
\noindent and using the fact that $\mathbb P(\mathcal H>u)\leq \mathbb E \left(e^{\gamma(\mathcal H-u)}\right)$ for any real random variable $\mathcal H$, $u \in \mathbb R$, and $\gamma>0$ with $\gamma=2\pi$ and $u=\frac{1}{\pi}(1+\epsilon/2) \log N$, we find that the right-hand side of \eqref{eq:max_bound1} converges to $0$
provided that $|S| = \mathcal{O}(N)$. Our choice of $S = \{s_i\}_{i=1}^{cNk} \subset J(\eta)$ is such that
\begin{align*}
a_{j} + \eta = s_{cN(j-1) + 1} < s_{cN(j-1) + 2} < \dots < s_{cNj} = b_j - \eta
\end{align*}

\noindent are $cN$ equally spaced points on $[a_j + \eta, b_j - \eta]$ for each $j = 1, \dots, k$. For $c = c(\mu_V) > 0$ sufficiently large (and independent of $N$), we may assume without loss of generality that any interval of the form $[s_l, s_{l+1}]$ satisfies $\mu_V([s_l, s_{l+1}]) < 1/N$, and in particular contains at most one element in $\{\rho_n: n \in I_N(\delta)\}$.

\paragraph{Step 2: a preliminary rigidity estimate.} We claim that the event
\begin{align*}
\mathcal{G}_N := \left\{ |\lambda_{(n)}-\rho_n| < N^{-1+\epsilon} \quad \text{for all } n \in I_N(\delta) \right\}
\end{align*}

\noindent has a probability tending to $1$ as $N \to \infty$ for any fixed $\epsilon > 0$. This follows from a result of Li \cite{Li2017}, but we provide a self-contained proof below.

Indeed, suppose there exists some $n \in I_N(\delta)$ such that $\lambda_{(n)} > \rho_n + N^{-1+\epsilon}$. Since $\rho_n \in J(2\eta)$, we also have $\rho_n +N^{-1+\epsilon} \in J(\eta)$ for $N$ sufficiently large. If we choose $s_l \in S$ to be the largest element in $S$ such that $s_l <  \rho_n + N^{-1+\epsilon}$, then $s_l$ cannot be one of the ``endpoints" $\{b_j - \eta\}_{j=1}^k$ and hence
\begin{align*}
s_l - \rho_n 
= s_{l+1} - \mathcal{O}(N^{-1}) - \rho_n 
\ge \rho_n + N^{1-\epsilon}  - \mathcal{O}(N^{-1}) - \rho_n
\ge \tfrac{1}{2} N^{-1+\epsilon}
\end{align*}
\noindent for $N$ sufficiently large. In particular,
\begin{align*}
\mathcal{H}_N(s_l)
&=   \sum_{i=1}^N \mathbf{1}\{\lambda_j \le s_l\} - N \mu_V((-\infty, s_l])\\
& \le   \sum_{i=1}^N \mathbf{1}\{\lambda_j \le \lambda_{(n)}\} - N \mu_V((-\infty, \rho_n]) - N\mu_V((\rho_n, s_l])\\
& = - N\mu_V((\rho_n, s_l]) \le - \frac{\sqrt 2}{\pi} C N^{\epsilon}
\end{align*}

\noindent for some constant $C>0$ independent of $N$, and this has a vanishing probability since \eqref{eq:max_bound1} converges to $0$. The event that $\lambda_{(n)} < \rho_n - N^{-1+\epsilon}$ for some $n \in I_N(\delta)$ may be treated similarly.

\paragraph{Step 3: the strong rigidity estimates.} Thanks to the last step, we may assume without loss of generality the event $\mathcal{G}_N$, and in particular $ \lambda_{(n)} \in J(\eta)$ for all $n \in I_N(\delta)$ because
\begin{align*}
& \mathbb{P}\left(\exists n \in I_N(\delta):  \lambda_{(n)} \not \in J(\eta) \right)
\le \mathbb{P}\left(\exists n \in I_N(\delta):  |\lambda_{(n)} - \rho_n| > \eta \right)
\xrightarrow{N \to \infty} 0.
\end{align*}

\noindent But then every eigenvalue $\lambda_{(n)} \in J(\eta)$ must be contained in some closed interval $[s_l, s_{l+1}]$ and this means that
\begin{align*}
\mathcal{H}_N(\lambda_{(n)})
&=  \sum_{i=1}^N \mathbf{1}\{\lambda_j \le \lambda_{(n)} \} - N \mu_V((-\infty, \lambda_{(n)}])\\
& \le   \sum_{i=1}^N \mathbf{1}\{\lambda_j \le s_{l+1} \} - N \mu_V((-\infty, s_{l}])
= \mathcal{H}_N(s_{l+1}) + N\mu_V((s_l, s_{l+1}])
\end{align*}

\noindent where $N\mu_V((s_l, s_{l+1}]) \le N \left(\max_{x\in J} \psi_V(x)\right)|s_{l+1}-s_l| \le C$ for some constant $C> 0$. Since the supremum of $\mathcal{H}_N(\cdot)$ can only be realized at the eigenvalues, we have
\begin{align*}
& \mathbb{P}\left(\sup_{x \in J(\eta)} \mathcal{H}_N(x) > \frac{1}{\pi} (1+\epsilon) \log N\right) 
\\
& \le \mathbb{P}\left(\sup_{n \in I_N(\delta)} \mathcal{H}_N(\lambda_{(n)})> \frac{1}{\pi}(1+\epsilon) \log N\right) \\
& \le \mathbb{P}\left(\sup_{l \le cNk} \mathcal{H}_N(s_l) + C > \frac{1}{\pi}(1+\epsilon) \log N\right) \to 0.
\end{align*}

\noindent A similar argument also holds for bounding $\sup_{x \in J(\eta)} (-\mathcal{H}_N(x))$ and therefore
\begin{align*}
\lim_{N \to \infty} \mathbb{P}\left(\sup_{x \in J(\eta)} |\mathcal{H}_N(x)| >\frac{1}{\pi} (1+\epsilon) \log N\right) = 0.
\end{align*}

Now, for every $n \in I_N(\delta)$,
\begin{align*}
\frac{1}{N}\left| \mathcal{H}_N(\lambda_{(n)})\right|
&=  \left| \frac{n}{N} - \mu_V((-\infty, \lambda_{(n)}]) \right|\\
& =  \left| \mu_V((-\infty, \rho_n]) - \mu_V((-\infty, \lambda_{(n)}]) \right|\\
& =   \left| \int_{\rho_n}^{\lambda_{(n)}} \psi_V(x)dx \right| 
\ge   |\lambda_{(n)} - \rho_n| \min_{x \in [\rho_n, \lambda_{(n)}]^*} \psi_V(x)
\end{align*}

\noindent where we wrote $[\rho_n, \lambda_{(n)}]^*$ to mean the interval $[\rho_n, \lambda_{(n)}]$ if $\rho_n \le \lambda_{(n)}$, and $[\lambda_{(n)}, \rho_n]$ otherwise. In either case, on our good event $\mathcal{G}_N$ this interval is contained in $J(\eta)$ on which $\psi_V(x)$ is uniformly continuous and strictly bounded away from $0$ and hence $\min_{x \in [\rho_n, \lambda_{(n)}]^*} \psi_V(x) \ge (1+\epsilon)^{-1} \psi_V(\rho_n)$ for all $n\in I_N(\delta)$ when $N$ is sufficiently large. Combining everything, we obtain simultaneously for all $n \in I_N(\delta)$
\begin{align*}
|\lambda_{(n)} - \rho_n| 
&\le \frac{1+\epsilon}{ \psi_V(\rho_n)} \frac{|\mathcal{H}_N(\lambda_{(n)})|}{ N}\\
&\le \frac{1+\epsilon}{ \psi_V(\rho_n)} \frac{\max_{x \in J(\eta)} |\mathcal{H}_N(x)|}{ N}
\le \frac{1+\epsilon}{\pi \psi_V(\rho_n)} \frac{\log N}{N}
\end{align*}

\noindent with high probability, and this concludes our proof.
\end{proof}

\subsection{Fisher-Hartwig singularities: a literature overview.} \label{LitOverFH}Asymptotics of structured determinants (such as Toeplitz, Hankel and Fredholm determinants) with FH singularities have attracted considerable attention over the years, and we briefly review here this rich literature. 

 In the pioneering work \cite{FisherHartwig}, Fisher and Hartwig conjectured a formula for the asymptotics of large Toeplitz determinants with root-type and jump-type singularities; see also the subsequent work \cite{Lenard} by Lenard. An important motivation for studying such determinants was to better understand the long range correlation of the Ising model and the momentum of one dimensional impenetrable bosons, see \cite{Ehrhardt, Impetus} for reviews. The Fisher-Hartwig conjecture was proved by Widom \cite{Widom2}, Basor \cite{Basor, Basor2}, B\"ottcher and Silbermann \cite{BS1986}, and Ehrhardt and Silberman \cite{EhrSil}, for parameters where the conjecture was expected to hold. A different type of asymptotics were observed for certain parameter sets (for example when $\Re(\beta_1-\beta_2)=1$) by  Basor and Tracy in \cite{BasTra}, and finally the problem was solved in full generality by Deift, Its and Krasovsky in \cite{DIK, DeiftItsKrasovsky}. The aforementioned  works  deal with FH singularities that are bounded away from each other. In recent years, we have witnessed important progress in understanding the asymptotics of Toeplitz determinants with merging singularities \cite{CIK, CK2015, FahsUniform}.
 
 The works \cite{Krasovsky, Garoni, IK, BWW, Charlier} have already been mentioned earlier in the introduction and concern Hankel determinants with FH singularities in the bulk; see also \cite{CharlierGharakhloo} for further results in this direction. Asymptotic formulas for Hankel determinants with FH singularities have also been obtained in other regimes of the parameters: see \cite{ClaeysFahs} for two merging root-type singularities in the bulk, \cite{BCI2016, WXZ2018} for singularities close to the edges, and \cite{ChDeano} for a large jump-type singularity.
 
 Asymptotics of Fredholm determinants with FH singularities have also been widely studied, see \cite{BW1983, BB1995, BDIK2015, ChSine} for the sine-kernel determinant, \cite{BB2018, ChCl3} for the Airy-kernel determinant, \cite{BIP, ChBessel, ChCl4} for the Bessel-kernel determinant, and \cite{DXZ2020 thinning, ChMoreillon} for the Pearcey-kernel determinant.

Finally, we also mention the works \cite{BasorEhrhardt1, ForKea} on Toeplitz+Hankel determinants, \cite{CharlierMB} on Muttalib-Borodin determinants, and \cite{WebbWong, DeanoSimm, Charlier 2d jumps} for asymptotic results on determinants with ``planar" FH singularities.

All the results mentioned above deal with a one-cut setting; the present work provides the first result on FH asymptotics in the multi-cut regime.

\subsection{Outline of the remainder of the article}\label{sec:outline} $ $

\medskip 

In Section \ref{sec:RHP}, we recall the connection between Hankel determinants, orthogonal polynomials and Riemann-Hilbert (RH) problems. We also discuss some deformations of the symbol $\nu$ we will be considering. Based on these deformations and results we prove later in the article, we prove Theorem \ref{th:pfasy}, Theorem \ref{th:smoothasy}, and Theorem \ref{th:FHasy}. In Section \ref{sec:thetaids}, we recall some background information about Riemann surfaces and theta functions, and develop some identities for $\theta$-functions that will play an important role in our analysis. In Section \ref{sec:trans} we open the lens of our RH problem. In Sections \ref{sec:global} and \ref{sec:Local}, we present the main parametrix and the local parametrices respectively. In Section \ref{sec:snorm}, we show that using the parametrices, we can transform the actual problem we are interested in into one that can be solved asymptotically. In Section \ref{sec:nonsing}, we use our parametrices and asymptotic solution to prove the main estimates required for Theorem \ref{th:smoothasy}. In Section \ref{sec:FH}, we use Theorem \ref{th:smoothasy} as well as our parametrices and asymptotic solution to prove the main estimates required for Theorem \ref{th:FHasy}. In Section \ref{sec:special}, we analyze $H_N(e^{-NV})$ for a very special $V$, for which things can be solved rather explicitly. In Section \ref{sec:pasy}, we finally prove our main estimates for Theorem \ref{th:pfasy} in full generality by using the results of Section \ref{sec:special} and our parametrices and asymptotic solution.

\medskip 

{\bf Acknowledgements: } We are grateful to the referees for useful remarks. C.C. was supported by the European Research Council, Grant Agreement No. 682537, the Swedish Research Council, Grant No. 2015-05430, the Ruth and Nils-Erik Stenb\"ack Foundation, and the Novo Nordisk Fonden Project Grant 0064428. B.F. was supported by the G\"oran Gustafsson Foundation (UU/KTH),  by the Leverhulme Trust research programme grant RPG-2018-260, and by project KAW 2015.0270. C.W. was supported by the Academy of Finland grant 308123 and the Ruth and Nils-Erik Stenb\"ack Foundation. M.D.W. was supported by ERC Advanced Grant 740900 (LogCorRM).

\section{Differential identities and overview of the proofs}\label{sec:RHP}

The goal of this section is to give an overview of our strategy to prove Theorems \ref{th:pfasy}, \ref{th:smoothasy} and \ref{th:FHasy}. We first review the connections between Hankel determinants, orthogonal polynomials and Riemann-Hilbert (RH) problems, and then we present three differential identities for $\log H_{N}(\nu)$.

\subsection{Orthogonal polynomials and a RH problem} Let $\nu$ be a non-negative function on $\mathbb R$, whose support has positive Lebesgue measure, and which is  H\"older continuous and integrable with respect to $x^jdx$ for any $j\in \mathbb N:=\{0,1,2,...\}$. We consider the family of orthonormal polynomials $\{p_j(z)=\kappa_j z^j+\dots\}_{j\geq 0}$ characterized by
\begin{equation} \label{eq:ortho}
\int_\R p_j(x)p_l(x)\nu(x)dx=\delta_{j,l}, \qquad j,l = 0,1,2,\ldots,
\end{equation}
where the degree of $p_{j}$ is $j$ and its leading coefficient $\kappa_{j}$ is positive. Since $\nu$ is non-negative and its support has  positive Lebesgue measure, $p_j$ exist by the Gram-Schmidt orthogonalization procedure. The connection to Hankel determinants comes from the well-known fact that (see e.g. \cite[Chapter II]{Szego})
\begin{equation}\label{eq:hprod}
H_N(\nu)=\prod_{j=0}^{N-1}\kappa_j(\nu)^{-2}.
\end{equation}
This formula expresses $H_N(\nu)$ in terms of the leading coefficients of $p_{0}(z),$ $p_{1}(z),$ $\ldots$, $p_{N-1}(z)$. In the next subsections, we will obtain differential identities which express certain $\log$-derivatives of $H_N(\nu)$ in terms of $p_{N-1}(z)$ and $p_{N}(z)$ only. 
Let us define for $z\notin \R$
\begin{equation}\label{eq:Ydef}
Y(z)=Y_N(z;\nu)=\begin{pmatrix}
\frac{1}{\kappa_N}p_N(z) & \frac{1}{2\pi i \kappa_N}\int_\R \frac{p_N(x)}{x-z}\nu(x)dx\\
-2\pi i \kappa_{N-1}p_{N-1}(z) & -\kappa_{N-1}\int_\R \frac{p_{N-1}(x)}{x-z}\nu(x)dx
\end{pmatrix}.
\end{equation}
It was noticed in \cite{FIK} (see also e.g. \cite[Chapter 3.2]{Deift}) that $Y$ is the unique solution to the following RH problem.

\subsubsection*{The RH problem for $Y$} 
\begin{itemize}[leftmargin=0.75cm]
\item[$(a)$] $Y: \C\setminus \R\to \C^{2\times 2}$ is analytic.
\item[$(b)$] $Y$ has continuous boundary values on $\R $, namely the limits $Y_\pm(x)=\lim_{\epsilon \to 0^+}Y(x\pm i\epsilon)$ exist and $x\mapsto Y_\pm(x)$ are continuous functions on $\R$. Moreover, $Y_{+}$ and $Y_{-}$ are related by the following jumps
\begin{equation}\label{eq:Yjump}
Y_+(x)=Y_-(x)\begin{pmatrix}
1 & \nu(x)\\
0 & 1
\end{pmatrix} \qquad \text{for} \qquad x \in \R.
\end{equation}
\item[$(c)$] As $z\to \infty$, $z\in \C\setminus \R$, we have
\begin{equation}\label{eq:Yinfty}
Y(z)=(I+\mathcal O(z^{-1}))\begin{pmatrix}
z^N & 0\\
0 & z^{-N}
\end{pmatrix}.
\end{equation}
\end{itemize}

By relying on the Deift-Zhou steepest descent analysis, a method first developed in \cite{DZ}, the RH problem for $Y$ may be analyzed asymptotically as $N\to \infty$ for a wide range of symbols $\nu$. Of particular importance to us in this regard, is the analysis of $Y_N(z;e^{-NV})$ for multi-cut regular $V$ undertaken by Deift et. al. \cite{DKMVZ}, the work by Claeys, Its and Krasovsky on RH problems with emerging singularities \cite{CIK}, and the analysis of $Y$ in the multi-cut setting with a singularity at the origin undertaken by Kuijlaars and Vanlessen \cite{KuijVanlessen}.

\subsection{Differential identities}\label{sec:DI}
Let $\hat{s}_{1},\hat{s}_{0},\hat{s} \in \mathbb{R}$ be parameters satisfying $\hat{s}_{1}>\hat{s}_{0}$ and $\hat{s} \in [\hat{s}_{0},\hat{s}_{1}]$. Assume that $\nu=\nu_{\hat{s}}$ depends on $\hat{s}$, and that for all $x$, $\nu_{\hat{s}}(x)$ is differentiable with respect to $\hat{s} \in (\hat{s}_{0},\hat{s}_{1})$. In this subsection, we derive a general differential identity which expresses $\partial_{\hat{s}} \log H_{N}(\nu_{\hat{s}})$ in terms of $Y_{N}(\cdot;\nu_{\hat{s}})$, which we will subsequently rely on in three distinct settings to prove Theorems \ref{th:pfasy}, \ref{th:smoothasy} and \ref{th:FHasy}.

Let us write $p_j(x)=p_j^{(\hat{s})}(x)=\kappa_j^{(\hat{s})}x^j+\dots$ with $\kappa_j^{(\hat{s})}>0$ for the orthonormal polynomials with respect to $\nu_{\hat{s}}$. It readily follows from \eqref{eq:ortho} that
\begin{equation*} 
-2 \frac{1}{\kappa_j^{(\hat{s})}}\frac{\partial \kappa_j^{(\hat{s})}}{\partial \hat{s}} = -2\int_{\mathbb R}p_j(x)\left(\frac{\partial}{\partial \hat{s}}p_j(x)\right)\nu_{\hat{s}}(x)dx=\int_{\mathbb R} p_{j}^{2}(x)\frac{\partial }{\partial \hat{s}} \nu_{\hat{s}}(x)dx.
\end{equation*}
Hence, taking the log in \eqref{eq:hprod} and then differentiating with respect to $\hat{s}$, we obtain
\begin{multline} \label{eq:DI_gen}
\frac{\partial}{\partial \hat{s}} \log H_N(\nu_{\hat{s}}) =-\int_\R\left[\frac{\partial}{\partial \hat{s}}\sum_{j=0}^{N-1}p_j(x)^2\right] \nu_{\hat{s}}(x)dx
 \\
=\int_\R\left[\sum_{j=0}^{N-1}p_j(x)^2\right]\frac{\partial}{\partial \hat{s}} \nu_{\hat{s}}(x) dx.
\end{multline}
By the well-known Christoffel-Darboux identity (see e.g. \cite[equation (3.49)]{Deift}), 
\begin{multline}\label{CD formula Y}
\sum_{j=0}^{N-1} p_j(x)^2 = \frac{\kappa_{N-1}^{(\hat{s})}}{\kappa_{N}^{(\hat{s})}} \left[p_{N}'(x) p_{N-1}(x) - p_{N-1}'(x) p_{N}(x)\right] \\
=\frac{1}{2\pi i}[Y(x)^{-1} Y'(x)]_{21},
\end{multline}
where $Y=Y_{N}(\cdot;\nu_{\hat{s}})$. Substituting \eqref{CD formula Y} in \eqref{eq:DI_gen}, we find
\begin{equation}\label{der of log Hn on the real line}
\frac{\partial}{\partial \hat{s}} \log H_N(\nu_{\hat{s}})
=\frac{1}{2\pi i}\int_\R[Y(x)^{-1} Y'(x)]_{21}\frac{\partial}{\partial \hat{s}} \nu_{\hat{s}}(x) dx, \qquad \hat{s} \in (\hat{s}_{0},\hat{s}_{1}). 
\end{equation}
We will find it convenient to deform the  contour of integration. Let us assume that there are $l$  open intervals $I_1,\dots,I_l$  such that $\log \nu_{\hat{s}}$ is defined on $I_1,\dots, I_l$ and that  $\frac{\partial}{\partial \hat{s}}\log \nu_{\hat{s}}$ extends to an analytic function (in the variable $x$) on an open neighbourhood of each interval $I_j$, $j=1,\ldots,l$. In what follows, we will take either $l=1$ or $l=k$ depending on the differential identity we are considering. Within the open neighbourhood of $I_{j}$, we let $\Gamma_j$ be a smooth, closed curve, oriented counter-clockwise, which encloses $I_j$ and intersects the real line only at the two endpoints of $I_j$. Denote $\Gamma=\cup_{j=1}^l\Gamma_j$ and $I=\cup_{j=1}^l I_j$.

By the jump conditions \eqref{eq:Yjump} of $Y$,
\begin{align}
\label{eq:DI_gen2}
[Y(x)^{-1} Y'(x)]_{21}\, \nu_{\hat{s}}(x)=[Y^{-1}Y']_{11,-}(x)-[Y^{-1}Y']_{11,+}(x).
\end{align}
After performing a contour deformation in \eqref{der of log Hn on the real line} using \eqref{eq:DI_gen2}, we obtain
\begin{equation}\label{eq:DI}
\frac{\partial}{\partial \hat{s}} \log H_N(\nu_{\hat{s}})=
 \oint_{\Gamma} [Y(z)^{-1}Y'(z)]_{11}\frac{\partial}{\partial \hat{s}}\log \nu_{\hat{s}}(z) \frac{dz}{2\pi i}+r(\nu_{\hat{s}}), 
\end{equation}
where $ \hat{s} \in (\hat{s}_{0},\hat{s}_{1}),$
where $Y=Y_{N}(\cdot;\nu_{\hat{s}})$, and where
\begin{equation} \label{defrnu}
r(\nu_{\hat{s}})=\int_{\R\setminus I} [Y(x)^{-1}Y'(x)]_{21}\frac{\partial}{\partial \hat{s}} \nu_{\hat{s}}(x)\frac{dx}{2\pi i}.\notag \end{equation}

The advantage of expressing $\frac{\partial}{\partial \hat{s}} \log H_N(\nu_{\hat{s}})$ in terms of $Y_{N}(\cdot;\nu_{\hat{s}})$ is that the RH problem for $Y_{N}$ can be asymptotically analyzed as $N \to + \infty$. To prove Theorems \ref{th:pfasy}, \ref{th:smoothasy} and \ref{th:FHasy}, we will use \eqref{eq:DI} in three different ways. We provide a detailed outline of this in the next subsections.

\subsection{Smooth ratio asymptotics: outline of the proof of Theorem \ref{th:smoothasy}}\label{sec:smasy}
Suppose $F$ and $f$ are functions satisfying the conditions in Theorem \ref{th:smoothasy}, and that $F$ is H\"older continuous on $\mathbb R$. Then there is a union of $k$ open intervals $I=\cup_{j=1}^kI_j$ such that $f(x)=\log F(x)$ is well-defined and analytic on a neighbourhood of $I$, and $[a_j,b_j]\subset I_j$ for $j=1,\dots, k$. For $t\in[0,1]$, we define

\begin{equation} \label{deformnut} \nu_t(x)=F_t(x)e^{-NV(x)}, \qquad F_t(x)=\begin{cases}e^{tf(x)}&{\rm for } \,\,x\in I,\\
g_t(x)&{\rm for }\,\,x\in \mathbb R\setminus  I,
\end{cases} \end{equation}
where $g_t$ is given on $\{x\in \R: \mathrm{dist}(x,I)<\delta\}$ by
\begin{equation*}
g_t(x)  = 
e^{tf(x)}+\frac{\mathrm{dist}(x,I)}{\delta}\left(1-t+tF(x)-e^{tf(x)}\right), 
\end{equation*} 
 and otherwise, on  $\{x\in \R: \mathrm{dist}(x,I)\geq\delta\}$, 
\begin{equation*}
g_t(x)=1-t+tF(x).
\end{equation*}
In the above, $\delta > 0$ is a small but fixed constant, for which $f=\log F$ is analytic in a $\delta$-neighborhood of $I$ (i.e. the set $\{z\in \C: \mathrm{dist}(z,I)<\delta\}$) . The function $g_t$ is added to ensure that $F_t$ is H\"older continuous.

Note that $\nu_{t}(x)$ is differentiable in $t \in (0,1)$ for each $x \in \mathbb{R}$, satisfies $\log \nu_t(x)=-NV(x)+tf(x)$ for $x\in I$, and
\begin{align*}
\nu_0(x)=e^{-NV(x)}, \qquad \nu_1(x)=F(x)e^{-NV(x)}, \qquad x \in \mathbb{R}.
\end{align*}
Hence, we can apply \eqref{eq:DI} with $l=k$, $\hat{s}=t$, $\hat{s}_{0}=0$, $\hat{s}_{1}=1$ and $\nu_{\hat{s}}=\nu_{t}$. Integrating this identity in $t$ from $t=0$ to $t=1$, we obtain
\begin{equation}\label{eq:DIsmasy}\begin{aligned}
\log \frac{H_N(Fe^{-NV})}{H_N(e^{-NV})}&=\int_0^1
 \oint_{\Gamma} [Y(z)^{-1}Y'(z)]_{11}f(z) \frac{dz}{2\pi i}dt+\int_0^1r(\nu_t) dt, \\
r(\nu_t)&=\int_{\R\setminus I} [Y(x)^{-1}Y'(x)]_{21}\frac{\partial}{\partial t} F_t(x)e^{-NV(x)}\frac{dx}{2\pi i},
\end{aligned}
\end{equation}
where $Y=Y_{N}(\cdot;\nu_{t})$.

In Section 8 we prove that $r(\nu_t)\to 0$ as $N\to \infty$ uniformly for $t\in (0,1)$, and we obtain large $N$ asymptotics for $Y(z)^{-1}Y'(z)$ uniformly for $z\in \Gamma$ and $t\in [0,1]$. Theorem \ref{th:smoothasy} will then be proved by substituting these asymptotics in \eqref{eq:DIsmasy} and performing the integrations in $z$ and $t$.

\subsection{Fisher--Hartwig asymptotics: outline of the proof of Theorem \ref{th:FHasy}}\label{sec:FHproof}
Let $\epsilon_{0}>0$ be fixed. For $\epsilon\in(0,\epsilon_{0}]$, let $\nu_\epsilon(x)=F(x)e^{-NV(x)}\omega_\epsilon(x)$, where $\omega_\epsilon(x)$ is defined by
\begin{equation}\begin{aligned}\label{eq:omega_eps}
\omega_{\epsilon}(z) &= \prod_{j=1}^p \omega_{\alpha_j, \epsilon}(z)\omega_{\beta_j, \epsilon}(z)\\
\omega_{\alpha_j, \epsilon}(z) &= (z - (t_j - i\epsilon))^{\frac{\alpha_j}{2}}(z - (t_j + i\epsilon))^{\frac{\alpha_j}{2}}, \\
\omega_{\beta_j, \epsilon}(z) &= (z - (t_j - i\epsilon))^{\beta_j}(z - (t_j + i\epsilon))^{-\beta_j}e^{-i \pi \beta_j},
\end{aligned}\end{equation}
with branches chosen such that $\omega_\epsilon(z)$ is real for $z\in \mathbb R$  and analytic for $z\in \mathbb C\setminus \Sigma_{\omega_\epsilon}$, where $\Sigma_{\omega_\epsilon}$ consists of contours connecting $t_j+i\epsilon$ to $+i\infty$  and $t_{j}-i\epsilon$ to $-i\infty$.\footnote{There is some freedom in how the contour is chosen. However, in a neighbourhood of $t_j\pm i \epsilon$, we have to be  careful about the definition of the contour -- this is done in Section \ref{sec:Painleve} below equation \eqref{eq:sjdef} as the preimage of the jump contour of a model RH problem under a certain conformal mapping. Outside of this neighbourhood, we can choose it to be say smooth, not having self intersections, and so that none of the contours intersect each other.} We also define $\omega_{0}(x)$ by $\omega_{0}(x)=\lim_{\epsilon\to 0}\omega_\epsilon(x)$. The definition \eqref{eq:omega_eps} of $\omega_\epsilon(x)$ ensures that $\omega_{0}(x)=\omega(x)$, where $\omega$ is defined in \eqref{eq:omega}. We also let
\begin{multline}\label{logomegaepsilon}
\log \omega_\epsilon(z)=\sum_{j=1}^p\bigg[ (\alpha_j/2+\beta_j)\log(z-(t_j-i\epsilon))\\  +(\alpha_j/2-\beta_j)\log(z-(t_j+i\epsilon))-\pi i \beta_j\bigg],
\end{multline}
be analytic in $\mathbb{C}\setminus \Sigma_{\omega_\epsilon}$. The branch of $\log(z-(t_j-i\epsilon))$ is such that the cut is a path from $t_j-i\epsilon$ to $-i\infty$ and the argument is chosen to satisfy $\arg (z-(t_j-i\epsilon)) = 0$ whenever $z-(t_j-i\epsilon)>0$. Similarly, the branch of $\log(z-(t_j+i\epsilon))$ is such that the cut is a path from $t_j+i\epsilon$ to $i\infty$ and the argument is chosen to satisfy $\arg (z-(t_j+i\epsilon)) = 0$ whenever $z-(t_j+i\epsilon)>0$.

Clearly, $\nu_{\epsilon}(x)$ is differentiable in $\epsilon \in (0,\epsilon_{0})$ for each $x \in \mathbb{R}$, and $\frac{\partial }{\partial \epsilon}\log \nu_\epsilon(x)$ extends to a meromorphic function in $x \in \mathbb{C}\setminus \{t_{j}+ i\epsilon,t_{j}- i\epsilon\}_{j=1}^{p}$. In particular, for each $\epsilon \in (0,\epsilon_{0})$, $x \mapsto \frac{\partial }{\partial \epsilon}\log \nu_\epsilon(x)$ is analytic in a neighbourhood of $\mathbb{R}$. Hence, we can apply \eqref{eq:DI} with $l=1$, $I_{1}=(-R,R)$, $\hat{s}=\epsilon$, $\hat{s}_{0}=0$, $\hat{s}_{1}=\epsilon_{0}$, $\nu_{\hat{s}}=\nu_{\epsilon}$ and $R$ is a new parameter that is chosen sufficiently large so that $[a_{1},b_{k}]\subset I_{1}$. In this situation, the contour $\Gamma=\Gamma_{1}$ appearing in \eqref{eq:DIsmasy} encloses $I_1$, and passes  between the singularities $\{t_{j}+ i\epsilon,t_{j}- i\epsilon\}_{j=1}^{p}$ and the real line, so that $z\mapsto \frac{\partial }{\partial \epsilon}\log \nu_\epsilon(z)$ is analytic on the interior of $\Gamma_1$. We also observe that $\left[Y(z)^{-1}Y'(z)\right]_{11}$ is analytic for $z\in \mathbb C\setminus \mathbb R$. Hence, by deforming $\Gamma_1$ to a circle $B_R$ of radius $R$ centered at the origin, we pick up some residue contributions of the poles at $t_j\pm i\epsilon$, and we obtain
\begin{align*}
& \partial_\epsilon \log H_N(\nu_{\epsilon})  = r(\nu_\epsilon)\\
& + i\sum_{j=1}^p \left[ \left(\frac{\alpha_j}{2} - \beta_j\right) \left(Y^{-1} Y'\right)_{11}(t_j + i\epsilon) - \left(\frac{\alpha_j}{2} + \beta_j\right) \left(Y^{-1} Y'\right)_{11}(t_j - i\epsilon)\right]\\
& +\oint_{B_R} [Y(z)^{-1}Y'(z)]_{11}\sum_{j=1}^p \bigg[ \Big( \frac{\alpha_{j}}{2}-\beta_{j} \Big) \frac{-i}{z-(t_{j}+i\epsilon)} \\
& \hspace{4.2cm} + \Big( \frac{\alpha_{j}}{2}+\beta_{j} \Big) \frac{i}{z-(t_{j}-i\epsilon)} \bigg] \frac{dz}{2\pi i},
\end{align*}
where $Y=Y_{N}(\cdot;\nu_{\epsilon})$ and
\begin{align*}
r(\nu_\epsilon)&=\int_{\R\setminus (-R,R)} [Y(x)^{-1}Y'(x)]_{21} \frac{\partial}{\partial \epsilon} \nu_{\epsilon}(x) \frac{dx}{2\pi i}.
\end{align*}

Taking $R\to \infty$, we observe that $r(\nu_\epsilon)\to 0$, and by condition (c) in the RH problem for $Y$, the integral over $B_R$ also converges to zero. Thus, for $\epsilon \in (0,\epsilon_{0})$, we obtain
\begin{multline}\label{eq:DIFH}
\partial_\epsilon \log H_N(\nu_\epsilon)  \\ =  i\sum_{j=1}^p \left[ \left(\frac{\alpha_j}{2} - \beta_j\right) \left(Y^{-1} Y'\right)_{11}(t_j + i\epsilon) - \left(\frac{\alpha_j}{2} + \beta_j\right) \left(Y^{-1} Y'\right)_{11}(t_j - i\epsilon)\right],
\end{multline}
where $Y=Y_{N}(\cdot;\nu_{\epsilon})$. Integrating this identity from $\epsilon=0$ to $\epsilon=\epsilon_{0}$ yields
\begin{align}
\label{eq:FHrat intro}
\frac{H_N(\nu_0)}{H_N(e^{-NV})}=\frac{H_N(\nu_{\epsilon_0})}{H_N(e^{-NV})}\exp\left(-\int_0^{\epsilon_0} \partial_\epsilon \log H_N(\nu_\epsilon)d\epsilon\right).
\end{align}
Since $x\mapsto \nu_{\epsilon_0}(x)$ is a H\"older continuous function on $\mathbb{R}$, the asymptotics of $H_N(\nu_{\epsilon_0})/H_N(e^{-NV})$ as $N \to + \infty$ are given by Theorem \ref{th:smoothasy}. Considering this limit is the topic of Section \ref{Sec:RatioFH}. We will obtain an extra simplification by assuming that $\epsilon_0$ is fixed but small in the right hand side of \eqref{th:smoothasy1} in exchange for an error term $\mathcal O(\epsilon_0)$, resulting in \eqref{eq:rat2} which provides asymptotics of the form
\begin{equation}\label{doublelim1}
\frac{H_N(\nu_{\epsilon_0})}{H_N(e^{-NV})}=\textrm{Main Asymptotics}\times \left(1+\mathcal O \left(N^{-1}\right)+\mathcal O \left( \epsilon_0 \right) \right).
\end{equation}
The implicit constants in the $\mathcal O(N^{-1})$ term depend on $\epsilon_0$ while the implicit constants in the $\mathcal O(\epsilon_0)$ term are independent of both $N$ and $\epsilon_0$.  See \eqref{eq:rat2} for details, for now we give an overview of the structure of the error terms. 

Now consider the second term on the right hand side of \eqref{eq:FHrat intro}. We need to determine the large $N$ asymptotics of $\partial_\epsilon \log H_N(\nu_\epsilon)$ uniformly in $\epsilon \in [0,\epsilon_{0}]$. Note that $\nu_{\epsilon}$ is a regular weight for each $\epsilon \in (0,\epsilon_{0}]$, while $\nu_{0}$ is singular. In fact, a critical transition takes place in the asymptotics of $H_{N}(\nu_{\epsilon})$ as $N \to + \infty$ and simultaneously $\epsilon \to 0$. A similar phase transition was studied in \cite{CIK}. In Section \ref{sec:FH}, using \eqref{eq:DIFH} and a local analysis from \cite{CIK}, we obtain uniform asymptotics of $\partial_\epsilon \log H_N(\nu_\epsilon)$ up to a term of order $\mathcal O(1)$. Thus, integrating with respect to $\epsilon$ and taking the exponential, we obtain in \eqref{eq:intermediate} asymptotics of the form
\begin{equation}\label{doublelim2} \exp\left(-\int_0^{\epsilon_0} \partial_\epsilon \log H_N(\nu_\epsilon)d\epsilon\right)=\textrm{Main Asymptotics}\times \left(1+\mathcal O(\epsilon_0)\right), \end{equation}
where the implicit constants are independent of $N$. Substituting \eqref{doublelim1} and \eqref{doublelim2} into \eqref{eq:FHrat intro} we obtain asymptotics of the form
\begin{equation*} \frac{H_N(\nu_0)}{H_N(e^{-NV})}=\textrm{Main Asymptotics} \times \left(1+\mathcal O \left(N^{-1}\right)+\mathcal O \left( \epsilon_0 \right)\right),
\end{equation*}
with error terms as in \eqref{doublelim1}.
For Theorem \ref{th:FHasy} we would like to replace $\mathcal O \left(N^{-1}\right)+\mathcal O \left( \epsilon_0 \right)$ with $o(1)$. This is possible because for any $\delta>0$ there is an $\epsilon_0$ such that $\mathcal O \left( \epsilon_0 \right)<\delta/2$ and an  $N^{*}$ (depending on $\epsilon_0$) such that for $N>N^{*}$ we have $\mathcal O \left(N^{-1}\right)<\delta/2$. This concludes the overview of the limits taken to obtain Theorem \ref{th:FHasy}, for the explicit formulas see Section \ref{sec:FH}.

\medskip The method that was used in \cite{Krasovsky, IK, Charlier, CharlierGharakhloo} to obtain asymptotics of Hankel determinants with a one-cut regular potential and Fisher-Hartwig singularities proceeds via differential identities with respect to the parameters $\{\alpha_{j},\beta_{j}\}_{j=1}^{p}$. While it is, in principle, also possible to apply the same strategy to our situation, in practice this would represent a considerable challenge. Indeed, in the multi-cut regime, the large $N$ analysis of $Y_{N}(\cdot;\nu)$ involves Riemann $\theta$-functions that depend on the Fisher-Hartwig singularities in a complicated way. For this reason, the method we use here for the multi-cut regime differs substantially from these earlier works; instead of deforming in the Fisher-Hartwig parameters $\{\alpha_{j},\beta_{j}\}_{j=1}^{p}$, we use \eqref{eq:DIFH}, where $\nu_{\epsilon}$ is a deformation of $\nu$ in the locations $\{t_{j}\}_{j=1}^{p}$ of the singularities. 

\subsection{Partition function asymptotics: outline of the proof of Theorem \ref{th:pfasy}}\label{sec:pfasy}
Let $V$ be an arbitrary $k$-cut regular potential. To obtain the large $N$ asymptotics of $H_{N}(e^{-NV})$, we will consider a family of $k$-cut regular potentials $\{V_{s}\}_{s \in [0,2]}$ that (piecewise) smoothly interpolates between $V_{2}:=V$ and a reference potential $V_{0}$, and we will apply \eqref{eq:DI} with $l=k$, $\hat{s}=s$, $\hat{s}_{0}=0$, $\hat{s}_{1}=2$ and $\nu_{\hat{s}}=e^{-NV_{s}}$. After integrating this identity in $s$ from $s=0$ to $s=2$, we obtain
\begin{align}\label{sec2lol8}
& H_{N}(e^{-NV}) = H_{N}(e^{-NV_{0}}) \exp \bigg( \int_{0}^{2} \partial_{s}\log H_{N}(e^{-NV_{s}})ds \bigg), \\
& \frac{\partial}{\partial s} \log H_N(e^{-NV_{s}})=
 -N\oint_{\Gamma} [Y(z)^{-1}Y'(z)]_{11}\frac{\partial}{\partial s}V_{s}(z) \frac{dz}{2\pi i}+r(e^{-NV_{s}}),  \nonumber
\end{align}
for $ s \in (0,2)$,
where $Y=Y_{N}(\cdot;e^{-NV_{s}})$. As can be seen from the right-hand side of \eqref{sec2lol8}, it is important to choose $V_{0}$ appropriately so that the large $N$ asymptotics of $H_{N}(e^{-NV_{0}})$ can be computed as accurately as possible. In the one-cut case, one can simply choose $V_{0}(x)=2x^{2}$, because in this case $H_{N}(e^{-NV_{0}})$ reduces to a Selberg integral and can thus be evaluated exactly. This fact was implicitly used in \cite{BWW}. In the multi-cut regime, the task of finding $V_{0}$ is much more challenging. For the two-cut case, Claeys, Grava and McLaughlin in \cite[Proposition 2.1]{CGML} used $V_{0}(x)=x^{4}-4x^{2}$ for their reference potential, however their approach for the reference potential does not generalize to the $k$-cut setting when $k\geq 3$. We now describe our choice of $V_{0}$.

Let $T_k(x)=2^{k-1}x^k+\dots$ be the $k$th Chebyshev polynomial of the first kind, namely the unique polynomial of degree $k$ satisfying
\begin{equation}\label{eq:Tkcos}
T_k(\cos \theta)=\cos k\theta, \qquad \theta\in [0,\pi].
\end{equation} 
Below, we prove that the potential $V_0(x)=V_0(x;k,\sigma)=\frac{2\sigma}{k}T_k(x)^2$ with $\sigma>1$ is $k$-cut regular. In Section \ref{sec:special}, as a first step in proving Theorem \ref{th:pfasy}, we   compute  the large $N$  asymptotics of $H_{N}\left(e^{-NV_{0}}\right)$ for $\sigma>1$. However, we are not immediately able to fully simplify these asymptotics, but they do simplify in the limit $\sigma\to 1$. More precisely, we will prove that
\begin{multline}\label{LimCheby}
\lim_{\sigma\downarrow 1} \limsup_{N\to \infty}
\Bigg|\log H_{N}\left(e^{-NV_{0}}\right)+N^2\left(\frac{3}{4k}+\frac{1}{2k}\log \sigma +\log 2\right)\\-N\log (2\pi)+\frac{k}{12} \log N-\frac{k-3}{12}\log k-k\zeta'(-1) 
+\frac{k-1}{8}\log (\sigma-1)
\Bigg|= 0.
\end{multline}
An important aspect when proving \eqref{LimCheby} is that we always view $\sigma>1$ as fixed when taking the limit $N\to \infty$, and subsequently take the limit $\sigma\to 1$.

The second step in the proof of Theorem \ref{th:pfasy} is to obtain the large $N$ asymptotics of the integral appearing on the right-hand side of \eqref{sec2lol8}. To this end, we now introduce the family $\{V_{s}\}_{s \in [0,2]}$ of $k$-cut regular potentials that we consider. The interpolation between $V_{0}$ and $V$ is done in two steps: for $s \in [0,1]$, $x\mapsto V_{s}(x)$ is a polynomial of degree $2k$ and the support ${\rm supp}\, \mu_{V_s}$ varies with $s$ in such a way that ${\rm supp}\, \mu_{V_1}={\rm supp}\, \mu_V$. For $s \in [1,2]$, the potential $x\mapsto V_{s}(x)$ is no longer necessarily a polynomial, ${\rm supp}\, \mu_{V_{s}}={\rm supp}\, \mu_{V}$ remains unchanged, and the relation $V_{2} \equiv V$ holds. We now describe $\{V_{s}\}_{s \in [0,2]}$ in more detail. 

By \eqref{eq:Tkcos}, if $\sigma>1$, $\sigma T_k(x)^2-1$ has $2k$ zeros on the interval $(-1,1)$, and we label them by $ a_1(0)< b_1(0)< a_2(0)<\dots<  a_k(0)< b_k(0)$, see also Figure \ref{fig: abhat and xihat}. As discussed before Corollary \ref{RemarkPoly1} (i), $V_0(x)=\frac{2\sigma}{k}T_k(x)^2$ is $k$-cut regular,
\begin{align}\label{support of tchebyshev}
J_{0}:={\rm supp}\, \mu_{V_{0}} = \cup_{j=1}^{k}[a_{j}(0),b_{j}(0)] = \{x:\sigma T_k(x)^2\leq 1\},
\end{align}
and by \eqref{eqmeasPoly},
\begin{equation*}
d\mu_{V_0}(x)=\frac{2\sigma}{\pi k}|T_k'(x)|\sqrt{1/\sigma-T_k(x)^2}dx, \qquad x \in  J_0.
\end{equation*}

\begin{figure}[h!]
\begin{center}
\includegraphics[scale=0.4]{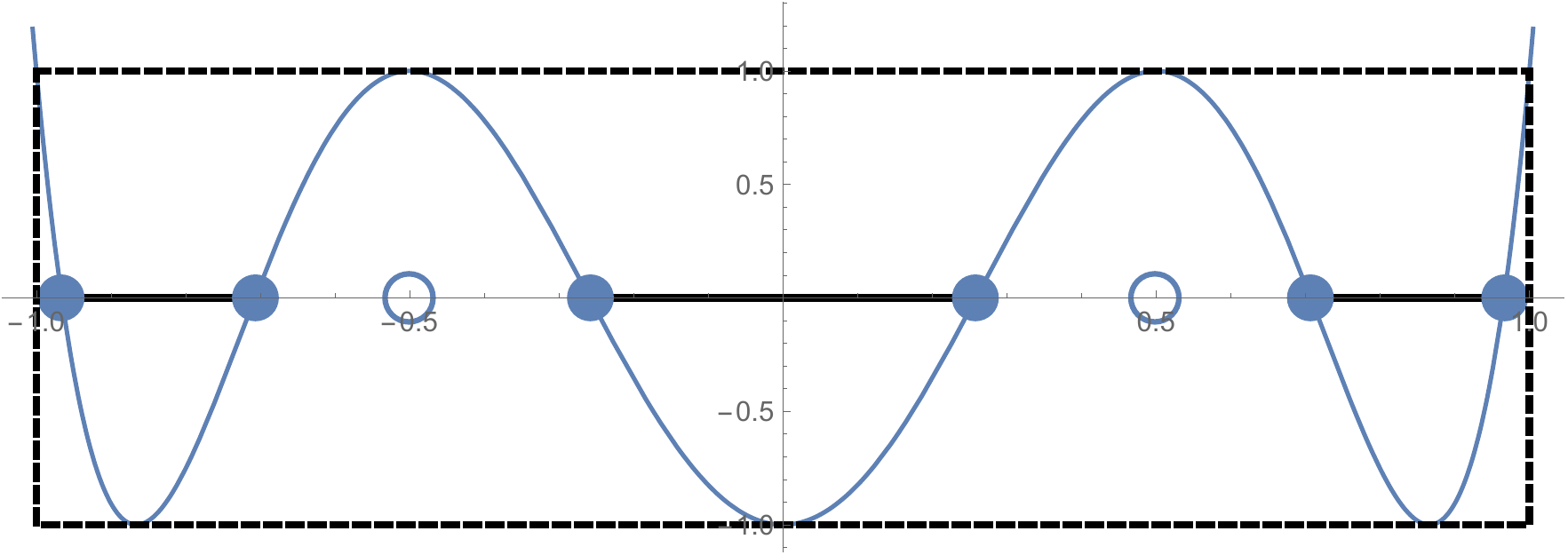}
\end{center}
\caption{\label{fig: abhat and xihat}The blue curve is $2T_{3}(x)^{2}-1$, the solid dots are $a_{j}(0),b_{j}(0)$, $j=1,2,3$, and the hollow dots are the zeros of $T_{3}'$. The dashed box has corners at $(\pm 1, \pm 1)$. The thick black lines represent $J_{0}$ with $\sigma=2$.}
\end{figure}

For $s\in [0,1]$, define
\begin{equation}\label{eq:Js}
J_s=\bigcup_{j=1}^k [a_j(s),b_j(s)]\quad \text{with} \quad \begin{cases}
a_j(s)=(1-s)a_j(0)+s a_j,\\
b_j(s)=(1-s)b_j(0)+sb_j,
\end{cases}
\end{equation}
where $a_j=a_{j}(1),b_j=b_{j}(1)$, $j=1,\ldots,k$ are the endpoints of ${\rm supp}\, \mu_{V}$. Note that the intervals $[a_j(s),b_j(s)]$, $j=1,\ldots,k$ remain disjoint for each $s\in [0,1]$. Denote 
\begin{equation} \label{def:Rs} \mathcal R_s(w)=\prod_{j=1}^k(w-a_j(s))(w-(b_j(s)), \end{equation} and let $\mathcal R_s^{1/2}(z)$ have branch cuts on $J_s$ and be positive for $z>b_k(s)$. In Section \ref{sec:deform} we prove the following proposition, which given a disjoint union of intervals $J=\cup_{j=1}^k[a_j,b_j]$ provides an explicit potential $\mathcal{V}_J$ whose equilibrium measure $\mu_{\mathcal{V}_J}$ is supported on $J$. We will subsequently apply this proposition to $J_s$ defined by \eqref{eq:Js}.

\begin{proposition}
\label{pr:defo1} Given a disjoint union of $k$ intervals $J=\bigcup_{j=1}^k[a_j,b_j]$ and a real constant $\widehat V \in \mathbb R$, the following claims hold.
\begin{itemize}
\item[(a)] There is a unique real monic polynomial $q$ of degree $k-1$ satisfying
\begin{equation} \label{zerointxil} 
\int_{b_j}^{a_{j+1}}q(x)\mathcal R^{1/2}(x)dx=0, \qquad j=1,\dots,k-1.
\end{equation}
The polynomial $q$ has a root in each interval $(b_j,a_{j+1})$, $j=1,\dots,k-1$.
\item[(b)] Let $q$ be the polynomial from (a) and let $\mathcal{V}_J$ be the degree $2k$ polynomial\footnote{Note that since we can expand $\mathcal R^{1/2}q$ around $\infty$ as a Laurent series with only finitely many terms with positive exponents, a routine residue at infinity calculation shows that the integral \eqref{eq:Vs} is a polynomial.} determined by $\mathcal{V}_J(1)=\widehat V$ and
\begin{equation}\label{eq:Vs}
\mathcal{V}_J'(z):=-\frac{1}{ic_J}\oint_{\Gamma}\frac{\mathcal R^{1/2}(x)q(x)}{z-x}dx,
\end{equation}
where $\Gamma$ is a counter-clockwise oriented contour surrounding $J$ and $z$, and
\begin{equation}\label{def:cs}
c_J=\frac{i}{2}\oint_{\Gamma}q(x)\mathcal R^{1/2}(x)dx>0.
\end{equation}
Then the equilibrium measure of $\mathcal{V}_J$ is $k$-cut regular and is given by 
\begin{equation} \label{SpecEqM}
d\mu_{\mathcal{V}_J}(x)=-\frac{i}{c_J}q(x)\mathcal R_{+}^{1/2}(x)dx, \qquad x \in J.
\end{equation}
\item[(c)] Suppose that $V$ is a polynomial of degree $2k$ with a positive leading coefficient and that ${\rm supp}{\mu_{V}}=J$ and $V(1)=\widehat V$. Then $V=\mathcal{V}_{J}$.
\item[(d)] The coefficients of the polynomials $\mathcal{V}_J$ and $q$ are smooth when considered as functions of $\{a_j,b_j\}_{j=1}^k$.
  \end{itemize}
\end{proposition}
We defer the proof of Proposition \ref{pr:defo1} to Section \ref{sec:deform}.

For each $s \in [0,1]$, let $V_{s}:=\mathcal{V}_{J_{s}}$, where $\mathcal{V}_{J_{s}}$ is the degree $2k$ polynomial given by Proposition \ref{pr:defo1}. By Proposition \ref{pr:defo1} (d), $s \to V_{s}$ is a smooth deformation from the Chebyshev potential $V_0$ to a second potential $V_1$ which has the same support as $V$. We now construct a smooth deformation from $V_1$ to $V=V_2$. 
 For $s\in [1,2]$, let 
 \begin{equation}\label{eq:Vs2}
V_s(x)=(2-s)V_1(x)+(s-1)V(x),
\end{equation}
and we define 
\begin{equation*}
\mu_{V_s}(dx)=(2-s)\mu_{V_1}(dx)+(s-1)\mu_V(dx), \qquad x \in J_1=J.
\end{equation*}
We emphasize that the support of $\mu_{V_s}$ is independent of $s\in[1,2]$. Using the Euler-Lagrange equations \eqref{eq:EL1}--\eqref{eq:EL2}, it is easily verified that $\mu_{V_s}$ is the equilibrium measure of $V_s$. In particular, $V_s$ is a $k$-cut regular potential for every $s\in[1,2]$.

In sections \ref{sec:trans}-\ref{sec:snorm} we obtain asymptotics for $Y_N(z;e^{-NV_s})$ as $N\to \infty$, which we subsequently rely on  in Section  \ref{sec:pasy} to analyze the large $N$ asymptotics of $\frac{d}{ds} \log H_N(e^{-NV_s})$. Integrating the asymptotics of $\frac{d}{ds} \log H_N(e^{-NV_s})$ as in \eqref{sec2lol8}, we  prove in Section \ref{sec:pasy} that
\begin{multline} \label{MainResSec11}
\lim_{\sigma \downarrow 1}\limsup_{N\to \infty}\Bigg |\log \frac{H_N(e^{-NV})}{H_N(e^{-NV_0})}+N^2 \iint \log|x-y|^{-1}d\mu_{V}(x)d\mu_{V}(y)\\ +N^2\int V(x)d\mu_{V}(x)-N^2\left(\frac{3}{4k}+\frac{1}{2k}\log \sigma +\log 2\right)  - \log \frac{\theta(N\Omega)}{\theta(0)}  \\
 -\frac{1}{8}\Bigg(\sum_{1\leq l<j \leq k} [\log(a_j-a_l)+\log(b_j-b_l)]-\sum_{l,j=1}^k\log|b_j-a_l|\Bigg)\\
+\frac{1}{24}\sum_{x\in\{a_j,b_j\}_{j=1}^k} \log \tilde \psi(x) -\frac{k}{4} \log 2+\frac{k-3}{12}\log k-\frac{k-1}{8}\log (\sigma-1) \Bigg|=0,
\end{multline}
where $\tilde \psi$ was defined in Theorem \ref{th:pfasy}. Theorem \ref{th:pfasy} is obtained by taking the sum of \eqref{LimCheby} and \eqref{MainResSec11} (without the absolute value so that the $\log H_N\left(e^{-NV_0}\right)$ terms cancel) and rearranging the terms. To verify that this results in the desired $o(1)$ error term in Theorem \ref{th:pfasy}, we observe that given $\delta>0$, we may pick $\sigma>1$ and $N^*$ such that for $N>N^*$ the sum of \eqref{LimCheby} and \eqref{MainResSec11} is less than $\delta$. This concludes the proof of Theorem \ref{th:pfasy}.

The proof of \eqref{MainResSec11} relies a number of identities for $\theta$-functions, and we found the work of Claeys, Grava and McLaughlin \cite{CGML} in the two-cut case as a most useful starting point when determining what these identities should be, also in the $k$-cut case.

\subsection{Proof of Proposition \ref{pr:defo1}} \label{sec:deform}

We split the proof in four parts.

\underline{Proof of (a): existence and uniqueness of $q$.} Let us write 
\begin{align*}
q(x)=x^{k-1}+\sum_{j=0}^{k-2}q_j x^j, \qquad \mbox{with } q_0,\ldots,q_{k-2}\in \mathbb C.
\end{align*}
With this notation, \eqref{zerointxil} can be rewritten as
\begin{equation} \label{zerointxil2}
\int_{b_j}^{a_{j+1}}x^{k-1}\mathcal R^{1/2}(x)dx+\sum_{l=0}^{k-2}q_l \int_{b_j}^{a_{j+1}}x^l\mathcal R^{1/2}(x)dx=0, \qquad j=1,\dots,k-1.
\end{equation}
This is a system of $k-1$ equations for the $k-1$ variables $q_0,\dots,q_{k-2}$. Let $\mathcal B$ be the $(k-1)\times (k-1)$ matrix defined by $\mathcal B_{l,j}=\int_{b_j}^{a_{j+1}}x^{l-1}\mathcal R^{1/2}(x)dx$, and let $\mathcal B_k$ be the row vector given by $(\mathcal B_k)_{j} = \int_{b_{j}}^{a_{j+1}}x^{k-1}\mathcal{R}^{1/2}(x)dx$. Then \eqref{zerointxil2} is equivalent to 
\begin{equation}\label{zerointxil3}
\mathcal B_k+ (q_0 \;\; q_1 \;\; \dots \;\; q_{k-2}) \; \mathcal B=(0 \;\; 0 \;\; \dots \;\; 0).
\end{equation}
In a similar way as in \eqref{eq:detA}, we can realize $\det\mathcal B$ as the integral of a Vandermonde determinant, and we obtain that $\det \mathcal B\neq 0$, which implies that \eqref{zerointxil3} has a unique solution. 

Since the entries of $\mathcal B$ are real, it follows that $q_{0},\ldots,q_{k-2}$ are real too, and thus $x\mapsto q(x)$ is real-valued for $x \in \mathbb{R}$. Thus, for \eqref{zerointxil} to hold, $q$ must have a root on each of the intervals $(b_j,a_{j+1})$ for $j=1,\dots,k-1$.

\underline{Proof of (b): $\mathcal{V}_J$ is a $k$-cut regular potential.} Recall that $\mathcal{V}_{J}$ is defined by \eqref{eq:Vs} and $\mathcal{V}_{J}(1)=\widehat{V}$. Let us first show that $\mu_{\mathcal{V}_J}$ defined in \eqref{def:cs}--\eqref{SpecEqM} is the equilibrium measure of $\mathcal{V}_J$. For this, we must show that the Euler-Lagrange equations \eqref{eq:EL1}--\eqref{eq:EL2} hold.

Let $q$ be the unique monic polynomial of degree $k-1$ satisfying \eqref{zerointxil}. For $z\notin J$, define
\begin{align}\label{lol2}
G(z)=\int_{J}\frac{d\mu_{\mathcal{V}_J}(y)}{z-y}=\frac{1}{i c_J}\int_{J}\frac{\mathcal R_{+}^{1/2}(y)q(y)}{z-y}dy.
\end{align}
Let $\Gamma$ be a counter-clockwise oriented loop enclosing both $J$ and $z$. By deforming the contour $J$ to $\Gamma$ in \eqref{lol2} and using Cauchy's integral formula, we see that
\begin{equation}
G(z)=-\frac{\mathcal R^{1/2}(z)\pi q(z)}{c_J}-\frac{1}{2  ic_J}\oint_{\Gamma}\frac{\mathcal R^{1/2}(w)q(w)}{z-w}dw. \notag
\end{equation}
By \eqref{eq:Vs}, the above equation can be rewritten as
\begin{equation}
G(z)=-\frac{\mathcal R^{1/2}(z)\pi q(z)}{c_J}+\frac{1}{2}\mathcal{V}_J'(z).\label{eq:Gks}
\end{equation}
Also, from \eqref{lol2}, we note that
\begin{equation}\label{eq:Gks2} 
2 \, \frac{d}{dx} \int \log |x-y|d\mu_{\mathcal{V}_J}(y)=G_{+}(x)+G_{-}(x), \qquad x \in \mathbb{R}.
\end{equation}
Therefore, using \eqref{eq:Gks}, we find
\begin{equation} \label{ForEL1a}
 2\, \frac{d}{dx} \int \log |x-y|d\mu_{\mathcal{V}_J}(y)-\mathcal{V}_J'(x)=0, \qquad x\in J.
\end{equation}
Additionally, from \eqref{zerointxil} and \eqref{eq:Gks}-\eqref{eq:Gks2}, we obtain
\begin{multline}\label{ForEL1b}
\left(2\int\log |b_j-y|d\mu_{\mathcal{V}_J}(y)-\mathcal{V}_J(b_j)\right) \\
-\left(2\int\log |a_{j+1}-y|d\mu_{\mathcal{V}_J}(y)-\mathcal{V}_J(a_{j+1})\right)\\ =-\int_{b_j}^{a_{j+1}} \frac{d}{dx}\left(2\int\log |x-y|d\mu_{\mathcal{V}_J}(y)-\mathcal{V}_J(x)\right) dx=0.
\end{multline}
By \eqref{ForEL1a} and \eqref{ForEL1b} we obtain equality in the first Euler-Lagrange equation \eqref{eq:EL1} upon setting 
\begin{equation*}
\ell=-2\int\log |b_1-y|d\mu_{\mathcal{V}_J}(y)+\mathcal{V}_J(b_1). 
\end{equation*}

Let $\xi_{1},\ldots,\xi_{k-1}$ be the roots of $q$. By \eqref{eq:Gks}, 
\begin{equation*} 
2 \, \frac{d}{dx} \int \log |x-y|d\mu_{\mathcal{V}_J}(y)-\mathcal{V}_J'(x)<0,
\end{equation*}
for $b_k<x$ and for $b_j<x<\xi_j$ with $j=1,2,\dots, k-1$, and similarly
\begin{equation*} 
2\, \frac{d}{dx} \int \log |x-y|d\mu_{\mathcal{V}_J}(y)-\mathcal{V}_J'(x)>0, 
\end{equation*}
for $x<a_1$ and $\xi_j<x<a_{j+1}$ for $j=1,\dots,k-1$, which proves that \eqref{eq:EL2} holds with a strict inequality. This shows that $\mu_{\mathcal{V}_J}$ is the equilibrium measure of $\mathcal{V}_{J}$. Furthermore, since \eqref{eq:EL2} is strict and $q$ is non-zero on a neighbourhood of $J$, it follows that $\mu_{\mathcal{V}_J}$ is $k$-cut regular.

\underline{Proof of (c): uniqueness of $\mathcal{V}_{J}$.} Assume that $V$ is a polynomial of degree $2k$ with positive leading coefficient satisfying ${\rm supp}{\smash{\mu_{V}}}=J$ and $V(1)=\smash{\widehat{V}}$. We will prove that $V=\mathcal{V}_{J}$. By deforming the contour of \eqref{eq:h} to $\infty$, we find that $h_V$ is a polynomial of degree $k-1$. Since $\psi_V$ satisfies \eqref{eq:psi} and \eqref{zerointpsiV}, it follows from part (a) that $h_V(x)$ is given by $c\, q(x)$ for a certain constant $c$, and thus $\mu_V$ is given by \eqref{SpecEqM}. By \eqref{eq:psi} and \eqref{ddzV}, $V$ satisfies \eqref{eq:Vs}, and thus $V=\mathcal{V}_J$.

\underline{Proof of (d): smoothness of $q$ and $\mathcal{V}_{J}$.} Using the multiple integral representation of $\det \mathcal{B}$ (which is similar to \eqref{eq:detA}), we see that $\det \mathcal B$ remains bounded and remains bounded away from 0. Hence, it directly follows from \eqref{zerointxil3} and \eqref{eq:Vs} that the coefficients of the polynomials of $V_J$ and $q$ are smooth as functions of $\{a_j,b_j\}_{j=1}^k$. 

This concludes the proof of Proposition \ref{pr:defo1}.

\subsection{Proof of Corollary \ref{RemarkPoly1}}\label{ProofPik}
In this section we prove Corollary \ref{RemarkPoly1}.

For the reader's convenience we first provide a proof that the equilibrium measure of $V(x)=\frac{2\nu}{k}\Pi_k(x)^2$ is given by \eqref{eqmeasPoly}.

  Let $a_1<b_1<\dots<a_k<b_k$ be the ordered zeros of $\Pi_k(x)^2-1/\nu$, and let $J=\cup_{j=1}^k[a_j,b_j]$. Using Proposition \ref{pr:defo1}, we will prove that  $V(x)=\frac{2\nu}{k}\Pi_k(x)^2$ is equal to $\mathcal{V}_{J}$, where $\mathcal{V}_{J}$ is the unique polynomial of degree $2k$ satisfying $\mathrm{supp}( \mu_{\mathcal V_J})=J$ and $\mathcal V_J(1)=\frac{2\nu}{k}\Pi_k(1)^2$. In this setting, we have 
\begin{equation*}
\mathcal R^{1/2}(x)=\sqrt{\Pi_k(x)^2-1/\nu}.
\end{equation*}
Using the change of variables $y=\Pi_k(x)$, we verify that
\begin{align*}
\int_{b_{j}}^{a_{j+1}}\Pi_k'(x)\mathcal{R}^{1/2}(x)dx = 0, \qquad j=1,\ldots,k-1,
\end{align*}
and therefore the function $q$ appearing in \eqref{zerointxil} is given by $q(x)=\frac{1}{k}\Pi_k'(x)$. To evaluate the integral on the right-hand side of \eqref{def:cs}, we compute the residue of the integrand at $\infty$ using
\begin{equation}\label{lollies}
\mathcal R^{1/2}(x)=\Pi_k(x)-\frac{1}{2\nu x^k}+\mathcal O(x^{-(k+1)}), \qquad \mbox{as } x \to \infty,
\end{equation}
and obtain $c_J=\frac{\pi}{2\nu}$. Using also \eqref{lollies} in \eqref{eq:Vs}, we obtain
\begin{equation*} 
\mathcal{V}_J'(z)=-\frac{1}{ikc_J}\oint_\Gamma \frac{\mathcal{R}^{1/2}(x)\Pi_k'(x)}{z-x}dx=-\frac{1}{ikc_J}\oint_\Gamma \frac{\Pi_k(x)\Pi_k'(x)}{z-x}dx.
\end{equation*}
This last integral can be evaluated by a direct residue computation at $x=z$, and we find $\mathcal{V}_J(x)=\frac{2\nu}{\pi}\Pi_k(x)^2=V(x)$, as desired. Thus the fact that $V$ has the equilibrium measure \eqref{eqmeasPoly} follows by part (b) of Proposition \ref{pr:defo1}.

We now proceed to simplify the terms in Theorem \ref{th:pfasy} for $V(x)=\frac{2\nu}{k}\Pi_k(x)^2$.

Using \eqref{eqmeasPoly} and the change of variables $y=\Pi_{k}(x)$, we infer that $\int_{a_j}^{b_j}d\mu_V=1/k$ for each $j=1,\dots, k$. Hence, the quantities $\Omega_j$ defined in \eqref{eq:Omega} are explicitly given by $\Omega_j=\frac{k-j}{k}$ for each $j=1,\dots,k-1$. Also, since $|\Pi_k(x)|=1/\sqrt \nu$ for each $x\in\{a_j,b_j\}_{j=1}^k$, by \eqref{tildepsi} we have
\begin{equation*} 
\tilde \psi(x)=\frac{2^{3/2}\nu^{3/4}}{k}\left|\Pi_k'(x)\right|^{3/2}. 
\end{equation*}
Assuming Theorem \ref{th:pfasy} (which will be proved in Sections \ref{sec:special} and \ref{sec:pasy}), it remains to compute $I_{V}(\mu_{V})$ explicitly, which we do now.

By the Euler-Lagrange equation \eqref{eq:EL1},
\begin{align}\label{IV in proof for polynomial}
I_{V}(\mu_{V}) & = \iint \log|x-y|^{-1}d\mu_{V}(x)d\mu_{V}(y)+\int V(x)d\mu_{V}(x)\\
& =\frac{\ell}{2}+\frac{1}{2}\int V(x)d\mu_{V}(x). \nonumber
\end{align}
This last integral can be evaluated by first deforming the contour to a large loop containing $J$, then performing the change of variables $y=\Pi_{k}(x)$, and then evaluating the residue at $\infty$. We find
\begin{equation}\label{int of V dmu}
\int V(x)d\mu_{V}(x)=\frac{1}{2k}.
\end{equation}
We now evaluate $\ell$. For this, we first replace $x$ in \eqref{eq:EL1} by the roots of $\Pi_{k}$; this yields $k$ different equations for $\ell$. By summing these $k$ equations, and then using a change of variables, we find
\begin{equation*}\begin{aligned} 
k\ell&=-\frac{4\nu}{\pi k}\int_J|\Pi_k'(x)| \sqrt{1/\nu - \Pi_k^2(x)} \log |\Pi_k(x)|dx\\&=-\frac{4\nu}{\pi}\int_{-1/\sqrt \nu}^{1/\sqrt \nu}\sqrt{1/\nu-y^2}\log |y|dy. 
\end{aligned}
\end{equation*}
This last integral can be evaluated explicitly, and we obtain
\begin{equation}\label{ell explicit in proof}
\ell=\frac{1}{k}\left(1+\log \nu+2\log 2\right).
\end{equation}
Substituting \eqref{int of V dmu} and \eqref{ell explicit in proof} in \eqref{IV in proof for polynomial}, we have 
\begin{equation}\label{LeadingPik} 
I_{V}(\mu_{V}) = \frac{1}{2k}\left(3/2+\log \nu+2\log 2\right).
\end{equation}
This finishes the proof of Corollary \ref{RemarkPoly1}.

\section{Analysis of $\theta$ functions}\label{sec:thetaids}
We recall some classical results from the theory of Riemann surfaces and $\theta$ functions, and in addition some identities for $\theta$ functions which are due to Fay. We provide ample details for the benefit of readers less well versed in the theory of Riemann surfaces and $\theta$-functions. Finally we combine known results into several identities which are specific to our situation and which we will rely on in later sections.

\subsection{$\theta$ has no zeros on $\mathbb R$}
Recall $\theta$ defined in \eqref{eq:thetadef}.

\begin{lemma}\label{le:jacobi}
	For any $\xi\in \C^{k-1}$, with $-i\tau$ real and positive definite, we have 
	\[
	\theta(\tau^{-1}\xi|-\tau^{-1})=e^{\pi i \xi^T\tau^{-1}\xi}\sqrt{\det (-i\tau)} \theta(\xi|\tau)
	\] 
\end{lemma}
\begin{proof} The lemma is proven in a more general setting in \cite[Chapter 2 $\S$5]{Mumford1}, we include a proof for the reader's convenience.
	The proof is by (multidimensional) Poisson summation. More precisely, for given $\xi\in \C^{k-1}$, consider the function $f_\xi:\R^{k-1}\to \C$, $f_\xi(x)= e^{2\pi i \xi^T x+\pi i x^{T}\tau x}$. Taking the change of variables $y=\tau^{-1}(\xi-\eta)+x$, we obtain that the Fourier transform of $f_\xi$  is given by
\begin{equation}\label{mystery}
	\widehat f_\xi(\eta)=\int_{\R^{k-1}} e^{-2\pi i \eta^T x}f_\xi(x)dx=e^{-\pi i (\xi-\eta)^T \tau^{-1}(\xi-\eta)}\widehat f_0(0).
	\end{equation}
	The functions $f_\xi$ and $\widehat f_\xi$ are Schwartz functions, and we recall the standard identity
	\begin{equation} \label{sumsfxi}
	\sum_{j\in \Z^{k-1}} f_\xi(j)=\sum_{j\in \Z^{k-1}}\widehat f_\xi(j).
	\end{equation}
	(The identity follows from the fact that if we define $F_\xi(x)=\sum_{j\in \mathbb Z^{k-1}}f_\xi(x-j)$ on $[0,1]^{k-1}$, then clearly $F_\xi(0)$ equals the left hand side of \eqref{sumsfxi}. On the other hand, by expanding $F_\xi(x)$ in the basis $e^{2\pi i j^T x}$, we find that
	$ F_\xi(0)=\sum_{j\in \mathbb Z^{k-1}} \int_{[0,1]^{k-1}} e^{-2\pi i j^Tx}F_\xi(x) dx $,  which equals the right-hand side of \eqref{sumsfxi}.)
		
	 Taking the definition of $f_\xi$ on the left hand side of \eqref{sumsfxi} and substituting \eqref{mystery} into the right-hand side, and recalling the definition of $\theta$ in \eqref{eq:thetadef}, we obtain
	\[
	\widehat f_0(0)e^{-\pi i \xi^T\tau^{-1}\xi}\theta(\tau^{-1}\xi|-\tau^{-1})=\theta(\xi|\tau).
	\]
	The claim now follows from noting that 
	\[
	\widehat f_0(0)=\int_{\R^{k-1}}e^{\pi i x^T\tau x}dx=\frac{1}{\sqrt{\det (-i\tau)}}.
	\]
\end{proof}

As we now see, this readily yields that $\xi\mapsto\theta(\xi|\tau)$ does not vanish on $\R^{k-1}$.
\begin{lemma}\label{le:norealzeros}
	If $-i\tau$ is real and positive definite, then $\theta(\xi|\tau)\neq 0$ for $\xi\in \R^{k-1}$.
\end{lemma}
\begin{proof}
	Lemma \ref{le:jacobi} implies that 
	\begin{align*}
	\theta(\xi|\tau)&=\frac{1}{\sqrt{\det(-i\tau)}}e^{-i\pi \xi ^{T} \tau ^{-1}\xi}\theta(\tau^{-1}\xi|-\tau ^{-1})\\
	&=\frac{1}{\sqrt{\det(-i\tau)}}e^{-i\pi \xi^{T} \tau ^{-1}\xi}\sum_{n\in \Z^{k-1}}e^{2\pi i n^{T}\tau^{-1}\xi-\pi i n^{T}\tau ^{-1}n}.
	\end{align*}
	Recalling that $-i\tau$ is real and positive definite, the sum is a sum of positive terms and the prefactors are also positive. This concludes the proof.
\end{proof}

\subsection{Divisor of $\theta(u(z)\pm u(\lambda))$ and of $\Theta(z,\lambda)$}\label{Sec:Divisors}
The functions $\theta(u(z)\pm u(\lambda))$ and $\theta\left[{\substack{{ \ualpha} \\ { \ubeta}}}\right](u(z)\pm u(\lambda))$, introduced in Section \ref{sec:intro_theta}, play a fundamental role throughout the entire paper, and in this section we analyze these functions. To study their zero sets, we consider  $u$ to be a multivalued function on the Riemann surface  $\mathcal S$ (defined in Section \ref{Sec:calS}) of genus $k-1$ as follows.

 Define
\begin{equation}\label{eq:lattice}
\Lambda=\left\{\sum_{j=1}^{k-1}(n_j e_j+m_j\tau_j): n_1,...,n_{k-1},m_1,...,m_{k-1}\in \Z\right\},
\end{equation} 
with $\tau$ as in \eqref{eq:period},
and extend the Abel map \eqref{eq:Abel} to $\mathcal S$. Then $u$ is not well defined as a function from $\mathcal S\to \mathbb C^{k-1}$, but by \eqref{eq:Aint} and \eqref{eq:period}, $u:\mathcal S\to \mathbb C^{k-1}/\Lambda$ is well defined. Thus it follows that $\theta(u(z) \pm u(\lambda))$ may be considered as a multivalued function on $\mathcal S$, with  a well defined zero set. To analyze this zero set, we will rely on the classical theory of Riemann surfaces, and start by recalling a few definitions and classical theorems which we will rely on (see  \cite{FK} for details).

Throughout the section, if $z\in \mathcal S$, we denote by $\mathbf Pz$ the projection from $\mathcal S$ to $\mathbb C$.\footnote{More precisely, if we represent $\mathcal S$  as $ \left\{ (z,+\mathcal R^{1/2}(z)):z\in \mathbb C \right\} \cup\left\{ (z,-\mathcal R^{1/2}(z)):z\in \mathbb C \right\}$, then
\begin{equation*} \mathbf P(z,+ \mathcal R^{1/2}(z))=\mathbf P(z,-\mathcal R^{1/2}(z))=z.
\end{equation*}} If $z\in \mathcal S\setminus \{a_j,b_j\}_{j=1}^k$, we denote $z^*$ the point in $\mathcal S$ satisfying $z^*\neq z$ and $\mathbf Pz=\mathbf Pz^*$. If $z\in \{a_j,b_j\}_{j=1}^k$, we denote $z^*=z$.

\begin{itemize}
\item[(a)] \underline{Divisors.} If $\alpha$ is an integer valued function on $\mathcal S$ which is non-zero for at most a finite number of points on $\mathcal S$, then a divisor is a formal symbol $D=\prod_{z\in \mathcal S}z^{\alpha(z)}$. A divisor is said to be \textit{integral} if $\alpha$ is non-negative on $\mathcal S$. We will vary notation slightly, and use the following notation when it suits. If  $z_1,\dots,z_{n_1}\in \mathcal S$ and $\lambda_1,\dots, \lambda_{n_2}\in \mathcal S$, such that $\{z_1,\dots,z_{n_1}\}$ and $\{\lambda_1,\dots, \lambda_{n_2}\}$ are disjoint, we will denote

\begin{equation*} D=\prod_{j=1}^{n_1}z_j\prod_{j=1}^{n_2}\lambda_j^{-1}=\prod_{z\in \mathcal S}z^{\alpha(z)}, \qquad \alpha(z)=\#\{j:z_j=z\}-\#\{j:\lambda_j=z\}.\end{equation*}
Then the degree of $D$ is $\deg D=n_1-n_2$.  We will write $D_1\geq D_2$ if $D_1/D_2$ is an integral divisor.
\item[(b)] \underline{Divisors of functions.} Denote the space of meromorphic functions  on $\mathcal S$ by $\mathcal H$. For $f\in \mathcal H$, we denote the divisor of $f$ by

\begin{equation*}{\rm Div}(f)=\prod_{z\in \mathcal S}z^{\textrm{ord}_z(f)}, 
\end{equation*}
where $\textrm{ord}_z(f)$ is the order of $f$ at $z$ (more specifically, $\textrm{ord}_z(f)$ is defined as the smallest integer such that $\lim_{\lambda\to \mathbf{P}z} \frac{f(\varphi^{-1}(\lambda))}{(\lambda-\mathbf{P}z)^{\textrm{ord}_z(f)}}\neq 0$, where $\varphi$ is a local chart). For any meromorphic function $f$, we have $\deg (f):=\deg \textrm{Div} (f)=0$ (see \cite[ I.1.6]{FK}). We will similarly refer to divisors of multivalued functions with well defined zeros and poles.
\item[(c)] \underline{Special divisors.} An integral divisor $D=\prod_{z\in \mathcal S} z^{\alpha(z)}$ of degree $k-1$ is said to be \textit{special} (see \cite[page 95]{FK}) if 
\begin{equation*}\textrm{dim}\{f\in \mathcal H:\textrm{ord}_z(f)\geq -\alpha(z) \textrm{ for all }z\in \mathcal S\}>1.\end{equation*}
This is an especially important notion, because if an integral divisor $D=\prod_{z\in \mathcal S} z^{\alpha(z)}$ of degree $k-1$ is not special and $f$ is a meromorphic function whose  poles are a subset of $\{z:\alpha(z)>0\}$ and at most of order $\alpha(z)$, then $f$ is constant on $\mathcal S$.
\item[(d)] \underline{Riemann-Roch Theorem.} 
By the Riemann-Roch theorem (see \cite[III.4.8]{FK}), an integral divisor $D$ of degree $k-1$ is special if and only if
\begin{equation}\label{Riemanroch}
\dim \{\mbox{holomorphic differentials }\tilde{\omega} : \mbox{Div}(\tilde{\omega}) \geq D\} > 0,
\end{equation}
where the divisor of a holomorphic differential is defined in a similar fashion as the divisor of a meromorphic function, see \cite[III.4.1]{FK}. 
While \cite{FK} take the definition in (c) to be the definition of a special divisor, some authors rather take \eqref{Riemanroch} as the definition of a special divisor. 
Recall that any holomorphic differential on $\mathcal{S}$ can be written in the form $\tilde{\omega}=\frac{p(z)}{\sqrt{\mathcal{R}(z)}}dz$, where $p$ is a polynomial of degree at most $k-2$. Thus, by \eqref{Riemanroch},   a divisor $D = \prod_{j=1}^{k-1}z_{j}$ is special if and only if there exists $j_{1},j_{2} \in \{1,\ldots,k-1\}$, $j_{1} \neq j_{2}$, such that $z_{j_{1}}=z_{j_{2}}^{*}$. 

\item[(e)] {\underline{Abel Theorem.}} Given a divisor $D=\prod_{j=1}^nz_j\prod_{j=1}^n\lambda_j^{-1}$,  there exists $f\in \mathcal H$ such that ${\rm Div}(f)=D$ if and only if
\begin{equation*}\sum_{j=1}^nu(z_j)-\sum_{j=1}^nu(\lambda_j)\equiv 0\mod \Lambda. \end{equation*}
(We recall that $u$ is well defined $\textrm{mod} \; \Lambda$). See \cite[III.6.3]{FK}.

\item[(f)] {\underline{Zeros of $\theta \circ u$ on $\mathcal S$.}} Let $K$ be the Riemann vector of constants (see \cite[VI.2.4]{FK}), let $D=\prod_{j=1}^{k-1}z_j$ be a divisor, and consider  
\begin{equation}\label{thetainlem}\theta\left(u(z)-K-\sum_{j=1}^{k-1}u(z_j)\right),\end{equation}
which is multivalued on $\mathcal S$ but has a well defined zero set. Then \eqref{thetainlem} is identically zero on $\mathcal S$ if and only if $D$ is special, and if $D$ is not special then \eqref{thetainlem} has a zero at each point $z_j$ for $j=1,\dots, k-1$, and no other zeros on $\mathcal S$. If $z\in \{z_j\}_{j=1}^{k-1}$, then the order of \eqref{thetainlem} at $z$ is $\#\{j:z_j=z\}$.

\end{itemize}

 By \cite[Section VI.2]{FK}  $\theta(u(z))$ is either identically zero or has exactly $k-1$ zeros. Since (by e.g. Lemma \ref{le:norealzeros}) $\theta(0)\neq 0$, it is not identically zero, and we denote the zeros by $z_1,\dots, z_{k-1}$. Furthermore, by \cite[Theorem p 308--309]{FK}, they satisfy 
\begin{equation}\label{eq:rvectarg}\sum_{j=1}^{k-1}u(z_j)\equiv -K \mod \Lambda. 
\end{equation} 

It is easy to verify by relying on \eqref{eq:Aint2} and \eqref{eq:period}, that \begin{equation}\nonumber u(b_j)\equiv \frac{1}{2}\left(\tau_j+\sum_{l=j}^{k-1}e_l\right)\,\,{\rm mod} \; \Lambda, \textrm{ for }j=1,\dots,k-1.\end{equation} If  $\ualpha=\frac{1}{2}e_j$ and   $\ubeta=\frac{1}{2}\sum_{l=j}^{k-1}e_l$ then $u(b_j)=\tau \ualpha+\ubeta$. Since $4\ualpha^T\ubeta=1$, and in particular is odd, it follows that 
$\theta\left[{\substack{\ualpha \\ \ubeta}}\right](0)=0$. Thus $\theta(u(b_j))=0$, and by \eqref{eq:rvectarg} we obtain:
\begin{itemize} \item[(g)] The Riemann vector of constants satisfies

\begin{equation}\label{eq:rvect}
K\equiv -\sum_{j=1}^{k-1}u(b_j)\, \,  \mod \Lambda.
\end{equation}
\end{itemize}

Observe that $\gamma^2(z)$ (defined in \eqref{eq:gamma}) extends to a meromorphic function on $\mathcal S$, and define $q_\pm(z)=\gamma(z)^{-2}\pm\gamma(\lambda)^{-2}$. 
We will describe the zeros of $z \mapsto \theta(u(z)\pm u(\lambda))$  in terms of the zeros of $z \mapsto q_\mp(z)$. For this, we first prove the following lemma.

\medskip For now, we view $q_\pm$ as a function of $z \in \mathcal{S}$ which depends on a parameter $\lambda \in \mathbb{C}$.

\begin{lemma} \label{lemzerosgamma}
Assume that $\lambda \in \mathbb C\setminus \{a_j,b_j\}_{j=1}^k$. Then $q_+$ has precisely $k$ (not necessarily distinct) zeros in $\mathcal S$ which we denote $z_{1,+}, \dots z_{k-1,+}, \lambda^{(+)}$, and similarly $q_-$ has $k$ (not necessarily distinct) zeros in $\mathcal S$ which we denote $z_{1,-},\dots, z_{k-1,-}, \lambda^{(-)}$, where $\lambda^{(+)}$ is on the second sheet of $\mathcal S$, $\lambda^{(-)}$ is on the first sheet of $\mathcal S$, and $\mathbf{P}\lambda^{(+)} = \mathbf{P}\lambda^{(-)} = \lambda$. Denote the divisor $D_\pm=\prod_{j=1}^{k-1}z_{j,\pm}$. Then $D_\pm$ is not special.
\end{lemma}

Recall that if $z\in \mathcal S\setminus \{a_j,b_j\}_{j=1}^k$, we denote $z^*$ the point in $\mathcal S$ satisfying $z^*\neq z$ and $\mathbf Pz=\mathbf Pz^*$.
\begin{proof} 
Since $\lambda \notin  \{a_j,b_j\}_{j=1}^k$, it follows that $\gamma(\lambda)^2\neq 0$. Thus, since $\gamma(z)^2=-\gamma(z^*)^2$, the zero sets of $q_+$ and $q_-$ are disjoint, and if $z$ is a zero of $q_+$, then: 1) $z^*$ is a zero of $q_-$, 2) $z^*$ is not a zero of $q_+$, 3) $z,z^*\notin \{a_j,b_j\}_{j=1}^k$ because $\{a_j,b_j\}_{j=1}^k$ are the poles and zeros of $\gamma^{2}$, and 4) $\lambda^{(\pm)}$ is a zero of $q_\pm$. Also, if $z$ and $z^*$ are zeros of $q_+$ and $q_-$ respectively, then $\gamma(z)^4=\gamma(z^*)^4=\gamma(\lambda)^4$, and by the definition of $\gamma$,
\begin{equation}\label{gamma4zeros}\prod_{j=1}^k(\mathbf{P}z-b_j)-\gamma(\lambda)^4\prod_{j=1}^k(\mathbf{P}z-a_j)=0. \end{equation}
If $\gamma(\lambda)\neq 1$, then \eqref{gamma4zeros} has $2k$ (not necessarily  distinct) zeros in the variable $z \in \mathcal{S}$, and it follows that $q_\pm$ has $k$ (not necessarily distinct) zeros in $\mathcal{S}$. If $\gamma(\lambda)^4=1$, then the left hand side of \eqref{gamma4zeros} is a polynomial of degree $k-1$ in the variable $\mathbf{P}z$, and thus has $2k-2$ zeros in the variable $z \in \mathcal{S}$. As $z \to \infty$ (on any sheet), we have
\begin{equation} \label{gamma4zeros2}\frac{\prod_{j=1}^k(\mathbf{P}z-b_j)}{\prod_{j=1}^k(\mathbf{P}z-a_j)}=\frac{\prod_{j=1}^k(\mathbf{P}z-b_j)-\prod_{j=1}^k(\mathbf{P}z-a_j)}{\prod_{j=1}^k(\mathbf{P}z-a_j)}+1=1+\mathcal O\left(z^{-1}\right), \end{equation}
from which we deduce that the two points at $\infty$ on $\mathcal{S}$ are zeros of $q_+$ and $q_{-}$.
Thus, if $\gamma(\lambda)^4=1$, we conclude from \eqref{gamma4zeros} and \eqref{gamma4zeros2} that $q_+$ and $q_{-}$ have each $k$ zeros on $\mathcal{S}$.

Recall from earlier in the proof that if $z$ is a zero of $q_+$ then $z^*$ is not a zero of $q_+$. Since $q_+$ has $k$ zeros, it follows by the Riemann-Roch Theorem (see (d) above) that $D_+$ is not special. Similarly, $D_-$ is not special.
\end{proof}

We are now in a position to prove the following lemma.
\begin{lemma}\label{lemDivthetadiff}
 Let $\lambda \in \mathcal S\setminus \{a_j,b_j\}_{j=1}^k$, and consider $\theta(u(z) \pm u(\lambda))$ as a multivalued function of $z\in \mathcal S$. With the notation of Lemma \ref{lemzerosgamma}, the  divisor of 
$\theta(u(z) \pm u(\lambda))$ is $D_\mp$. 
\end{lemma}
\begin{proof}
The proof follows \cite[Lemma 3.27]{DIZ}. Consider the multi-valued functions
\begin{equation}\label{thetazfns}
\theta\left(u(z)-\sum_{j=1}^{k-1}u(z_{j,+})-K\right)  \qquad {\rm and} \qquad \theta\left(u(z)-\sum_{j=1}^{k-1}u(z_{j,-})-K\right),\end{equation}
where we recall that $K$ is the vector of Riemann constants. Since $D_+$ and $D_-$ are not special, it follows by (f) above that the divisors of \eqref{thetazfns} are given by $D_+$ and $D_-$ respectively. 

Observe that the divisor of $q_\pm$ is given by $\lambda^{(\pm)}\prod_{j=1}^{k-1}z_{j,\pm}\prod_{j=1}^{k}b_j^{-1}$. Hence, by the Abel theorem (see (e) above) and \eqref{eq:rvect} we obtain
\begin{equation}\label{Abel1}
-\sum_{j=1}^{k-1}u(z_{j,\pm})-K\equiv u(\lambda^{(\pm)})\,\,  \mod \Lambda. \end{equation}
Since $u(\lambda)=-u(\lambda^*) $, it follows that the zeros of $\theta(u(z)+u(\lambda))$ and $\theta(u(z)-u(\lambda))$, viewed as functions of $z \in \mathcal{S}$, are given by the zeros of \eqref{thetazfns} respectively. Since \eqref{thetazfns} had the same zeros as $q_-$ and $q_+$, the lemma follows.
\end{proof}

Next we consider the divisor of $\theta\left[{\substack{\ualpha \\ \ubeta}}\right](u(z)\pm u(\lambda))$. 
\begin{lemma}\label{Prime}
Let $\ualpha=\frac{1}{2}e_1$ and let $\ubeta=\frac{1}{2}\sum_{j=1}^{k-1}e_j$. Let $\lambda \in \mathcal S\setminus \{b_2,\dots, b_{k-1}\}$.  Then the divisor of $\theta\left[{\substack{\ualpha \\ \ubeta}}\right](u(z)-u(\lambda))$, as a multi-valued function of $z\in \mathcal S$, is $\lambda\prod_{j=2}^{k-1}b_j$. Furthermore, the divisor of $\theta\left[{\substack{\ualpha \\ \ubeta}}\right](u(z)+u(\lambda))$ is $\lambda^*\prod_{j=2}^{k-1}b_j$, and the divisor of 
$\Theta(z,\lambda)=\frac{\theta\left[{\substack{\ualpha \\ \ubeta}}\right](u(z)-u(\lambda))}{\theta\left[{\substack{\ualpha \\ \ubeta}}\right](u(z)+u(\lambda))}$ is $\lambda \left(\lambda^*\right)^{-1}$.
\end{lemma}

\begin{proof}
We observe that $\left[\substack{\ualpha \\ \ubeta}\right]$ is an odd characteristic, and that $u_{+}(b_1)=\tau \ualpha+\ubeta$. By  \eqref{eq:rvect}, we have $\tau\ualpha+\ubeta=-K-\sum_{j=2}^{k-1}u(b_j)$. Let $\lambda \in  \mathcal S\setminus \{b_2,\dots, b_{k-1}\}$ be fixed. By the Riemann-Roch theorem (see (d) above), $\prod_{j=2}^{k-1}b_j\lambda$ is not a special divisor,  and by (f) above, it follows that $\theta(u(z)-K-\sum_{j=2}^k u(b_j)-u(\lambda))$ is not identically zero and has zeros at $b_2,\dots, b_k,\lambda$. By the definition of $\theta\left[{\substack{\ualpha \\ \ubeta}}\right]$ in \eqref{eq:thetachar}, it follows that $\theta\left[{\substack{\ualpha \\ \ubeta}}\right](u(z)-u(\lambda))$ also has divisor $\prod_{j=2}^{k-1}b_j\lambda$.

The second statement follows from the fact that $u(\lambda)=-u(\lambda^*)\mod \Lambda$, and the third statement follows from the first two statements.
\end{proof}

Lemma \ref{Prime} above, and also Lemma \ref{Prime2} below, are standard results related to what is known as the prime form, see \cite[Chapter 2]{Fay}. We include proofs for the reader's convenience.

\subsection{Some further identities for $\theta$-functions}
In this section we prove the following proposition, which we will rely on in Section \ref{sec:pasy}, and which also yields the identity for $W$ in \eqref{id:W}.

\begin{proposition}\label{PropForM}
Recall that $\boldsymbol{\omega}$ is defined by \eqref{eq:Aint}.
For $z,\lambda \in \mathbb C\setminus [a_1,b_k]$, and any fixed $v\in \mathbb C^{k-1}$ such that $\theta(v)\neq 0$,
\begin{multline}\label{ForM11M22}
\theta\left(\int_z^\lambda \boldsymbol{\omega}+v\right)\theta\left(\int_z^\lambda\boldsymbol \omega-v\right)\frac{\theta(0)^2\left(\frac{\gamma(z)}{\gamma(\lambda)}+\frac{\gamma(\lambda)}{\gamma(z)}\right)^2}{4\theta(v)^2\theta(u(z)-u(\lambda))^2}\\= (z-\lambda)^2 \left(\frac{1}{2}W(z,\lambda)+\frac{1}{2(z-\lambda)^2}+ \sum_{j=1}^{k-1}\sum_{i=1}^{k-1}u_j'(z)u_i'(\lambda)\partial_i\partial_j\log \theta (v)\right),
\end{multline}
where the contour of integration does not cross $[a_1,b_k]$,
and
\begin{multline}\label{ForM11M12}
\frac{\theta(0)^2\left(\left(\frac{\gamma(z)}{\gamma(\lambda)}\right)^2-\left(\frac{\gamma(\lambda)}{\gamma(z)}\right)^2\right)\theta\left(u(z)+u(\lambda)+v \right)\theta\left(u(z)-u(\lambda)-v \right)}{4\theta(u(z)-u(\lambda))\theta(u(z)+u(\lambda))\theta(v)^2}\\
=\frac{(z-\lambda)^2\theta\left(2u(\lambda)+v \right)\theta(0)}{4\theta \left(2u(\lambda)\right)\theta\left(v \right)}
\left(\sum_{j=1}^k\frac{1}{\lambda-b_j}-\frac{1}{\lambda-a_j}\right) \\
\times \sum_{j=1}^{k-1} u_j'(z)\Big( \partial_j \log \theta\left(2u(\lambda)+v \right)-\partial_j\log \theta\left[{\substack{\ualpha \\ \ubeta}}\right]\left(u(z)+u(\lambda)\right)\\ +\partial_j\log \theta\left[{\substack{\ualpha \\ \ubeta}}\right]\left(u(z)-u(\lambda)\right)-\partial_j \log \theta\left(v \right)
\Big),
\end{multline}
where $\ualpha=\frac{1}{2}e_1$ and  $\ubeta=\frac{1}{2}\sum_{j=1}^{k-1}e_j$.
\end{proposition}
Setting $v=0$ in \eqref{ForM11M22}, we obtain \eqref{id:W}.

In the remainder of the section we prove Proposition \ref{PropForM}.

\begin{lemma}\label{Prime2}
Let $\ualpha=\frac{1}{2}e_1$ and let $\ubeta=\frac{1}{2}\sum_{j=1}^{k-1}e_j$. Then
\begin{equation*}{\rm Div}\left( \sum_{j=1}^{k-1}u_j'(z)\partial_j\theta\left[{\substack{\ualpha \\ \ubeta}}\right](0)\right)=\infty_1^2\infty_2^2\prod_{j=2}^{k-1}b_jb_1^{-1}b_k^{-1}\prod_{j=1}^ka_j^{-1},\end{equation*}
where $\infty_1$ is $\infty$ on the first sheet and $\infty_2$ is $\infty$ on the second sheet.
\end{lemma}

\begin{proof} From Lemma \ref{Prime}, the divisor of $\theta\left[{\substack{\ualpha \\ \ubeta}}\right](u(z)-u(\lambda))$ is $\lambda\prod_{j=2}^{k-1}b_j$.
   Thus, for $\lambda \notin \{b_2,\dots,b_{k-1}\}$, it follows that $\theta\left[{\substack{\ualpha \\ \ubeta}}\right](u(z)-u(\lambda))$ has a zero of order 1 at $\lambda$ as a function of $z$, while for $\lambda \in \{b_2,\dots, b_{k-1}\}$ the divisor is special and  $\theta\left[{\substack{\ualpha \\ \ubeta}}\right](u(z)-u(\lambda))$   is identically zero as a function of $z$. Taking also into account that $u'(z)=\mathcal O(|\mathbf P z|^{-2})$ as $z\to \infty_1, \, \infty_2$,
\begin{equation*}\begin{aligned} \lim_{z\to \lambda}\frac{\partial}{\partial z}\theta\left[{\substack{\ualpha \\ \ubeta}}\right](u(z)-u(\lambda)) &\neq 0&& \textrm{for $\lambda\in  \mathcal S\setminus \{b_2,\dots, b_{k-1}, \infty_1, \infty_2\}$}, \\
\lim_{z\to \lambda}\frac{\partial}{\partial z}\theta\left[{\substack{\ualpha \\ \ubeta}}\right](u(z)-u(\lambda)) &= 0&& \textrm{for $\lambda\in  \{b_2,\dots, b_{k-1}, \infty_1,\infty_2\}$}. \end{aligned} \end{equation*}

Thus,
\begin{align*}\sum_{j=1}^{k-1}u_j'(z)\partial_j\theta\left[{\substack{\ualpha \\ \ubeta}}\right](0)&\neq 0, &&\textrm{for $z \in \mathcal S\setminus \{b_2,\dots, b_{k-1}, \infty_{1},\infty_{2}\}$,}\\
\sum_{j=1}^{k-1}u_j'(z)\partial_j\theta\left[{\substack{\ualpha \\ \ubeta}}\right](0)&=0 && \textrm{for $z \in \{b_2,\dots, b_{k-1}, \infty_1,\infty_2\}$.}
\end{align*}
The order of the zero at $\infty_1$ and $\infty_2$ is 2,  and thus the total order of the zeros is at least $k+2$. The only possible poles are at $a_1,\dots,a_k,b_1,b_k$, if there is a pole at each point it is of order $-1$. In particular it follows that the total degree of the poles is greater than or equal to $-k-2$. Since the degree of the divisor is zero, there must be a pole of order $1$ at each of the points $a_1,\dots,a_k,b_1,b_k$, which finishes the proof.
\end{proof}
Let $J_{\mathcal S}$ be $J_-$ on the first sheet,  namely:
\begin{equation*} J_{\mathcal S}=\{z\in \mathcal S: \textrm{  $\mathbf Pz\in J$ and $\mathcal R^{1/2}(z)=\mathcal R_-^{1/2}(\mathbf P z)$}\}, \end{equation*}
and let $\Sigma_{\mathcal S}$ be $\{z:\mathbf Pz \in \cup_{j=1}^{k-1}(b_j,a_{j+1})\}$.
The following identity is due to Fay \cite{Fay}.

\begin{lemma} Let $\ualpha=\frac{1}{2}e_1$ and let $\ubeta=\frac{1}{2}\sum_{j=1}^{k-1}e_j$. Assume that $z\in \mathcal S\setminus \{b_2,\dots,b_{k-1}, \infty_1, \infty_2\}$. Then, for any $\lambda_1,\lambda_2\in \mathcal S\setminus\{b_2,\dots,b_{k-1}\}$ such that $\lambda_1\neq \lambda_2$ and $v\in \mathbb C^{k-1}$,
\begin{multline}\label{identity1}
\theta\left(\int_z^{\lambda_1}\boldsymbol \omega+v \right)\theta\left(\int_{\lambda_2}^z\boldsymbol \omega+v \right) \\
=\frac{\theta\left(v \right)\theta\left(\int_{\lambda_2}^{\lambda_1}\boldsymbol \omega+v \right)\theta\left[{\substack{\ualpha \\ \ubeta}}\right]\left( \int^{\lambda_1}_z\boldsymbol\omega\right)\theta\left[{\substack{\ualpha \\ \ubeta}}\right]\left(\int^{\lambda_2}_z\boldsymbol\omega \right)}{\theta\left[{\substack{\ualpha \\ \ubeta}}\right]\left(\int_{\lambda_2}^{\lambda_1}\boldsymbol\omega\right)
\left(\sum_{j=1}^{k-1}u_j'(z)\partial_j\theta\left[{\substack{\ualpha \\ \ubeta}}\right](0)\right)}
\\ \times \sum_{j=1}^{k-1} u_j'(z)\bigg( \partial_j \log \theta\left(\int_{\lambda_2}^{\lambda_1}\boldsymbol\omega+v \right)-\partial_j\log \theta\left[{\substack{\ualpha \\ \ubeta}}\right]\left(\int^{\lambda_1}_z\boldsymbol\omega\right)\\ +\partial_j\log \theta\left[{\substack{\ualpha \\ \ubeta}}\right]\left(\int^{\lambda_2}_z\boldsymbol\omega\right)-\partial_j \log \theta\left(v \right)
\bigg),
\end{multline}
where none of the integrals pass through $J_\mathcal S \cup \Sigma_{\mathcal S}$, and
where in the event that $\theta(v)=0$ we interpret $\theta(v)\sum_{j=1}^{k-1}u_j'(z) \partial_j\log \theta(v)=\sum_{j=1}^{k-1}u_j'(z) \partial_j \theta(v)$, and likewise with the other logarithmic derivatives. 

Furthermore, for any $\lambda\in \mathcal S\setminus\{b_2,\dots,b_{k-1}\}$,
\begin{multline}\label{identity2}
\theta\left(\int_z^\lambda \boldsymbol{\omega}+v\right)\theta\left(\int_z^\lambda\boldsymbol{\omega}-v\right) \\
=\frac{\theta(v)^2\theta\left[{\substack{\ualpha \\ \ubeta}}\right]\left(\int_z^\lambda\boldsymbol{\omega}\right)^2}{\left(\sum_{j=1}^{k-1}u_j'(\lambda)\partial_j\theta\left[{\substack{\ualpha \\ \ubeta}}\right](0)\right)\left(\sum_{j=1}^{k-1}u_j'(z)\partial_j\theta\left[{\substack{\ualpha \\ \ubeta}}\right](0)\right)}\\
\times \sum_{j=1}^{k-1}\sum_{i=1}^{k-1}u_j'(z)u_i'(\lambda)\left(-\partial_i\partial_j\log \theta\left[{\substack{\ualpha \\ \ubeta}}\right]\left(\int_z^\lambda\boldsymbol{\omega}\right)+\partial_i\partial_j\log \theta (v)\right).
\end{multline}
\end{lemma}
\begin{remark} By Lemma \ref{Prime2}, the right-hand side of \eqref{identity1} and \eqref{identity2} are well defined. \end{remark}
\begin{proof} \eqref{identity2} follows from \eqref{identity1} by taking the limit $\lambda_1\to \lambda_2$.

We provide a proof of \eqref{identity1} inspired by Gus Schrader's simple proof of Fay's trisecant identity in \cite{Schrader}.
 
We first prove that as a function of $\lambda_1$, the right-hand side of \eqref{identity1} is pole free.
Observe that, as a function of $\lambda_1$, both
\begin{equation}\label{poleidentity1}\frac{\theta\left[{\substack{\ualpha \\ \ubeta}}\right]\left( \int^{\lambda_1}_z\boldsymbol{\omega}\right)}{\theta\left[{\substack{\ualpha \\ \ubeta}}\right]\left(\int_{\lambda_2}^{\lambda_1}\boldsymbol{\omega}\right)}, \quad \textrm{and} \quad \frac{\sum_{j=1}^{k-1}u_j'(z) \partial_j \theta\left[{\substack{\ualpha \\ \ubeta}}\right]\left( \int^{\lambda_1}_z\boldsymbol{\omega}\right)}{\theta\left[{\substack{\ualpha \\ \ubeta}}\right]\left(\int_{\lambda_2}^{\lambda_1}\boldsymbol{\omega}\right)} \end{equation} 
have a single pole of order 1 at $\lambda_2$ (this follows from Lemma \ref{Prime}, by Lemma \ref{Prime2}, and by the fact that $\theta\left[{\substack{\ualpha \\ \ubeta}}\right](u(z)-u(b_{j}))$ is identically zero for $j=2,\dots,k-1$). As $\lambda_1\to \lambda_2$, the sum of the logarithmic derivatives on the right-hand side of \eqref{identity1} converges to zero, canceling the pole coming from \eqref{poleidentity1}, thus the right-hand side of \eqref{identity1} is pole free.

It is a simple exercise to prove that the left and right-hand side of \eqref{identity1} have the same monodromy over the loops $A_j$ and $B_j$ (in particular observe that
\begin{equation*} \partial_j \log \theta\left(\int_{\lambda_2}^{\lambda_1}\boldsymbol{\omega}+v \right)-\partial_j\log \theta\left[{\substack{\ualpha \\ \ubeta}}\right]\left(\int^{\lambda_1}_z\boldsymbol{\omega}\right)\end{equation*}
has no monodromy over the loops $A_j$ and $B_j$). By the Jacobi inversion theorem (see \cite[Section III.6.6]{FK}), there are points $z_1,\dots,z_{k-1}\in \mathcal S$ such that $v\equiv-K-\sum_{j=1}^{k-1}u(z_j)\mod \Lambda$. If $\prod_{j=1}^{k-1}z_j$ is not a special divisor, then the $LHS$ has $k-1$ zeros at $z_1,\dots,z_{k-1}$, and thus the function $RHS/LHS$ is a meromorphic function with possible poles at $z_1,\dots,z_{k-1}$. However since $\prod_{j=1}^{k-1}z_j$ is not special, it follows that any meromorphic function whose poles form a subset of $\{z_1,\dots, z_{k-1}\}$ is a constant, and thus $RHS/LHS$ is a constant. By taking the limit $\lambda_1\to z$, we obtain that the $RHS/LHS$ is equal to 1. On the other hand, if $\prod_{j=1}^{k-1} z_j$ is special, then we pick a sequence $v_n\equiv-K-\sum_{j=1}^{k-1} u(z^{(n)}_j)\mod \Lambda$ such that $z^{(n)}_j\to z_j$ and $\prod_{j=1}^{n-1}z_j^{(n)}$ is not special. Then the lemma holds for each $v_n$, and taking the limit $n\to \infty$ we obtain the lemma for $v$.
\end{proof}

Define $\gamma$ on $\mathcal S$ such that $\gamma(z)=\gamma(\mathbf Pz)$ for $z$ in the first sheet of $\mathcal S\setminus J_\mathcal S$, and extend $\gamma$ analytically to $\mathcal S\setminus J_\mathcal S$. Then for $\lambda$ in the second sheet of $\mathcal S$, we have $\gamma(\lambda)=i\gamma(\mathbf P \lambda)$.

Similarly, $u$ is uniquely defined on $\mathcal S\setminus J_\mathcal S$ by letting $u(z)=u(\mathbf{P} z)$ for $z$ in the first sheet of $\mathcal S$ and extending analytically to $\mathcal S\setminus J_\mathcal S$.
\begin{lemma} \label{lem:thetaids2}
Let $\ualpha=\frac{1}{2}e_1$ and let $\ubeta=\frac{1}{2}\sum_{j=1}^{k-1}e_j$. For $z, \lambda\in \mathcal S\setminus \left( \{a_j,b_j\}_{j=1}^{k-1}\cup \infty_1\cup \infty_2\right)$, we have 
\begin{multline}\label{thetaids2}\frac{\theta(0)^2\left(\frac{\gamma(z)}{\gamma(\lambda)}+\frac{\gamma(\lambda)}{\gamma(z)}\right)^2}{4\theta(u(z)-u(\lambda))^2} \\
= \frac{(\mathbf P z-\mathbf P \lambda)^2\left( \sum_{j=1}^{k-1}u_j'(z)\partial_j \theta\left[{\substack{\ualpha \\ \ubeta}}\right](0)\right)\left(\sum_{j=1}^{k-1}u_j'(\lambda)\partial_j\theta\left[{\substack{\ualpha \\ \ubeta}}\right](0)\right)}{\theta\left[{\substack{\ualpha \\ \ubeta}}\right](u(z)-u(\lambda))^2}.
\end{multline}
\end{lemma}
\begin{proof}
The left and the right-hand side have the same monodromy around the loops $A_j$ and $B_j$, and thus the $LHS/RHS$ is a meromorphic function. By Lemma \ref{lemzerosgamma}, the left hand side has divisor  ${\rm Div}(LHS)=(\lambda^*)^2\prod_{j=1}^ka_j^{-1}b_j^{-1}$, as a function of $z$ (where $\lambda^*$ is such that $\mathbf{P}\lambda=\mathbf{P}\lambda^*$ and $\lambda \neq \lambda^*$). By Lemma \ref{Prime} and Lemma \ref{Prime2}, the right-hand side has divisor  ${\rm Div}(RHS)=(\lambda^*)^2\prod_{j=1}^ka_j^{-1}b_j^{-1}$, as a function of $z$. Thus the $LHS/RHS$ is constant, and the constant is determined by setting $z=\lambda$.
\end{proof}

We now continue the proof of Proposition \ref{PropForM}.

Multiplying the left hand side of \eqref{identity2} with the left hand side of \eqref{thetaids2}, 
\begin{multline}\label{ForM11M223}
\theta\left(\int_z^\lambda \boldsymbol{\omega}+v\right)\theta\left(\int_z^\lambda\boldsymbol{\omega}-v\right)\frac{\theta(0)^2\left(\frac{\gamma(z)}{\gamma(\lambda)}+\frac{\gamma(\lambda)}{\gamma(z)}\right)^2}{4\theta(v)^2\theta(u(z)-u(\lambda))^2}\\= (\mathbf P z-\mathbf P \lambda)^2 \left( \frac{\partial}{\partial z}\frac{\partial}{\partial \lambda}\log \theta \left[{\substack{\ualpha \\ \ubeta}}\right] (u(z)-u(\lambda))+ \sum_{j=1}^{k-1}\sum_{i=1}^{k-1}u_j'(z)u_i'(\lambda)\partial_i\partial_j\log \theta (v)\right).
\end{multline}
Setting $v=0$, we obtain
\begin{multline} \label{ForM11M224}
\frac{\partial}{\partial z}\frac{\partial}{\partial \lambda}\log \theta \left[{\substack{\ualpha \\ \ubeta}}\right] (u(z)-u(\lambda))=\frac{1}{4(\mathbf P z-\mathbf P \lambda)^2}\left(\frac{\gamma(z)}{\gamma(\lambda)}+\frac{\gamma(\lambda)}{\gamma(z)}\right)^2\\ -\sum_{j=1}^{k-1}\sum_{i=1}^{k-1}u_j'(z)u_i'(\lambda)\partial_i\partial_j\log \theta (0).
\end{multline}
In particular, since $u(\lambda)=-u(\lambda^*)$ and $\gamma(\lambda)^2=-\gamma(\lambda^*)^2$, it follows that 
\begin{multline}\label{ForM11M224extra}
\frac{\partial}{\partial z}\frac{\partial}{\partial \lambda}\log \theta \left[{\substack{\ualpha \\ \ubeta}}\right] (u(z)+u(\lambda))
=\frac{\partial}{\partial z}\frac{\partial}{\partial \lambda}\log \theta \left[{\substack{\ualpha \\ \ubeta}}\right] (u(z)-u(\lambda^*))
\\=-\frac{1}{4(\mathbf P z-\mathbf P \lambda)^2}\left(\frac{\gamma(z)}{\gamma(\lambda)}-\frac{\gamma(\lambda)}{\gamma(z)}\right)^2 +\sum_{j=1}^{k-1}\sum_{i=1}^{k-1}u_j'(z)u_i'(\lambda)\partial_i\partial_j\log \theta (0).
\end{multline}
Substituting \eqref{ForM11M224} and \eqref{ForM11M224extra} into the definition of $W$ in \eqref{def:W}, we obtain that
\begin{equation}\label{ForM11M225}\frac{\partial}{\partial z}\frac{\partial}{\partial \lambda}\log \theta \left[{\substack{\ualpha \\ \ubeta}}\right] (u(z)-u(\lambda))=\frac{1}{2}W(z,\lambda)+\frac{1}{2(\mathbf Pz-\mathbf P\lambda)^2}.
\end{equation}
Substituting \eqref{ForM11M225} into \eqref{ForM11M223}, we obtain \eqref{ForM11M22}.

It follows from \eqref{thetaids2} that for $z,\lambda\in \mathcal S\setminus J_\mathcal S$,
\begin{multline}\label{thetaids3}\frac{\theta(0)\left(\frac{\gamma(z)}{\gamma(\lambda)}+\frac{\gamma(\lambda)}{\gamma(z)}\right)}{2\theta(u(z)-u(\lambda))}\\
= \frac{(\mathbf Pz- \mathbf P\lambda)\sqrt{\left( \sum_{j=1}^{k-1}u_j'(z)\partial_j \theta\left[{\substack{\ualpha \\ \ubeta}}\right](0)\right)\left(\sum_{j=1}^{k-1}u_j'(\lambda)\partial_j\theta\left[{\substack{\ualpha \\ \ubeta}}\right](0)\right)}}{\theta\left[{\substack{\ualpha \\ \ubeta}}\right](u(z)-u(\lambda))}.
\end{multline}
The square root is chosen such that it is equal to $\left( \sum_{j=1}^{k-1}u_j'(z)\partial_j \theta\left[{\substack{\ualpha \\ \ubeta}}\right](0)\right)$ when $\lambda=z$, and by Lemma \ref{Prime2} the branches may be chosen such that it is analytic on $\mathcal S\setminus J_{\mathcal S}$.
In particular it follows from Lemma \ref{Prime2} that  $\sum_{j=1}^{k-1}u_j'(z)\partial_j \theta\left[{\substack{\ualpha \\ \ubeta}}\right](0)$ has a pole at $b_k$, and by analytically continuing around $b_k$ and recalling the definition of $J_{\mathcal S}$, it follows that if $z$ is in the first sheet, then
\begin{equation}\label{sqrtpp*}
\sqrt{\sum_{j=1}^{k-1}u_j'(z^*)\partial_j \theta\left[{\substack{\ualpha \\ \ubeta}}\right](0)}=i\sqrt{\sum_{j=1}^{k-1}u_j'(z)\partial_j \theta\left[{\substack{\ualpha \\ \ubeta}}\right](0)}.
\end{equation}
Recall that for $z$ in the first sheet, $\gamma(z^*)=i\gamma(z)$. Denote the $LHS$ of \eqref{thetaids3} by $LHS(z,\lambda)$. Taking the product $LHS(z,\lambda)\times LHS(z^*,\lambda)$ and setting it equal to $RHS(z,\lambda)\times RHS (z^*,\lambda)$, and substituting $\gamma(z^*)=i\gamma(z)$, $u(z^*)=-u(z)$, and \eqref{sqrtpp*},  we obtain 
\begin{multline}\label{ForM11M12part1}
\frac{\theta(0)^2\left(\left(\frac{\gamma(z)}{\gamma(\lambda)}\right)^2-\left(\frac{\gamma(\lambda)}{\gamma(z)}\right)^2\right)}{4\theta(u(z)-u(\lambda))\theta(u(z)+u(\lambda))}\\ = \frac{(\mathbf Pz-\mathbf P\lambda)^2\left( \sum_{j=1}^{k-1}u_j'(z)\partial_j \theta\left[{\substack{\ualpha \\ \ubeta}}\right](0)\right)\left(\sum_{j=1}^{k-1}u_j'(\lambda)\partial_j\theta\left[{\substack{\ualpha \\ \ubeta}}\right](0)\right)}{\theta\left[{\substack{\ualpha \\ \ubeta}}\right](u(z)-u(\lambda))\theta\left[{\substack{\ualpha \\ \ubeta}}\right](u(z)+u(\lambda))}.
\end{multline}
Setting $\lambda_1=\lambda_2^*=\lambda$ in \eqref{identity1}, and recalling that $u(\lambda)=-u(\lambda^*)$, we obtain
\begin{multline}
\label{ForM11M12part2}
\theta\left(u(z)+u(\lambda)+v \right)\theta\left(u(z)-u(\lambda)-v \right)
\\
=-\frac{\theta\left[{\substack{\ualpha \\ \ubeta}}\right]\left( u(z)+u(\lambda)\right)\theta\left[{\substack{\ualpha \\ \ubeta}}\right]\left(u(\lambda)-u(z) \right)}{\theta\left[{\substack{\ualpha \\ \ubeta}}\right]\left(2u(\lambda)\right)
\left(\sum_{j=1}^{k-1}u_j'(z)\partial_j\theta\left[{\substack{\ualpha \\ \ubeta}}\right](0)\right)} \theta\left(v \right)\theta\left(2u(\lambda)+v \right)
\\ \times \sum_{j=1}^{k-1} u_j'(z)\Big( \partial_j \log \theta\left(2u(\lambda)+v \right)-\partial_j\log \theta\left[{\substack{\ualpha \\ \ubeta}}\right]\left(u(z)+u(\lambda)\right)\\ 
+\partial_j\log \theta\left[{\substack{\ualpha \\ \ubeta}}\right]\left(u(z)-u(\lambda)\right)-\partial_j \log \theta\left(v \right)
\Big),
\end{multline}
Multiplying \eqref{ForM11M12part1} and \eqref{ForM11M12part2} we obtain 
\begin{multline}\label{ForM11M12part3}
\frac{\theta(0)^2\left(\left(\frac{\gamma(z)}{\gamma(\lambda)}\right)^2-\left(\frac{\gamma(\lambda)}{\gamma(z)}\right)^2\right)\theta\left(u(z)+u(\lambda)+v \right)\theta\left(u(z)-u(\lambda)-v \right)}{4\theta(u(z)-u(\lambda))\theta(u(z)+u(\lambda))}\\
=\frac{(\mathbf Pz-\mathbf P\lambda)^2\left(\sum_{j=1}^{k-1}u_j'(\lambda)\partial_j\theta\left[{\substack{\ualpha \\ \ubeta}}\right](0)\right)}{\theta\left[{\substack{\ualpha \\ \ubeta}}\right]\left(2u(\lambda)\right)}
\theta\left(v \right)\theta\left(2u(\lambda)+v \right) \\
\times \sum_{j=1}^{k-1} u_j'(z)\Big( \partial_j \log \theta\left(2u(\lambda)+v \right)-\partial_j\log \theta\left[{\substack{\ualpha \\ \ubeta}}\right]\left(u(z)+u(\lambda)\right)\\ 
+\partial_j\log \theta\left[{\substack{\ualpha \\ \ubeta}}\right]\left(u(z)-u(\lambda)\right)-\partial_j \log \theta\left(v \right)
\Big),
\end{multline}
Now consider
\begin{equation} \label{funtheta1} \frac{\theta(0)^2\theta \left[{\substack{\ualpha \\ \ubeta}}\right] (u(z)-u(\lambda))^2 \left(\frac{\gamma(z)}{\gamma(\lambda)}+\frac{\gamma(\lambda)}{\gamma(z)}\right)^2}{4(\mathbf Pz-\mathbf P\lambda)^2\theta(u(z)-u(\lambda))^2\left(\sum_{j=1}^{k-1}u_j'(\lambda)\partial_j\theta\left[{\substack{\ualpha \\ \ubeta}}\right](0)\right)\left(\sum_{j=1}^{k-1}u_j'(z)\partial_j\theta\left[{\substack{\ualpha \\ \ubeta}}\right](0)\right)}
\end{equation} 
as a function of $z$.
By Lemma \ref{lemDivthetadiff}, Lemma \ref{Prime}, and Lemma \ref{Prime2}, \eqref{funtheta1} has no zeros or poles. By \eqref{eq:quasi}, \eqref{eq:Aint} and \eqref{eq:period}, \eqref{funtheta1} has no monodromy on $\mathcal S$ and is  a meromorphic function. Thus it is constant. By taking $z\to \lambda$, we obtain that \eqref{funtheta1} is identically equal to $1$. Recall from \eqref{eq:gamma} that $\gamma(\lambda)^2=-\gamma(\lambda^*)^2$, and that
 \begin{equation*}\gamma'(\lambda)=\frac{\gamma(\lambda)}{4}\left(\sum_{j=1}^k \frac{1}{\mathbf P\lambda-b_j}-\frac{1}{\mathbf P\lambda-a_j}\right).
\end{equation*}
By taking $z\to \lambda^*$ and taking the square root in \eqref{funtheta1}, we obtain
\begin{equation}\label{funtheta2}
\frac{\theta(0)\theta\left[{\substack{\ualpha \\ \ubeta}}\right](2u(\lambda))}{4\theta(2u(\lambda))\left(\sum_{j=1}^{k-1}u_j'(\lambda)\partial_j\theta\left[{\substack{\ualpha \\ \ubeta}}\right](0)\right)}\left(\sum_{j=1}^k\frac{1}{\mathbf P\lambda-b_j}-\frac{1}{\mathbf P\lambda-a_j}\right)=1.
\end{equation}
(To verify that the correct sign was taken when taking the square root we simply take $\lambda\to b_k$ in \eqref{funtheta2} and verify that the left hand side equals $1$ in this limit). Substituting \eqref{funtheta2} into the right-hand side of \eqref{ForM11M12part3} and dividing by $\theta(0)^2\theta(v)^2$, we obtain \eqref{ForM11M12}, concluding the proof of Proposition \ref{PropForM}.

\section{Analysis of the Riemann-Hilbert problem $Y$}\label{sec:trans}
As is standard in the analysis of RH problems (see \cite{Deift}), we will apply the method of non-linear steepest descent developed by Deift/Zhou to evaluate the large $N$-asymptotics for the Riemann-Hilbert problem for $Y$. In the multi-cut setting that we are interested in, this was developed by \cite{DKMVZ}, and we have also been inspired by the work of \cite{KuijVanlessen}. To avoid repetition, we consider a generic RH problem $Y(z) = Y_N(z)$ associated with the symbol $\nu_\epsilon(x)=F(x)e^{-NV(x)}\omega_\epsilon(x)$ from Section \ref{sec:FHproof}. Observe that we omit the parameters $s$ and $t$ in our notation but keep the parameter $\epsilon$ despite the fact that we deform the symbol in each of these variables. In this section, we discuss some transformations of the generic RH problem that are common to all of the RH problems relevant to us. Note that one should interpret $\omega_{\epsilon}(x) \equiv 1$ when $p = 0$, which corresponds to the situation when we deform our potential $V$ and interpolate our function $f$ as in Section \ref{sec:pfasy} and Section \ref{sec:smasy} respectively. 
\subsubsection*{The first transformation $Y \mapsto T$.}
Throughout this section, we will make use of the notation of Section \ref{sec:intro_em}. Let us introduce the ``$g$-function" 
\begin{align} \label{eq:g_fn}
g(z) = \int_{\R} \log (z-\lambda) d\mu_V(\lambda) = \int_{J} \log(z-\lambda) \psi_V(\lambda)d\lambda, \qquad z\in \mathbb C \setminus (-\infty, b_k],
\end{align}
where the principal branch is chosen for the logarithm. The $g$-function is analytic in $\mathbb C \setminus (-\infty, b_k]$ with continuous boundary values on $(-\infty,b_k]$ which we denote by $g_{\pm}$. Since
\begin{align*}
g_{+}(x) + g_{-}(x) = 2\int_\R \log|\lambda - x| d\mu_V(\lambda), \qquad x \in \R,
\end{align*}

\noindent the Euler-Lagrange equations \eqref{eq:EL1} and \eqref{eq:EL2}, combined with our assumption that the inequality in \eqref{eq:EL2} is strict, may be rewritten as
\begin{align}
\label{eq:EL1s}
g_{+}(x) + g_{-}(x) - V(x) + \ell = 0 & \qquad \text{for } x \in J, \\
\label{eq:EL2s}
g_{+}(x) + g_{-}(x) - V(x) + \ell < 0  & \qquad \text{for } x \in J^c.
\end{align}
We can then define our first transformation. For $z\in \C\setminus\R$, let
\begin{align}\label{eq:Tdef}
T(z) = e^{\frac{N}{2}\ell \sigma_3} Y(z) e^{-N(g(z) +\frac{\ell}{2}) \sigma_3}
\end{align}
\noindent where
\begin{align*}
\sigma_3 = \begin{pmatrix} 1 & 0 \\ 0 & -1 \end{pmatrix}, \qquad x^{\sigma_3}=\begin{pmatrix} x&0\\0&x^{-1}\end{pmatrix}.
\end{align*}
Relying on the fact that $g$ is analytic on $\mathbb C\setminus (-\infty,b_k]$ and the fact that $g(z)=\log(z)+\mathcal O(1/z)$ as $z\to \infty$, it follows that $T$ satisfies the the following RH problem.

\subsubsection{The RH problem for $T$} 

\begin{itemize}[leftmargin=0.75cm]
\item[$(a)$] $T: \C \setminus \R \to \C^{2 \times 2}$ is analytic.
\item[$(b)$] $T$ has continuous boundary values $T_\pm(x) = \lim_{\epsilon \to 0^+} T(x \pm i\epsilon)$ on $\R$ which satisfy the jump condition
\begin{align}\label{eq:Tjump}
T_+(x) = T_-(x) \begin{pmatrix}
e^{-N(g_{+}(x) - g_{-}(x))} &F(x)e^{N(g_{+}(x) + g_{-}(x) - V(x) + \ell)}\omega_{\epsilon}(x)\\
0 & e^{N(g_{+}(x) - g_{-}(x))}
\end{pmatrix}
\end{align}
\item[$(c)$] As $z \to \infty$,
\begin{align}\label{eq:Tinfty}
T(z) = I + \mathcal O(z^{-1}).
\end{align}
\end{itemize}

\subsubsection*{The second transformation $T \mapsto S$.} In this step we perform the ``opening of lenses". Consider the function
\begin{align}\label{eq:phidef}
\phi(z) := - \pi i \int_{b_k}^z \psi_V(\lambda) d\lambda,  \qquad z \in U_{V} \setminus (-\infty, b_k),
\end{align}
where the path of integration lies in $U_{V} \setminus (-\infty, b_k]$, and we recall that $U_{V}$ is the domain of analyticity of $V$. We observe that for $z\in J$, taking the derivative of \eqref{eq:g_fn}, we have by a residue calculation that  $\frac{d}{dz}(g_+(z)-g_-(z))=-2\pi i \psi_V(z)$, and combined with \eqref{eq:EL1s}, we obtain 
\begin{equation*}\begin{aligned}\frac{d}{dz} \left(g_+(z)-\frac{V(z)}{2}\right)&=\frac{d}{dz}\left(\frac{g_+(z)-g_-(z)}{2}+\frac{g_+(z)+g_-(z)-V(z)+\ell}{2}\right)\\ &=-\pi i \psi_V(z). \end{aligned} \end{equation*}
Since $\frac{d}{dz}\phi_+(z)=-\pi i \psi_V(z)$ on $J$, and $\phi(b_k)=g(b_k)-V(b_k)/2+\ell/2$
we obtain by analytic continuation onto $U_V\setminus (-\infty, b_k]$ that 
\begin{equation}\label{formulaphi} \phi(z) = g(z) - V(z)/2 + \ell/2.\end{equation} In particular, we have
\begin{align}\label{jumpphi1}
g_{+}(x) + g_{-}(x) - V(x) + \ell & = \phi_{+}(x) + \phi_{-}(x), && \text{for} \qquad x\in \R,\\
\phi_+(x)+\phi_-(x)&=0,&& \text{for} \qquad x\in J,\label{jumpphi2}\\
\pm(g_{+}(x) - g_{-}(x)) &= 2\phi_{\pm}(x),  && \text{for} \qquad x\in J,\label{jumpphi3}
\end{align}
where we observe that \eqref{jumpphi2} is obtained from \eqref{jumpphi1} and \eqref{eq:EL1s}.
Using these identities together, we rewrite the jumps for $T$ in terms of $\phi$. Furthermore, for $x \in J$, the matrix $T_{-}(x)^{-1}T_{+}(x)$ can be factorized as
\begin{equation}\begin{aligned}
\label{Factorization} &\begin{pmatrix}
e^{-2N\phi_{+}(x)} & e^{f(x)} \omega_{\epsilon}(x)\\
0 & e^{-2N \phi_{-}(x)}
\end{pmatrix}
= \begin{pmatrix}
1 & 0 \\
e^{-f(x)} \omega_{\epsilon}(x)^{-1} e^{-2N \phi_{-}(x)} & 1
\end{pmatrix}\\
& \qquad \qquad \times \begin{pmatrix}
0 & e^{f(x)} \omega_{\epsilon}(x) \\
- e^{-f(x)} \omega_{\epsilon}(x)^{-1} & 0
\end{pmatrix}
\begin{pmatrix}
1 & 0 \\
e^{-f(x)} \omega_{\epsilon}(x)^{-1} e^{-2N \phi_{+}(x)} & 1
\end{pmatrix},
\end{aligned}\end{equation}
where we observe that the left hand side indeed is the same jump matrix as that in \eqref{eq:Tjump} by \eqref{jumpphi2} and \eqref{jumpphi3}, and that equality between the left and the right-hand side hold by \eqref{jumpphi2}.
Around each interval $[a_j,b_j]$, we open lenses $J_j^+$ and $J_j^-$ as shown in Figure \ref{ContourS}, which are contained in $U_V$. The contour $J_j^+$ is in the upper half plane and $J_j^-$ is in the lower half plane. Moreover, both $J_j^+$ and $J_j^-$ start at $a_j$, end at $b_j$, and if any of the $t_l$ are in the interval $(a_j,b_j)$, then the contours pass nearby these points as well. There is some freedom in the choice of the lenses. In particular we require that $\Re \phi(z)$ is positive on $J_j^+$ and $J_j^-$,  see Lemma \ref{expsmalljump} below. The shape of the contours  $J_j^+$ and $J_j^-$ near the points $a_j,b_j$ and $t_l$ will be specified  in Section \ref{sec:para}. For now we mention that when $J_j^+$ and $J_j^-$ meets the points $a_j$ and $b_j$, each of them forms an angle $\pi/3$ with $(a_j,b_j)$.

We write $\Sigma_S= \cup_{j=1}^k (J_j^+\cup J_j^-)\cup \R$ for the contour consisting of the lenses and the real axis.
\begin{figure}[t]
	\begin{center}
		\begin{picture}(100,80)(-5,-40)
		\put(-138,-2){$a_1$}
		\put(-28,-2){$b_1$}
		\put(162,-2){$b_k$}
		\put(51,-2){$a_k$}
		
		\put(-102,2){$t_1$}
		\put(-62,2){$t_2$}
		\put(108,2){$t_p$}
		
		\put(-150,-5){\line(1,0){150}}
		\put(10,-6){$\dots\dots$}
		\put(50,-5){\line(1,0){150}}
		
		\put(140,10){\thicklines\vector(1,0){.0001}}
		\put(140,-5){\thicklines\vector(1,0){.0001}}
		\put(140,-20){\thicklines\vector(1,0){.0001}}		
		\put(90,10){\thicklines\vector(1,0){.0001}}
		\put(90,-5){\thicklines\vector(1,0){.0001}}
		\put(90,-20){\thicklines\vector(1,0){.0001}}		

		\put(-76,10){\thicklines\vector(1,0){.0001}}
		\put(-76,-5){\thicklines\vector(1,0){.0001}}
		\put(-76,-20){\thicklines\vector(1,0){.0001}}

		\put(-40,10){\thicklines\vector(1,0){.0001}}
		\put(-40,-5){\thicklines\vector(1,0){.0001}}
		\put(-40,-20){\thicklines\vector(1,0){.0001}}

		\put(-110,10){\thicklines\vector(1,0){.0001}}
		\put(-110,-5){\thicklines\vector(1,0){.0001}}
		\put(-110,-20){\thicklines\vector(1,0){.0001}}

		\put(-135,-5.15){\thicklines\vector(1,0){.0001}}
		\put(-10,-5.15){\thicklines\vector(1,0){.0001}}
		\put(180,-5.15){\thicklines\vector(1,0){.0001}}
		
		\put(190,-2){$\mathbb{R}$}
		\put(-85,20){$J_{1}^{+}$}
		\put(108,20){$J_{k}^{+}$}
		\put(-85,-35){$J_{1}^{-}$}
		\put(108,-35){$J_{k}^{-}$}
		
		\qbezier(-128,-5)(-113,25)(-98,-5)
		\qbezier(-98,-5)(-78,25)(-58,-5)
		\qbezier(-58,-5)(-43,25)(-28,-5)
		\qbezier(-128,-5)(-113,-35)(-98,-5)
		\qbezier(-98,-5)(-78,-35)(-58,-5)
		\qbezier(-58,-5)(-43,-35)(-28,-5)
        \qbezier(62,-5)(87,25)(112,-5)
        \qbezier(62,-5)(87,-35)(112,-5)
        \qbezier(112,-5)(137,25)(162,-5)
        \qbezier(112,-5)(137,-35)(162,-5)
		\end{picture}
		\caption{The contour $\Sigma_S$, in a situation where $t_1,t_2 \in (a_1,b_1)$ and $t_{p}\in (a_k,b_k)$. }
		\label{ContourS}
	\end{center}
\end{figure}
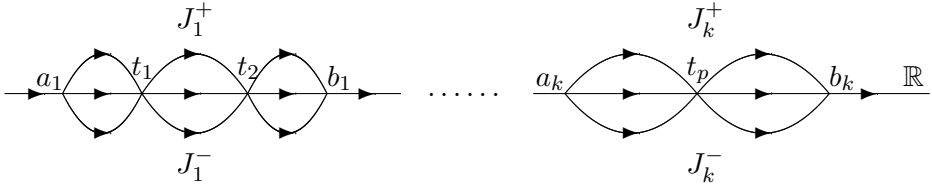
The next transformation is defined by
\begin{align} \label{eq:Sdef}
S(z) = T(z) \times
\begin{cases}
I, & \text{for } z \text{ outside the lenses},\\
\begin{pmatrix}
1 & 0 \\
-e^{-f(z)} \omega_{\epsilon}(z)^{-1} e^{-2N \phi(z)} & 1
\end{pmatrix},
& \begin{subarray}{l} \ds \text{for } z \text{ inside the lenses, } \\[0.1cm] \ds \hspace{0.6cm} \Im(z) > 0, \end{subarray} \\
\begin{pmatrix}
1 & 0 \\
e^{-f(z)} \omega_{\epsilon}(z)^{-1} e^{-2N \phi(z)} & 1
\end{pmatrix},
& \begin{subarray}{l} \ds \text{for } z \text{ inside the lenses, } \\[0.1cm] \ds \hspace{0.6cm} \Im(z) < 0. \end{subarray}
\end{cases}
\end{align}
Before stating the RH problem for $S$, we first note the following. For $j = 1, \dots, k-1$ and $x \in (b_j, a_{j+1})$,
\begin{align}\label{lol3}
g_{+}(x) - g_{-}(x)
= 2\pi i \int_{a_{j+1}}^{b_k} \mu_V(dx) = 2 \pi i\Omega_j.
\end{align}
By combining  \eqref{Factorization} and \eqref{lol3}, and the fact that $\phi$ is analytic on $\mathbb C\setminus(-\infty,b_k]$, we obtain that $S$ defined by \eqref{eq:Sdef} solves the following RH problem.

\subsubsection*{The RH problem for $S$}

\begin{itemize}[leftmargin=0.75cm]
\item[$(a)$] $S: \C \setminus \Sigma_S \to \C^{2\times 2}$ is analytic, where $\Sigma_S$ is as in Figure \ref{ContourS}.
\item[$(b)$] $S$ has the  jump relation $S_+(z)=S_-(z)J_S(z)$ for $z\in\Sigma_S$ $($where the orientation is as in Figure \ref{ContourS}$)$.
\begin{align}
& J_S(z)=\begin{pmatrix}
1 &F(z) e^{N(\phi_{+}(z) + \phi_{-}(z))}\omega_{\epsilon}(z) \\
0 & 1
\end{pmatrix}, & & z \in \R \setminus [a_1, b_k], \label{eq:Sjump1} \\
& J_S(z)= \begin{pmatrix}
0 & e^{f(z)}\omega_{\epsilon}(z) \\
- e^{-f(z)}\omega_{\epsilon}(z)^{-1} & 0
\end{pmatrix}, & & z \in \cup_{j=1}^k(a_j,b_j), \label{eq:Sjump2} \\
& J_S(z)=\begin{pmatrix}
e^{-2 \pi i N \Omega_j} &  F(z) e^{N(\phi_{+}(z) + \phi_{-}(z))}\omega_{\epsilon}(z)\\
0 & e^{2 \pi i N \Omega_j}
\end{pmatrix}, & & z \in (b_j, a_{j+1}), \label{eq:Sjump3} \\
& J_S(z)= \begin{pmatrix}
1 & 0\\
 e^{-f(z)}\omega_{\epsilon}(z)^{-1} e^{-2N \phi(z)} & 1
\end{pmatrix}, & & z \in (J^{+}\cup J^{-})\setminus  \{a_j,b_j\}_{j=1}^k. \label{eq:Sjump4}
\end{align}
where $J^{\pm} := \cup_j J_j^{\pm}$.
\item[$(c)$] $S(z) = I + \mathcal{O}(z^{-1})$ as $z \to \infty$.
\item[$(d)$] $S(z) = \mathcal{O}(1)$ as $z$ approaches any of the points of self-intersection of the contour $\Sigma_S$.
\end{itemize}
A crucial fact to our future analysis, see Lemma \ref{expsmalljump} below,  is that the jumps \eqref{eq:Sjump1} and \eqref{eq:Sjump4} converge pointwise to the identity $I$ as $N\to \infty$, and for $z$ bounded away from $a_j,b_j,t_j$, they tend uniformly to the identity.

\section{Main parametrix} \label{sec:global}
We construct a function $M(z)$ which is analytic on $\mathbb C\setminus [a_1,b_k]$, which has the same jumps as $S$ on $[a_1,b_k]\setminus \{a_j,b_j\}_{j=1}^k$, and which satisfies $M(z)=I+\mathcal O\left(z^{-1}\right)$ as $z\to \infty$. Then the jumps of $SM^{-1}$ will converge uniformly  to the identity except on open neighbourhoods of the points $a_j,b_j,t_j$. We divide the construction of such a function into two parts. First we deal with the special case $F(z)\omega_\epsilon(z)\equiv 1$ which is precisely the situation considered in the fundamental works \cite{DKMVZ,DIZ}. Subsequently we extend to the general situation $F(z)\omega_\epsilon(z)\not \equiv 1$, where we are inspired by the construction of the main parametrix considered by  Kuijlaars and Vanlessen in \cite{KuijVanlessen}, but also provide a new representation for this type of main parametrix. In Remark \ref{KuijVan} below, we comment on how our representation may be brought to the form presented in \cite{KuijVanlessen}.
\subsection{Main parametrix for $F(z)\omega_\epsilon(z)\equiv 1$}
Let $\lambda \in \mathbb C\cup\{ \infty \} \setminus J$, and let $v_1,\dots, v_{k-1}$ be real. Consider the following RH problem
\subsubsection*{RH problem for $N_\lambda$}
\begin{itemize}
\item[$(a)$] $N_\lambda $ is analytic on $\C \setminus [a_1, b_k] $.
\item[$(b)$] $N_\lambda$ has continuous boundary values $N_{\lambda,\pm}$ on $[a_1, b_k] \setminus  \{a_j, b_j\}_{j=1}^k$ satisfying
\begin{align}
N_{\lambda,+}(x) =N_{\lambda,-}(x) \times
\begin{cases}
\begin{pmatrix}
0 & 1 \\ -1 & 0 \end{pmatrix}, & \qquad x \in \bigcup_{j=1}^k (a_j, b_j), \\
e^{-2 \pi i v_j \sigma_3}, & \qquad x \in (b_j, a_{j+1}), \quad j = 1, \dots, k-1.
\end{cases}
\end{align}

\item[$(c)$] $N_\lambda(z) = I + \mathcal{O}(z-\lambda)$ as $z \to \lambda$ for $\lambda \in \mathbb C\setminus J$, while if $\lambda=\infty$, then $N_\infty(z)=I+\mathcal O(z^{-1})$ as $z\to \infty$.

\item[$(d)$] $N_\lambda(z) = \mathcal{O}((z-z_\ast)^{-\frac{1}{4}})$ as $z \to z_\ast \in  \{a_j, b_j\}_{j=1}^k$.
\item[$(e)$] As $z\to \infty$, $N_\lambda(z)\to N_\lambda(\infty)$ for some constant $N_\lambda(\infty)$.
\end{itemize}

In \eqref{eq:Mlamdef} below, we construct an explicit solution to the RH problem for $N_\lambda$.
When we wish to emphasize the dependence of $N_\lambda(z)$ on $v=(v_1,\dots,v_{k-1})^T$, we denote $N_{\lambda}(z;v)$.
In the special case $F(z)\omega_\epsilon(z)\equiv 1$, our main parametrix is given by $N_\infty(z;N\Omega)$.
While we only rely on $N_\infty$ to construct a main parametrix for the RH problem for $S$, we will later rely on the fact that 
\begin{equation}\label{Ninftylambda}N_\infty(\lambda)^{-1}N_\infty(z)=N_\lambda(z).\end{equation} We will construct a function which solves the RH problem for $N_\lambda$. The first step is to construct a solution for the special case where $v_1,\dots,v_{k-1}=0$, which we denote by $N_\lambda(z;0)$. This is given by
\begin{align*}
N_\lambda(z;0) = 
\begin{pmatrix}
\frac{1}{2}\left(\frac{\gamma(z)}{\gamma(\lambda)} +\frac{\gamma(\lambda)}{ \gamma(z)}\right)
&\frac{1}{2i}\left(\frac{\gamma(z)}{\gamma(\lambda)} -\frac{\gamma(\lambda)}{ \gamma(z)}\right)\\
-\frac{1}{2i}\left(\frac{\gamma(z)}{\gamma(\lambda)} -\frac{\gamma(\lambda)}{ \gamma(z)}\right)
& \frac{1}{2}\left(\frac{\gamma(z)}{\gamma(\lambda)} +\frac{\gamma(\lambda)}{ \gamma(z)}\right)
\end{pmatrix},
\end{align*}
where $\gamma$ was defined in \eqref{eq:gamma}. Then it follows that $N_\lambda(z;0)$ is analytic on  $\C\setminus J$ and that as $z\to \lambda$, $N_\lambda(z;0)\to I$. For $x\in (a_j,b_j)$, we have $\gamma_+(x)=i\gamma_-(x)$, and thus
\begin{equation*}N_{\lambda,+}(x)=N_{\lambda,-}(x)\begin{pmatrix}0&1\\-1&0\end{pmatrix}. \end{equation*}
Thus $N_\lambda(z;0)$ solves the RH problem for $N_\lambda$ in the special case where $v_1,\dots,v_{k-1}=0$.

To construct $N_{\lambda}$ for general values of $v_{1},\ldots,v_{k-1}$, we follow \cite{DIZ}, and use $\theta$ functions and the Abel map.
Recall that $\boldsymbol{\omega_j}$ was the unique holomorphic 1-form on $\mathcal S$ satisfying \eqref{eq:Aint}, and is of the form \eqref{formomega}. 
 
Recall $\mathsf Q_j$ from \eqref{formomega}.
Since $\mathsf Q_j$ is of degree at most $k-2$, we note that the limit $u_j(\infty)=\lim_{z\to \infty}u_j(z)$ exists, and furthermore since the residue at $\infty$ is $0$, $u_j(z)$ extends to an analytic function on $\mathbb C\setminus[a_1,b_k]$. By \eqref{eq:Aint}, \eqref{eq:period}, and the fact that  $\frac{\mathsf Q_j}{\mathcal R^{1/2}}$ has zero residue at $\infty$, it is easily verified that
\begin{equation} \label{eq:ujump1} u_+(x)-u_-(x)=\tau_j, \qquad x\in (b_j,a_{j+1}),\end{equation}
for $j=1,\dots,k-1$, where $\tau_j$ denotes the $j$'th column vector of $\tau$ (so that $\tau=(\tau_1,\dots,\tau_{k-1})$), and that
\begin{equation}\label{eq:ujump2}
u_{+}(x)+u_{-}(x)=-2\sum_{s=j}^{k-1}\int_{b_s}^{a_{s+1}}\frac{\mathsf Q(y)}{\mathcal R^{1/2}(y)}dy=\sum_{s=j}^{k-1}\oint_{A_s}\boldsymbol\omega=\sum_{s=j}^{k-1}e_s,
\end{equation}
for $x\in (a_j,b_j)$ and $j=1,\dots,k$, where $(e_j)_l=\delta_{l,j}$. We observe that $u_+$ and $u_-$ are continuous at the points $a_j$ and $b_j$, and thus
\begin{equation}\label{uofa}
u_+(a_j)=\frac{u_+(a_j)+u_-(a_j)}{2}+\frac{u_+(a_j)-u_-(a_j)}{2}=\frac{1}{2}\sum_{s=j}^{k-1}e_s+\frac{\tau_{j-1}}{2},\end{equation}
for $j=2,\dots, k$, and
\begin{equation} \label{uofb} u_+(b_j)=\frac{1}{2}\sum_{s=j}^{k-1}e_s+\frac{\tau_j}{2},\end{equation}
for $j=1,\dots, k-1$.

Recall the definition of the $\theta$-function in \eqref{eq:thetadef}. For $z\in \mathbb C\setminus [a_1,b_k]$, define 
\begin{align*}
m_\lambda(z)=m_\lambda(z; v) = 
\frac{\theta(0)}{\theta(v)} \begin{pmatrix}
\frac{\theta(u(z) + v -u(\lambda))}{\theta(u(z) -u(\lambda))} 
& \frac{\theta(u(z) - v +u(\lambda))}{\theta(u(z) +u(\lambda))} \\
\frac{\theta(u(z) + v  +u(\lambda))}{\theta(u(z) +u(\lambda))} 
&  \frac{\theta(u(z) - v  -u(\lambda))}{\theta(u(z) -u(\lambda))}
\end{pmatrix}.
\end{align*}
It follows from Lemma \ref{lemDivthetadiff}  that $m$ is well defined and meromorphic on $z\in \mathbb C\setminus [a_1,b_k]$.

By \eqref{eq:quasi}, we have the following standard properties 
\begin{equation}\label{quasitheta}
\theta(\xi+e_j)=\theta(\xi) \qquad \text{and} \qquad \theta(\xi+\tau_j)=e^{-\pi i \tau_{j,j}-2\pi i \xi_j}\theta (\xi).
\end{equation}
Thus it follows by \eqref{eq:ujump2} that 
\begin{equation*}m_{\lambda,+}(x)=m_{\lambda,-}(x)\begin{pmatrix}0&1\\1&0\end{pmatrix}, \end{equation*}
for $x\in \cup_{j=1}^k (a_j,b_j)$, and by \eqref{eq:ujump1} that
\begin{equation*} m_{\lambda,+}(x)=m_{\lambda,-}(x)e^{-2\pi i v_j\sigma_3}, \end{equation*}
for $x\in (b_j,a_j+1)$, with $j=1,\dots,k-1$.
As a consequence, if we define
\begin{multline} \label{eq:Mlamdef}
N_\lambda(z)=\frac{\theta(0)}{2\theta(v)}\\
\times  \begin{pmatrix}
\left(\frac{\gamma(z)}{\gamma(\lambda)} +\frac{\gamma(\lambda)}{ \gamma(z)}\right)\frac{\theta(u(z) + v -u(\lambda))}{\theta(u(z) -u(\lambda))} 
&\frac{1}{i}\left(\frac{\gamma(z)}{\gamma(\lambda)} -\frac{\gamma(\lambda)}{ \gamma(z)}\right) \frac{\theta(u(z) - v +u(\lambda))}{\theta(u(z) +u(\lambda))} \\
-\frac{1}{i}\left(\frac{\gamma(z)}{\gamma(\lambda)} -\frac{\gamma(\lambda)}{ \gamma(z)}\right)\frac{\theta(u(z) + v  +u(\lambda))}{\theta(u(z) +u(\lambda))} 
& \left(\frac{\gamma(z)}{\gamma(\lambda)} +\frac{\gamma(\lambda)}{ \gamma(z)}\right) \frac{\theta(u(z) - v  -u(\lambda))}{\theta(u(z) -u(\lambda))}
\end{pmatrix},
\end{multline}
then we observe that
\begin{equation*}
\left(N_\lambda(z)\right)_{ij}=\left(m_{\lambda}(z)\right)_{ij}\left(N_{\lambda}(z;0)\right)_{ij}, \end{equation*}
and  so it is easily verified by using the jumps of $m$ and $N_\lambda(z;0)$ that $N_\lambda(z)$ satisfies condition $(b)$ in the RH problem for $N_\lambda$.
By Lemma \ref{lemzerosgamma} and Lemma \ref{lemDivthetadiff}, the poles of $\frac{1}{\theta(u(z)\pm u(\lambda))}$ are cancelled by the zeros of $\frac{\gamma(z)}{\gamma(\lambda)} \mp \frac{\gamma(\lambda)}{ \gamma(z)}$, and thus condition (a) of the RH problem for $N_\lambda$ is satisfied. By further recalling that $\gamma(z)=\mathcal O((z-z_*)^{-1/4})$ as $z\to z_*$ for $z_*\in \{a_j,b_j\}_{j=1}^k$, condition (d) is satisfied. Conditions (c) and (e) are easily verified by the definition of $N_\lambda$. Thus $N_\lambda$ defined in \eqref{eq:Mlamdef} solves the RH problem for $N_\lambda$.

\subsection{Main parametrix for $F(z)\omega_\epsilon(z) \not \equiv 1$}
To construct a main parametrix valid also in the general setting $F(z)\omega_\epsilon(z)\not \equiv 1$ for $z\in J$, we first consider some properties of $\Theta$ defined in \eqref{Theta}.

By Lemma \ref{Prime}, $\Theta(z,\lambda)$ is analytic as a function of $z\in \mathbb C\setminus[a_1,b_k]$, with a zero of order $1$ at $z=\lambda$. By \eqref{eq:quasi}, \eqref{eq:ujump1}, and \eqref{eq:ujump2},
\begin{align}\label{JumpsTheta1}
\Theta(z_+,\lambda)\Theta(z_-,\lambda)&=1&&\textrm{for }z\in J,\\
\Theta(z_+,\lambda)&=e^{4\pi i u_j(\lambda)} \Theta(z_-,\lambda) &&\textrm{for }z\in (b_j,a_{j+1}) \textrm{ with }j=1,\dots,k-1, \label{JumpsTheta2}
\end{align}
where $\Theta(z_\pm,\lambda)=\lim_{\delta\downarrow 0}\Theta(z\pm i\delta, \lambda)$ for $z\in \mathbb R$.

Recall the definition of $w_z(\lambda)$ in \eqref{defwlambda}.

Since $\Theta(z,\lambda)$ has a zero of order $1$ at $z=\lambda$, it follows that 
\begin{align} \label{polewlambda} w_z(\lambda)&=\frac{1}{\lambda-z}+\mathcal O(1)&& \textrm{as }\lambda \to z\neq \infty,\\
w_\infty(\lambda)&=-\frac{1}{\lambda}+\mathcal O\left(\lambda^{-2}\right) &&\textrm{as }\lambda \to \infty.\label{polewlambda2}
 \end{align}
As a function of $z$, $w_z(\lambda)$ is analytic on $\mathbb C\setminus ([a_1,b_k]\cup\{\lambda\})$, while  as a function of $\lambda$, $w_z(\lambda)$ is analytic on $\mathbb C\setminus (J\cup \{z\})$ (the fact that there are no jumps on $(b_j,a_{j+1})$ follows by reversing the role of $z,\lambda$ in  \eqref{JumpsTheta2} and taking the derivative).
By \eqref{JumpsTheta1} and \eqref{JumpsTheta2}, we obtain
\begin{align}
\label{jumpswlambda1} w_z(\lambda_+)+w_z(\lambda_-)&=0&&\textrm{for }\lambda\in J,\\
\label{jumpswlambda2} w_{z_+}(\lambda)+w_{z_-}(\lambda)&=0&&\textrm{for } z \in J,\\
\label{jumpswlambda3} w_{z_+}(\lambda)-w_{z_-}(\lambda)&=4\pi i u_j'(\lambda)&&\textrm{for }z\in (b_j,a_{j+1})\,\, \textrm{with }j=1,\dots,k-1.
\end{align}
By the above jump conditions and \eqref{polewlambda}, and the fact that $w_z(\lambda)$ is bounded as a function of $z$ as $z\to x\in \{a_j,b_j\}_{j=1}^k$, it is easily verified, by considering it as a function of $z$, that
\begin{equation}\label{2ndrepwlambda}
\frac{w_z(\lambda)}{\mathcal R^{1/2}(z)}=\frac{1}{\mathcal R^{1/2}(\lambda)(\lambda-z)}+2 \sum_{j=1}^{k-1} u_j'(\lambda)\int_{b_j}^{a_{j+1}}\frac{dx}{\mathcal R^{1/2}(x)(x-z)}.
\end{equation}
Given a function $\mathfrak g$ analytic on a neighbourhood of $J$, define
\begin{equation}\label{defdg}
d_{\mathfrak g}(z)=-\frac{1}{4\pi i} \oint_{\gamma_{z,\mathfrak g}}\mathfrak g(\lambda)w_z(\lambda)d\lambda=\frac{1}{2\pi i}\int_J\mathfrak g(\lambda)w_z(\lambda_+)d\lambda,
\end{equation}
where $\gamma_{z,\mathfrak g}$ is a curve oriented counterclockwise enclosing $J$ but not $z$, and is such that $\mathfrak g$ is analytic in a neighbourhood of $\gamma_{z,\mathfrak g}$. For $z\in J$, it follows by 
\eqref{polewlambda} that
\begin{multline*} 
d_{\mathfrak g}(z+i\epsilon)+d_{\mathfrak g}(z-i\epsilon)=\frac{\mathfrak g(z+i\epsilon)+\mathfrak g(z-i\epsilon)}{2} \\
-\frac{1}{4\pi i}\oint_{\tilde \gamma} \mathfrak g(\lambda)(w_{z+i\epsilon}(\lambda)+w_{z-i\epsilon}(\lambda))d\lambda,
\end{multline*}
where $\tilde \gamma$ encloses both $J$ and $z\pm i\epsilon$. Taking the limit $\epsilon\to 0$, we obtain by \eqref{jumpswlambda2} that
\begin{equation}\label{jumpsdf1}
d_{\mathfrak g,+}(z)+d_{\mathfrak g,-}(z)=\mathfrak g(z), \end{equation}
for $z\in J$. Recall the definition of $\Upsilon$ in \eqref{def:Omegahat}. By \eqref{jumpswlambda3},
\begin{equation}\label{jumpsdf2}
d_{\mathfrak g,+}(z)-d_{\mathfrak g,-}(z)=-\oint_{\gamma_{z,\mathfrak g}}u_j'(\lambda)\mathfrak g(\lambda)d\lambda=-2\pi i \Upsilon_j(\mathfrak g),\end{equation}
for $z\in (b_j,a_{j+1})$ with $j=1,\dots,k-1$.

Recall that $F(z)$ has the form $F(z)=e^{f(z)}$ for $z$ in a neighbourhood of $J$, where $f$ is analytic.
Now define
\begin{equation}\label{eq:Ddef}
D(z)=\exp(d_f(z)+d_\epsilon(z)),
\end{equation}
where $d_\epsilon=d_{\log \omega_\epsilon}$.

The function $D$ will appear in our construction of the main parametrix, and the properties set out in the following lemma are important
\begin{lemma}\label{le:D}
 $D$ satisfies the following properties.
\begin{itemize}[leftmargin=0.75cm]
\item[$(a)$] $D: \C \setminus [a_1, b_k] \to \C$ is analytic.
\item[$(b)$] $D$ has continuous boundary values $D_\pm$ on $[a_1, b_k] \setminus  \{a_j, b_j\}_{j=1}^k$ satisfying
\begin{align}\label{eq:Djump}
\begin{cases}
D_+(z) D_-(z) =  e^{f(z)}\omega_{\epsilon}(z), & z \in \cup_{j=1}^k (a_j, b_j), \\
D_+(z) D_-(z)^{-1} = e^{-2\pi i \Upsilon_j}, & z \in (b_j, a_{j+1}), \quad j = 1, \dots, k-1.
\end{cases}
\end{align}
\item[$(c)$] There exists some $D_\infty \in \C$ such that as $z \to \infty$,
\begin{align}\label{eq:Dinfty}
D(z) = D_\infty (1 + \mathcal{O}(z^{-1})).
\end{align}
\item[$(d)$] $D$ and $D^{-1}$ are bounded for $z\in \mathbb C\setminus \left([a_1,b_k]\cup_j U_{t_j}\right)$, for any fixed open neighbourhoods $U_{t_j}$ containing $t_j$, uniformly over all deformation parameters in our differential identities, namely $\epsilon \in (0,\epsilon_0)$, $s\in [0,2]$, and $t\in [0,1]$.
\item[$(e)$] For $z$ in a neighbourhood of $t_j$, and $\Im z>0$,
\begin{equation}  \nonumber D(z)^{\pm 1}=\mathcal O\left((z-(t_j-i\epsilon))^{\pm \left(\alpha_j/2+\beta_j\right)}\right), \quad \frac{d}{dz}\log D(z)=\frac{\alpha_j/2+\beta_j}{z-(t_j-i\epsilon)}+\mathcal O(1), 
\end{equation}
uniformly for $0<\epsilon<\epsilon_0$, while for $\Im z<0$ we have
\begin{equation} \nonumber 
D(z)^{\pm 1}=\mathcal O\left((z-(t_j+i\epsilon))^{\pm \left(\alpha_j/2-\beta_j\right)}\right), \quad \frac{d}{dz}\log D(z)=\frac{\alpha_j/2-\beta_j}{z-(t_j+i\epsilon)}+\mathcal O(1).
\end{equation}
\end{itemize}
\end{lemma}
\begin{proof}
Properties (a), (b), and (c) are easily verified by the above considerations of $d_{\mathfrak g}$. Property (d) follows by noting that if $\tilde \gamma$ encloses $z$ and $J$, then $d_f(z)=\frac{1}{2}f(z)-\frac{1}{4\pi i} \int_{\tilde \gamma}f(\lambda)w_z(\lambda)d\lambda$, and that the integral over $\tilde \gamma$ is bounded as $z\to a_j$ and $z\to b_j$ for $j=1,\dots,k$. 

For property (e), we recall the definition of $\log \omega_\epsilon$ in \eqref{logomegaepsilon}, and observe that for $z$ in a neighbourhood of $t_j$, the only non-trivial terms to consider are $\log(z-(t_j\pm i\epsilon))$. If $\Im z>0$ in a neighbourhood of $t_j$, then we consider the second equality in \eqref{defdg} setting $\mathfrak g(\lambda)=\log(\lambda-(t_j+i\epsilon))$, and observe that
\begin{equation*}\frac{1}{2\pi i}\int_J\log(\lambda-(t_j+i\epsilon))w_z(\lambda_+)d\lambda
\end{equation*}
remains uniformly bounded in the neighbourhood because we can deform the contour of integration downwards in the complex plane, while for 
\begin{multline*}
\frac{1}{2\pi i}\int_J\log(\lambda-(t_j-i\epsilon))w_z(\lambda_+)d\lambda \\
=\frac{1}{2\pi i}\int_{\widetilde J}\log(\lambda-(t_j-i\epsilon))w_z(\lambda)d\lambda+\log(z-(t_j-i\epsilon)), 
\end{multline*}
where $\widetilde J$ is a deformation of $J$ which passes above $t_j$ in the complex plane. The integral over $\widetilde J$ is bounded. Recalling \eqref{defdg},  the definition of $d_\epsilon=d_{\log \omega_\epsilon}$, and the definition of $\log \omega_\epsilon$ in \eqref{logomegaepsilon}, this concludes the proof of (e) for $\Im z>0$. The case where $\Im z<0$ is similar, and is left to the reader. The derivative of $\log D$ also follows similarly.
\end{proof}
Now define
\begin{align}\label{eq:Pdef}
M(z) := D_{\infty}^{\sigma_3} N_\infty (z; N \Omega + \Upsilon) D(z)^{-\sigma_3}.
\end{align}

The function $M$ will act as our main parametrix, approximating $S$ as described in the introduction of Section \ref{sec:global}. It is easily verified by relying on the properties of the RH problem for $N_\infty$ and the RH problem for $D$ that $M$ solves the following RH problem.
\subsubsection*{The RH problem for M}
\begin{itemize}[leftmargin=0.75cm]
\item[(a)] $M: \C \setminus [a_1, b_k] \to \C^{2\times 2}$ is analytic.
\item[(b)] $M$ has continuous boundary values $M_\pm$ on $[a_1, b_k] \setminus  \{a_j, b_j\}_{j=1}^k$ satisfying
\begin{align*}
M_{+}(x) = M_{-}(x) \times
\begin{cases}
\begin{pmatrix}
0 &  e^{f(x)}\omega_{\epsilon}(x) \\ -e^{-f(x)}\omega_{\epsilon}(x)^{-1} & 0
\end{pmatrix}, & x \in \bigcup_{j=1}^k (a_j, b_j), \\
e^{-2 \pi i N \Omega_j \sigma_3}, & \begin{subarray}{l} \ds x \in (b_j, a_{j+1}), \\[0.1cm] \ds  j = 1, \dots, k-1. \end{subarray}
\end{cases}
\end{align*}

\item[(c)] $M(z) = I + \mathcal O(z^{-1})$ as $z \to \infty$.

\item[(d)]  $M(z) = \mathcal{O}((z-z_\ast)^{-\frac{1}{4}})$ as $z \to z_\ast \in  \{a_j, b_j\}_{j=1}^k$.
\end{itemize}

\begin{remark}\label{KuijVan}
By \eqref{2ndrepwlambda} and \eqref{defdg} we find that
\begin{multline}\nonumber d_{\mathfrak g}(z)=-\frac{\mathcal R^{1/2}(z)}{4\pi i}\oint_{\gamma_{z,\mathfrak g}}\frac{\mathfrak g(\lambda)d\lambda}{\mathcal R^{1/2}(\lambda)(\lambda-z)}\\ -\frac{\mathcal R^{1/2}(z)}{2\pi i}\sum_{j=1}^k\oint_{\gamma_{z,\mathfrak g}} \mathfrak g(\lambda)u_j'(\lambda)d\lambda \int_{b_j}^{a_{j+1}} \frac{dx}{\mathcal R^{1/2}(x)(x-z)}.
\end{multline}
By \eqref{def:Omegahat}, this becomes
\begin{multline}\label{Altdg} 
d_{\mathfrak g}(z)=-\frac{\mathcal R^{1/2}(z)}{4\pi i}\oint_{\gamma_{z,\mathfrak g}}\frac{\mathfrak g(\lambda)d\lambda}{\mathcal R^{1/2}(\lambda)(\lambda-z)}\\
-\mathcal R^{1/2}(z)\sum_{j=1}^k\Upsilon_j(\mathfrak g) \int_{b_j}^{a_{j+1}} \frac{dx}{\mathcal R^{1/2}(x)(x-z)}.
\end{multline}
Clearly \eqref{Altdg} is of the same form as  \cite[formula (4.8)]{KuijVanlessen}, and by recalling the form for $\Upsilon(\mathfrak g)$ in Remark \ref{explicittau}, we can compare with \cite[formula (4.12)]{KuijVanlessen}.

\end{remark}

\section{Local parametrices} \label{sec:Local}
In this section we aim to construct functions $P^{(x)}(z)$ on neighbourhoods $U_x$ of $x\in\{a_j,b_j\}_{j=1}^k\cup \{t_j\}_{j=1}^m$, which have the same jump contours and jumps as $S$ on each neighbourhood $U_x$, and additionally satisfy $P^{(x)}(z)M(z)^{-1}$ $\to I$ as $N\to \infty$ uniformly for $z\in \partial U_x$. We start by constructing parametrices at $\{a_j,b_j\}_{j=1}^k$.
\subsection{The Airy local parametrix}\label{sec:Airy}
In this section, we construct a solution to an RH problem which provides good approximation to $S$ close to the branch points $a_i$, $b_i$, and we rely on the following model RH problem for $\Phi_{\rm Ai}$.

\subsubsection*{The RH problem for $\Phi_{\rm Ai}$}
\begin{figure}[h!]
\begin{center}
\begin{tikzpicture}
\node [below] at (5,0) {$0$};
\node [above] at (5.5,0.4) {$2\pi/3$};

\draw[decoration={markings, mark=at position 0.25 with {\arrow[thick]{>}}},
        postaction={decorate}][decoration={markings, mark=at position 0.75 with {\arrow[thick]{>}}},
        postaction={decorate}]  (2,0)--(8,0);
\draw[decoration={markings, mark=at position 0.5 with {\arrow[thick]{<}}},
        postaction={decorate}]  (5,0)--(4,-2);
\draw[decoration={markings, mark=at position 0.5 with {\arrow[thick]{<}}},
        postaction={decorate}]  (5,0)--(4,2);

\draw (5.4,0) arc (0:120:0.4) ;

\end{tikzpicture} 
\caption{The jump contour $\Sigma_{\mathrm{Ai}}$ for $\Phi_{\mathrm{Ai}}$.}\label{fig:Airy_contour}\end{center}
\end{figure}
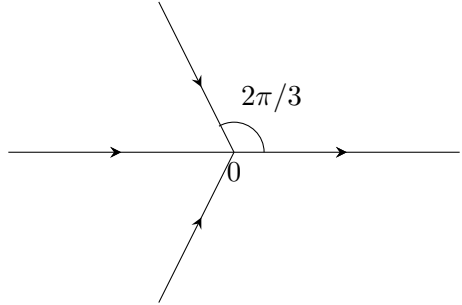
\begin{itemize}[leftmargin=0.75cm]
\item[(a)] $\Phi_{\mathrm{Ai}}: \C \setminus \Sigma_{\mathrm{Ai}} \to \C^{2\times 2}$ is analytic, where the contour is given by $\Sigma_{\mathrm{Ai}} = \R \cup e^{\frac{2\pi i}{3}} \R^+ \cup e^{-\frac{2\pi i}{3}} \R^+$ and is oriented as in \Cref{fig:Airy_contour}.
\item[(b)] $\Phi_{\mathrm{Ai}}$ has continuous boundary values $\Phi_{\mathrm{Ai}, \pm}$ on $\Sigma_{\mathrm{Ai}} \setminus \{0\}$ satisfying
\begin{align*} 
\Phi_{\mathrm{Ai}, +}(\zeta)
& = \Phi_{\mathrm{Ai}, -}(\zeta) \begin{pmatrix} 0 & 1 \\ -1 & 0 \end{pmatrix}, \ \qquad \zeta \in (-\infty, 0), \\
\Phi_{\mathrm{Ai}, +}(\zeta) 
& = \Phi_{\mathrm{Ai}, -}(\zeta) \begin{pmatrix} 1 & 1 \\ 0 & 1 \end{pmatrix}, \quad \qquad  \zeta \in (0, \infty), \nonumber \\
\Phi_{\mathrm{Ai}, +}(\zeta) 
& = \Phi_{\mathrm{Ai}, -}(\zeta) \begin{pmatrix}1 & 0 \\ 1 & 1 \end{pmatrix}, \quad \qquad  \zeta \in \left(e^{\frac{2 \pi i}{3}} \R^+ \cup e^{-\frac{2\pi i}{3}} \R^+\right) \setminus \{0\}. \nonumber
\end{align*}
\item[(c)] As $\zeta \to 0$, $\Phi_{\mathrm{Ai}}(\zeta) = \mathcal{O}(1)$. 
\item[(d)] As $\zeta \to \infty$,
\begin{align}\label{eq:Aiinfty}
\Phi_{\mathrm{Ai}}(\zeta) = \frac{1}{\sqrt{2}} \zeta^{-\frac{1}{4}\sigma_3} 
\begin{pmatrix} 1 & i \\ i & 1 \end{pmatrix} \left( I + \frac{1}{8\zeta^{3/2}} \begin{pmatrix}\frac{1}{6}&i\\i&-\frac{1}{6}\end{pmatrix}+\mathcal O\left(\frac{1}{\zeta^3}\right)\right) e^{-\frac{2}{3} \zeta^{3/2} \sigma_3}.
\end{align}
\end{itemize}
An explicit solution to the RH problem for $\Phi_{\rm Ai}$ was constructed in \cite[Section 7]{DKMVZfirst} in terms of Airy functions.\footnote{
Our $\Phi_{\mathrm{Ai}}$ is given by $\sqrt{2\pi}e^{-\frac{\pi i}{12}} \Psi^{\sigma}$ in the notation of \cite{DKMVZfirst}.
}

\noindent {\bf{Local parametrices at the edge.}} \label{sec:para} Denote by $U_x$  an open disc of radius $\delta > 0$ centred at the point $x$. We choose $\delta > 0$ small enough such that the distance between any $U_{a_j}, U_{b_j}, U_{t_j}$ is at least $\delta$, and that all these open discs are enclosed by the contour $\Gamma$ (or one of the contours $\Gamma_j$) which appeared in the differential identity of Section \ref{sec:RHP}. For the purposes of our argument in Section \ref{sec:pasy}, when we deformed $V$, it is actually important that we can choose $\delta$ independent of $s$. We will also need to choose the discs so small that $h_V$ does not vanish in the disc. Moreover, note that as the discs are surrounded by $\Gamma$ (or one of the $\Gamma_j$), they are also contained in $U_V$ -- the domain of analyticity of $V$.

On each open disc $U_x$ where $x \in \{a_j, b_j\}_{j=1}^k$, let $\psi_V$ be as in \eqref{eq:psi} and $\zeta_{x}$ be the local variable given by the conformal mapping
\begin{align}\label{eq:zetadef}
\zeta_x(z) = \left(\frac{3 \pi i}{2} \int^z_x \psi_V(\lambda) d\lambda\right)^{\frac{2}{3}}, \qquad z \in U_x.
\end{align}

\noindent We choose the branch such that $\zeta_{b_j}(z)$ is positive for $z > b_j$ on $U_{b_j}$, and $\zeta_{a_j}(z)$ is positive for $z < a_j$ on $U_{a_j}$. 

We can finally be more precise about how we define the lenses used in the transformation \eqref{eq:Sdef}: inside $U_{a_j}$ and $U_{b_j}$ we choose them such that under $\zeta$, they are mapped to the contour $\Sigma_{\mathrm{Ai}}$ in \Cref{fig:Airy_contour}. We will return to how we define them close to $t_j$ later on. We then define the local parametrices $P^{(a_j)}(z)$ and $P^{(b_j)}(z)$ as
\begin{align}\label{eq:tildePdef}
P^{(a_j)}(z)& = 
E(z) \sigma_3 \Phi_{\mathrm{Ai}}\left(N^{2/3} \zeta_{a_j}(z)\right) \sigma_3 e^{-N \phi(z) \sigma_3} e^{-\frac{f(z)}{2} \sigma_3} \omega_{\epsilon}(z)^{-\frac{\sigma_3}{2}}, & z \in U_{a_j}, \\
P^{(b_j)}(z)&=E(z) \Phi_{\mathrm{Ai}}\left(N^{2/3} \zeta_{b_j}(z)\right) e^{-N \phi(z) \sigma_3} e^{-\frac{f(z)}{2} \sigma_3} \omega_{\epsilon}(z)^{-\frac{\sigma_3}{2}}, & z \in U_{b_j},\notag 
\end{align}

\noindent where $\phi$ is as in \eqref{eq:phidef} and
\begin{align*}
E(z) = \begin{cases}
M(z) e^{\frac{f(z)}{2} \sigma_3} \omega_{\epsilon}(z)^{\frac{\sigma_3}{2}} e^{N \pi i \Omega_{j-1} \eta(z) \sigma_3} \frac{1}{\sqrt{2}} \begin{pmatrix} 1 & i \\ i & 1 \end{pmatrix} \zeta_{a_{j}}(z)^{\frac{1}{4}\sigma_3} N^{\frac{\sigma_3}{6}}, & z \in U_{a_j}, \\
M(z) e^{\frac{f(z)}{2} \sigma_3} \omega_{\epsilon}(z)^{\frac{\sigma_3}{2}} e^{N \pi i \Omega_{j} \eta(z) \sigma_3} \frac{1}{\sqrt{2}} \begin{pmatrix} 1 & -i \\ -i & 1 \end{pmatrix} \zeta_{b_{j}}(z)^{\frac{1}{4}\sigma_3} N^{\frac{\sigma_3}{6}}, & z \in U_{b_j},
\end{cases}
\end{align*}

\noindent with $M$ as in Section \ref{sec:global}, $\Omega_0 = \Omega_k = 0$ and
\begin{align}\label{eq:etasign}
\eta(z) = \begin{cases}
1, & \Im z > 0, \\
-1, & \Im z < 0.
\end{cases}
\end{align}

\begin{lemma}\label{le:edgerhp}
For $x\in \{a_j,b_j\}_{j=1}^k$, the function $P^{(x)}$ defined in \eqref{eq:tildePdef} satisfies the following RH problem.
\begin{itemize}[leftmargin=0.75cm]
\item[$(a)$] $P^{(x)}: U_{x} \setminus \Sigma_{S} \to \C^{2\times 2}$ is analytic. Here $\Sigma_S$ is as in Section \ref{sec:trans}.
\item[$(b)$] On $U_x\cap \Sigma_S$,
\begin{equation*}P^{(x)}_+=P^{(x)}_-J_S,\end{equation*}
where we recall that $J_S$ is the jump matrix defined in condition $(b)$ for the RH problem for $S$.
\item[$(c)$] We have the matching condition
\begin{align}\label{eq:gl_match}
P^{(x)}(z) M(z)^{-1} = I + \mathcal{O}(N^{-1}), \qquad \mbox{ as } N \to + \infty,
\end{align}
uniformly for  $z \in \partial U_{x} $, and also uniformly over the parameters appearing in our deformations $V_s,\, e^{tf},\, \omega_\epsilon$.
\item[$(d)$] As $z\to x$, $P^{(x)}(z)=\mathcal O(1)$.
\end{itemize}
\end{lemma}

\begin{proof}  Condition $(a)$ follows from the definition of $\Sigma_S$ on $U_x$ (we recall that $\Sigma_S$ was defined on $U_x$ so that $\zeta_x$ maps $\Sigma_S$ to $\Sigma_{\textrm{Ai}}$). Condition $(b)$ follows from condition $(b)$ for the RH problem for $\Phi_{\rm Ai}$, combined with the fact that $E,f,$ and $\omega_\epsilon$ are analytic on $U_x$, and that $\phi$ is analytic on $\mathbb C \setminus (-\infty,b_k]$ and on $J$ we have the condition \eqref{jumpphi2} while on $(b_j,a_{j+1})$ we have \eqref{lol3}.

A straightforward calculation making use of \eqref{eq:Aiinfty}, the definition of $\phi$ from \eqref{eq:phidef}, the definition of $\zeta$ from \eqref{eq:zetadef}, as well as the fact that $\Omega_j$ are real, shows that we have for $z\in \partial U_x$ (and $x$ being the relevant $a_j$ or $b_j$)
\begin{multline*}
P^{(x)}(z)M(z)^{-1}=I \\
+M(z)e^{\frac{f(z)}{2}\sigma_3}\omega_\epsilon(z)^{\frac{\sigma_3}{2}}\mathcal O\big(N^{-1}\zeta_x(z)^{-3/2}\big)\omega_\epsilon(z)^{-\frac{\sigma_3}{2}}e^{-\frac{f(z)}{2}\sigma_3}M(z)^{-1},
\end{multline*}
where the implied constant is universal. We have that $e^{\pm f(z)}$ and $\omega_\epsilon(z)^{\pm1}$ are uniformly  bounded on $\partial U_x$. The same holds for $M$, since it is analytic on $\mathbb C\setminus J$, and has continuous boundary values on $J\cap U_x\setminus\{x\}$. Thus we obtain condition $(c)$. Condition $(d)$ follows from a direct computation which we omit here and the proof is complete.

\end{proof}

We observe that by including the subleading term in condition $(d)$ in the RH problem for $\Phi_{\rm Ai}$ we obtain that as $N\to \infty$,
\begin{equation}\label{PM-1fine}\begin{aligned}
&P^{(x)}(z) M(z)^{-1}=I+\frac{1}{N} \Delta(z)+\mathcal O(1/N^2),\\
&\Delta(z)=
\frac{M(z)e^{N\pi i \Omega_{j-1}\eta(z)\sigma_3}}{8\zeta_{a_j}(z)^{3/2}}\sigma_3\begin{pmatrix}
\tfrac{1}{6} & i\\
i & -\tfrac{1}{6}
\end{pmatrix}\sigma_3 e^{-N\pi i \Omega_{j-1}\eta(z)\sigma_3}M(z)^{-1}, && z \in \partial U_{a_j},\\
&\Delta(z)=
\frac{M(z)e^{N\pi i \Omega_{j}\eta(z)\sigma_3}}{8\zeta_{b_j}(z)^{3/2}}\begin{pmatrix}
\frac{1}{6} & i\\
i & -\frac{1}{6}
\end{pmatrix}e^{-N\pi i \Omega_{j}\eta(z)\sigma_3}M(z)^{-1}, && z \in \partial U_{b_j}
,\end{aligned}
\end{equation}

We now turn to approximations close to the points $t_j$.

\subsection{The Painlev\'e local parametrix}\label{sec:Painleve}

Similarly to the Airy local parametrix, we will consider an approximation for $S$ near the points $t_j$ through a certain model RH problem for $\Psi$. As opposed to the Airy model RH problem, the solution to the RH problem for $\Psi$ may not be explicitly constructed, however the unique existence of a solution to $\Psi$ was proven in \cite{CIK}.\footnote{Let us point out a minor difference between our description of the problem and that of \cite{CIK}. Our parameters $\alpha$ and $\beta$ correspond to $2\alpha$ and $-\beta$ in \cite{CIK}. } The model RH problem $\Psi$ is associated with the Painlev\'{e} V  equation,  and depends on two complex parameters $\alpha,\beta$ and a positive one $s>0$.

\subsubsection*{RH problem for $\Psi$}

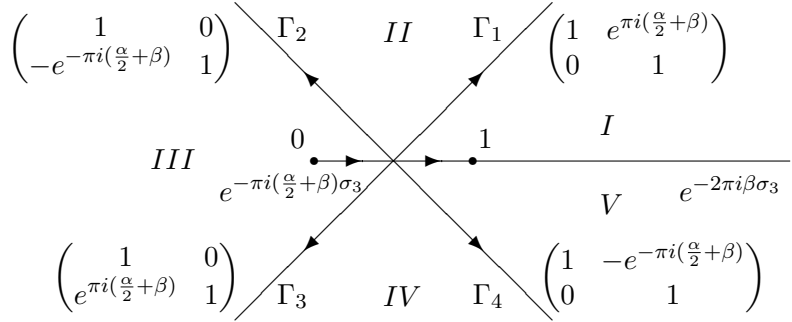
\begin{figure}[t]
	\begin{center}
		\begin{picture}(120,140)(-5,-65)
		
		\put(3,0){$0$}
		\put(74,0){$1$}
		
		\put(72,42){$\Gamma_{1}$}
		\put(-2,42){$\Gamma_{2}$}
		\put(-2,-59){$\Gamma_{3}$}
		\put(72,-59){$\Gamma_{4}$}
		
		\put(120,5){$I$}
		\put(38,40){$II$}
		\put(-50,-8){$III$}
		\put(38,-60){$IV$}
		\put(120,-25){$V$}
		
		\put(99,35){$\begin{pmatrix} 1 & e^{\pi i(\frac{\alpha}{2}+\beta)}\\ 0 &1\end{pmatrix}$}
		\put(-104,35){$\begin{pmatrix} 1 &0\\ -e^{-\pi i(\frac{\alpha}{2}+\beta)} &1\end{pmatrix}$}
		\put(-87,-52){$\begin{pmatrix} 1 &0\\ e^{\pi i(\frac{\alpha}{2}+\beta)} &1\end{pmatrix}$}
		\put(97,-52){$\begin{pmatrix} 1 & -e^{-\pi i(\frac{\alpha}{2}+\beta)}\\0 &1\end{pmatrix}$}
		\put(-24,-20){$e^{-\pi i (\frac{\alpha}{2}+\beta) \sigma_3}$}
		\put(150,-20){$e^{-2\pi i \beta \sigma_3}$}
		
		\put(-18,-65){\line(1,1){120}}
		\put(-18,55){\line(1,-1){120}}
		\put(12,-5){\line(1,0){60}}
		\put(72,-5){\line(1,0){120}}
		
		\put(12,-5){\thicklines\circle*{2.5}}
		\put(72,-5){\thicklines\circle*{2.5}}
		
		\put(8,-39){\thicklines\vector(-1,-1){.0001}}
		\put(76,29){\thicklines\vector(1,1){.0001}}
		\put(8,29){\thicklines\vector(-1,1){.0001}}
		\put(76,-39){\thicklines\vector(1,-1){.0001}}
		\put(30,-5){\thicklines\vector(1,0){.0001}}
		\put(60,-5){\thicklines\vector(1,0){.0001}}
		
		\end{picture}
		\caption{The jump contour $\Sigma_{\Psi}$ and the jump matrices for $\Psi$.}
		\label{figure: contour Psi-}
	\end{center}
\end{figure}

\begin{itemize}
	\item[(a)] $\Psi: \mathbb{C}\setminus \Sigma_{\Psi} \rightarrow \mathbb{C}^{2\times 2}$ is analytic, with $\Sigma_{\Psi}=\Gamma_1 \cup\Gamma_2 \cup\Gamma_3 \cup\Gamma_4 \cup[0,+\infty]$, and
	\[\Gamma_1 =\frac{1}{2}+e^{\frac{\pi i}{4}}\mathbb R^+,\ \Gamma_2 =\frac{1}{2}+e^{\frac{3\pi i}{4}}\mathbb R^+,\ \Gamma_3 =\frac{1}{2}+e^{-\frac{3\pi i}{4}}\mathbb R^+,\ \Gamma_4 =\frac{1}{2}+e^{-\frac{\pi i}{4}}\mathbb R^+,\]
	oriented as in Figure \ref{figure: contour Psi-}.
	\item[(b)] $\Psi$ has continuous boundary values on $\Sigma_\Psi\setminus \{0,\frac{1}{2},1\}$ and the following jumps:
	\begin{equation*}
	\begin{aligned}
	&\Psi_+(z)=\Psi_-(z)\begin{pmatrix} 1 & e^{\pi i(\frac{\alpha}{2}+\beta)}\\ 0 &1\end{pmatrix}  &&\textrm{for $z \in \Gamma_1$,}  \\
	&\Psi_+(z)=\Psi_-(z)\begin{pmatrix} 1 &0\\ -e^{-\pi i(\frac{\alpha}{2}+\beta)} &1\end{pmatrix}  &&\textrm{for $z \in \Gamma_2 $,}  \\
	&\Psi_+(z)=\Psi_-(z)\begin{pmatrix} 1 &0\\ e^{\pi i(\frac{\alpha}{2}+\beta)} &1\end{pmatrix}  &&\textrm{for $z \in \Gamma_3$,}  \\
	&\Psi_+(z)=\Psi_-(z)\begin{pmatrix} 1 & -e^{-\pi i(\frac{\alpha}{2}+\beta)}\\0 &1\end{pmatrix}  &&\textrm{for $z \in \Gamma_4$,}  \\
	&\Psi_+(z)=\Psi_-(z)e^{-\pi i(\frac{\alpha}{2}+\beta)\sigma_3}  &&\textrm{for $z \in (0,1)\setminus \{1/2\}$,}  \\
	&\Psi_+(z)=\Psi_-(z)e^{-2\pi i\beta \sigma_3}  &&\textrm{for $z >1$.}
	\end{aligned}
	\end{equation*}
	\item[(c)]There exist functions $p=p(s,\alpha,\beta),\,q=q(s,\alpha,\beta),\,r=r(s,\alpha,\beta)$  such that $\Psi(z)$ has the following behavior as $z\to \infty$: 
	\begin{equation}\label{Psi as}
	\Psi(z)=\left(I+\frac{1}{z}\begin{pmatrix}q&r\\p&-q\end{pmatrix}+\mathcal{O}(z^{-2})\right)z^{\beta\sigma_3}\exp\Big(-\frac{s}{2}z\sigma_3\Big).
	\end{equation}
	\item[(d)] The function $H_0(z):=\Psi(z)z^{-(\frac{\alpha}{4}+\frac{\beta}{2})\sigma_3}$, where the branch cut is taken on $(0,+\infty)$, is analytic in a neighbourhood of $0$. Similarly, the function
	\begin{equation*}
	H_1(z) = \Psi(z)(z-1)^{(\frac{\alpha}{4}-\frac{\beta}{2})\sigma_{3}} \left\{ \begin{array}{l l}
	e^{\pi i \beta \sigma_{3}}, & \Im z > 0, \\
	e^{-\pi i \beta \sigma_{3}}, & \Im z < 0,
	\end{array} \right.
	\end{equation*}
	where we choose the principal branch for the roots, is analytic in a neighbourhood of $1$. Furthermore, $\Psi$ is bounded near $1/2$.
\end{itemize}

For our purposes, we find the following modification of this RH problem to be convenient. 
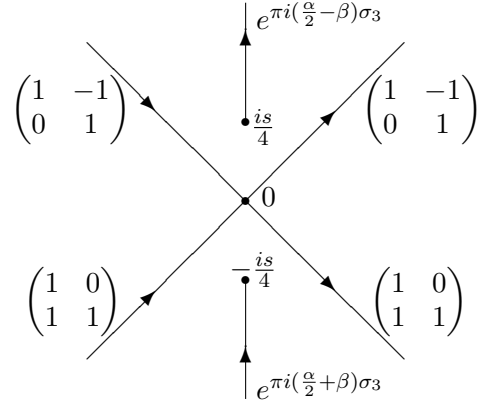
\begin{figure}[t]
	\begin{center}
		\begin{picture}(120,140)(-5,-70)
		\put(48,-7){$0$}
		\put(36,-34){$-\frac{is}{4}$}
		\put(44,20){$\frac{is}{4}$}
		
		\put(86,27){$\begin{pmatrix}1&-1\\0&1\end{pmatrix}$}
		\put(-47,27){$\begin{pmatrix}1&-1\\0&1\end{pmatrix}  $}
		\put(-42,-46){$\begin{pmatrix}1&0\\1&1\end{pmatrix} $}
		\put(89,-46){$\begin{pmatrix}1&0\\1&1\end{pmatrix} $}
		\put(46,60){$e^{\pi i (\frac{\alpha}{2}-\beta) \sigma_3}$}
		\put(46,-80){$e^{\pi i (\frac{\alpha}{2}+\beta) \sigma_3}$}
		
		\put(42,25){\line(0,1){45}}
		\put(42,-35){\line(0,-1){45}}
		\put(-18,-65){\line(1,1){120}}
		\put(-18,55){\line(1,-1){120}}
		
		\put(42,25){\thicklines\circle*{2.5}}
		\put(42,-5){\thicklines\circle*{2.5}}
		\put(42,-35){\thicklines\circle*{2.5}}
		
		\put(42,60){\thicklines\vector(0,1){.0001}}
		\put(42,-60){\thicklines\vector(0,1){.0001}}
		\put(8,-39){\thicklines\vector(1,1){.0001}}
		\put(76,29){\thicklines\vector(1,1){.0001}}
		\put(8,29){\thicklines\vector(1,-1){.0001}}
		\put(76,-39){\thicklines\vector(1,-1){.0001}}
		
		\end{picture}
		\caption{The jump contour $\Sigma_{\Phi_{\mathrm{V}}}$ and the jump matrices for $\Phi_{\mathrm{V}}$.}
		\label{figure: contour Psi(1)}
	\end{center}
\end{figure}

\subsubsection*{Modified Painlev\'{e} V model RH problem}\label{section: PV}

We define $\Phi_{\mathrm{V}}(z;s,\alpha,\beta) = \Phi_{\mathrm{V}}(z)$ by
\begin{equation}\label{eq:PhiV}
\Phi_{\mathrm{V}}(z;s,\alpha,\beta) = e^{\frac{s}{4}\sigma_{3}}s^{\beta \sigma_{3}} \Psi\Big( -\frac{2 i z}{s}+\frac{1}{2};s,\alpha,\beta \Big) \left\{ \begin{array}{l l}
e^{-\frac{\pi i}{2}(\frac{\alpha}{2}+\beta)\sigma_{3}}, & \Re z > 0, \\
e^{\frac{\pi i}{2}(\frac{\alpha}{2}+\beta)\sigma_{3}}, & \Re z < 0.
\end{array} \right.
\end{equation}
This is closely related to a function defined in \cite[(3.9)]{CIK}, denoted by $\Phi(\lambda,s)$ there. More precisely, we have with this notation
\begin{equation} \label{eq:PhiV2}
\Phi_{\rm V}(z,s)=\Phi(-2i z,s) \widehat G(z)^{\sigma_3},
\end{equation}
where
\begin{equation}\label{eq:Ghatdef}
\widehat{G}(z) = 2^{\beta}e^{\pi i \beta}e^{-\frac{\pi i \alpha}{4}}e^{i z} \Big( z+\frac{is}{4} \Big)^{\frac{1}{2}(\frac{\alpha}{2}+\beta)}\Big( z-\frac{is}{4} \Big)^{-\frac{1}{2}(\frac{\alpha}{2}-\beta)},
\end{equation}
with
\begin{equation*}
-\frac{\pi}{2}< \arg \Big(z+\frac{is}{4}\Big) < \frac{3\pi}{2}, \qquad - \frac{3\pi}{2} < \arg\Big(z-\frac{is}{4}\Big)<\frac{\pi}{2}.
\end{equation*}
It is straightforward to verify that $\Phi_{\rm V}$ satisfies the following RH problem
\subsubsection*{RH problem for $\Phi_{\rm V}$}
\begin{itemize}
\item[(a)] $\Phi_{\rm V}$ is analytic on $\mathbb C\setminus \Sigma_{\Phi_{\rm V}}$, where $\Sigma_{\Phi_{\rm V}}$ is as in Figure \ref{figure: contour Psi(1)}.
\item[(b)] On $\Sigma_{\Phi_{\rm V}}$, $\Phi_{\rm V}$ has the following jumps:
\begin{equation*}\begin{aligned}
\Phi_{\rm V,+}(z)&=\Phi_{\rm V,-}(z)\begin{pmatrix}1&-1\\0&1\end{pmatrix}&&\textrm{for $\arg z=\pi/4$ and $\arg z=3\pi/4$,}\\
\Phi_{\rm V,+}(z)&=\Phi_{\rm V,-}(z)\begin{pmatrix}1&0\\1&1\end{pmatrix}&&\textrm{for $\arg z=-\pi/4$ and $\arg z=-3\pi/4$,}\\
\Phi_{\rm V,+}(z)&=\Phi_{\rm V,-}(z)e^{\pi i \left(\frac{\alpha}{2}+\beta\right)\sigma_3}&&\textrm{for $\arg(z+is/4)=-\pi/2$,}\\
\Phi_{\rm V,+}(z)&=\Phi_{\rm V,-}(z)e^{\pi i \left(\frac{\alpha}{2}-\beta\right)\sigma_3}&&\textrm{for $\arg(z-is/4)=\pi/2$.}
\end{aligned} \end{equation*}
\item[(c)] The behaviour of $\Phi_{\rm V}(z)$ as $z\to \infty$ is inherited from condition (c) in the RH problem for $ \Psi$.
\item[(d)] The behaviour of $\Phi_{\rm V}(z)$ as $z\to \pm is/4$ is inherited from condition (d) in the RH problem for $\Psi$.
\end{itemize}
 
For the following results, we rely on \cite{CIK}.
\begin{proposition} \label{Prop:CIK}
Let $\alpha > -1$ and $\beta \in i\mathbb{R}$. The following statements hold.
\begin{itemize} \item[(a)] $\Psi$ is solvable all $s>0$.
\item[(b)] When considered as a function of $s$, $H_0(0)$ and $H_1(1)$ are analytic on an open set in the complex plane containing $(0,\infty)$. 
\item[(c)] As $s\to 0$,
\begin{equation*} 
H_0(0)=\mathcal O\left(\begin{pmatrix}s^{\alpha/2}&s^{-\alpha/2}+s^{\alpha/2}\\ s^{\alpha/2} &s^{-\alpha/2}+s^{\alpha/2}\end{pmatrix} \right), \quad H_1(1)=\mathcal O\left(\begin{pmatrix}s^{-\alpha/2}+s^{\alpha/2}&s^{\alpha/2}\\s^{-\alpha/2}+s^{\alpha/2} &s^{\alpha/2}\end{pmatrix} \right),
\end{equation*}
for $\alpha\neq 0$, while
\begin{align*}
H_{0}(0) = \bigO \left( \begin{pmatrix}
1 & \log s \\ 1 & \log s
\end{pmatrix} \right), \qquad H_{1}(1) = \bigO \left( \begin{pmatrix}
\log s & 1 \\ \log s & 1
\end{pmatrix} \right),
\end{align*} as $s\to 0$ for $\alpha=0, \, \beta \neq 0$.
\item[(d)] As $s\to \infty$, we have $\Phi(s/2;s),\, \Phi(-s/2;s)= I+\mathcal O\left(e^{-cs}\right)$ (where $\Phi$ is as in \cite{CIK}, see \eqref{eq:PhiV2} above for relation to $\Phi_{\rm V}$), for some constant $c>0$. Similarly, $\Phi'(s/2;s),\, \Phi'(-s/2;s)=\mathcal O\left(e^{-cs}\right)$ as $s\to \infty$.
\item[(e)] There hold the identities
\begin{align*}
(H_0(0)^{-1}H_0'(0))_{11}&=-\frac{s}{2}-\frac{\sigma(s)}{\frac{\alpha}{2}+\beta}+\frac{\alpha}{4}-\frac{\beta}{2},\\
(H_1(1)^{-1}H_1'(1))_{22}&=\frac{s}{2}+\frac{\sigma(s)}{\frac{\alpha}{2}-\beta}-\frac{\alpha}{4}-\frac{\beta}{2},
\end{align*}
where $\sigma$ is a real analytic function on $(0,+\infty)$ and satisfies
\begin{equation}\label{sigmalim}
\sigma(s)=\begin{cases}\alpha^2/4-\beta^2+\mathcal O(s)+\mathcal O(s^{1+\alpha}\log s)&\textrm{as } s\to 0,\\
\mathcal O\left(s^{-1+2\alpha}e^{-s}\right) &\textrm{as } s\to + \infty.\end{cases} \end{equation}
Furthermore, as $\hat s\to \infty$,
\begin{multline*}
\int_{0}^{\hat s} \bigg( \sigma(s) - \Big( \frac{\alpha^{2}}{4}-\beta^{2} \Big) \bigg) \frac{ds}{s} = - \bigg( \frac{\alpha^{2}}{4}-\beta^{2} \bigg) \log (\hat s) \\
- \log \frac{G(1+\frac{\alpha}{2}+\beta)G(1+\frac{\alpha}{2}-\beta)}{G(1+\alpha)} + o(1),
\end{multline*}
where $G$ is the Barnes G-function.
\item[(f)] As $z \to \infty$, we have
\begin{equation}\label{eq:PhiVasy}
\Phi_{\mathrm{V}}(z)\widehat{G}(z)^{-\sigma_{3}} = I + \bigO(z^{-1}),
\end{equation}
uniformly for $s \in (0,+\infty)$.
\end{itemize}
\end{proposition}
\begin{remark}
In \cite{CIK} it was furthermore proven that $\sigma$ satisfies the Jimbo-Miwa-Okamoto form of the  Painlev\'e V equation:
\begin{multline*}
\left(s\frac{d^2\sigma}{ds^2}\right)^2=\left(\sigma-s\frac{d\sigma}{ds}+2\left(\frac{d\sigma}{ds}\right)^2+\alpha \frac{d\sigma}{ds}\right)^2\\
-4\left(\frac{d\sigma}{ds}\right)^2\left(\frac{d\sigma}{ds}+\alpha/2-\beta\right)\left(\frac{d\sigma}{ds}+\alpha/2+\beta\right).
\end{multline*}
\end{remark}
\begin{proof} Part (a) is the vanishing lemma proven in \cite[Section 4.4.1]{CIK}. For part (b), it is generally known that such functions are meromorphic functions of $s$ with only a finite number of possible poles  for values of $s$ where the RH problem is not solvable, and by part (a) it follows that they are analytic on an open set containing $(0,\infty)$. 

For part (c), we rely on the analysis of $\Psi$ as $s\to 0$ in \cite[Section 4.2.2]{CIK}, taking care that our $\alpha$ and $\beta$ correspond to $2\alpha$ and $-\beta$ in \cite{CIK}.
When $\alpha \neq 0,1,2,\dots$, it follows from  \cite{CIK}[(4.53), (4.23), and (4.75)], writing $\xi=\frac{\lambda}{s}+1$, that
\begin{multline*}
\Psi(\xi)=s^{\beta \sigma_3}E(\lambda)
\begin{pmatrix}1&c_0J(\lambda;s,\alpha/2,-\beta)\\0&1 \end{pmatrix}(\lambda+s)^{\left(\frac{\alpha}{4}+\frac{\beta}{2}\right)\sigma_3}\lambda^{\left(\frac{\alpha}{4}-\frac{\beta}{2}\right)\sigma_3}\begin{pmatrix}1&\hat g\\0&1 \end{pmatrix}\\ \times (I+o(1)), 
\end{multline*}
as $s\to 0$, uniformly for $\lambda$ in a fixed neighbourhood of $0$ intersected with the region where $\xi=\frac{\lambda}{s}+1$ is in region III, where $\hat g$ and $c_0$ are some constants and $E$ is an analytic function (and thus bounded) on a neighbourhood of $0$. By   \cite[formula (4.60)]{CIK}, $J$ is bounded, from which part (c) follows for $H_0(0)$ and for $\alpha\neq 0,1,2,\dots$, by the definition of $H_0$. The function $H_1(1)$ is evaluated similarly. When 
$\alpha=0,1,2,\dots$, by \cite[formulas  (4.67),  (4.23) and (4.75)]{CIK},
\begin{multline*}\Psi(\xi)=(I+o(1))s^{\beta \sigma_3}\widetilde E(\lambda)\begin{pmatrix}1&c_2 \widetilde J(\lambda;s,\alpha/2,-\beta)\\0&1 \end{pmatrix}\\ \times (\lambda+s)^{\left(\frac{\alpha}{4}+\frac{\beta}{2}\right)\sigma_3}\lambda^{\left(\frac{\alpha}{4}-\frac{\beta}{2}\right)\sigma_3}\begin{pmatrix}1&\frac{(-1)^{\alpha+1}}{\pi} \sin \pi(\alpha/2-\beta) \log (\lambda e^{-\pi i}) \\0&1 \end{pmatrix}\begin{pmatrix}1&\hat g\\0&1 \end{pmatrix}, \end{multline*}
as $s\to 0$, uniformly for $\lambda$ in a fixed neighbourhood of $0$ intersected with the region where $\xi=\frac{\lambda}{s}+1$ is in region III, where $\hat g$ and $c_2$ are some constants and $\widetilde E$ is an analytic function (and thus bounded) on a neighbourhood of $0$. By  \cite[formula  (4.71)]{CIK}, it follows that
\begin{multline*}
c_2\widetilde J(\lambda;s,\alpha,\beta)=-\frac{1}{\pi}\sin \pi (\alpha/2+\beta)\lambda^{\alpha/2-\beta}(\lambda+x)^{\alpha/2+\beta}\log(\lambda e^{-\pi i})\\ +\mathcal O\left(\textrm{max}\{1,s^\alpha|\log s |\}\right). \end{multline*}
Thus, by the definition of $H_0$, part (c) follows also for $\alpha=0,1,2,\dots$ for $H_0(0)$, and $H_1(1)$ is evaluated similarly.

Part (d) follows from   \cite[formulas (4.2) and (4.9)]{CIK}.

Now consider part (e). We will prove the identity for $H_0(0)^{-1}H_0'(0)$, the other one follows in a similar manner. It is easily verified that $\Psi'\Psi^{-1}$ is of the form
\begin{equation}\label{eq:Aident}
\Psi'(z,s)\Psi(z,s)^{-1}=A_\infty(s)+\frac{A_0(s)}{z}+\frac{A_1(s)}{z-1} 
\end{equation}
and similarly that
\begin{equation}\label{eq:Bident}
\left[\frac{\partial}{\partial s}\Psi(z,s)\right]\Psi(z,s)^{-1}=B_1(s)z+B_0(s)
\end{equation}
for some matrices $A_i,B_i$  which are independent of $z$. Inserting the expression $\Psi(z)=H_0(z)z^{(\frac{\alpha_j}{4}+\frac{\beta_j}{2})\sigma_3}$ into \eqref{eq:Aident} and evaluating the residue at $z=0$, we see that
\[
\left(\frac{\alpha_j}{4}+\frac{\beta_j}{2}\right)H_0(0)\sigma_3 H_0(0)^{-1}=A_0(s).
\]
On the other hand, expanding near zero $H_0(z)=H_{0,0}(I+H_{0,1} z+\mathcal O(z^2))$, we see that $H_0(0)^{-1}H_0'(0)=H_{0,1}$. Moreover, using $\Psi(z)=H_0(z)z^{(\frac{\alpha}{4}+\frac{\beta}{2})\sigma_3}$, we see by evaluating \eqref{eq:Bident} as well as its $z$-derivative, both at $z=0$, that
\[
\frac{\partial }{\partial s} H_{0,0}=B_0H_{0,0} \qquad \text{and} \qquad \frac{\partial }{\partial s}H_{0,1}=H_{0,0}^{-1}B_1H_{0,0}.
\]
Substituting the large $z$ expansion \eqref{Psi as} for $\Psi$ into \eqref{eq:Bident}, it follows that $B_1=-\frac{1}{2}\sigma_3$, so by our remark that $H_0^{-1}(0)H_0'(0)=H_{0,1}$, we obtain
\begin{equation}\nonumber \begin{aligned}
\frac{\partial}{\partial s} (H_0^{-1}(0)H_0'(0))_{11}&=-\frac{1}{2}(H_0^{-1}(0)\sigma_3 H_0(0))_{11}\\&=-\frac{1}{2}(H_0(0)\sigma_3 H_0(0)^{-1})_{11}=-\frac{1}{\frac{\alpha}{2}+\beta}A_{0,11}(s).
\end{aligned}\end{equation}

Using \cite[(4.103) and (4.110)]{CIK}, which state that   $A_{0,11}(s)=\sigma'(s)+\frac{\alpha}{4}+\frac{\beta}{2}$ (recall that our $\alpha$ and $\beta$ correspond to $2\alpha$ and $-\beta$ in \cite{CIK}), we find that
\begin{equation}\label{G-1G'sigmaeqn}
\partial_s (H_0^{-1}(0)H_0'(0))_{11}=-\frac{1}{\frac{\alpha}{2}+\beta}\sigma'(s)-\frac{1}{2},
\end{equation}
where $\sigma$ is a real analytic function satisfying \eqref{sigmalim}.
It remains to evaluate the integration constant here.  Recalling the notation $\Phi(z,s)=\Phi_{\rm V}( \frac{iz}{2},s)\widehat G(\frac{iz}{2})^{-\sigma_3}$, it follows by part (d) and the definition of $\widehat G$ in \eqref{eq:Ghatdef}, and by the relation between $\Psi$ and $\Phi$ in \eqref{eq:PhiV}, \eqref{eq:PhiV2}, that
\begin{align*}
&(H_0(0)^{-1}H_0'(0))_{11}= \lim_{z\to 0} \left( \big[ \Psi^{-1}(z)\Psi'(z) \big]_{11}- (\tfrac{\alpha}{4}+\tfrac{\beta}{2})\frac{1}{z} \right)\\
&=\lim_{z\to -i\frac{s}{4}}\left[-2i(\Phi(-2iz)^{-1}\Phi'(-2iz))_{11}+\frac{\widehat G'(z)}{\widehat G(z)}-(\tfrac{\alpha}{4}+\tfrac{\beta}{2})\frac{1}{z+i\frac{s}{4}}\right]\frac{is}{2}\\
&=-\frac{s}{2}+\frac{\alpha}{4}-\frac{\beta}{2}+o(1),
\end{align*}
as $s\to 0$.
Combined with \eqref{G-1G'sigmaeqn} and the limiting behaviour of $\sigma$ in \eqref{sigmalim}, we obtain part (e) for $H_0(0)^{-1}H_0'(0)$, the result for $H_1(1)^{-1}H_1'(-1)$ is proven in a similar manner.

\end{proof}

We are finally in a position to define our local parametrix near a point $t_j$.

\subsubsection*{The local parametrix near \texorpdfstring{$t_j$}{t}}

Let $\epsilon_{0}>0$  be sufficiently small but fixed, and let $U_{t_{j}}$ denote a sufficiently small but fixed disc centred at $t_{j}$. In this section, we follow \cite{CIK,ClaeysFahs} to construct a local parametrix in $U_{t_{j}}$ which will approximate $S$ uniformly for $0<\epsilon\leq \epsilon_0$. Let us define 
\begin{equation}\label{eq:fdef}
\zeta_{t_j}(z) = i \left( \left\{ \begin{array}{l l}
\phi(z), & \mbox{if } \Im z > 0 \\
-\phi(z), & \mbox{if } \Im z < 0
\end{array} \right\} - \frac{\phi(t_{j}+i\epsilon)-\phi(t_{j}-i\epsilon)}{2} \right).
\end{equation}
One can check e.g. from \eqref{eq:phidef} that $\phi(z) = \overline{\phi(\overline{z})}$, and thus
\begin{equation*}
\frac{\phi(t_{j}+i\epsilon)-\phi(t_{j}-i\epsilon)}{2} = i \Im \phi(t_{j}+i\epsilon).
\end{equation*}
One readily verifies that the function $\zeta_{t_j}$ is a conformal map in $U_{t_{j}}$, and satisfies $\zeta_{t_j}^{\prime}(t_{j}) = \pi  \psi_{V}(t_{j}) > 0$. This remark allows us to finally define what the lenses from Section \ref{sec:trans} look like in $U_{t_j}$: we choose the lenses so that $\zeta_{t_j}$ maps $\Sigma_S\setminus \mathbb R$ to $\Sigma_{\Phi_{\rm V}}$. We observe that $\zeta_{t_j}(t_j)$ is not necessarily equal to zero, but $t_j$ is not a singular point in the RH problem for $S$, so we simply deform the contour of $S$ slightly and assume that the lens intersects the real line at $t_j$ anyway.

To make use of our Painlev\'e model RH problem, we also need to define the $\alpha,\beta,$ and $s$ parameters. For $\alpha$ and $\beta$, we simply choose $\alpha_j,\beta_j$. For $s$, let us define
\begin{equation}\label{eq:sjdef}
s_{j,N} = 2N \Big( \phi(t_{j}+i\epsilon) + \phi(t_{j}-i\epsilon) \Big) = 4N \Re \phi(t_{j}+i\epsilon) > 0.
\end{equation}
Then $N\zeta_{t_j}(t_j\pm i\epsilon)=\pm \frac{i s_{j,N}}{4}$.

This allows us to define precisely what the branch cut of $\omega_\epsilon$ from Section \ref{sec:FHproof} is: we choose the cut going from $t_j+i\epsilon$ upward to be such that $\zeta_{t_j}$ maps the part of it in $U_{t_j}$ to $i(s_{j,N}/4,\infty)$, and similarly in the lower half plane -- one readily checks from basic properties of $\zeta_{t_j}$, that at least for small enough $U_{t_j}$, this gives rise to a valid branch cut (smooth, does not intersect itself, or spiral into itself). 

We can finally define our local parametrix:
\begin{equation}\label{eq:local}
P^{(t_{j})}(z) = E_{t_{j}}(z)\Phi_{\mathrm{V}}(N\zeta_{t_j}(z); s_{j,N},\alpha_{j},\beta_{j} ) Q(z)e^{-N\phi(z)\sigma_{3}}e^{-\frac{f(z)}{2}\sigma_{3}}\omega_{\epsilon}(z)^{-\frac{\sigma_{3}}{2}},
\end{equation}
with
\begin{equation}\label{eq:Qdef}
Q(z) = \left\{  \begin{array}{l l}
\begin{pmatrix}
0 & 1 \\ -1 & 0
\end{pmatrix}, & \Im z > 0, \\
I, & \Im z < 0,
\end{array} \right.
\end{equation}
and where $\Phi_{\mathrm{V}}$ is the solution to the modified Painlev\'{e} V model RH problem discussed above, and $E_{t_{j}}$ is given by
\begin{multline}\label{eq:Edef}
E_{t_{j}}(z) = M(z) \omega_{\epsilon}(z)^{\frac{\sigma_{3}}{2}}e^{\frac{f(z)}{2}\sigma_{3}} Q(z)^{-1} \Big( N\zeta_{t_j}(z)+\frac{is_{j,N}}{4} \Big)^{-(\frac{\alpha_{j}}{4}+\frac{\beta_{j}}{2})\sigma_{3}} \\ \times \Big( N\zeta_{t_j}(z)-\frac{is_{j,N}}{4} \Big)^{(\frac{\alpha_{j}}{4}-\frac{\beta_{j}}{2})\sigma_{3}}e^{-iN \Im \phi(t_{j}+i\epsilon)\sigma_{3}}e^{\frac{\pi i \alpha_{j}}{4}\sigma_{3}}e^{-\pi i \beta_{j} \sigma_{3}}2^{-\beta_{j} \sigma_{3}}.
\end{multline}
Here we choose the branch of the roots so that the cuts of the $N\zeta_{t_j}(z)\pm \frac{is_{j,N}}{4}$-terms are on $\Sigma_{\omega_\epsilon}$ discussed in Section \ref{sec:FHproof} -- to be precise, one defines $\log (N\zeta_{t_j}(z)\pm i \frac{s_{j,N}}{4})=\int_1^{N \zeta_{t_j}(z)\pm i \frac{s_{j,N}}{4}}\frac{dw}{w}$ for all $z$ off of the above mentioned cuts, in such a manner that the integration contour does not cross the cut.

The claim about $P^{(t_j)}$ is the following. 
\begin{lemma}\label{le:singrhp}
	The function $P^{(t_j)}$ satisfies the following RH problem.
	\begin{itemize}[leftmargin=0.75cm]
		\item[$(a)$] $P^{(t_j)}: U_{t_j} \setminus \Sigma_{S} \to \C^{2\times 2}$ is analytic -- here $\Sigma_S$ is as in Section \ref{sec:trans}.
		\item[$(b)$] $P^{(t_j)}$ satisfies the same jump conditions as $S$ on $\Sigma_S\cap U_{t_j}$. 
		\item[$(c)$] For $z \in \bigcup_{j=1}^k (\partial U_{a_j} \cup \partial U_{b_j})$, 
		\begin{align}\label{eq:loc_match}
		P^{(t_j)}(z)M(z)^{-1} = I + \mathcal{O}(N^{-1})
		\end{align}
		
		\noindent where the implied constant is uniform in $0<\epsilon\leq \epsilon_0$.
		\item[$(d)$] $P^{(t_j)}(z)=\mathcal O(1)$ as $z\to t_j$.
	\end{itemize}
\end{lemma}
\begin{proof}
	Verifying a), b), and d) is a routine calculation making use of properties of $\Phi_{\mathrm V}$ and very similar to ones in \cite{CIK,ClaeysFahs} so we omit the details. Concerning the matching condition, in particular the uniformity in it, we point out that using \eqref{eq:PhiVasy} along with the definition of $\zeta_{t_j}$ from \eqref{eq:fdef}, that for $z\in \partial U_{t_j}$ 
	\begin{align*}
	P^{(t_j)}(z)M(z)^{-1}=I+E_{t_j}(z)\mathcal O (N^{-1}) E_{t_j}(z)^{-1}.
	\end{align*}
	with an implied constant that is uniform in $z$ and $0<\epsilon<\epsilon_0$. One also readily verifies that $P$ and  $\omega_\epsilon^{\pm 1}$ are bounded on $\partial U_{t_j}$ uniformly in $0<\epsilon<\epsilon_0$, and thus we obtain \eqref{eq:loc_match}.
\end{proof}

We now turn to performing our final transformation and solving our RH problem approximately.

\section{Solving the small norm problem}\label{sec:snorm}

As is standard in the RH analysis of these types of problems, we have the following lemma. Recall the contour $\Sigma_S$ from Figure \ref{ContourS} and the jumps $J_S$ from \eqref{eq:Sjump1}-\eqref{eq:Sjump4}.
\begin{lemma} \label{expsmalljump} Given $\delta>0$ (fixed but sufficiently small), there exists $c>0$, such that
\begin{equation*}J_S(z)=I+\mathcal O\left(e^{-cN}/(|z|^2+1)\right), \end{equation*}
as $N\to \infty$, uniformly for $z\in \Sigma_S\setminus J$ and $|z-z_0|>\delta$ for all $z_0\in \{a_j,b_j\}_{j=1}^k\cup\{t_j\}_{j=1}^p$.\end{lemma}
\begin{proof}
Combining $\phi(z) = g(z) - V(z)/2 + \ell/2$ with the (strict) Euler-Lagrange inequality \eqref{eq:EL2}, we infer that $\phi_{+}(x)+\phi_{-}(x) < 0$ for $x \in \mathbb{R}\setminus J$.  This yields the lemma for any fixed $x\in \mathbb R\setminus J$ with $x$ bounded away from $\{a_j,b_j\}_{j=1}^k\cup\{t_j\}_{j=1}^p$. For uniformity as $x\to \infty$,  we recall our assumption that $\frac{V(x)}{\log (|x|+1)}\to \infty$, and since $g$ has logarithmic growth at $\infty$, we obtain the uniformity required.

Now consider the contour $J^\pm$.
For $x \in J$, $\Re \phi_{\pm}(x) = 0$, and furthermore $\frac{d}{dx}  \phi_{\pm}(x) = \mp i \pi \psi_{V}(x)$. Thus, by the Cauchy-Riemann equations, $\Re \phi(z) > 0$ on $J^+$ and $J^-$ assuming they are chosen to be sufficiently near the real line. We observe that at the intersection points $J^\pm \cap \partial U_x$ for $x\in\{a_j,b_j\}_{j=1}^k$, we can not choose the contours $J^\pm$ to be sufficiently near the real line, because on $U_x$ the contours $J^\pm$ were already determined in Section \ref{sec:para} and the derivative of $\phi$ at $x$ is zero. We consider the intersection points near $b_j$, the ones near $a_j$ are considered similarly. As $z\to b_j$,
\begin{equation} \nonumber\phi(z)=\phi(b_j)-c(z-b_j)^{3/2}+\mathcal O((z-b_j)^{5/2}),\end{equation}
for some $c>0$, where we recall that $\phi(b_j)$ is purely imaginary. If $z_\pm\in J^\pm\cap \partial U_{b_j}$ then the argument of $(z_\pm-b_j)$ is close to  $\pm 2\pi/3$ when we choose $\delta$ to be sufficiently small (see Section \ref{sec:para}), and thus $\Re \phi(z_\pm)$ is positive. 
 \end{proof}

Recall the main parametrix $M$ satisfying the RH problem for $M$ as in Section \ref{sec:global}, the local parametrices at $t_1,\dots, t_p$ constructed in Section \ref{sec:Painleve}, and the local parametrices at $a_j,b_j$ for $j=1,\dots, k$ constructed in Section  \ref{sec:Airy}. Define
\begin{equation}\label{eq:Rdef}
R(z)=\begin{cases}
S(z)M(z)^{-1}, & z\notin \cup_{x\in\mathcal T} U_x\\
S(z)[P^{(x)}(z)]^{-1}, & z\in U_x \textrm{ for }x\in \mathcal T,
\end{cases}
\end{equation}	
where $ \mathcal T= \{a_j,b_j\}_{j=1}^k\cup  \{t_j\}_{j=1}^p$, and let $\Sigma_R$ be as in Figure \ref{ContourR}, i.e.
\begin{multline}\nonumber 
\Sigma_R=\cup_{x\in \mathcal T}\partial U_x\cup \{y:y\in J^\pm \textrm{ and } |y-x|>\delta \textrm{ for all $x\in \mathcal T$}\}\\ \cup \{y\in \mathbb R\setminus J: \, |y-x|>\delta \textrm{ for all }x\in\{a_j,b_j\}_{j=1}^k\}.
\end{multline}

Combining Lemma \ref{expsmalljump}, Lemma \ref{le:edgerhp}, and Lemma \ref{le:singrhp} with the uniform boundedness of $M$ for $z$ bounded away from $\mathcal T$, we see that $R$ satisfies the following small norm RH problem.
\begin{itemize}
\item[(a)] $R$ is analytic on $\mathbb C\setminus \Sigma_R$.
\item[(b)] On $\Sigma_R$,
\begin{equation}\nonumber 
R_+(z)=R_-(z)(I+\widehat \Delta(z)),
\end{equation}
where $\widehat \Delta(z)=\mathcal O\left(N^{-1}\right)$ for $z\in \partial U_x$ with $x\in \mathcal T$, and 
$\widehat \Delta(z)=\mathcal O\left(\frac{e^{-cN}}{|z^2|+1}\right)$, uniformly for $z\in \Sigma_R\setminus \cup_{x\in \mathcal T} \partial U_x$ and uniformly in the deformation parameters $t,s,\epsilon$ introduced in the differential identities of Section \ref{sec:RHP}.
\item[(c)] $R(z)\to I$ as $z\to \infty$. 
\end{itemize}

It is easily verified that $R$ has the form
\begin{equation}\label{FormR}
R(z)=I+\int_{\Sigma_R} \frac{R_-(u)\widehat \Delta(u)du}{(u-z)2\pi i}, \end{equation}
where we emphasize that the orientation on $\partial U_x$ for $x\in \mathcal T$ is clockwise. 
By standard small norm analysis  (see e.g. \cite[Section 7]{Deift}), it follows that $R_-= I+\mathcal O(N^{-1})$  in $L^2\left(\Sigma_R,\frac{du}{|u^2|+1}\right)$ as $N\to \infty$. Thus by \eqref{FormR} and the asymptotics of $\widehat \Delta$ in condition (b) for the RH problem for $R$, it follows that
\begin{equation}\label{smallnorm}
R(z)=I+\mathcal O \left(N^{-1}\right) \qquad \text{and} \qquad R'(z)=\mathcal O(N^{-1})
\end{equation}
as $N\to \infty$, uniformly for $z$ in compact subsets of $ \mathbb C\setminus \Sigma_R$ and uniformly in the  deformation parameters $t,s,\epsilon$ introduced in the differential identities of Section \ref{sec:RHP}. Furthermore, the jumps of $R$ on $J^\pm$ and $\partial U_x$ for $x\in \mathcal T$ extend to  analytic functions in  neighbourhoods of $J^\pm$ and $\partial U_x$. Thus these jump contours can be deformed, and so \eqref{smallnorm} also holds as $z$ approaches $J^\pm$ and $\partial U_x$ for $x\in \mathcal T$. The upshot is that there are open intervals $I_j$ such that $[a_j,b_j]\subset I_j$ and that \eqref{smallnorm} holds uniformly for $z\in \mathbb C \setminus (\mathbb R\setminus \cup_{j=1}^k I_j)$. The reason that \eqref{smallnorm} does not hold uniformly on $\mathbb R\setminus \cup_{j=1}^k I_j$ is that we did not assume that $F$ is  analytic in this region, and thus we cannot deform the jumps of $R$ here.

For future reference we observe that in the case where $F(x)\omega_\epsilon(x)\equiv 1$, by substituting \eqref{PM-1fine} into \eqref{FormR}, we obtain
\begin{equation} \label{eq:expR}\begin{aligned}R(z)&=I+\frac{1}{N} R^{(1)}(z)+\mathcal O(1/N^2), \\
R'(z)&=\frac{1}{N}\frac{d}{dz}R^{(1)}(z)+\mathcal O(1/N^2),\end{aligned} \end{equation}
as $N\to \infty$, uniformly for $z\in \mathbb C\setminus \left( \mathbb R\setminus \cup_{j=1}^k I_j\right)$,
where
\begin{equation}\label{formulaR1}
R^{(1)}(z)= 
\sum_{x\in \mathcal \{a_j,b_j\}_{j=1}^k}\oint_{\partial U_x} \frac{\Delta(u)du}{(u-z)2\pi i}, 
\end{equation}
and where $\Delta(u)$ is as in \eqref{PM-1fine} (so that $\widehat \Delta(u)=\frac{1}{N}\Delta(u)+\mathcal O(1/N^2)$ on $\partial U_x$).

\begin{figure}[t]
	\begin{center}
		\begin{picture}(100,80)(-5,-40)
		\put(-138,5){$U_{a_1}$}
		\put(-28,5){$U_{b_1}$}
		\put(162,5){$U_{b_k}$}
		\put(51,5){$U_{a_k}$}
		
		\put(-105,8){$U_{t_1}$}
		\put(-65,8){$U_{t_2}$}
		\put(105,8){$U_{t_p}$}
		
		\put(-128,-5){\thicklines\circle{10}}
		\put(-98,-5){\thicklines\circle{10}}
		\put(-28,-5){\thicklines\circle{10}}	
		\put(-58,-5){\thicklines\circle{10}}
		\put(62,-5){\thicklines\circle{10}}
		\put(112,-5){\thicklines\circle{10}}
		\put(162,-5){\thicklines\circle{10}}
		
		\put(-150,-5){\line(1,0){17}}
		\put(-23,-5){\line(1,0){17}}
		
		\put(10,-6){$\dots\dots$}
		\put(45,-5){\line(1,0){12}}
		\put(167,-5){\line(1,0){15}}
		
		\put(140,10){\thicklines\vector(1,0){.0001}}
		\put(140,-20){\thicklines\vector(1,0){.0001}}		
		\put(90,10){\thicklines\vector(1,0){.0001}}
		\put(90,-20){\thicklines\vector(1,0){.0001}}		
		
		\put(-76,10){\thicklines\vector(1,0){.0001}}
		\put(-76,-20){\thicklines\vector(1,0){.0001}}
		
		\put(-40,10){\thicklines\vector(1,0){.0001}}
		\put(-40,-20){\thicklines\vector(1,0){.0001}}
		
		\put(-110,10){\thicklines\vector(1,0){.0001}}
		\put(-110,-20){\thicklines\vector(1,0){.0001}}

		\put(-137,-5.15){\thicklines\vector(1,0){.0001}}
		\put(-10,-5.15){\thicklines\vector(1,0){.0001}}
		\put(180,-5.15){\thicklines\vector(1,0){.0001}}
		\put(55,-5.15){\thicklines\vector(1,0){.0001}}

     	\put(-125,0){\thicklines\vector(1,0){.00001}}
		\put(-95,0){\thicklines\vector(1,0){.00001}}
		\put(-55,0){\thicklines\vector(1,0){.00001}}
		\put(-25,0){\thicklines\vector(1,0){.00001}}

        \put(65,0){\thicklines\vector(1,0){.00001}}
        \put(115,0){\thicklines\vector(1,0){.00001}}
        \put(165,0){\thicklines\vector(1,0){.00001}}

		\qbezier(-125,-1.3)(-113,22)(-101,-1.3)
		\qbezier(-95,-1.3)(-78,22)(-61,-1.3)
		\qbezier(-55,-1.3)(-43,22)(-31,-1.3)
		\qbezier(-125,-8.7)(-113,-32)(-101,-8.7)
		\qbezier(-95,-8.7)(-78,-32)(-61,-8.7)
		\qbezier(-55,-8.7)(-43,-32)(-31,-8.7)
		\qbezier(65,-1.3)(87,22)(109,-1.3)
		\qbezier(65,-8.7)(87,-32)(109,-8.7)
		\qbezier(115,-1.3)(137,22)(159,-1.3)
		\qbezier(115,-8.7)(137,-32)(159,-8.7)
		\end{picture}
		\caption{The contour $\Sigma_R$. }
		\label{ContourR}
	\end{center}
\end{figure}
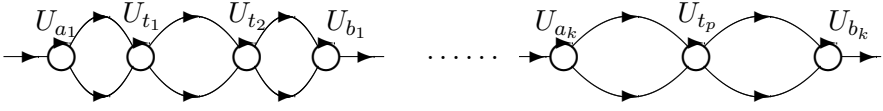

\section{Ratio asymptotics for non-singular symbols}\label{sec:nonsing}
In this section, we obtain large $N$ asymptotics for ratio asymptotics of Hankel determinants with non-singular symbols, thereby proving Theorem \ref{th:smoothasy}. To do so, we will rely on the asymptotics of $Y$ obtained in the Riemann-Hilbert analysis, and specialize to symbols of the form $\nu_t(z)=F_t(z)e^{-NV(z)}$  as denoted in \eqref{deformnut}, with $tf=\log F_t$ analytic in a neighbourhood of $J$, and we fix $\omega_\epsilon=1$. Substituting these asymptotics for $Y$ into  the right-hand side of \eqref{eq:DIsmasy} and integrating, we obtain Theorem \ref{th:smoothasy}. 

We denote  $\Upsilon(tf)=t\Upsilon(f)$, where the quantity $\Upsilon$ was defined in \eqref{def:Omegahat}.  Fix a contour $\Gamma_j$  oriented counter-clockwise enclosing $[a_j,b_j]$, and furthermore enclosing the lenses around $[a_j,b_j]$ and the discs $U_{a_j},$ $U_{b_j}$ (defined in Section \ref{sec:Local}) such that $f$ and $V$ are analytic in an open set containing $\Gamma_j$ and $[a_j,b_j]$, such that $\Gamma_j$ intersects $\mathbb{R}$ only at two points. Recall that this is precisely the setting of Section \ref{sec:DI} and Section \ref{sec:smasy}, and that we denoted by $I_j$ the largest subinterval of $\mathbb R$ enclosed by $\Gamma_j$. Finally, we denoted $I=\cup_{j=1}^kI_j$ and $\Gamma=\cup_{j=1}^k\Gamma_j$.

For $z\in \bigcup_j \Gamma_j\cup(\R\setminus I)$ and $t\in (0,1)$, we have by \eqref{eq:Tdef}, \eqref{eq:Sdef}, and \eqref{eq:Rdef}
\[
Y(z)=e^{-N\frac{\ell}{2}\sigma_3} R(z)M(z)e^{N(g(z)+\frac{\ell}{2})\sigma_3},
\]
where we recall the definition of $g$ in \eqref{eq:g_fn}, and  that $Y$, $R$ and $M$ depend on $t$. So
\begin{multline*}
Y(z)^{-1}Y'(z)=e^{-N(g(z) +\frac{\ell}{2})\sigma_3}\bigg(M(z)^{-1}R(z)^{-1}R'(z)M(z)\\+M(z)^{-1}M'(z)+Ng'(z)\sigma_3\bigg) e^{N(g(z) +\frac{\ell}{2})\sigma_3},
\end{multline*}
and we recall that $r(\nu_t)$ is expressed in terms of the quantity
\begin{equation}\label{Y-1Y'1}
\big(Y(z)^{-1}Y'(z)\big)_{21}=e^{2N(g(z)+\frac{\ell}{2})}\left(M(z)^{-1}R(z)^{-1}R'(z)M(z)+M(z)^{-1}M'(z)\right)_{21},
\end{equation}
while
\begin{equation}\label{Y-1Y'2}
\big(Y(z)^{-1}Y'(z)\big)_{11}=\left(M(z)^{-1}R(z)^{-1}R'(z)M(z)+M(z)^{-1}M'(z)\right)_{11}+Ng'(z)
\end{equation}
appears in formula \eqref{eq:DIsmasy}.
\begin{lemma} As $N\to \infty$, \label{lem:Y-1Y'f}
\begin{equation}\label{eq:ratdiasy1}
 \sup_{t\in(0,1)}|r(\nu_t)|=
\sup_{t\in(0,1)}\left|\int_{\R\setminus I}[Y(x)^{-1} Y'(x)]_{21}\frac{d}{dt}F_t(x)e^{-NV(x)}\frac{dx}{2\pi}\right|=\mathcal O\left(e^{-cN}\right),
\end{equation}
for some $c>0$,
where $r(\nu_t)$ is as defined in \eqref{eq:DIsmasy}.
\end{lemma}
\begin{proof} Let 
\begin{equation}\label{defH} H(z)=e^{-N(2g(z)+\ell)}\left[Y(z)^{-1}Y(z)'\right]_{21}.\end{equation} Then \eqref{eq:ratdiasy1} becomes 
\begin{equation}\label{intH+}
\sup_{t\in(0,1)}\left|\int_{\R\setminus I}H_+(x)\frac{d}{dt}F_t(x)e^{N(2g_+(x)+\ell-V(x))}dx\right|=\mathcal O\left(e^{-cN}\right).
\end{equation}
We will show that $H_+(x)$ is uniformly bounded for $x\in \R\setminus I$. By \eqref{lol3} combined with the Euler-Lagrange equation \eqref{eq:EL2} (which is a strict inequality by our assumptions) the integrand on the left-hand side of \eqref{intH+} is exponentially small for any fixed $x\in \mathbb R\setminus I$. By \eqref{eq:Vgrowth} and the fact that $g(x)=\mathcal O(\log |x|)$ as $x\to \pm \infty$, the definition of $F_t$ in \eqref{deformnut} and assumption (c) for $F$, we obtain 
\begin{equation}\frac{d}{dt} F_t(x)e^{N(2g_+(x)+\ell-V(x))}=\mathcal O\left(e^{-cN}/x^2\right)\end{equation} 
as $x\to \infty$. Thus, upon proving the boundedness of $H_+$, we have proven the lemma.

Since $e^{-2Ng(z)}=\mathcal O(z^{-2N})$ as $z\to \infty$ and $\left[Y(z)^{-1}Y(z)'\right]_{21}=\mathcal O(z^{2N-2})$ as $z\to \infty$, $H(z)=\mathcal O(z^{-2})$ as $z\to \infty$. Also, since $\left[Y(z)^{-1}Y(z)'\right]_{21}$ is an entire function and since $e^{g(z)}$ is analytic on $\mathbb C\setminus [a_1,b_k]$, it follows that $H$ is analytic on $\mathbb C\setminus [a_1,b_k]$. Combining the facts that $H(z)=\mathcal O(z^{-1})$ as $z\to \infty$,  that $H(z)$ is analytic on $\mathbb C\setminus [a_1,b_k]$, and  that $H$ is bounded in a neighbourhood of $[a_1,b_k]$, we obtain that
\begin{equation}\nonumber H(z)=\int_{a_1}^{b_k} \frac{(H_+(x)-H_-(x))dx}{2\pi i(x-z)}. \end{equation}
Thus, if $\gamma$ is a counter-clockwise oriented loop containing $[a_1,b_k]$, and $z$ is not contained in $\gamma$, then
\begin{equation}\label{formH} H(z)=\oint_\gamma \frac{H(x)dx}{2\pi i(z-x)}.  \end{equation}
Along $\gamma$, we have by \eqref{Y-1Y'1} that
\begin{align*}
H(z)=\big(M(z)^{-1}R(z)^{-1}R'(z)M(z)+M(z)^{-1}M'(z)\big)_{21}.
\end{align*}
We recall that \eqref{smallnorm} holds uniformly for $z\in \mathbb C\setminus(\mathbb R\setminus \cup_{j=1}^k I_j)$, where each $I_j$ is an open interval such that $[a_j,b_j]\subset I_j$. We let $\gamma$ be such that it intersects $\mathbb R$ in $I_1$ and $I_k$ (and not in the support $\mu_V$), and so in particular \eqref{smallnorm} holds uniformly for $z\in \gamma$. Thus, since additionally $M(z)$ and $M'(z)$ are uniformly bounded for $z\in \gamma$, it follows that $H(z)$ is bounded uniformly for $z\in \mathbb R\setminus [a_1-\delta,b_k+\delta]$ for $\delta>0$. 

Now we consider $H_+(z)$ for $z\in [b_j+\delta,a_{j+1}-\delta]$. Let $\Sigma_j$ be a curve connecting $b_j$ to $a_{j+1}$ in the lower half of the complex plane. Let $O_j$ be the region contained by $[b_j,a_{j+1}]$ and $\Sigma_j$. Let $H_j(x)=H(x)$ for $x\in \mathbb C\setminus O_j$, and
$H_j(x)=H(x)e^{-4\pi  i N \Omega_j}$ for $x\in O_j$. Then by \eqref{lol3}, $H_{j,+}(x)=H_{j,-}(x)$ for $x\in (b_j,a_{j+1})$.  Then, similarly to \eqref{formH}, we have
\begin{equation}H_+(z)=\oint_{\gamma_j}\frac{H_j(x)dx}{2\pi i (z-x)}\nonumber \end{equation}
for $z\in (b_j+\delta,a_{j+1}-\delta)$, where $\gamma_j$ is a loop containing $[a_1,b_k]\setminus(b_j,a_{j+1})$ but not $z$ (and is such that $H_j$ is analytic on $\gamma_j$). We assume that $\gamma_j$ only intersects $\mathbb R$ in $\cup_{i=1}^k I_i$, and thus $H_j(x)$ is uniformly bounded for $x\in \gamma_j$ (similarly to the boundedness of $H(x)$ for $x\in \gamma$ above). Thus it follows that $H_+(z)$ is uniformly bounded for $z\in \mathbb R\setminus I$, which concludes the proof of the lemma.
\end{proof}

\medskip

We now integrate $[Y(z)^{-1}Y'(z)]_{11}$ appearing in \eqref{eq:DIsmasy}.  Recall $N_\lambda$ defined in \eqref{eq:Mlamdef}, satisfying the RH problem for $N_\lambda$. By considering the Riemann-Hilbert problem that $N_\lambda(z)$ satisfies, it is easily verified that $N_{\lambda}(z) = N_\infty(\lambda)^{-1} N_\infty(z)$ for any $\lambda, z \in \C\setminus [a_1,b_k]$. Thus we may rewrite
\[
N_\infty (z)^{-1}N_\infty '(z)=N_\infty(z)^{-1}N_\infty(\lambda)N_\infty(\lambda)^{-1}N_\infty  '(z)=N_\lambda(z)^{-1}N_\lambda '(z).
\]

On the other hand, since $\lambda$ was arbitrary and $N_z(z)=I$ by construction, we find that
\begin{equation}\label{M-1M'infty}
N_\infty (z)^{-1}N_\infty  '(z)=\lim_{\lambda\to z}N_\lambda '(z).
\end{equation}
We rely on this to evaluate \eqref{Y-1Y'2}. Since $V$ and $f$ are analytic on a neighbourhood of $\Gamma_j$, for $j=1,\dots,k$, the small norm estimate for $R$ in \eqref{smallnorm} holds uniformly for $z\in \cup_{j=1}^k\Gamma_j$.
Recall the definition of $M$ in \eqref{eq:Pdef},
and let $D(z)=\exp(d_f(z))$ with $d_f$ is as in \eqref{defdg}. By \eqref{M-1M'infty},  \eqref{Y-1Y'2} and \eqref{smallnorm},
\begin{multline}\label{Y-1Y'f}
[Y(z)^{-1} Y'(z)]_{11}=Ng'(z)-t\frac{D'(z)}{D(z)}+\lim_{\lambda\to z}N_\lambda '(z)_{11}+\mathcal O(N^{-1})\\
=Ng'(z)-t\frac{d}{dz}\log D(z)+\sum_{j=1}^{k-1} \frac{\partial_j \theta(N\Omega+t\Upsilon(f))}{\theta(N\Omega+t\Upsilon(f))}u_j'(z)+\mathcal O(N^{-1}),
\end{multline}
as $N\to \infty$, uniformly for $t\in (0,1)$ and $z\in \bigcup_j\Gamma_j$. Substituting the definition of $D(z)$ from \eqref{eq:Ddef}, we thus obtain
\begin{multline}\label{eq:tdiasy}
\oint_{\Gamma} [Y(z)^{-1}Y'(z)]_{11}f(z)\frac{dz}{2\pi i}
=N \oint_{\Gamma} g'(z)f(z)\frac{dz}{2\pi i}
-t\oint_{\Gamma}d_f'(z)f(z)\frac{dz}{2\pi i}\\
\qquad +\sum_{m=1}^{k-1}\frac{\partial_m \theta(N\Omega+t\Upsilon(f))}{\theta(N\Omega+t\Upsilon(f))}\oint_{\Gamma} u_m'(z)f(z)\frac{dz}{2\pi i}+o(1),
\end{multline}
as $N\to \infty$,
where $o(1)$ is uniform in $t\in (0,1)$. We now evaluate the right-hand side.

\underline{The first term:} For the first term here, we note that by the definition of $g$ from \eqref{eq:g_fn}, Fubini,  and Cauchy's integral formula, we have
\begin{equation}\label{eq:1stterm}
\oint_{\Gamma} g'(z)f(z)\frac{dz}{2\pi i}=\int_J \left(\oint_{\Gamma} \frac{1}{z-x}f(z)\frac{dz}{2\pi i}\right)d\mu_V(x)=\int_J f(x)d\mu_V(x).
\end{equation}
\underline{The second term:} By \eqref{defdg} and the fact that $\frac{\partial}{\partial z}w_z(\lambda)=W(z,\lambda)$, we obtain
\begin{equation}- \oint_{\Gamma}d_f'(z)f(z)\frac{dz}{2\pi i}=\frac{1}{2}\oint_\Gamma\oint_{\widetilde \Gamma} W(z,\lambda)f(z)f(\lambda) \frac{d\lambda}{2\pi i}\frac{dz}{2\pi i}. \end{equation}
\underline{The third term:} By \eqref{jumpsdf2}
\begin{equation}\begin{aligned}\label{eq:3rdtermpart} 
\sum_{m=1}^{k-1}\frac{\partial_m \theta(N\Omega+t\Upsilon(f))}{\theta(N\Omega+t\Upsilon(f))}\oint_{\Gamma}u_m'(z)f(z)\frac{dz}{2\pi i}&=\sum_{m=1}^{k-1}\frac{\partial_m \theta(N\Omega+t\Upsilon(f))}{\theta(N\Omega+t\Upsilon(f))}\Upsilon_m(f)\\& =\frac{d}{dt}\log \theta(N\Omega+\Upsilon(tf)).
\end{aligned}\end{equation}
\underline{Combining everything:} Combining the first, second and third term we obtain
\begin{multline}\label{eq:tdiasy2}
\oint_{\Gamma} [Y(z)^{-1}Y'(z)]_{11}f(z)\frac{dz}{2\pi i}
=\frac{d}{dt}\Bigg[tN \int_J f(x)d\mu_V(x)
\\+\frac{t^2}{4}\oint_\Gamma\oint_{\widetilde \Gamma} W(z,\lambda)f(z)f(\lambda) \frac{d\lambda}{2\pi i}\frac{dz}{2\pi i}
\qquad +\log \theta(N\Omega+\Upsilon(tf))\Bigg]+o(1),
\end{multline}
as $N\to \infty$,
where $o(1)$ is uniform in $t\in (0,1)$. Substituting \eqref{eq:tdiasy2} into \eqref{eq:DIsmasy} we obtain Theorem \ref{th:smoothasy}. 

\subsection{Proof of \eqref{smoothasy2}}\label{SecProofSmasy2}
To prove \eqref{smoothasy2}, observe that
\begin{multline}
\oint_\Gamma\oint_{\widetilde \Gamma} W(z,\lambda)f(z)f(\lambda) \frac{d\lambda}{2\pi i}\frac{dz}{2\pi i}\\
=\oint_\Gamma\int_J \left(W(z,\lambda_-)-W(z, \lambda_+)\right)f(z)f(\lambda) \frac{d\lambda}{2\pi i}\frac{dz}{2\pi i}.
\end{multline}
Taking integration by parts,
 \begin{multline}\oint_\Gamma\oint_{\widetilde \Gamma} W(z,\lambda)f(z)f(\lambda) \frac{d\lambda}{2\pi i}\frac{dz}{2\pi i}\\ =\frac{1}{2\pi i}\sum_{j=1}^k\oint_\Gamma \left[f(b_j)\left(w_{b_{j,-}}(z)-w_{b_{j,+}}(z)\right)-f(a_j)\left(w_{b_{j,-}}(z)-w_{b_{j,+}}(z)\right)\right]f(z)\frac{dz}{2\pi i}
\\+\oint_\Gamma \int_J (w_{\lambda_+}(z)-w_{\lambda_-}(z))f'(\lambda)f(z) \frac{d\lambda}{2\pi i}
\frac{dz}{2\pi i}.
\end{multline}
By \eqref{jumpswlambda3},
\begin{equation}w_{b_{j,-}}(z)-w_{b_{j,+}}(z)=w_{a_{(j+1),-}}(z)-w_{a_{(j+1),+}}(z)=-4\pi i u_j'(z), \end{equation} 
for $j=1,\dots,k-1$, and since $w_\lambda(z)$ is analytic for $\lambda\in \mathbb C\setminus [a_1,b_k]$, 
\begin{equation} w_{b_{k,-}}(z)-w_{b_{k,+}}(z)=w_{a_{1,-}}(z)-w_{a_{1,+}}(z)=0. \end{equation}
Thus, using additionally \eqref{jumpswlambda2},
\begin{multline}\label{ProofSmasy2}\oint_\Gamma\oint_{\widetilde \Gamma} W(z,\lambda)f(z)f(\lambda) \frac{d\lambda}{2\pi i}\frac{dz}{2\pi i} =2\sum_{j=1}^{k-1}(f(a_{j+1})-f(b_j))\Upsilon_j(f)
\\+2\oint_\Gamma \int_J w_{\lambda_+}(z)f'(\lambda)f(z) \frac{d\lambda}{2\pi i}
\frac{dz}{2\pi i}.
\end{multline}
Then \eqref{smoothasy2} follows from \eqref{ProofSmasy2} and \eqref{repwlambda}.

\subsection{Proof of \eqref{eqlaplace} } \label{SecProofeqlaplace}
The asymptotics of \eqref{eqlaplace} follow from a second integration by parts and subsequently changing the homology basis from $u_j$ to $ \widehat u_j$.
Our starting point is \eqref{ProofSmasy2}. By \eqref{repwlambda}, $w_{\lambda_+}(z_{\pm})=\pm \frac{1}{z-\lambda}+\mathcal O(1)$ as $z\to \lambda$ (for $z,\lambda$ in $J$), and thus
\begin{equation}\label{ointwlambda+f} \oint_\Gamma  w_{\lambda_+}(z)f(z) 
\frac{dz}{2\pi i}=\mathcal P.\mathcal V. \int_J \left[w_{\lambda_+}(z_-)-w_{\lambda_+}(z_+)\right]f(z) 
\frac{dz}{2\pi i},
\end{equation}
where $\mathcal P.\mathcal V.\int$ is the principal value integral. By \eqref{repwlambda}, by the fact that $\mathcal R^{1/2}_{\pm}(x)$ and $u_{j,\pm}'(x)$ are  imaginary on $J$, and by the fact that $\mathcal R^{1/2}(x)$ is real on $\mathbb R\setminus J$, it follows that $w_{\lambda_+}(z_{\pm})$ is real on $J$. Thus
\begin{equation}w_{\lambda_+}(z_{\pm})=\frac{\partial}{\partial z}\log |\Theta(\lambda_+,z_{\pm})|. \end{equation}
Thus, taking integration by parts in \eqref{ointwlambda+f},
\begin{multline}\oint_\Gamma  w_{\lambda_+}(z)f(z) 
\frac{dz}{2\pi i}=\int_J\log \left|\frac{\Theta(\lambda_+,z_+)}{\Theta(\lambda_+,z_-)}\right|f(z)\frac{dz}{2\pi i}\\
+\frac{1}{2\pi i}\sum_{j=1}^k\left(f(b_j)\log \left| \frac{\Theta(\lambda_+,b_{j,-})}{\Theta(\lambda_+,b_{j,+})}\right|-f(a_j)\log \left| \frac{\Theta(\lambda_+,a_{j,-})}{\Theta(\lambda_+,a_{j,+})}\right|\right)
\end{multline}
From \eqref{JumpsTheta2}, it follows that  if $x\in \{b_j,a_{j+1}\}$, then
\begin{equation} \label{jumpsThetaabs2}
\log\left|\Theta(z,x_\pm )\right|=\pm \Re \left(2\pi i  u_j(z)\right), \end{equation}
for $z \in J$.  By \eqref{JumpsTheta1}, $\left| \Theta(\lambda_+,z_+)\Big/\Theta(\lambda_+,z_-)\right|=\left| \Theta(\lambda_+,z_+)\right|^2$. Thus,
\begin{multline}\oint_\Gamma  w_{\lambda_+}(z)f(z) 
\frac{dz}{2\pi i}=2i\sum_{j=1}^{k-1}\left(f(a_{j+1})-f(b_j)\right)\textrm{Im} (u_{j,+}(\lambda))
\\ +2\int_J\log \left|\Theta(\lambda_+,z_+)\right|f'(z)\frac{dz}{2\pi i}
\end{multline}
Substituting into the right hand side of \eqref{ProofSmasy2} we obtain
\begin{multline}\label{Proofeql3}\oint_\Gamma\oint_{\widetilde \Gamma} W(z,\lambda)f(z)f(\lambda) \frac{d\lambda}{2\pi i}\frac{dz}{2\pi i} =2\sum_{j=1}^{k-1}(f(a_{j+1})-f(b_j))\Bigg(\Upsilon_j(f)\\ +\frac{1}{\pi}\int_Jf'(\lambda)\textrm{Im}u_{j,+}(\lambda)d\lambda\Bigg)
+4\int_J\int_J \log \left|\Theta(\lambda_+,z_+)\right|f'(z)f'(\lambda)\frac{dz}{2\pi i}\frac{d\lambda}{2\pi i}.
\end{multline}
Finally, we want to bring the right hand side of \eqref{Proofeql3} to a slightly different form. 
Taking integration by parts in the definition of $\Upsilon_j$ \eqref{def:Omegahat} and relying on \eqref{eq:ujump1},
\begin{equation}\nonumber
	\Upsilon_j=-\oint_{\Gamma} u_j(z)f'(z)\frac{dz}{2\pi i}-\frac{1}{2\pi i}\sum_{l=1}^{k-1}(f(b_l)-f(a_{l+1}))\tau_{l,j}.
\end{equation}
On $J$, we have $\Re u_{j,+}=\Re u_{j,-}$, and so
	\begin{equation}\label{Omegajintparts} \Upsilon_j=\frac{1}{\pi}\int_J \mathrm{Im}(u_{j,+}(z))f'(z)dz-\frac{1}{2\pi i}\sum_{l=1}^{k-1}(f(b_l)-f(a_{l+1}))\tau_{l,j}.
\end{equation}

We recall that $\tau$ is symmetric and $\mathrm{Im}\tau=(-i)\tau$ is positive definite. From this, and using the fact that $\sum_{l=1}^{k-1}\tau_{m,l} (\tau ^{-1})_{l,j}=\delta_{m,j}$, we find
\begin{multline}
\sum_{j=1}^{k-1}(f(a_{j+1})-f(b_j))\left(\Upsilon_j(f)+\frac{1}{\pi}\int_Jf'(\lambda)\textrm{Im}u_{j,+}(\lambda)d\lambda\right)
\\ =2\pi i \sum_{j,l=1}^{k-1}\Upsilon_j(f)\Upsilon_l(f) (\tau^{-1})_{j,l}\\
	-\frac{2}{\pi}\sum_{j,l}\iint_{J\times J}\mathrm{Im}(u_{j,+}(z))\mathrm{Im}(u_{l,+}(\lambda))f'(z)f'(\lambda)(\mathrm{Im}\tau)^{-1}_{l,j}dzd\lambda. 
\end{multline}

Thus, recalling the definition of $\mathcal L (f)$ from \eqref{defLG}, we obtain
 \begin{equation}\oint_\Gamma\oint_{\widetilde \Gamma} W(z,\lambda)f(z)f(\lambda) \frac{d\lambda}{2\pi i}\frac{dz}{2\pi i} =2\mathcal L(f) -4\pi i \sum_{j,l=1}^{k-1}\Upsilon_j(f)\Upsilon_l(f) (\tau^{-1})_{j,l}.
\end{equation}
Substituting into \eqref{th:smoothasy1} and relying on \eqref{idchange2}, we obtain \eqref{eqlaplace}.

\section{Ratio asymptotics for Fisher-Hartwig symbols}\label{sec:FH}
The goal of this section is to prove the main technical estimates needed for the proof of Theorem \ref{th:FHasy} in Section \ref{sec:FHproof} by using the asymptotics of Section \ref{sec:snorm} and the parametrices from Section \ref{sec:para} as well as the results for the smooth symbol, namely Theorem \ref{th:smoothasy}.

To be more precise, the main idea of the proof is to write (in the notation of Section \ref{sec:FHproof})
\begin{equation}
\label{eq:FHrat}
\frac{H_N(\nu_0)}{H_N(e^{-NV})}=\frac{H_N(\nu_{\epsilon_0})}{H_N(e^{-NV})}\exp\left(-\int_0^{\epsilon_0} \partial_\epsilon \log H_N(\nu_\epsilon)d\epsilon\right)
\end{equation}
and use the identity \eqref{eq:DIFH} to express the logarithmic derivative in terms of the solution to a RH problem. We evaluate the asymptotics of the RH problem by relying on  the local parametrix from Section \ref{sec:Painleve} and the asymptotics of the small norm problem from Section \ref{sec:snorm}. We  use Theorem \ref{th:smoothasy} to evaluate the asymptotics of $H_N(\nu_{\epsilon_0})$ for fixed $\epsilon_0>0$. Actually, the asymptotics of the integral above are not quite as precise as those in the previous section in that we will have error terms of size $\mathcal O(\epsilon_0)$, but taking in the end $\epsilon_0\to 0$ (after we take $N\to\infty$) will provide the desired estimates. We now turn to analyzing the integral term above -- after this, we will see how its $\epsilon_0\to 0$ asymptotics combine with those of the ratio term to produce the result we are after.

\subsection{The integral term}

Our main task in evaluating $\int_0^{\epsilon_0} \partial_\epsilon \log H_N(\nu_\epsilon)d\epsilon$  is to analyze the asymptotics of the differential identity \eqref{eq:DIFH}.
	
Let us begin by noting that for $z$ close enough to $t_j\pm i\epsilon$, and in particular for $z\in U_{t_j}$ (recall that $U_{t_j}$ is a fixed disc centered at $t_j$), we can make use of \eqref{eq:Tdef}, \eqref{eq:Sdef}, and \eqref{eq:Rdef}, to write
\[
Y(z)=e^{-\frac{N\ell }{2}\sigma_3}R(z)P^{(t_j)}(z)e^{Ng(z)\sigma_3}e^{\frac{N\ell}{2}\sigma_3}.
\]
Recalling the definition of $P^{(t_j)}$ from \eqref{eq:local}, a short calculation  shows that  (suppressing the dependence of $\Phi_{\rm V}$ on $\alpha_j,\beta_j$, and $s_j$ and that of $E_{t_j}$ on $t_j$) 
\begin{equation}(Y(z)^{-1}Y'(z))_{11}=\mathcal Y_{1,j}(z)+\mathcal Y_{2,j}(z)+\mathcal Y_{3,j}(z)+\mathcal Y_{4,j}(z), 
\label{e:Y-1Y}
\end{equation}
where
\begin{align*}
\mathcal Y_{1,j}(z) & =  \big[Ng'(z)-N\phi'(z)-\tfrac{f'(z)}{2}-\tfrac{\omega_\epsilon'(z)}{2\omega_\epsilon(z)}\big],\\
\mathcal Y_{2,j}(z) & = N \zeta_{t_j}'(z)\big[\Phi_{\rm V}(N\zeta_{t_j}(z))^{-1}\Phi_{\rm V}'(N\zeta_{t_j}(z))\big]_{rr},\\
\mathcal Y_{3,j}(z) &= \big[\Phi_{\rm V}(N\zeta_{t_j}(z))^{-1}E(z)^{-1}E'(z)\Phi_{\rm V}(N\zeta_{t_j}(z))\big]_{rr},\\
\mathcal Y_{4,j}(z) &=\big[\Phi_{\rm V}(N\zeta_{t_j}(z))^{-1}E(z)^{-1}(R^{-1}R')(z)E(z)\Phi_{\rm V}(N\zeta_{t_j}(z))\big]_{rr},
\end{align*}
for $z$ outside the lenses and in $U_{t_j}$, where $r=2$ for $\Im z>0$ and $r=1$ for $\Im z<0$.

\medskip 

\underline{The $\mathcal Y_{3,j}$- and $\mathcal Y_{4,j}$-terms}. We focus on the term $\mathcal Y_{3,j}(z)$ in the situation where $z\to t_j+i\epsilon$. The remaining cases, namely $\mathcal Y_{3,j}(z)$ with $z\to t_j-i\epsilon$ or $\mathcal Y_{4,j}(z)$ with $z\to t_j\pm i\epsilon$ will be similar and we will leave the details to the reader. 
We will show below that in the $N\to\infty$ and $\epsilon_0\to 0$ limit, $\lim_{z\to t_j+i\epsilon}\mathcal Y_{j,3}(z)=\mathcal O(1)$ uniformly in $\epsilon\in(0,\epsilon_0]$, so after integration over $(0,\epsilon_0]$, this term tends to zero when we first let $N\to\infty$ and then $\epsilon_0\to 0$.

We begin by arguing that $E(z)^{\pm 1}$ and $E'(z)$ are uniformly bounded in $z\in U_{t_j}$, $N$, and $\epsilon\leq \epsilon_0$, so it remains to control the contribution from $\Phi_{\rm V}$. To do this we rely on Proposition \ref{Prop:CIK}.

To see the boundedness of the $E$-terms, we recall the definition of $E$ in \eqref{eq:Edef}, the definition of $M$ in \eqref{eq:Pdef}, the fact that $N_\infty(z;N\Omega+\Upsilon)$ from \eqref{eq:Mlamdef} and $N_\infty'(z;N\Omega+\Upsilon)$ are uniformly bounded on $U_{t_j}$, so using part (e) of Lemma \ref{le:D}, one can check with a routine calculation that $E(z)^{\pm 1}$ and $E'(z)$ are  bounded as $N \to + \infty$ uniformly for $z \in U_{t_{j}}$ and $0<\epsilon\leq \epsilon_0$ -- the boundedness in $N$ makes use of the fact that we are assuming that $\beta\in i\R$. Note that from \eqref{smallnorm}, we see that by the same reasoning also $E^{-1}R^{-1}R'E$ is uniformly bounded in $U_{t_j}$ (actually uniformly $\mathcal O(N^{-1})$).

We turn to the contribution of the $\Phi_{\rm V}$-terms. Recall that for any matrix  $A\in \C^{2\times 2}$ and $q\in \C\setminus \{0\}$, we have  $(q^{-\sigma_3}Aq^{\sigma_3})_{22}=A_{22}$. From Proposition \ref{Prop:CIK} (b) and (d), formula \eqref{eq:PhiV}, and condition (d) for the RH problem for $\Psi$ from Section \ref{sec:Painleve},  we see that for some $c_0>0$ fixed, we have for $\epsilon>c_0 N^{-1}$ (so that $s_{j,N}$ is bounded away from $0$), that the $\Phi_{\rm V}$-terms contribute only a uniformly bounded amount. We now turn to the case $\epsilon=\mathcal O(N^{-1})$, corresponding to small $s$.

Again using the fact that $(q^{-\sigma_3}Aq^{\sigma_3})_{22}=A_{22}$, one finds from \eqref{eq:PhiV}, \eqref{eq:fdef}, and condition (d) in the RH problem for $\Psi$, that
\begin{multline*}
\lim_{z\to t_j+i\epsilon}\mathcal Y_{3,j}(z)=\Bigg[\left(H_1(1)s_{j,N}^{\frac{\alpha_j}{2}\sigma_3}\right)^{-1}s_{j,N}^{-\beta_{j}\sigma_{3}}e^{-\frac{s_{j,N}}{4}\sigma_{3}}E(t_j+i\epsilon)^{-1}\\ E'(t_j+i\epsilon) e^{\frac{s_{j,N}}{4}\sigma_{3}}s_{j,N}^{\beta_{j}\sigma_{3}}H_1(1) s_{j,N}^{\frac{\alpha_j}{2}\sigma_3}\Bigg]_{22}.
\end{multline*}
We  use Proposition \ref{Prop:CIK} item c) (and the fact that $H_1$ has determinant one) to see that the second column of $H_1(1) s^{\frac{\alpha_j}{2}\sigma_3}$ is bounded as $s\to 0$.  Since the other quantities are bounded, we conclude that for $\epsilon \leq c_0/N$
\begin{equation*}\begin{aligned}
\lim_{z\to t_j+i\epsilon}\mathcal Y_{3,j}(z)&=\mathcal O(1)
\end{aligned}
\end{equation*}
uniformly in $\epsilon\in(0,c_0/N]$. Note that this bound also is true for $\epsilon\geq c_0/N$ by our discussion above.  

In fact, by \eqref{eq:sjdef}, and the fact that $\phi$ has a non-zero derivative at $t_j$, it follows that $s_{j,N}>\epsilon$ for $N$ sufficiently large, and so $\lim_{z\to t_j+ i\epsilon}\mathcal Y_{3,j}(z)=\mathcal O\left(1\right)$. As we mentioned before, an analogous statement holds for $\lim_{z\to t_j-i\epsilon}\mathcal Y_{3,j}(z)$ and $\lim_{z\to t_j\pm i\epsilon}\mathcal Y_{4,j}(z)$.

Thus, as $N\to \infty$,
\begin{multline}\label{epsilonFirst}
\left(\frac{\alpha_j}{2}-\beta_j\right)(Y(t_j+i\epsilon)^{-1}Y'(t_j+i\epsilon))_{11}-\left(\frac{\alpha_j}{2}+\beta_j\right)(Y(t_j-i\epsilon)^{-1}Y'(t_j-i\epsilon))_{11} \\ = \lim_{u\to \epsilon}\bigg[\left(\frac{\alpha_j}{2}-\beta_j\right)\left(\mathcal Y_{1,j}(t_j+iu)+\mathcal Y_{2,j}(t_j+iu)\right)\\-\left(\frac{\alpha_j}{2}+\beta_j\right)\left(\mathcal Y_{1,j}(t_j-iu)+\mathcal Y_{2,j}(t_j-iu)\right)\bigg]+\mathcal O\left(1\right) ,
\end{multline}
uniformly in $\epsilon \in(0,\epsilon_0]$.

\medskip

\underline{The $\mathcal Y_{1,j}$-term.} Directly from the definition \eqref{eq:omega_eps}, one readily verifies that
\begin{equation}\label{derivlogomeps}
\frac{\omega_\epsilon'(z)}{\omega_\epsilon(z)}=\sum_{l=1}^p \left(\frac{\frac{\alpha_l}{2}+\beta_l}{z-(t_l-i\epsilon)}+\frac{\frac{\alpha_l}{2}-\beta_l}{z-(t_l+i\epsilon)}\right).
\end{equation}
Thus, we conclude that 
\begin{multline}\label{eq:Y4asy}
\lim_{u\to \epsilon}\left[\left(\frac{\alpha_j}{2}-\beta_j\right)\mathcal Y_{1,j}(t_j+iu)-\left(\frac{\alpha_j}{2}+\beta_j\right)\mathcal Y_{1,j}(t_j-iu)-i\frac{\frac{\alpha_j^2}{4}+\beta_j^2}{u-\epsilon}\right]\\ 
=N\left((\tfrac{\alpha_j}{2}-\beta_j)(g'-\phi')(t_j+i\epsilon)-(\tfrac{\alpha_j}{2}+\beta_j)(g'-\phi')(t_j-i\epsilon)\right) \\
+\frac{i}{2\epsilon}\bigg(\frac{\alpha_j^2}{4}-\beta_j^2\bigg)+\mathcal O(1),
\end{multline}
uniformly in $\epsilon\in(0,\epsilon_0]$. In particular, when we integrate $\epsilon$ over $(0,\epsilon_0]$ and take $\epsilon_0\to 0$, the $\mathcal O(1)$-term is irrelevant. 

\underline{The $\mathcal Y_{2,j}$-terms.} We turn to the analysis of the last term -- $\mathcal Y_{2,j}$, and aim to show that after adding $i\frac{\frac{\alpha_j^2}{4}+\beta_j^2}{u-\epsilon}$ to cancel the corresponding term from $\mathcal Y_{1,j}$, there exists a finite limit which we can control. Here we rely heavily on results and notation in \cite{CIK}. We will also use arguments similar to \cite[Section 5.7]{ClaeysFahs}. 

Let us first express the relevant quantities in terms of $\Psi$ from Section \ref{sec:Painleve}. By \eqref{eq:PhiV},
\begin{multline*}
\big[\Phi_{\rm V}(N\zeta_{t_j}(t_j\pm iu))^{-1})\Phi_{\rm V}'(N\zeta_{t_j}(t_j\pm iu))\big]_{rr}\\=-\tfrac{2i}{s_{j,N}}\big[\Psi(\tfrac{1}{2}-\tfrac{2i N\zeta_{t_j}(t_j\pm iu)}{s_{j,N}})^{-1}\Psi'(\tfrac{1}{2}-\tfrac{2i N\zeta_{t_j}(t_j\pm iu)}{s_{j,N}})\big]_{rr}.
\end{multline*}

Let us now note that from \eqref{eq:fdef} and \eqref{eq:sjdef}, we see that as $u\to \epsilon$, $-2i N\zeta_{t_j}(t_j\pm iu)/s_{j,N}\to \pm \frac{1}{2}$, so we see that this is related to the asymptotics of $\Psi^{-1}\Psi'$ near $0$ and $1$. From part (d) of the RH problem for $\Psi$ from Section \ref{sec:Painleve}, we see that this is most conveniently expressed in terms of the analytic functions $H_0$ and $H_1$ (now used with $\alpha=\alpha_j$, $\beta=\beta_j$, $s=s_{j,N}$). Using the definition and analyticity of $H_0$, and $H_1$, we find that 
\begin{multline} \nonumber
\lim_{u\to \epsilon}\Bigg\{(\Phi_{\rm V}(N\zeta_{t_j}(t_j+ iu))^{-1}\Phi_{\rm V}'(N\zeta_{t_j}(t_j+iu)))_{22}+\tfrac{2i}{s_{j,N}}(H_1^{-1}(1)H_1'(1))_{22}\\
-\left(\tfrac{\alpha_j}{4}-\tfrac{\beta_j}{2}\right)\frac{1}{N\zeta_{t_j}(t_j+iu)-\frac{is_{j,N}}{4}}\Bigg\} = 0,
\end{multline}
and
\begin{multline}\nonumber
\lim_{u\to \epsilon}\Bigg\{(\Phi_{\rm V}(N\zeta_{t_j}(t_j- iu))^{-1}\Phi_{\rm V}'(N\zeta_{t_j}(t_j-iu)))_{11}+\tfrac{2i}{s_{j,N}}(H_0^{-1}(0)H_0'(0))_{11}\\
 -\left(\tfrac{\alpha_j}{4}+\tfrac{\beta_j}{2}\right)\frac{1}{N\zeta_{t_j}(t_j-iu)+\frac{is_{j,N}}{4}}\Bigg\}=0,
\end{multline}
where $H_0$ and $H_1$ depend of course in a non-trivial manner on $s$ and thus $\epsilon$ and $N$, so we still need to understand their behavior. 

By Proposition \ref{Prop:CIK} (e), we conclude that 
\begin{equation}\label{eq:Y3asy} \begin{aligned}
&\lim_{u\to \epsilon}\Bigg\{\left(\frac{\alpha_j}{2}-\beta_j\right)\mathcal Y_{2,j}(t_j+iu)-\left(\frac{\alpha_j}{2}+\beta_j\right)\mathcal Y_{2,j}(t_j-iu)
\\&  +(\tfrac{\alpha_j}{2}-\beta_j)N\zeta_{t_j}'(t_j+iu) \\
& \times \left[\tfrac{2i}{s_{j,N}}\left(\tfrac{s_{j,N}}{2}+\tfrac{\sigma(s_{j,N},\alpha_j,\beta_j)}{\frac{\alpha_j}{2}-\beta_j}-\tfrac{\alpha_j}{4}-\tfrac{\beta_j}{2}\right)-\frac{\tfrac{\alpha_j}{4}-\tfrac{\beta_j}{2}}{N\zeta_{t_j}(t_j+iu)-\frac{is_{j,N}}{4}}\right]\\
&   -(\tfrac{\alpha_j}{2}+\beta_j)N\zeta_{t_j}'(t_j-iu)\\
& \times \left[\tfrac{2i}{s_{j,N}}\left(-\tfrac{s_{j,N}}{2}-\tfrac{\sigma(s_{j,N},\alpha_j,\beta_j)}{\frac{\alpha_j}{2}+\beta_j}+\tfrac{\alpha_j}{4}-\tfrac{\beta_j}{2}\right)-\frac{\tfrac{\alpha_j}{4}+\tfrac{\beta_j}{2}}{N\zeta_{t_j}(t_j-iu)+\frac{is_{j,N}}{4}}\right]\Bigg\} \notag\\
&   =0. 
\end{aligned}\end{equation}

Expanding $\zeta_{t_j}(t_{j} \pm iu) = \pm \frac{is_{j,N}}{4N} \pm i\zeta_{t_j}'(t_{j} \pm i\epsilon)(u-\epsilon) + \bigO((u-\epsilon)^{2})$, and using the notation $\sigma_{j}(s):=\sigma(s;\alpha_j,\beta_j)$, we obtain
\begin{align}\label{eq:Y3asy 2}
&\lim_{u\to \epsilon}\left[\left(\frac{\alpha_j}{2}-\beta_j\right)\mathcal Y_{2,j}(t_j+iu)-\left(\frac{\alpha_j}{2}+\beta_j\right)\mathcal Y_{2,j}(t_j-iu)+i\frac{\frac{\alpha_j^2}{4}+\beta_j^2}{u-\epsilon}\right]\\ \notag
&\qquad = (\tfrac{\alpha_j}{2}-\beta_j)N\zeta_{t_j}'(t_j+i\epsilon)\left(\tfrac{-2i}{s_{j,N}}\right)\left(\tfrac{s_{j,N}}{2}+\tfrac{\sigma_{j}(s_{j,N})}{\frac{\alpha_j}{2}-\beta_j}-\tfrac{\alpha_j}{4}-\tfrac{\beta_j}{2}\right) \\
&\qquad  \quad -(\tfrac{\alpha_j}{2}+\beta_j)N\zeta_{t_j}'(t_j-i\epsilon)\left(\tfrac{-2i}{s_{j,N}}\right)\left(-\tfrac{s_{j,N}}{2}-\tfrac{\sigma_{j}(s_{j,N})}{\frac{\alpha_j}{2}+\beta_j}+\tfrac{\alpha_j}{4}-\tfrac{\beta_j}{2}\right)+\mathcal O(1),\notag
\end{align}
uniformly in $\epsilon\in(0,\epsilon_0]$.
\medskip

\underline{Combining the $\mathcal Y$-terms.} By \eqref{eq:fdef} and  \eqref{eq:sjdef}, 
we have 
\begin{equation*}\begin{aligned}\frac{ds_{j,N}(\epsilon)}{d\epsilon}&=2N(\zeta_{t_j}'(t_j+i\epsilon)+\zeta_{t_j}'(t_j-i\epsilon)),\\
\zeta_{t_j}'(t_j\pm i\epsilon)&=\pm i \phi'(t_j\pm i\epsilon),\end{aligned} \end{equation*}
where here and below we emphasize the $\epsilon$ dependence of $s_{j,N}=s_{j,N}(\epsilon)$.

Thus we obtain
\begin{align}\label{eq:Y3asy 2b}
	&\lim_{u\to \epsilon}\left[\left(\frac{\alpha_j}{2}-\beta_j\right)\mathcal Y_{2,j}(t_j+iu)-\left(\frac{\alpha_j}{2}+\beta_j\right)\mathcal Y_{2,j}(t_j-iu)+i\frac{\frac{\alpha_j^2}{4}+\beta_j^2}{u-\epsilon}\right]\\ \notag
	& =\left(\frac{\alpha_j}{2}-\beta_j\right)N\phi'(t_j+i\epsilon)-\left(\frac{\alpha_j}{2}+\beta_j\right)N\phi'(t_j-i\epsilon)\\
	\notag &  +\frac{1}{i} \frac{d s_{j,N}(\epsilon)}{d\epsilon} \frac{1}{s_{j,N}(\epsilon)}\left( \sigma_j(s_{j,N}(\epsilon))-\frac{\alpha_j^2}{4}+\beta_j^2\right) \nonumber \\
& +\frac{1}{2i} \frac{d\log s_{j,N}(\epsilon)}{d\epsilon} \left( \frac{\alpha_j^2}{4}-\beta_j^2\right)
	+\mathcal O(1). \nonumber 
\end{align}

By substituting  \eqref{eq:Y4asy} and \eqref{eq:Y3asy 2b} into \eqref{epsilonFirst},  relying on \eqref{eq:DIFH} and \eqref{e:Y-1Y}, we obtain after a short calculation that as $N\to\infty$ 
\begin{align*}
& \log \frac{H_N(\nu_{\epsilon_0})}{H_N(\nu_0)} = \sum_{j=1}^{p} \int_{0}^{\epsilon_0}
 \bigg( \sigma_{j}(s_{j,N}(\epsilon)) - \Big( \frac{\alpha_{j}^{2}}{4}-\beta_{j}^{2} \Big) \bigg) \frac{1}{s_{j,N}(\epsilon)} \frac{ds_{j,N}(\epsilon)}{d\epsilon}d\epsilon\\
&\quad   +N\sum_{j=1}^p\left(\frac{\alpha_j}{2}\left( g(t_j+i\epsilon_0)+g(t_j-i\epsilon_0)\right) +\beta_j \left(g(t_j-i\epsilon_0)-g(t_j+i\epsilon_0)\right)-i\pi\beta_j\right)\\
&\quad  -N\sum_{j=1}^p\left(\frac{\alpha_j}{2}\left( g_+(t_j)+g_-(t_j)\right) -\beta_j \left(g_+(t_j)-g_-(t_j)\right)-i\pi\beta_j\right)\\
&\quad  +\sum_{j=1}^p \frac{\frac{\alpha_j^2}{4}-\beta_j^2}{2}\int_0^{\epsilon_0} \left(\frac{d\log s_{j,N}(\epsilon)}{d\epsilon}-\frac{1}{\epsilon}\right)d\epsilon+\mathcal O(\epsilon_0).
\end{align*}

To obtain asymptotics for the first term on the right-hand side, we rely on Proposition \ref{Prop:CIK} part (e).
For the second and third terms, we recall the definition of $\omega_\epsilon$ from Section \ref{sec:FHproof} and $g$ from \eqref{eq:g_fn}, to obtain that 
\begin{multline}\nonumber
\int \log \omega_\epsilon(x)d\mu_V(x)\\ =\sum_{j=1}^p \left(\frac{\alpha_j}{2}\left( g(t_j+i\epsilon)+g(t_j-i\epsilon)\right) +\beta_j \left(g(t_j-i\epsilon)-g(t_j+i\epsilon)\right)-i\pi \beta_j\right)
\end{multline}
and by the definition of $\omega$ in \eqref{eq:omega}
\begin{align*}
\int \log \omega(x)d\mu_V(x)=\sum_{j=1}^p&\left(\frac{\alpha_j}{2}\left( g_+(t_j)+g_-(t_j)\right) -\beta_j \left(g_+(t_j)-g_-(t_j)\right)-i\pi\beta_j\right)
\end{align*}	

For the last term we observe that by \eqref{eq:sjdef},
\begin{equation}
\frac{d\log s_{j,N}(\epsilon)}{d\epsilon}=\frac{1}{\epsilon}+\mathcal O(1),
\label{eq:logs}
\end{equation}
as $\epsilon\to 0$ -- in fact, this logarithmic derivative is independent of $N$. Thus we obtain that as $N\to \infty$,
\begin{align}\label{eq:intermediate}
H_N(\nu_0)&=(1+\mathcal O(\epsilon_0))H_N(\nu_{\epsilon_0})\prod_{j=1}^p\frac{G(1+\frac{\alpha_j}{2}+\beta_j)G(1+\frac{\alpha_j}{2}-\beta_j)}{G(1+\alpha_j)}s_{j,N}(\epsilon_0)^{\frac{\alpha_j^2}{4}-\beta_j^2}\\
&\quad \times  e^{-N\int \log \omega_{\epsilon_0}(x)d\mu_V(x)+N\int \log \omega(x)d\mu_V(x)}.\notag
\end{align}

Naturally the right-hand side should not depend on $\epsilon_0$. Our next task is to see how the $\epsilon_0$-dependence cancels in the limit where $N\to\infty$. For this, we look more carefully at $H_N(\nu_{\epsilon_0})$ using Theorem \ref{th:smoothasy}.

\subsection{The ratio term} \label{Sec:RatioFH}

In this section, we study the ratio term in the right-hand side of \eqref{eq:FHrat}. For notational simplicity, we replace here $\epsilon_0$ by $\epsilon$. By Theorem \ref{th:smoothasy}, if $\epsilon$ is fixed, then
\begin{multline}
\frac{H_N(\nu_{\epsilon})}{H_N(e^{-NV})}=e^{N\int \left(f(x)+ \log \omega_\epsilon(x)\right)d\mu_V(x)}\frac{\theta(N\Omega+\Upsilon_{\epsilon})}{\theta(N\Omega)}\\ \times \exp\left(\frac{1}{4}\oint_{\Gamma}\oint_{\widetilde \Gamma}W(z,\lambda)\left(f(z)+\log \omega_\epsilon(z)\right)\left(f(\lambda)+\log \omega_\epsilon(\lambda)\right)\frac{dz}{2\pi i}\frac{d\lambda}{2\pi i}\right) \\
\times (1+\mathcal O(N^{-1})),
\label{eq:rat}
\end{multline}
as $N\to \infty$, where $\Gamma=\cup_{j=1}^k \Gamma_j$, $\widetilde\Gamma$ is as in Theorem \ref{th:smoothasy}, and the implied constant in $\mathcal O(N^{-1})$ may depend on $\epsilon$. Moreover, we have written $\Upsilon_\epsilon$ to emphasize that $\Upsilon$ depends on $\epsilon$. To reiterate, we wish to study the behavior of this in the limit where first $N\to \infty$ and then $\epsilon\to 0$. We begin by expanding the $f+\log\omega_\epsilon$ product above and looking at the cross term.
 
For the cross term, using integration by parts and \eqref{derivlogomeps} (and recalling the definition of $W$, $w$, $d_f$, and $\omega$ from  \eqref{def:W}, \eqref{defwlambda}, \eqref{defdg}, and \eqref{eq:omega_eps})
\begin{equation}\label{e:cross} \begin{aligned} &\frac{1}{2}\oint_{\Gamma}\oint_{\widetilde \Gamma}W(z,\lambda)f(z)\log \omega_\epsilon(\lambda) \frac{dz}{2\pi i}\frac{d\lambda}{2\pi i }=\frac{1}{2}\oint_{\widetilde \Gamma}\partial_\lambda \left(\oint_{\Gamma} w_\lambda(z)f(z)\frac{dz}{2\pi i}\right)\log \omega_\epsilon(\lambda) \frac{d\lambda}{2\pi i }\\
&\quad =-\oint_{\widetilde \Gamma} d_{f}'(\lambda)\log \omega_\epsilon(\lambda)\frac{d\lambda}{2\pi i}\\
 &\quad
=-\sum_{j=1}^p\left[
\left(\frac{\alpha_j}{2}+\beta_j\right)d_f(t_j-i\epsilon)+\left(\frac{\alpha_j}{2}-\beta_j\right)d_f(t_j+i\epsilon)-\alpha_jd_f(\infty)\right] \\&\quad \to 
-\sum_{j=1}^p\left[
\left(\frac{\alpha_j}{2}+\beta_j\right)d_{f,-}(t_j)+\left(\frac{\alpha_j}{2}-\beta_j\right)d_{f,+}(t_j)-\alpha_jd_f(\infty)\right],
\end{aligned}
\end{equation}
as $\epsilon\to 0$.
To simplify this, note first that by \eqref{jumpsdf1}, $d_{f,+}(t_j)+d_{f,-}(t_j)=f(t_j)$, which takes care of the $\alpha_j$-terms above. For the $\beta_j$-terms we claim that 
\begin{equation*}
d_{f,+}(t_j)-d_{f,-}(t_j)=\frac{1}{\pi i} \mathcal{P.V.}\int_J  w_{t_j,+}(\lambda_+)f(\lambda) d\lambda,
\end{equation*} 
where $w_{t_j,+}(\lambda_+)$ is the boundary value of $w_z(\lambda)$ from the $+$ side in both the $z$ variable and $\lambda$ variable as $z\to t_j$, and where $\mathcal{P.V.}\int$ is the principal value integral and $(\lambda-t_j)w_{t_j,+}(\lambda_+)$ is differentiable on the interior of $J$.

This follows from the following reasoning: denote
\begin{equation*} \phi_f(z)=\frac{1}{2\pi i}\int_J \frac{f(x)dx}{\mathcal R_+^{1/2}(x)(x-z)}. \end{equation*}
Then, by \eqref{2ndrepwlambda}--\eqref{defdg},
\begin{equation*} d_f(z)=\mathcal R^{1/2}(z)\left(\phi_f(z)+2\int_Jf(\lambda)\sum_{j=1}^{k-1}u_{j,+}'(\lambda)\int_{b_j}^{a_{j+1}} \frac{dx}{\mathcal R^{1/2}(x)(x-z)}\frac{d\lambda}{2\pi i}\right). \end{equation*}
By Sokhotski-Plemelj formula
\begin{equation*}
	\phi_{f,\pm}(z)=\pm \frac{f(z)}{2\mathcal R^{1/2}_+(z)}+\frac{1}{2\pi i}\mathcal {P.V.}\int_J \frac{f(x)dx}{\mathcal R^{1/2}_+(x)(x-z)},
\end{equation*}
for $z\in J$. Thus
\begin{multline*}
	d_{f,\pm}(z)=\pm \mathcal R^{1/2}_+(z)\Bigg(\pm \frac{f(z)}{2\mathcal R^{1/2}_+(z)}+\frac{1}{2\pi i}\mathcal {P.V.}\int_J \frac{f(\lambda)d\lambda}{\mathcal R^{1/2}_+(\lambda)(\lambda-z)}\\+\int_J2f(\lambda)\sum_{j=1}^{k-1}u_{j,+}'(\lambda)\int_{b_j}^{a_{j+1}} \frac{dx}{\mathcal R^{1/2}(x)(x-z)}\frac{d\lambda}{2\pi i}\Bigg),
\end{multline*}
which implies that 
\begin{equation*}d_{f,+}(t_j)-d_{f,-}(t_j)=\frac{1}{\pi i} \mathcal{P.V.}\int_J f(\lambda) w_{t_j,+}(\lambda_+)d\lambda
\end{equation*}
as claimed. 

We conclude that
\begin{multline}\label{formulamix}-\sum_{j=1}^p\left[
\left(\frac{\alpha_j}{2}+\beta_j\right)d_{f,-}(t_j)+\left(\frac{\alpha_j}{2}-\beta_j\right)d_{f,+}(t_j)-\alpha_jd_f(\infty)\right]\\ =-\sum_{j=1}^p\left(\frac{\alpha_j}{2}f(t_j)-\frac{\beta_j}{\pi i}\mathcal{P.V.}\int_J w_{t_j,+}(\lambda_+)f(\lambda)d\lambda-\alpha_jd_f(\infty)\right).\end{multline}

We now turn to the $\log \omega_\epsilon\log\omega_\epsilon$-term in \eqref{eq:rat} whose behavior is more complicated and we describe in the following lemma.

\begin{lemma}
As $\epsilon\to 0$,
\begin{multline*}
\exp\left(\frac{1}{4}\oint_{\Gamma}\oint_{\widetilde \Gamma} W(z,\lambda)\log \omega_\epsilon(z)\log \omega_\epsilon(\lambda)\frac{dz}{2\pi i}\frac{d\lambda}{2\pi i}\right)\\ =
(1+o(1))\prod_{j=1}^p\left(2\epsilon \right)^{-\frac{\alpha_j^2}{4}+\beta_j^2} e^{-\frac{\pi i \beta_j\alpha_j}{4}-\sum_{l<j}\frac{\alpha_l\beta_j\pi i}{2}}\exp \Big( \frac{\alpha_{j}}{2}d_{\epsilon=0}(\infty) \Big)
 \\ \times
\left(\prod_{l\neq j}|t_j-t_l|^{\beta_j\beta_l-\alpha_j\alpha_l/4}\right)\left(\prod_{l,j=1}^p\Theta(t_{j,+},\infty)^{\frac{\beta_j\alpha_l}{2}}\left|\widetilde \Theta(t_{l,+},t_{j,+})\right|^{\beta_j\beta_l} \right) \\
 \times \prod_{j=1}^{p} \exp \bigg( -\sum_{l=1}^{j-1} \frac{i \pi \beta_{j} \alpha_{l}}{4} + \sum_{l=j+1}^{p}\frac{i \pi \beta_{j} \alpha_{l}}{4} \bigg),
\end{multline*}
where $\Theta$ is as in \eqref{Theta}, and 
\[
\widetilde \Theta(z,w)=\frac{\Theta(z,w)}{w-z}.
\]
\label{le:logwlogw}
\end{lemma}
Substituting the asymptotics of Lemma \ref{le:logwlogw}, \eqref{formulamix}, and \eqref{e:cross} into \eqref{eq:rat}, and using the fact that $\Upsilon_\epsilon\to \Upsilon_0$ as $\epsilon\to 0$, we obtain
\begin{equation}\begin{aligned}
\frac{H_N(\nu_{\epsilon})}{H_N(e^{-NV})}&=e^{N\int \left(f(x)+ \log \omega_\epsilon(x)\right)d\mu_V(x)}\frac{\theta(N\Omega+\Upsilon_{0})}{\theta(N\Omega)}\\ & \times \exp\left(\frac{1}{4}\oint_{\Gamma}\oint_{\widetilde \Gamma}W(z,\lambda)f(z)f(\lambda)\frac{dz}{2\pi i}\frac{d\lambda}{2\pi i}\right)\\
& \times\exp \sum_{j=1}^p\left(-\frac{\alpha_j}{2}f(t_j)+\frac{\beta_j}{\pi i}\mathcal{P.V.}\int_J w_{t_j,+}(\lambda_+)f(\lambda)d\lambda+\alpha_jd_f(\infty)\right)
\\ & \times 
\prod_{j=1}^p\left(2\epsilon \right)^{-\frac{\alpha_j^2}{4}+\beta_j^2} e^{-\frac{\pi i \beta_j\alpha_j}{4}-\sum_{l<j}\frac{\alpha_l\beta_j\pi i}{2}}\exp \Big( \frac{\alpha_{j}}{2}d_{\epsilon=0}(\infty) \Big)
 \\ & \times
\left(\prod_{l\neq j}|t_j-t_l|^{\beta_j\beta_l-\alpha_j\alpha_l/4}\right)\left(\prod_{l,j=1}^p\Theta(t_{j,+},\infty)^{\frac{\beta_j\alpha_l}{2}}\left|\widetilde \Theta(t_{l,+},t_{j,+})\right|^{\beta_j\beta_l} \right) \\ &
 \times \prod_{j=1}^{p} \exp \bigg( -\sum_{l=1}^{j-1} \frac{i \pi \beta_{j} \alpha_{l}}{4} + \sum_{l=j+1}^{p}\frac{i \pi \beta_{j} \alpha_{l}}{4} \bigg)(1+\mathcal O(N^{-1})+\mathcal O(\epsilon)),
\label{eq:rat2}
\end{aligned}
\end{equation}
for any fixed $\epsilon>0$ as $N\to \infty$ (the implicit constants in the $\mathcal O(N^{-1})$ term depends on $\epsilon$ while the implicit constants in the $\mathcal O(\epsilon)$ term are independent of both $N$ and $\epsilon$).

We now prove Lemma \ref{le:logwlogw}.
\begin{proof}
Recall from \eqref{defdg} that 
\begin{equation} \label{Qfepsilon}
\frac{1}{4}\oint_{\Gamma}\oint_{\widetilde \Gamma} W(z,\lambda)\log \omega_\epsilon(z)\log \omega_\epsilon(\lambda)\frac{dz}{2\pi i}\frac{d\lambda}{2\pi i} = -\frac{1}{2} \oint_{\widetilde\Gamma} d_\epsilon'(z) \log \omega_\epsilon(z)\frac{dz}{2\pi i}, \end{equation}
where $d_\epsilon=d_{\log \omega_\epsilon}$ and where $\log \omega_\epsilon$ was defined in  \eqref{logomegaepsilon}. For notational simplicity, we will from now on not distinguish between $\widetilde \Gamma$ and $\Gamma$.

Since $d_\epsilon'(z)$ is analytic on $\mathbb C\setminus J$ and (by \eqref{defdg} and \eqref{id:W}) $d_\epsilon'(z)=\frac{1}{2\pi i}\int_{J}\log \omega_\epsilon(\lambda)W(z,\lambda_+)d\lambda =\mathcal O (z^{-2})$ as $z\to \infty$, we obtain by deforming the contour of integration to the branch cut of the logarithm, 
\begin{equation}\label{resintdeps}
\oint_{\Gamma} d_\epsilon'(z) \log (z-(t_j\pm i\epsilon))\frac{dz}{2\pi i}=d_\epsilon(t_j\pm i\epsilon)-d_\epsilon(\infty),
\end{equation}
where throughout the proof $\log(z-(t_j-i\epsilon))$ will take arguments in $(-\pi/2, 3\pi /2)$ and $\log(z-(t_j+i\epsilon))$ will take arguments in $(-3\pi  /2, \pi /2)$ (similarly as in the definition of $\log \omega_\epsilon$ in \eqref{logomegaepsilon}). Thus to study asymptotics of \eqref{Qfepsilon}, it is in fact sufficient to understand asymptotics of $d_\epsilon(t_j\pm i \epsilon)$:
\begin{align*}
& \frac{1}{4}\oint_{\Gamma}\oint_{\widetilde \Gamma} W(z,\lambda)\log \omega_{\epsilon}(z) \log \omega_{\epsilon}(\lambda) \frac{dz}{2\pi i}\frac{d\lambda}{2\pi i}  \\
& = -\frac{1}{2}\sum_{j=1}^{p} \bigg[ \bigg( \frac{\alpha_{j}}{2}-\beta_{j} \bigg)d_{\epsilon}(t_{j}+i\epsilon)+\bigg( \frac{\alpha_{j}}{2}+\beta_{j} \bigg) d_{\epsilon}(t_{j}-i\epsilon) - \alpha_{j} d_{\epsilon}(\infty) \bigg].
\end{align*}

Denote 
\begin{align} \label{depsilonj}d_{\epsilon,j}(z)&=-\frac{1}{4\pi i}\oint_\Gamma w_z(\lambda)\log \omega_{\epsilon,j}(\lambda)d\lambda,\\ 
\label{logomegaj}\log  \omega_{\epsilon,j}(z)&=\log \omega_\epsilon(z)-\left(\frac{\alpha_j}{2}-\beta_j\right)\log(z-(t_j+i\epsilon)) \\ & -\left(\frac{\alpha_j}{2}+\beta_j\right)\log(z-(t_j-i\epsilon)). \nonumber
\end{align} 
With such a definition of $\omega_{\epsilon,j}$, we have in the notation \eqref{eq:omega} 
 \begin{equation}\label{defomegaj}\lim_{\epsilon\to 0}\omega_{\epsilon,j}(x)={\omega_j}(x)e^{-\pi i \beta_j}, \qquad  {\omega_j}(x)=\frac{\omega(x)}{\omega_{\alpha_j}(x) \omega_{\beta_j}(x)}. \end{equation}
The purpose of introducing this notation is that we will find it convenient to decompose $d_\epsilon(t_j\pm i \epsilon)$ in the following way:
\begin{multline}\label{depsilonsplit}
d_\epsilon(t_j\pm i\epsilon)=d_{\epsilon,j}(t_j\pm i\epsilon)+(\alpha_j/2-\beta_j)d_{\log(z-(t_j+i\epsilon))}(t_j\pm i\epsilon)\\
+(\alpha_j/2+\beta_j)d_{\log (z-(t_j-i\epsilon))}(t_j\pm i \epsilon),
\end{multline}
and then analyze $d_{\epsilon,j}$ and the remaining terms separately, starting with  $d_{\epsilon,j}(t_j\pm i\epsilon)$, defined in \eqref{depsilonj}.

\medskip 

\underline{The $d_{\epsilon,j}$-term:} By definition, $\log \omega_{\epsilon,j}$ has logarithmic singularities at $t_l\pm i\epsilon$ for $l\neq j$, but no singularity at $t_j\pm i\epsilon$. On the other hand, by \eqref{polewlambda}, the  only singularity of $w_{t_j\pm i\epsilon}$ is a pole at $t_j\pm i\epsilon$ and $w_{t_j\pm i \epsilon}(\lambda)=\frac{1}{\lambda-(t_j\pm i\epsilon)}+\mathcal O(1)$ as $\lambda \to  t_j\pm i\epsilon$, and $w_{t_j\pm i\epsilon}$ is  analytic on $\mathbb C\setminus (J\cup \{t_j\pm i\epsilon\})$. 
From \eqref{2ndrepwlambda}, $w_z(\lambda)=\mathcal O(\lambda^{-2})$ as $\lambda \to \infty$.  We deform $\Gamma$ in \eqref{depsilonj} to the logarithmic branches, and additionally evaluate the residue at $\lambda=t_j\pm i\epsilon$, to obtain that as $\epsilon\to 0$
\begin{multline} \label{depsilontj+ieps}
d_{\epsilon,j}(t_j\pm i\epsilon)=\frac{1}{2}\log \omega_{\epsilon,j}(t_j\pm i\epsilon) \\-\frac{1}{2}\sum_{l\neq j}\left[(\alpha_l/2+ \beta_l)\int_{-i\infty}^{t_l-i\epsilon}w_{t_j\pm i \epsilon}(\lambda)d\lambda+(\alpha_l/2- \beta_l)\int_{+i\infty}^{t_l+i\epsilon}w_{t_j\pm i \epsilon}(\lambda)d\lambda\right]\\
\to\frac{1}{2}\log \omega_{j}(t_j) -\frac{\pi i \beta_j}{2}-\frac{1}{2}\sum_{l\neq j} \Bigg[ \frac{\alpha_l}{2}\left(\int_{-i\infty}^{t_{l,-}}w_{t_{j,\pm}}(\lambda)d\lambda+\int_{+i\infty}^{t_{l,+}}w_{t_{j,\pm}}(\lambda)d\lambda\right)\\
+\beta_l\left(\int_{-i\infty}^{t_{l,-}}w_{t_{j,\pm}}(\lambda)d\lambda-\int_{+i\infty}^{t_{l,+}}w_{t_{j,\pm}}(\lambda)d\lambda\right)
\Bigg]
.
\end{multline}

We now evaluate the integrals of $w_{t_j,\pm}$. We will be taking the logarithm of $\Theta$, and we begin by fixing the branch of the logarithm. For $x\in (a_1,b_k)$,  $\log \Theta(x_\pm ,\lambda)$ is analytic for $\lambda \in \mathbb C\setminus J$ because $w_z(\lambda)$ has no residue at $\infty$ as a function of $\lambda$.  By Lemma \ref{Prime} and the fact that $u(b_k)=0$, $\theta \left[\substack{\ualpha \\ \ubeta}\right](u(z))$ is not identically zero. Thus, recalling the definition of $\Theta$ in \eqref{Theta}, it follows that $\Theta(z,b_k)=1$ for all $z$. We fix the branch of the logarithm so that
\begin{equation}\label{BranchTheta} \log \Theta(z,b_k)=0\end{equation}
for all $z$.

By combining \eqref{JumpsTheta1} and \eqref{JumpsTheta2}, if $x\in\{b_j,a_{j+1}\}$ for $j=1,\dots,k-1$, then $\Theta(x_\pm, \lambda)^2=e^{\pm 4\pi i u_j(\lambda)}$.  By our choice of branch for the logarithm,  $\log \Theta(x_\pm, \lambda)=\pm 2\pi i u_j(\lambda)$ for $x\in \{b_j,a_{j+1}\}$. Thus for any $z$,
\begin{equation}\label{zerointgap}
\int_{b_j}^{a_{j+1}}w_{z}(\lambda_\pm)d\lambda=\log \Theta(a_{j+1,\pm},z)-\log \Theta(b_{j,\pm},z)=0. \end{equation}
Now by \eqref{2ndrepwlambda}, if $\lambda\in J$, we have $w_{t_{j,\pm}}(\lambda_+)=-w_{t_{j,\pm}}(\lambda_-)$. Thus, by combining with \eqref{zerointgap}, if $t_l>t_j$, contour deformation and our choice of the branch of $\log \Theta(x_\pm ,\lambda)$ yield that
\begin{align*}
& -\left(\int_\infty^{t_{l,-}}w_{t_{j,\pm}}(\lambda)d\lambda+\int_\infty^{t_{l,+}}w_{t_{j,\pm}}(\lambda)d\lambda\right)  \\
& = \int_{(t_l,+\infty)}(w_{t_j,\pm}(\lambda_+)+w_{t_j,\pm}(\lambda_-)) d\lambda =2\int_{b_k}^\infty w_{t_{j,\pm}}(\lambda)d\lambda=2\log \Theta(t_{j,\pm},\infty).
\end{align*}
Since $w_{t_{j,\pm}}(\lambda)$ is purely imaginary for $\lambda>b_k$ by \eqref{2ndrepwlambda}, it follows that $\log \Theta(t_{j,\pm},\infty)$ is purely imaginary.
If $t_l<t_j$, then by \eqref{polewlambda} we pick up an additional residue at $\lambda=t_j$:
\begin{equation*}\begin{aligned}
-\left(\int_\infty^{t_{l,-}}w_{t_{j,\pm}}(\lambda)d\lambda+\int_\infty^{t_{l,+}}w_{t_{j,\pm}}(\lambda)d\lambda\right)&=\mp 2\pi i +2\int_{b_k}^\infty w_{t_{j,\pm}}(\lambda)d\lambda\\&=\mp 2\pi i+2\log \Theta(t_{j,\pm},\infty).
\end{aligned}
\end{equation*}

Similarly, we have 
\begin{equation} \label{sum1wtj} \int_\infty^{t_{l,-}}w_{t_{j,\pm}}(\lambda)d\lambda-\int_\infty^{t_{l,+}}w_{t_{j,\pm}}(\lambda)d\lambda=\log \frac{\Theta(t_{j,\pm},t_{l,-})}{\Theta(t_{j,\pm},t_{l,+})}, \end{equation}
and since we can assume without loss that the integrals are taken to follow the real line, it follows by \eqref{2ndrepwlambda} that \eqref{sum1wtj} is real,
 and by \eqref{JumpsTheta1} is equal to $-2\log \left|\Theta(t_{j,\pm},t_{l,+})\right|$.

Thus, by \eqref{depsilontj+ieps},
\begin{multline*}
\lim_{\epsilon\to 0} \exp(d_{\epsilon,j}(t_j\pm i\epsilon))\\ =e^{-\pi i \beta_j/2}\sqrt{{\omega_j}(t_j)}\prod_{l\neq j}\left|\Theta(t_{l,+},t_{j,\pm})\right|^{\beta_l} \Theta(t_{j,\pm},\infty)^{\alpha_l/2}e^{\mp \pi i \mathbf 1\{t_l<t_j\}\alpha_l/2}, \end{multline*}
where ${\omega_j}$ was given in \eqref{defomegaj}.
Relying again on \eqref{JumpsTheta1},
\begin{multline}\label{prod1}
\lim_{\epsilon\to 0}\exp\left[\left(\frac{\alpha_j}{2}-\beta_j\right)d_{\epsilon,j}(t_j+i\epsilon)+\left(\frac{\alpha_j}{2}+\beta_j\right)d_{\epsilon,j}(t_j-i\epsilon)\right]\\=e^{-\pi i \alpha_j\beta_j/2}{\omega_j}(t_j)^{\alpha_j/2}\prod_{l\neq j}\left|\Theta(t_{l,+},t_{j,+})\right|^{-2\beta_j\beta_l}
\Theta(t_{j,+},\infty)^{-\beta_j\alpha_l}e^{\pi i \alpha_l\beta_j\mathbf 1\{t_l<t_j\}}
\end{multline}
The above concludes the interaction between the singularities -- namely the contribution of the $d_{\epsilon,j}$-terms.

\medskip 

\underline{The $d_{\log(x-(t_j\pm i \epsilon))}$-terms:} We define the limit
\begin{equation}\label{hatdj}
\begin{aligned}
\hat d_{j,\pm}&=\lim_{\epsilon\to 0} d_{\log(z-(t_j\pm i\epsilon))}(t_j\pm i\epsilon)\\
&=\lim_{\epsilon\to 0}\frac{1}{2\pi i}\int_J\log(\lambda-(t_j\pm i \epsilon))w_{t_j\pm i\epsilon}(\lambda_+)d\lambda
\end{aligned}
\end{equation}
 which is well defined since we can deform the integration contour in a downward or upward direction around the singularity as needed, and recall that $\log(z-(t_j-i\epsilon))$ takes arguments in $(-\pi/2, 3\pi /2)$ and $\log(z-(t_j+i\epsilon))$ takes arguments in $(-3\pi  /2, \pi /2)$.
 
On the other hand, writing 
\[
d_{\log (z-(t_j\pm i \epsilon))}(t_j\mp i\epsilon)=\frac{1}{2\pi i}\int_J \log(\lambda-(t_j\pm i \epsilon)) w_{t_j\mp i\epsilon}(\lambda_+)d\lambda
\]
and deforming the contour of integration from $J$  past the pole $t_j\mp i \epsilon$ (possibly using \eqref{jumpswlambda1}) to a contour $\widetilde J_\mp$, we find by the residue theorem
\[
d_{\log (z-(t_j\pm i \epsilon))}(t_j\mp i\epsilon)=\mp \frac{1}{2\pi i}\int_{\widetilde J_\mp}\log(\lambda-(t_j\pm i \epsilon))w_{t_j\mp i\epsilon}(\lambda)d\lambda+ \log(\mp 2i\epsilon).
\]
By our choice of $\widetilde J_\mp$ (both the pole and logarithmic singularity are on the same side of $\widetilde J_\mp$), we find (using \eqref{jumpswlambda2}, contour deforming, and again possibly using \eqref{jumpswlambda1}) 
\begin{equation}\label{prod2}\begin{aligned}
d_{\log(z-(t_j\pm i\epsilon))}(t_j\mp i\epsilon)&=\pm \frac{1}{2\pi i}\int_{\widetilde J_\mp}\log(\lambda-t_j)w_{t_j}(\lambda)d\lambda+ \log(\mp 2i\epsilon)+o(1)\\ &= \mp \frac{\pi i}{2}+\log ( 2\epsilon)-\hat d_{j,\pm}+o(1),
\end{aligned}\end{equation}
as $\epsilon\to 0$.

Combining \eqref{hatdj} and \eqref{prod2}, we find for the latter two terms in \eqref{depsilonsplit} 
	\begin{align}
	\label{eq:prod31}
	&(\alpha_j/2-\beta_j)d_{\log(z-(t_j+i\epsilon))}(t_j+ i\epsilon)	+(\alpha_j/2+\beta_j)d_{\log (z-(t_j-i\epsilon))}(t_j+ i \epsilon)\\
	&\qquad =(\alpha_j/2-\beta_j)\hat d_{j,+}+(\alpha_j/2+\beta_j)\left(\frac{\pi i}{2}+\log(2\epsilon)-\hat d_{j,-}\right)+o(1) \notag 
	\end{align}
	and 
	\begin{align}
	\label{eq:prod32}
	&(\alpha_j/2-\beta_j)d_{\log(z-(t_j+i\epsilon))}(t_j- i\epsilon)	+(\alpha_j/2+\beta_j)d_{\log (z-(t_j-i\epsilon))}(t_j- i \epsilon)\\
	&\qquad =(\alpha_j/2-\beta_j)\left(-\frac{\pi i}{2}+\log(2\epsilon)-\hat d_{j,+}\right)+(\alpha_j/2+\beta_j)\hat d_{j,-}+o(1) \notag 
	\end{align}
	
    Plugging \eqref{eq:prod31}, \eqref{eq:prod32}, and \eqref{prod1}, into \eqref{depsilonsplit}, which we then use in \eqref{resintdeps}, and ultimately (recalling \eqref{logomegaepsilon}) in \eqref{Qfepsilon}, we find 
    \begin{equation} 
    \label{FHform1}\begin{aligned} &\exp \left(\frac{1}{4}\oint_{\Gamma}\oint_{\widetilde \Gamma} W(z,\lambda)\log \omega_\epsilon(z)\log \omega_\epsilon(\lambda)\frac{dz}{2\pi i}\frac{d\lambda}{2\pi i}\right)
     =  
    \prod_{j=1}^p\left(2\epsilon \right)^{-\frac{\alpha_j^2}{4}+\beta_j^2}\\ &\times \exp\left(-\beta_j^2\left(\hat d_{j,+}+\hat d_{j,-}\right)\right)e^{\pi i \frac{\beta_j\alpha_j}{4}}\exp\left(\frac{\alpha_j}{2}d_\epsilon(\infty)\right) \exp\left(\frac{\alpha_j\beta_j}{2}\left(\hat d_{j,+}-\hat d_{j,-}\right)\right)\\
     &  \times    
    {\omega_j}(t_j)^{-\alpha_j/4}\prod_{l\neq j}\left|\Theta(t_{l,+},t_{j,+})\right|^{\beta_j\beta_l}\Theta(t_{j,+},\infty)^{\frac{\beta_j\alpha_l}{2}}e^{-\frac{\pi i}{2}\alpha_l\beta_j \mathbf 1\{t_l<t_j\}}(1+o(1)).
    \end{aligned}\end{equation}

We now consider $\hat d_{j,\pm}$. Observe that if $\Gamma$ encloses $J$ but not $t_j\pm i\epsilon$, then 
\begin{equation*}\frac{1}{2}\oint_{\Gamma} \frac{\log(\lambda-(t_j\pm i\epsilon))}{\lambda-(t_j\pm i\epsilon)}\frac{d\lambda}{2\pi i}=0. \end{equation*}
Thus we add this term to the  definition of $d_{\log(z-(t_j\pm i\epsilon))}(t_j\pm i\epsilon)$ in \eqref{defdg} to obtain
\begin{multline*}
d_{\log(z-(t_j\pm i\epsilon))}(t_j\pm i\epsilon)\\=-\frac{1}{2}\oint_\Gamma \log(\lambda-(t_j\pm i\epsilon)) \left(
\frac{d}{d\lambda} \log \Theta(t_j\pm i\epsilon, \lambda)-\frac{1}{\lambda-(t_j\pm i\epsilon)}\right)\frac{d\lambda}{2\pi i}.
\end{multline*}
We deform the contour to the branch of the logarithm, and since $\frac{d}{d\lambda} \log \Theta(t_j\pm i\epsilon,\lambda)=\mathcal O(1/\lambda^2)$ as $\lambda \to \infty$, we obtain that as $R\to \infty$,
\begin{multline}\label{intint}
d_{\log(z-(t_j\pm i\epsilon))}(t_j\pm i\epsilon)=\frac{1}{2}\int_{t_j\pm i \epsilon}^{t_j\pm iR}\frac{d}{d\lambda} \log \frac{\Theta(t_j\pm i\epsilon,\lambda)}{\lambda-(t_j\pm i\epsilon)} d\lambda\\+\frac{1}{2}\oint_{\partial B_R} \log(\lambda-(t_j\pm i\epsilon)) \left(
\frac{1}{\lambda-(t_j\pm i\epsilon)}\right)\frac{d\lambda}{2\pi i} +\mathcal O(1/R)
\end{multline}
where $\partial B_R$ is a circle of radius $R$ centered at $t_j$, oriented counterclockwise. Recalling the notation $\widetilde \Theta(z,\lambda)=\frac{\Theta(z,\lambda)}{\lambda-z}$ and evaluating the $\partial B_R$ integral explicitly, this becomes
\begin{multline*}
	d_{\log(z-(t_j\pm i\epsilon))}(t_j\pm i\epsilon)=\frac{1}{2}\Big[\log \Theta(t_j\pm i \epsilon,\infty) -\log (\pm i R)\\-\log \widetilde \Theta(t_j\pm i \epsilon,t_j\pm i\epsilon)
	+\log R \mp \frac{\pi i}{2}
	\Big]+\mathcal O(1/R)\\
	\to \frac{1}{2}\Big[\log \Theta(t_{j,\pm},\infty)-\log \widetilde \Theta(t_{j,\pm},t_{j,\pm})\mp \pi i \Big],
\end{multline*}
as $\epsilon \to 0$ and $R\to \infty$. Now write
\begin{multline*}
\log \Theta(t_{j,\pm},\infty)-\log \widetilde \Theta(t_{j,\pm},t_{j,\pm})\\=\lim_{R\to +\infty}\left(\int_{t_j}^R\left( w_{t_{j,\pm}}(\lambda_\pm)-\frac{1}{\lambda-t_j}\right) d\lambda+\log R\right).
\end{multline*}
By \eqref{2ndrepwlambda}, we obtain that $w_{t_{j,+}}(\lambda_+)=w_{t_{j,-}}(\lambda_-)\in \mathbb R$ for $\lambda \in J$ and $w_{t_{j,+}}(\lambda_+)=-w_{t_{j,-}}(\lambda_-)$ is purely imaginary for $\lambda \in \mathbb R\setminus J$. 
Thus, the contribution to $\hat d_{j,+}+\hat d_{j,-}$ is the real part of the above integral, and the contribution to $\hat d_{j,+}-\hat d_{j,-}$ is the imaginary part of it -- this can also be written as the integral of $w_{t_j,+}$ over $(b_k,\infty)$:
\begin{equation*}\begin{aligned}
\hat d_{j,+}+\hat d_{j,-}&=
\log \left|\frac{\Theta(t_{j,+},\infty)}{\widetilde \Theta(t_{j,+},t_{j,+})}\right|,\\
\hat d_{j,+}-\hat d_{j,-}&=\int_{b_k}^{+\infty} w_{t_{j,+}}(\lambda_+)-\pi  i.
\end{aligned}
\end{equation*}

By \eqref{zerointgap} and the fact that $\Theta(z,b_k)=1$, it follows that 
\begin{equation*}
\hat d_{j,+}-\hat d_{j,-}=\log \Theta(t_{j,+},\infty)-\pi i.
\end{equation*}

Thus, by \eqref{FHform1}, we obtain the lemma.
\end{proof}

\subsection{Combining the terms}
We will now substitute \eqref{eq:intermediate} and \eqref{eq:rat2} into \eqref{eq:FHrat}. First consider the term $s_{j,N}(\epsilon)$ appearing in \eqref{eq:intermediate}. By \eqref{eq:sjdef},
\begin{equation*} s_{j,N}(\epsilon)=4Ni\phi_+'(t_j)\epsilon \left(1 +\mathcal O(\epsilon)\right), \end{equation*}
as $\epsilon\to 0$. By \eqref{eq:phidef},
\begin{equation*} s_{j,N}(\epsilon)=4\pi N\psi_V(t_j)\epsilon\left(1  +\mathcal O(\epsilon)\right). \end{equation*}
 Thus, substituting \eqref{eq:intermediate} and \eqref{eq:rat2} into \eqref{eq:FHrat}, we obtain that
\begin{equation*}\begin{aligned}
&\frac{H_N(\nu_0)}{H_N(e^{-NV})}=
e^{ N\int \left(f(x)+ \log \omega(x)\right)d\mu_V(x)}\frac{\theta(N\Omega+\Upsilon_{0})}{\theta(N\Omega)}\\ & \qquad \times \exp\left(\frac{1}{4}\oint_{\Gamma}\oint_{\widetilde\Gamma}W(z,\lambda)f(z)f(\lambda)\frac{dz}{2\pi i}\frac{d\lambda}{2\pi i}\right)\\
&\qquad  \times\exp\sum_{j=1}^p\left(-\frac{\alpha_j}{2}f(t_j)+\frac{\beta_j}{\pi i}\mathcal{P.V.}\int_J w_{t_j,+}(\lambda_+)f(\lambda)d\lambda+\alpha_jd_f(\infty)\right)\\ 
&\qquad \times \prod_{j=1}^p\frac{G(1+\frac{\alpha_j}{2}+\beta_j)G(1+\frac{\alpha_j}{2}-\beta_j)}{G(1+\alpha_j)}(2\pi N\psi_V(t_j))^{\frac{\alpha_j^2}{4}-\beta_j^2}
\exp \Big( \frac{\alpha_{j}}{2}d_{\epsilon=0}(\infty) \Big)\\ &\qquad \times 
\prod_{l,j=1}^p e^{-\frac{\pi i \beta_j\alpha_l}{4}}\prod_{j<l}e^{\frac{\pi i}{2}(\alpha_l\beta_j-\alpha_j\beta_l)} 
 \\ &\qquad  \times
\left(\prod_{l\neq j}|t_j-t_l|^{\beta_j\beta_l-\alpha_j\alpha_l/4}\right)\left(\prod_{l,j=1}^p\Theta(t_{j,+},\infty)^{\frac{\beta_j\alpha_l}{2}}\left|\widetilde \Theta(t_{l,+},t_{j,+})\right|^{\beta_j\beta_l} \right) \\ &\qquad 
 \times (1+\mathcal O(N^{-1})+\mathcal O(\epsilon)),
\end{aligned}
\end{equation*}
as $N\to \infty$ for any fixed and sufficiently small $\epsilon>0$ (the implicit constants in the $\mathcal O(N^{-1})$ term depends on $\epsilon$ while the implicit constants in the $\mathcal O(\epsilon)$ term are independent of both $N$ and $\epsilon$). Since $\epsilon$ was arbitrary, Theorem \ref{th:FHasy} follows
from the following lemma.
\begin{lemma}
\begin{equation} \label{Finalobs}2d_f(\infty)+d_{\epsilon=0}(\infty)=-\int_J\frac{\widetilde{Q}(x)f(x)}{\mathcal R_+^{1/2}(x)}\frac{dx}{\pi i}-\frac{\mathcal A}{2}\mathcal C_{\mathcal S} +\sum_{j=1}^{p} \beta_{j} \bigg(\log \Theta(t_{j,+},\infty) - \frac{\pi i}{2}\bigg), \end{equation}
where $\widetilde Q$ is the unique monic polynomial of degree $k-1$ satisfying \eqref{intQhat} and $\mathcal C_{\mathcal S}$ was defined in \eqref{defCS}.
\end{lemma}

\begin{proof}

\noindent \underline{Step 1:} We first prove that
\begin{equation} \label{Step1goalQ}2d_f(\infty)+d_{\epsilon=0}(\infty)=-\int_J\frac{\widetilde{Q}(x)(2f(x)+\log \omega(x))}{\mathcal R_+^{1/2}(x)}\frac{dx}{2\pi i}. \end{equation}
Recall the definition of $d_f$ from \eqref{defdg} and note that
\begin{equation*}-\mathcal R^{1/2}(\lambda)w_\infty(\lambda)\end{equation*}
is a monic polynomial of degree $k-1$, which follows from the fact that  $w_\infty(\lambda)$ is analytic on $\mathbb C\setminus J$
 and from \eqref{polewlambda2}, \eqref{jumpswlambda1}. Secondly we have \eqref{zerointgap}, from which it follows that 
\begin{equation}
\int_{b_j}^{a_{j+1}}\frac{\mathcal R^{1/2}(\lambda)w_\infty(\lambda)}{\mathcal R^{1/2}(\lambda)}d\lambda=0
\end{equation}
for $j=1,\dots,k-1$.
 The fact that $\widetilde{Q}$ exists and is unique, is identical (up to replacing $\mathcal R(x)$ by $\frac{1}{\mathcal R(x)}$) to  the proof of Proposition \ref{pr:defo1} (a).

\noindent \underline{Step 2:} We prove that
\begin{equation}\label{Step2goalQ}\int_J\frac{\widetilde{Q}(x)\log |x-t|}{\mathcal R_+^{1/2}(x)}\frac{dx}{2\pi i}=\frac{1}{2} \mathcal C_S,
\end{equation}
for any $t\in J$. 

To prove \eqref{Step2goalQ}, first observe that
\begin{equation} \int_J\frac{\widetilde{Q}(x)\log |x-t|}{\mathcal R_+^{1/2}(x)}\frac{dx}{2\pi i}=\int_J\frac{\widetilde{Q}(x)\log (x-t)_+}{\mathcal R_+^{1/2}(x)}\frac{dx}{4\pi i}-\int_J\frac{\widetilde{Q}(x)\log (x-t)_-}{\mathcal R_-^{1/2}(x)}\frac{dx}{4\pi i},\end{equation}
where  $\arg (x-t) \in (-\pi,\pi)$ in the logarithm. By \eqref{intQhat}, it follows that
\begin{equation}\label{Step2int1} \int_J\frac{\widetilde{Q}(x)\log |x-t|}{\mathcal R_+^{1/2}(x)}\frac{dx}{2\pi i}=\int_{a_1}^{b_k}\frac{\widetilde{Q}(x)\log (x-t)_+}{\mathcal R_+^{1/2}(x)}\frac{dx}{4\pi i}-\int_{a_1}^{b_k}\frac{\widetilde{Q}(x)\log (x-t)_-}{\mathcal R_-^{1/2}(x)}\frac{dx}{4\pi i}.\end{equation}
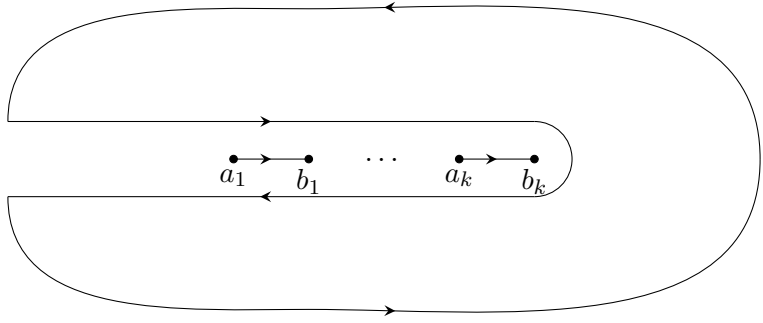
\begin{figure}
\begin{center}
\begin{tikzpicture}
\node [below] at (-1,0) {$a_1$};
\node [below] at (0,0) {$b_1$};
\node [below] at (2,0) {$a_k$};
\node [below] at (3,0) {$b_k$};
\node at (1,0) {$\dots$};

\draw [fill] (-1,0) circle [radius=0.05];
\draw [fill] (0,0) circle [radius=0.05];
\draw [fill] (2,0) circle [radius=0.05];
\draw [fill] (3,0) circle [radius=0.05];

\draw[decoration={markings, mark=at position 0.5 with {\arrow[thick]{<}}},
        postaction={decorate}]  (-4,0.5) to [out=90,in=180] (0,2) to [out=0,in=90] (6,0);

\draw[decoration={markings, mark=at position 0.5 with {\arrow[thick]{>}}},
        postaction={decorate}]  (-4,-0.5) to [out=-90,in=180] (0,-2) to [out=0,in=-90] (6,0);

\draw (3,0.5) to [out=0,in=90] (3.5,0) to [out=-90,in=0] (3,-0.5);

\draw[decoration={markings, mark=at position 0.5 with {\arrow[thick]{>}}},
        postaction={decorate}]  (-1,0)--(0,0);

\draw[decoration={markings, mark=at position 0.5 with {\arrow[thick]{>}}},
        postaction={decorate}]  (2,0)--(3,0);

\draw[decoration={markings, mark=at position 0.5 with {\arrow[thick]{>}}},
        postaction={decorate}]  (-4,0.5)--(3,0.5);

\draw[decoration={markings, mark=at position 0.5 with {\arrow[thick]{<}}},
        postaction={decorate}]  (-4,-0.5)--(3,-0.5);
\end{tikzpicture} 
\caption{The contour $\gamma_0$.}
\label{FigDeform}
\end{center}
\end{figure}
Let $C_R$ be a circle of radius $R$ oriented counterclockwise, beginning at argument $-\pi$ and ending at argument $\pi$. 
Now let $\gamma_0$ be the contour in Figure \ref{FigDeform}, namely the union of $C_R$ and  an indentation along $(-R,b_k]$ along which the contour follows first the $+$ side from $-R$ to $b_k$ and then the $-$ side from $b_k$ to $-R$.  
 Then by \eqref{Step2int1}
\begin{multline} \label{Step2Int2}0=\int_{\gamma_0} \frac{\widetilde Q(x) \log(x-t)}{\mathcal R^{1/2}(x)}\frac{dx}{4\pi i}=\int_J\frac{\widetilde{Q}(x)\log |x-t|}{\mathcal R_+^{1/2}(x)}\frac{dx}{2\pi i}\\+\int_{-R}^{a_1}\frac{\widetilde{Q}(x)(\log (x-t)_+-\log(x-t)_-)}{\mathcal R^{1/2}(x)}\frac{dx}{4\pi i}+\int_{C_R}\frac{\log(x-t)\widetilde Q(x)}{\mathcal R^{1/2}(x)}\frac{dx}{4\pi i}.
\end{multline}
 As $R\to \infty$,
\begin{equation} \int_{C_R}\frac{\log(x-t)\widetilde Q(x)}{\mathcal R^{1/2}(x)}\frac{dx}{4\pi i}= \frac{\log R}{2}+o(1), \end{equation}
and thus
\begin{equation} \int_{-R}^{a_1}\frac{\widetilde{Q}(x)(\log (x-t)_+-\log(x-t)_-)}{\mathcal R^{1/2}(x)}\frac{dx}{4\pi i}+\int_{C_R}\frac{\log(x-t)\widetilde Q(x)}{\mathcal R^{1/2}(x)}\frac{dx}{4\pi i}\to -\frac{1}{2}\mathcal C_{\mathcal S},
\end{equation}
as $R\to \infty$. Substituting this into \eqref{Step2Int2} in the limit $R\to \infty$, we obtain \eqref{Step2goalQ}.

\noindent \underline{Step 3:} Since $\widetilde{Q}(\lambda) = -\mathcal R^{1/2}(\lambda)w_\infty(\lambda)$, a direct primitive computation using \eqref{defwlambda} yields
\begin{equation}\label{Step3goalQ}\int_J\frac{\widetilde{Q}(x)\big( \mathbf{1}_{x<t} - \mathbf{1}_{x>t}\big)}{\mathcal R_+^{1/2}(x)}\frac{dx}{2\pi i} = -\frac{1}{\pi i}\bigg(\log \Theta(t_{+},\infty) - \frac{\pi i}{2}\bigg),
\end{equation}
for any $t\in J$.  Substituting \eqref{Step2goalQ} and \eqref{Step3goalQ} into \eqref{Step1goalQ}, and recalling that 
\begin{align*}
\log \omega(x)=\sum_{j=1}^p \bigg( \alpha_j \log |x-t_j| + \pi i \beta_{j}\begin{cases}
1, & \mbox{if } x<t_{j} \\
-1, & \mbox{if } x>t_{j} 
\end{cases} \bigg),
\end{align*}
we obtain the lemma.
\end{proof}

\section{Partition function asymptotics for the \texorpdfstring{$k$}{k}-cut Chebyshev potential}\label{sec:special}

Recall that Chebyshev polynomial of the first kind $T_0(x)=1$ and $T_k(x)=2^{k-1}x^k+\dots$ for $k=1,2,\dots$  is the unique polynomial of degree $k$ satisfying $T_k(\cos \theta)=\cos k\theta$, and that we denoted the Chebyshev potential by $V_0(x)=\frac{2\sigma}{k}T_k(x)^2$. To emphasize the $k$ dependency we will in this section denote $V_{0,k}=V_0$.

In this section we obtain asymptotics of $H_N\left(e^{-NV_0}\right)$ as $N\to \infty$, and our main goal is to prove \eqref{LimCheby}. We do this in three parts. First, in Section \ref{Secperiodic}, we obtain general and exact results for polynomials that are orthogonal with respect to a ``Chebyshev-type" weight. Secondly, in Section \ref{SecasymCheb}, we apply this general result to our specific situation by relying on Theorem \ref{th:smoothasy}. Finally, in Section \ref{Sec:limsigma}, we take the limit $\sigma \downarrow 1$.

\subsection{Results for Chebyshev-type orthogonal polynomials}\label{Secperiodic}
In this section, we consider a ``Chebyshev-type"  weight. The particular structure of these weights allow us, in Lemma \ref{le:op2gue} below, to obtain exact results.

 An analogue of Lemma \ref{le:op2gue}, valid for orthogonal polynomials on the unit circle, was proven by \cite{Baik}, see also \cite[eq (1.61)]{Bthesis}. Lemma \ref{le:op2gue} and the analogue in \cite{Baik,Bthesis} are related by the connection between  orthogonal polynomials on the unit circle and orthogonal polynomials on the interval $[-1,1]$ obtained in \cite{DIK}. Thus, although we prove Lemma \ref{le:op2gue} directly, one could alternatively obtain it by combining results from \cite{Baik,Bthesis} and \cite{DIK}.

Denote the Chebyshev polynomial of the second kind by $U_k(x)=2^kx^k+\dots$, which satisfies $U_k(\cos \theta)=\frac{\sin((k+1)\theta)}{\sin \theta}$, and recall that $T_k'(x)=kU_{k-1}(x)$.

\begin{lemma}\label{le:op2gue}
Let $w_1$ be a non-negative, even function on $[-1,1]$, and for $k\geq 1$ define 
\begin{equation}\nonumber
w_k(x)=|U_{k-1}(x)|w_1(T_k(x)).
\end{equation}
Let $P_j^{(k)}$ denote the monic orthogonal polynomials associated with the weight $w_k$ on $[-1,1]$,  
and let $\kappa_j^{(k)} > 0$ satisfy 
\begin{equation}\label{chebkappas}
\int_{-1}^1 P_j^{(k)}(x)^2 w_{k}(x) dx = \left(\kappa_j^{(k)}\right)^{-2}. \end{equation}
Denote the $N \times N$ Hankel determinant associated with the weight $w_k$ by $H_N(w_k)$, as defined in \eqref{eq:HNnu}.
 Then
\begin{itemize}
\item[(a)] The monic orthogonal polynomials satisfy
\begin{equation}
P_{nk}^{(k)}(x) =2^{-n(k-1)} P_n^{(1)}(T_k(x)), \qquad \mbox{for any } k \geq 1, \; n \geq 0, \label{orthok0} 
\end{equation}
and if $r=1,2,\dots,k-1$, then 
\begin{equation}\label{orthok}
P_{nk+r}^{(k)}(x) =2^{-n(k-1)-r+1}\frac{U_{r-1}(x) P_{n+1}^{(1)}(T_k(x))+\frac{ P_{n+1}^{(1)}(1)}{ P_n^{(1)}(1)}U_{k-r-1}(x) P_n^{(1)}(T_k(x))}{U_{k-1}(x)}. \end{equation}

	\item[(b)] For any $k \geq 1$ and $n \geq 0$, $\kappa_{nk}^{(k)}$ satisfies
	\begin{equation}\nonumber \left(\kappa_{nk}^{(k)}\right)^{-2}=\left(\kappa_n^{(1)}\right)^{-2} 2^{-2n(k-1)},
	\end{equation}
and if $r=1,2,\dots, k-1$, then
\begin{equation}\nonumber \left(\kappa_{nk+r}^{(k)}\right)^{-2}=\left(\kappa_n^{(1)}\right)^{-2} 2^{-2n(k-1) - 2r + 1} \frac{P_{n+1}^{(1)}(1) }{P_n^{(1)}(1) }.\end{equation}
\item[(c)] For any $k \geq 1$ and $n \geq 0$, the Hankel determinant satisfies
\begin{equation}\nonumber
	H_{nk}(w_{ k}) =\frac{H_n(w_1)^k P_n^{(1)}(1)^{k-1}}{2^{n(k-1)(nk-1)}},
\end{equation}
and if $r=1,2,\dots, k-1$, then
\begin{equation} \nonumber
H_{nk+r}(w_k)=\frac{H_n(w_1)^{k-r}H_{n+1}(w_1)^r P_n^{(1)}(1)^{k-r} P_{n+1}^{(1)}(1)^{r-1}}{2^{n(k-1)(nk-1)+2nr(k-1)+(r-1)^2}}.
\end{equation}
\end{itemize}
\end{lemma}
\begin{proof} Part (c) of the lemma follows in a straightforward manner by substituting part (b) of the lemma into the standard formula
\begin{equation}\label{repHankelwk}H_N(w_k)=\prod_{j=0}^{N-1}\left(\kappa_j^{(k)}\right)^{-2},\end{equation} 
and relying on the fact that $P_0^{(1)}(1)=1$. 

We structure the proof of parts (a) and (b) of the lemma as follows. First we prove that the functions \eqref{orthok0} and \eqref{orthok} are monic polynomials, secondly we prove that \eqref{orthok0} are the orthogonal polynomials of degree $nk$, thirdly we prove that \eqref{orthok} are the orthogonal polynomials of degree $nk+r$ for $r=1,\dots,k-1$, and finally we prove  part (b) of the lemma.

\underline{Step 1.} By the fact that the leading coefficient of $T_k$ and $U_{k-1}$ is $2^{k-1}$, it follows that \eqref{orthok0} is a monic polynomial, and that if \eqref{orthok} is a polynomial, then it is monic.

We prove that the functions defined in \eqref{orthok} indeed are polynomials. If $x_j=\cos( j\pi/k)$ is the $j$-th root of $U_{k-1}(x)$ for $j=1,\dots,k-1$, then 
\begin{equation*}
U_{k-r-1}(x_j)=\frac{\sin(k-r)j\pi/k}{\sin(j\pi/k)}=(-1)^{j+1}\frac{\sin \pi rj/k}{\sin \pi j/k}=(-1)^{j+1}U_{r-1}(x_j),
\end{equation*}
and furthermore $T_k(x_j)=(-1)^j$, and thus the numerator in the fraction of the right-hand side of \eqref{orthok} at $x=x_j$ is
\begin{equation}\label{zero1}
U_{r-1}(x_j)P_{n+1}^{(1)}((-1)^j)\left(1+(-1)^{j+1}\frac{P_{n+1}^{(1)}(1)P_n^{(1)}((-1)^j)}{P_n^{(1)}(1)P_{n+1}^{(1)}((-1)^j)}\right),
\end{equation}
which is zero, since the fact that  $w_1$ is even implies that  $P_{2n}^{(1)}$ is an even polynomial and $P_{2n+1}^{(1)}$ is an odd polynomial. Thus it follows that \eqref{orthok} is a polynomial.

\underline{Step 2.}
Now we verify that  \eqref{orthok0} satisfies the condition of orthogonality. We need to prove that if $p_j$ is a polynomial of degree $j$, then
\begin{equation}\label{int1instead}
\int_{-1}^1P_n^{(1)}(T_k(x))p_j(x)w_k(x)dx=0
\end{equation}
for any $j=0,\dots,nk-1$. 
Let $n'<n$ and $r'=0,1,\dots,k-1$, and consider  
\begin{equation} \label{int1}
\int_{-1}^1P_n^{(1)}(T_k(x))\left[P_{n'}^{(1)}(T_k(x))T_{r'}(x)\right]w_k(x)dx.
\end{equation}
Since $\{P_{n'}^{(1)}(T_k(x))T_{r'}(x)\}_{n<n', \, r'<k}$ forms a basis for the polynomials of degree less than $nk$, it is sufficient to show that
 \eqref{int1} is zero to prove  \eqref{int1instead}.
If $nk$ and $n'k+r'$ do not have the same parity, then the change of variables $x \to -x$, together with $T_{k}(-x)=(-1)^{k}T_{k}(-x)$ and $P_{n}^{(1)}(-x)=(-1)^{n}P_{n}^{(1)}(x)$ shows that the integral \eqref{int1} is zero. Let us now focus on the case where $nk$ and $n'k+r'$ have the same parity. Letting $x=\cos \phi$, \eqref{int1} becomes
\begin{equation*}
\int_0^\pi P_n^{(1)}(\cos k\phi)P_{n'}^{(1)}(\cos k\phi)\cos(r'\phi)|\sin k\phi| \, w_1(\cos k\phi)d\phi.
\end{equation*}
 Thus we obtain
 \begin{multline}\nonumber
 \sum_{s=0}^{k-1}\int_0^{\pi/k} P_n^{(1)}(\cos (k\theta+s\pi))P_{n'}^{(1)}(\cos (k\theta+s\pi))\\ \times \cos\left(r'\theta+\frac{r's\pi}{k}\right)\left|\sin \left(k\theta+s\pi\right)\right| w_1(\cos k\theta)d\theta, 
\end{multline}
which, by the fact that $P_{2n}^{(1)}$ is even and $P_{2n+1}^{(1)}$ is odd, yields
\begin{equation}\label{int2}
\sum_{s=0}^{k-1}(-1)^{s(n+n')}\int_0^{\pi/k} P_n^{(1)}(\cos k\theta)P_{n'}^{(1)}(\cos k\theta)\cos\left(r'\theta+\frac{r's\pi}{k}\right)\sin k \theta \,  w_1(\cos k\theta)d\theta.
\end{equation}
If $r'=0$,
\begin{multline}\nonumber
\int_0^{\pi/k} P_n^{(1)}(\cos k\theta)P_{n'}^{(1)}(\cos k\theta)\sin k \theta \, w_1(\cos k\theta)d\theta
\\=\frac{1}{k}\int_0^\pi P_n^{(1)}(\cos \alpha)P_{n'}^{(1)}(\cos \alpha)\sin \alpha \,w_1(\cos \alpha) d\alpha,
\end{multline}
which is zero by orthogonality of $P_n^{(1)}$ (assuming $n'<n$). If $r'\neq 0$, then
\begin{multline}\nonumber
\sum_{s=0}^{k-1}(-1)^{s(n+n')}\cos\left(r'\theta+\frac{r's\pi}{k}\right) \\
=\frac{1}{2}\sum_{s=0}^{k-1}(-1)^{s(n+n')}\left(e^{i\left(r'\theta+\frac{r's\pi}{k}\right)}+e^{-i\left(r'\theta+\frac{r's\pi}{k}\right)}\right).
\end{multline}
Since
\begin{equation} \label{sum2}
 \sum_{s=0}^{k-1}(-1)^{s(n+n')}e^{i\frac{r's\pi}{k}}=\frac{1-(-1)^{k(n+n')+r'}}{1-e^{\pi i(n+n'+r'/k)}},
\end{equation}
and we assumed that $kn$ and $kn'+r'$ had the same parity, meaning that the right-hand  side of \eqref{sum2} is zero, it follows that \eqref{int2} is zero, which implies that \eqref{int1} is zero, as desired.

\underline{Step 3.}
We now verify that \eqref{orthok} satisfies the condition of orthogonality. Consider 
\begin{multline} \label{int3}
\int_{-1}^1\frac{U_{r-1}(x)P_{n+1}^{(1)}(T_k(x))+\frac{P_{n+1}^{(1)}(1)}{P_n^{(1)}(1)}U_{k-r-1}(x)P_n^{(1)}(T_k(x))}{U_{k-1}(x)} \\
\times P_{n'}^{(1)}(T_k(x))T_{r'}(x)w_k(x)dx.
\end{multline}
We will prove that \eqref{int3} is zero for $n'k+r'<nk+r$ for $r=1,2,\dots,k-1$ and $r'=0,1,\dots,k-1$.
Again, the change of variable $x \mapsto -x$ shows immediately that the above integral is $0$, provided that $n'k+r'$ and $nk+r$ do not have the same parity. We now focus on the case where $n'k+r'$ has the same parity as $nk+r$. Letting $x=\cos \phi$, and summing over $\phi=\frac{\pi  s}{k}+\theta$ for $s=0,1,2,\dots,k-1$ (with $\theta\in(0,\pi/k)$), we obtain
\begin{multline*}
\sum_{s=0}^{k-1}\int_0^{\pi/k} \Bigg(\sin \left(r\theta+\frac{\pi rs}{k}\right)P_{n+1}^{(1)}(\cos(k\theta+\pi s)) \\
+\frac{P_{n+1}^{(1)}(1)}{P_n^{(1)}(1)} \sin\left((k-r)\theta+\frac{ \pi (k-r)s}{k}\right)P_n^{(1)}(\cos(k\theta+\pi s))\Bigg)\\ \times P_{n'}^{(1)}(\cos(k\theta+s\pi))\cos\left(r' \theta+\frac{\pi sr'}{k}\right)\frac{|\sin (k\theta+\pi s)|}{\sin(k\theta+\pi s)}w_1(\cos k\theta)d\theta.
\end{multline*}
Since $P_{2n}^{(1)}$ are even functions and $P_{2n+1}^{(1)}$ are odd functions, we obtain that \eqref{int3} is given by
\begin{multline}\label{int4}
\sum_{s=0}^{k-1}(-1)^{s(n+n')}\int_0^{\pi/k} \Bigg(\sin \left(r\theta+\frac{\pi rs}{k}\right)P_{n+1}^{(1)}(\cos k\theta)\\
-\frac{P_{n+1}^{(1)}(1)}{P_n^{(1)}(1)} \sin\left((r-k)\theta+\frac{ \pi rs}{k}\right) P_n^{(1)}(\cos k\theta)\Bigg)\\ \times P_{n'}^{(1)}(\cos k\theta)\cos\left(r' \theta+\frac{\pi sr'}{k}\right)w_1(\cos k\theta)d\theta.
\end{multline}

Write
\begin{multline}\label{contrib1}
\sin \left(r\theta+\frac{\pi rs}{k}\right)\cos\left(r' \theta+\frac{\pi sr'}{k}\right) \\
=\frac{1}{2}\Im\left( e^{i(r+r')(\theta+\pi  s/k)}+e^{i(r-r')(\theta+\pi s /k)}\right)
\end{multline}
and
\begin{multline}\label{contrib2}
\sin \left((r-k)\theta+\frac{\pi rs}{k}\right)\cos\left(r' \theta+\frac{\pi sr'}{k}\right) \\
=\frac{(-1)^s}{2} \Im \left( e^{i(r+r'-k)(\theta+\pi  s/k)}+e^{i(r-r'-k)(\theta+\pi s /k)}\right).
\end{multline}
Since we assume that $nk+r$ and $n'k+r'$ have the same parity, it follows that
\begin{equation}\label{2 sums}
\sum_{s=0}^{k-1}(-1)^{s(n+n')}e^{\frac{\pi i s}{k}(r\pm r')}=0,
\end{equation}
for $r\pm r' \neq 0\mod k$. 
Thus the first terms on the RHS of \eqref{contrib1} and \eqref{contrib2} do not make any contribution to \eqref{int4} for $r+r'\neq k$, and the same holds for the second terms but for $r-r'\neq 0$. In order to conclude that \eqref{int4} (and hence \eqref{int3}) is always zero, we inspect the remaining cases not covered by the geometric sum argument \eqref{2 sums}. 

Let us start with $\Im e^{i(r-r')(\theta+\pi s/k)}$ on the RHS of \eqref{contrib1}. We only need to check the case $r - r' = 0$, but $\Im e^{i(r-r')(\theta+\pi s/k)}$ is trivially zero. The same argument holds for the term $\Im e^{i(r+r'-k)(\theta+\pi s/k)}$ on the RHS of \eqref{contrib2} but with $r + r' = k$.

Now consider the first term on the RHS of \eqref{contrib1} with $r+r' = k$. The contribution to \eqref{int4} from $\Im e^{i(r+r')(\theta+\pi s/k)}$ involves integrals of the form
\begin{equation*}
\int_0^{\pi/k} \sin \left(r\theta+r'\theta \right)P_{n+1}^{(1)}(\cos k\theta)P_{n'}^{(1)}(\cos k\theta)w_1(\cos k\theta) d\theta,
\end{equation*}
which is zero by orthogonality of the polynomials $P_n^{(1)}$ since $n+1 > n'$. Similarly, when $r - r' = 0$ the contribution from $\Im e^{i(r-r')(\theta+\pi s /k)}$ on the RHS of \eqref{contrib2} leads to integrals of the form
\begin{equation}\label{int5}
\frac{1}{2}\sum_{s=0}^{k-1}(-1)^{s(n+n')} \int_0^{\pi/k} \frac{P_{n+1}^{(1)}(1)}{P_n^{(1)}(1)} \sin k\theta \,  P_n^{(1)}(\cos k\theta)P_{n'}^{(1)}(\cos k\theta)w_1(\cos k\theta) d\theta.
\end{equation}

\noindent which is equal to zero again thanks to orthogonality (since we assumed at the beginning that $nk+r > n'k + r'$, which implies $n > n'$ when $r=r'$).  Thus the polynomials \eqref{orthok} are orthogonal.

\underline{Step 4.} We prove part (b) of the lemma.

When $r=0$, it follows by the definition \eqref{orthok0} and by taking the successive changes of variables $x=\cos \theta$, $\theta = \frac{\pi s}{k}+\alpha$ ($s=0,\ldots,k-1$) and $y = \cos k\alpha$, that
\begin{equation*}
\int_{-1}^1P_{nk}^{(k)}(x)^2w_k(x)dx=2^{-2n(k-1)}\int_{-1}^1P_n^{(1)}(y)^2w_1(y)dy,
\end{equation*}
from which the first claim in part (b) follows by \eqref{chebkappas}.

For $r\neq0$, we observe that all the arguments leading up to \eqref{int5} hold also for $n=n'$ and $r=r'$, and thus \eqref{int3} is given by \eqref{int5} which is equal to
\begin{equation*}\frac{P_{n+1}^{(1)}(1)}{2P_n^{(1)}(1)}\int_{-1}^1P^{(1)}_n(x)^2w_1(x)dx.\end{equation*}
 It follows by the orthogonality of $P_{nk+r}^{(k)}$ and the fact that $T_k$ has leading coefficient $2^{k-1}$, that
\begin{equation}\nonumber
\int_{-1}^1 P_{nk+r}^{(k)}(x)^2w_k(x)dx=2^{-2n(k-1)-2r+1}\frac{P_{n+1}^{(1)}(1)}{P_n^{(1)}(1)}\int_{-1}^1P_n^{(1)}(x)^2w_1(x)dx,
\end{equation}
which proves the second claim in part (b) by \eqref{chebkappas}.
\end{proof}

\subsection{Asymptotics of  $H_N(e^{-NV_0})$} \label{SecasymCheb}
In this Section we obtain asymptotics of $H_N(e^{-NV_0})$ as $N\to \infty$ for fixed $\sigma>1$, where we recall that $V_0(x)=V_{0,k}(x)=\frac{2\sigma}{k}T_k(x)^2$. An analogue of this result for Toeplitz determinants can be found in \cite{Baik}, \cite{Marchal}, and \cite[eq (1.61)]{Bthesis}.

In Section \ref{sec:pfasy}, we proved that $V_0$ is $k$-cut regular for $\sigma>1$, and that the equilibrium measure associated to $V_0$ is supported on 
 the set $J_{0}=\bigcup_{j=1}^k [a_j(0),b_j(0)]$, where $a_j(0),b_j(0)\in(-1,1)$ are the (ordered) zeros of $\sigma T_k(x)^2-1$, and on $J_{0}$ it is given by
\begin{equation}\label{eqlibCheby}
d\mu_{V_0}(x)=\frac{2\sigma}{\pi ki}T_k'(x)\left(T_k(x)^2-1/\sigma\right)^{1/2}_+dx
\end{equation}
where $\left(T_k(x)^2-1/\sigma\right)^{1/2}$ is analytic on $\mathbb C\setminus J_{0}$ and  behaves like $ 2^{k-1}x^k$ as $x\to +\infty$.
 We will obtain asymptotics for the Hankel determinant $H_N\left(e^{-NV_0}\right)$ as $N\to \infty$, by expressing $H_N\left(e^{-NV_0}\right)$ in terms of the leading coefficients of the (rescaled) Hermite polynomials.

For $\sigma>1$,  let $1(x)$ be an even non-negative  H\"older continuous function on $\mathbb R$ satisfying the properties that $1(x)=0$ for all $x$ in an  open neighbourhood of $\mathbb R\setminus (-1,1)$, and $1(x)=1$ in an open neighbourhood of $\{x: \, x^2\leq 1/\sigma\}$. For definiteness, we define $1(x)$ explicitly as follows. Given $\sigma_1>1$ we define $1(x)$ for $\sigma>\sigma_1$ as follows:
\begin{equation} \label{def1b}
1(x)=\begin{cases} 1 & \textrm{for } x^2< 1/\sigma_1+(1-1/\sigma_1)/3\\
0& \textrm{for }x^2>1-(1-1/\sigma_1)/3,\\
1-\frac{x^2-1/\sigma_1-(1-1/\sigma_1)/3}{(1-1/\sigma_1)/3}& \textrm{for } \frac{1}{\sigma_1} + \frac{1-1/\sigma_1}{3} \leq x^2 \leq 1- \frac{1-1/\sigma_1}{3}. \end{cases}
\end{equation}
 Fix $k\geq 1$, and define 
\begin{equation}\label{w1wk}w_1(x)=1(x) \frac{e^{-2n\sigma x^2}}{\sqrt{1-x^2}}, \quad w_{k}(x) =|U_{k-1}(x)|w_1(T_k(x))=1(T_k(x))\frac{e^{-nkV_{0,k}(x)}}{\sqrt{1-x^2}}.\end{equation}
Then the conditions of Lemma \ref{le:op2gue} hold, and in particular we will utilize part (c) of the lemma, which we will combine with Theorem \ref{th:smoothasy}. 
  By Theorem \ref{th:smoothasy}, we have
\begin{equation} \label{Hnk1}\begin{aligned}
H_{nk}(w_k)&=H_{nk}\left(e^{-nkV_{0,k}}\right)(1+\mathcal O(n^{-1}))\\
& \times \frac{\theta(nk \Omega+\Upsilon|\tau)}{\theta(nk\Omega|\tau)}e^{nk\int_{J_{0,k}}f(x)d\mu_{V_{0,k}}(x)} e^{\mathcal Q_{J_{0,k}}(f)},
\\  \mathcal Q_{J_{0,k}}(f)&=  \frac{1}{4}\oint_{\Gamma} \oint_{\widetilde \Gamma}W(z,\lambda)f(z)f(\lambda) \frac{dz}{2\pi i}\frac{d\lambda}{2\pi i},
\end{aligned}\end{equation}
as $n\to \infty$ for fixed $\sigma>1$, 
where $f(x) = -\frac{1}{2}\log (1-x^2)$ for $J_{0,k}=\{x:T_{k}(x)^{2} \leq \frac{1}{\sigma}\}$, and where we place extra emphasis on the $k$-dependence by writing $J_0=J_{0,k}$ and $V_0=V_{0,k}$. We recall that $W$ was defined in \eqref{def:W} for $k=2,3,\dots$  and when $k=1$ we use the formula for $W$ in \eqref{wzonecut}. If $k=1$, the $\theta$ functions should be interpreted as equal to $1$. 
 Formula \eqref{Hnk1} gives asymptotics for $H_{nk}$, and we now give a similar formula for $H_{nk+r}$, with $r=1,\dots,k-1$. Define
\begin{equation*} \widetilde V_{0,k}(x)=\frac{2\widetilde \sigma}{k}T_k(x)^2, \qquad \widetilde \sigma=\left(1+\tfrac{r}{nk}\right)\sigma. \end{equation*}
Let  $\widetilde w_k$ be as $w_k$ defined in \eqref{w1wk}, but with $\sigma$ and $V_{0,k}$ replaced by $\widetilde \sigma $ and $\widetilde V_{0,k}$ for $k=2,3,\dots$. Observe that $(nk+r)V_{0,k}(x)=nk\widetilde V_{0,k}(x)$, so that
\begin{equation*}
 \widetilde w_k(x)=1(T_k(x))\frac{e^{-nk\widetilde V_{0,k}(x)}}{\sqrt{1-x^2}}=1(T_k(x))\frac{e^{-(nk+r) V_{0,k}(x)}}{\sqrt{1-x^2}}. 
\end{equation*}
 By Theorem \ref{th:smoothasy},
\begin{multline} \label{Hnk1b}
H_{nk+r}(\widetilde w_k)=H_{nk+r}\left(e^{-(nk+r)V_{0,k}}\right)\frac{\theta((nk+r) \Omega+\Upsilon|\tau)}{\theta((nk+r)\Omega|\tau)}\\ \times e^{(nk+r)\int_{J_{0,k}}f(x)d\mu_{V_{0,k}}(x)} e^{\mathcal Q_{J_{0,k}}(f)}(1+\mathcal O(n^{-1})),
\end{multline}
as $n\to \infty$.

By Lemma \ref{le:op2gue} (c) and \eqref{Hnk1},
\begin{multline}\label{Hnk1c}
H_{nk}\left(e^{-nkV_{0,k}}\right)=H_n\left(w_1\right)^k\frac{ P_n^{(1)}(1)^{k-1}}{2^{n(k-1)(nk-1)}}\frac{\theta(nk\Omega|\tau)}{\theta(nk \Omega+\Upsilon|\tau)}\\ \times e^{-nk\int_{J_{0,k}}f(x)d\mu_{V_{0,k}}(x)}  e^{-\mathcal Q_{J_{0,k}}(f)}(1+\mathcal O(1/n)),
\end{multline} 
as $n\to \infty$, where $P^{(1)}_n$ are the orthogonal polynomials associated with the weight $w_1$. Similarly, by Lemma \ref{le:op2gue} (c) and \eqref{Hnk1b},
\begin{multline}\label{Hnk1d}
H_{nk+r}\left(e^{-(nk+r)V_{0,k}}\right)=H_n\left(\widetilde w_{1,n}\right)^{k-r}H_{n+1}\left(\widehat w_{1,n+1}\right)^r \\
\frac{ \widetilde P_n^{(1)}(1)^{k-r}\widehat P_{n+1}^{(1)}(1)^{r-1}}{2^{n(k-1)(nk-1)+2nr(k-1)+(r-1)^2}} \frac{\theta((nk+r)\Omega|\tau)}{\theta((nk+r) \Omega+\Upsilon|\tau)} \\ 
\times  e^{-(nk+r)\int_{J_{0,k}}f(x)d\mu_{V_{0,k}}(x)}  e^{-\mathcal Q_{J_{0,k}}(f)}(1+\mathcal O(1/n)),
\end{multline} 
as $n\to \infty$,
where $\widetilde P^{(1)}_n$ and $\widehat P^{(1)}_{n+1}$ are the orthogonal polynomials associated with the weight $\widetilde w_{1,n}$ and $\widehat w_{1,n+1}$ respectively, given by:
\begin{equation*} \begin{aligned}\widetilde w_{1,n}(x)&=\frac{1(x)}{\sqrt{1-x^2}}e^{-2n\widetilde \sigma x^2}, &&\widetilde \sigma=\sigma(1+r/nk),\\
\widehat w_{1,n+1}(x)&=\frac{1(x)}{\sqrt{1-x^2}}e^{-2(n+1)\widehat \sigma x^2}, &&\widehat \sigma=\sigma\left(1+\frac{r/k-1}{n+1}\right).
\end{aligned} \end{equation*}

Observe that $\widetilde w_{1,n}(x)=\widehat w_{1,n+1}(x)$, however we give separate notations because it will be convenient when computing large $n$ asymptotics of $H_n (\widetilde w_{1,n})$ and  of $H_{n+1}(\widehat w_{1,n+1})$.

Observe also that $\widehat \sigma>1$ for $n$ sufficiently large.

Denote 
\begin{equation}\label{defDk} D_{k}(z) = \exp \Big( \frac{1}{2\pi i} \int_{J_{0,k}} f(\lambda)w_{z}(\lambda_{+})d\lambda \Big) \end{equation} and $w_{z}$ is defined in \eqref{defwlambda} for $k=2,3,\dots$ and \eqref{wzonecut} for $k=1$. This definition of $D_{k}$ coincides with the definition \eqref{eq:Ddef} of $D$ (see also \eqref{defdg}) after setting $d_{\epsilon}=0$ and $f(z)=-\frac{1}{2}\log(1-z^2)$.

\begin{lemma}\label{Unif} Let $w_1(x)$ be as defined in \eqref{w1wk} with $1(x)$ as in \eqref{def1b}. Given $1<\sigma_1<\sigma_2$,  the following statements hold.
\begin{itemize}
\item[(a)] As $n\to \infty$,
\begin{multline*}
\log \frac{H_n(w_1)}{H_n\left(e^{-2n\sigma x^2}\right)}=-\frac{\sigma n}{\pi} \int_{-1/\sqrt \sigma}^{1/\sqrt \sigma}\sqrt{1/\sigma-x^2}\log(1-x^2)dx+\mathcal Q_{J_{0,1}}(f)+\mathcal O(n^{-1}),
\end{multline*}
uniformly for $\sigma_1<\sigma<\sigma_2$.
\item[(b)] As $n\to \infty$,
\begin{multline*}P_n^{(1)}(1)=(1+\mathcal O(n^{-1}))\sqrt{\frac{2}{1+\sqrt{1-\sigma^{-1}}}}\left[ \gamma_1(1)+\gamma_1(1)^{-1}\right] \frac{1}{2}D_1(1)^{-1}\\ \exp \left(\frac{\sigma n}{\pi} \int_{-1/\sqrt \sigma}^{1/\sqrt \sigma}\sqrt{1/\sigma-x^2}\log(1-x^2)dx\right).
\end{multline*}
uniformly for $\sigma_1<\sigma<\sigma_2$, where we emphasize that $k=1$ by the notation $\gamma=\gamma_1$ with $\gamma$ defined in \eqref{eq:gamma}.
\end{itemize}
\end{lemma}
\begin{proof}
Observe that our Riemann-Hilbert analysis for the weight $w_1$ holds uniformly for $1<\sigma_1<\sigma<\sigma_2$ \footnote{This is because the singularities of $f(x)=-\frac{1}{2}\log (1-x^2)$, and also the singularities of $1(x)$, remain bounded away from the support of $V_{k,0}$. The technicalities are straightforward to verify: there are two local parametrices at $1/\sqrt \sigma$ and $-1/\sqrt \sigma$. Set the radius of each local parametrix to $\min \left \{ \frac{1}{3}( 1-\frac{1}{\sqrt{\sigma_1}}),\frac{2}{3\sqrt{\sigma_2}}
\right\}$. Then we leave it to the reader to verify that condition (c) in Lemma \ref{le:edgerhp} and Lemma \ref{expsmalljump} hold uniformly for $\sigma_1<\sigma<\sigma_2$, and thus \eqref{smallnorm} holds uniformly.}, and 
\begin{align*}
d\mu_{V_{0,1}}(x) = \frac{2\sigma}{\pi}\sqrt{\tfrac{1}{\sigma}-x^2}\,dx, \qquad x \in [-\tfrac{1}{\sqrt{\sigma}},\tfrac{1}{\sqrt{\sigma}}].
\end{align*}
Thus Theorem \ref{th:smoothasy} holds uniformly for $\sigma_1<\sigma<\sigma_2$, from which we obtain part (a) of the lemma.

We now prove part (b). By \eqref{eq:Ydef} we have $P_n^{(1)}(1)=Y_{11}(1)$ (where $Y(\cdot)=Y_{n}(\cdot;w_1)$), and by the definitions of $T$, $S$, and $R$ in \eqref{eq:Tdef}, \eqref{eq:Sdef}, and \eqref{eq:Rdef} respectively, we obtain
\begin{equation*}P_n^{(1)}(1)=(R(1)M(1))_{11}e^{ng(1)}.\end{equation*}
Note that  $d\mu_{V_{0,1}}(x)=d\mu_{V_{0,1}}(-x)$, and thus
\begin{equation*} g(1)=\int_{J_{0,1}} \log(1-x)d\mu_{V_{0,1}}(x)=-\int_{J_{0,1}} f(x)d\mu_{V_{0,1}}(x). 
\end{equation*}

Since $R$ is analytic in a neighbourhood of $1$ (because $w_1=0$ is analytic in a  neighbourhood of $1$), it follows that $R(1)=I+\mathcal O(n^{-1})$ as $n\to \infty$, uniformly for $\sigma_1<\sigma<\sigma_2$. Recalling the definition of $M$ in \eqref{eq:Pdef} (with $k=1$ so that the $\theta$ function is identically $1$), we obtain
\begin{equation}\label{Pn(1)}
P_n^{(1)}(1)=(1+\mathcal O(n^{-1}))D_1(\infty)\left[ \gamma_1(1)+\gamma_1(1)^{-1}\right] \frac{1}{2}D_1(1)^{-1}e^{-n\int_{J_{0,1}} f(x)d\mu_{V_{0,1}}(x)}.
\end{equation}

 Directly from the definition of $D_{1}$, we obtain
\begin{align}
D_1(\infty)&=
\exp \bigg( \lim_{z\to \infty} \frac{\mathcal R^{1/2}(z)}{8\pi i}\oint_\Gamma \frac{\log (1-\lambda^2)}{\mathcal R^{1/2}(\lambda)}\frac{d\lambda}{\lambda-z} \bigg) \nonumber \\
\label{D1inftyinter} & = \exp \bigg(\frac{-1}{8\pi i}\oint_\Gamma \frac{\log(1-\lambda^2) d\lambda}{\left(\lambda^2-1/\sigma\right)^{1/2}} \bigg), \end{align}
where $\Gamma$ is a counter-clockwise oriented curve surrounding $J_{0,1}$, and $\log(1-\lambda^2)$ is real on $(-1,1)$ and has a branch cut on $(-\infty,-1]\cup [1,+\infty)$.
For $\sigma>1$, the following identity holds:
\begin{equation}\label{D1infty}
D_1(\infty)=\sqrt{\frac{2}{1+\sqrt{1-\sigma^{-1}}}}.
\end{equation}

\begin{figure}[h!]

\begin{tikzpicture}
\draw (-4.5, 0) -- (4.5, 0);
\draw[fill=black] (-4,0) node[below left] {$-R$} circle (2pt)
				(4,0) node[below right] {$R$} circle (2pt)
				(-1,0) node[below right] {$-1$} circle (2pt)
				(1,0) node[below left] {$1$} circle (2pt);
\draw[thick, ->, dashed] (-2.5, 0.2) -- (-1, 0.2) (-4, 0.2) -- (-2.5, 0.2);
\draw[thick, ->, dashed] (-2.5, -0.2) -- (-4, -0.2) (-1, -0.2) -- (-2.5, -0.2);

\draw[thick, ->, dashed]  (4, 0.2) -- (2.5, 0.2) (1, 0.2) -- (2.5, 0.2);
\draw[thick, ->, dashed]  (1, -0.2) -- (2.5, -0.2) (4, -0.2) -- (2.5, -0.2);

\draw[->, dashed, thick] (4, 0.2)  arc [radius=4, start angle = 0, end angle = 90];
\draw[dashed, thick] (0, 4.2)  arc [radius=4, start angle = 90, end angle = 180];
\draw[->, dashed, thick] (-4, -0.2)  arc [radius=4, start angle = 180, end angle = 270];
\draw[dashed, thick] (0, -4.2)  arc [radius=4, start angle = 270, end angle = 360];

\draw[->, dashed, thick] (-1, 0.2)  arc [radius=0.2, start angle = 90, end angle = 0];
\draw[dashed, thick] (-1, -0.2)  arc [radius=0.2, start angle = 270, end angle = 360];

\draw[dashed, thick] (0.8, 0)  arc [radius=0.2, start angle = 180, end angle = 90];
\draw[->, dashed, thick] (1, -0.2)  arc [radius=0.2, start angle = 270, end angle = 180];

\end{tikzpicture}
\caption{Contour deformation for the evaluation of $D_1(\infty)$}
\label{fig:condef_D1}
\end{figure}
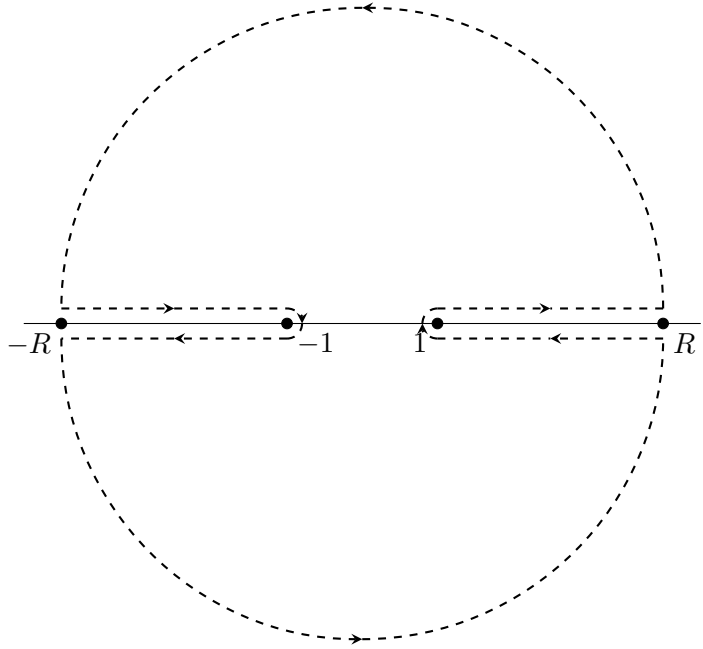

We obtain \eqref{D1infty} as follows. First, we deform the contour $\Gamma$ in \eqref{D1inftyinter} to the one in Figure \ref{fig:condef_D1}, which consists of one half-circle of radius $R > 0$ in each half plane, horizontal line contours right above/below the interval $(-R, -1]$ and $[1, R)$, and arbitrarily small arcs near $\pm 1$ (which are negligible). The contributions from the horizontal contours can be evaluated by
\begin{multline*}
-\frac{1}{8\pi i} \int_{-R}^{-1} \frac{\log(1-\lambda^2)_+-\log (1-\lambda^2)_-}{(\lambda^2 - 1/\sigma)^{1/2}}d\lambda\\
-\frac{1}{8\pi i} \int_{1}^R \frac{\log(1-\lambda^2)_+-\log (1-\lambda^2)_-}{(\lambda^2 - 1/\sigma)^{1/2}}d\lambda = \frac{1}{2} \int_{1}^R \frac{\d\lambda }{(\lambda^2 - 1/\sigma)^{1/2}}\\
= \frac{1}{2} \left[ \log\left(\frac{R + \sqrt{R^2 - 1/\sigma}}{2}\right) + \log \left(\frac{2}{1+\sqrt{1-1/\sigma}}\right)\right].
\end{multline*}

\noindent It is straightforward to verify that the first term in the last equality cancels with the contributions from the two half circles as $R \to \infty$, and the second term gives \eqref{D1infty} after exponentiation.

Substituting \eqref{D1infty} into \eqref{Pn(1)} we obtain part (b) of the lemma.
\end{proof}

By Lemma \ref{Unif} and \eqref{Hnk1c},
\begin{multline}\label{Hnk3}
\frac{H_{nk}\left(e^{-nkV_{0,k}}\right)}{H_n\left(e^{-2\sigma nx^2}\right)^k}=\frac{ 1}{2^{n(k-1)(nk-1)}}\frac{\theta(nk\Omega|\tau)}{\theta(nk \Omega+\Upsilon|\tau)}\left(\frac{2}{1+\sqrt{1-\sigma^{-1}}}\right)^{\frac{k-1}{2}}\\ 
\times \left[ \gamma_1(1)+\gamma_1(1)^{-1}\right]^{k-1} \frac{1}{2^{k-1}}D_1(1)^{-k+1}e^{-nk\int_{J_{0,k}}f(x)d\mu_{V_{0,k}}(x)}  e^{-\mathcal Q_{J_{0,k}}(f)} \\ 
\times \exp \left(-\frac{\sigma n}{\pi} \int_{-1/\sqrt \sigma}^{1/\sqrt \sigma}\sqrt{1/\sigma-x^2}\log(1-x^2)dx+k\mathcal Q_{J_{0,1}}(f)\right)(1+\mathcal O(1/n)),
\end{multline} 
as $n\to \infty$.

Since Lemma \ref{Unif} is uniform in $\sigma_1<\sigma< \sigma_2$, it may also be applied to evaluate the asymptotics of $H_n\left(\widetilde w_1\right)$ and $\widetilde P_n^{(1)}(1)$. In the next calculation, we emphasize the dependence of $\gamma_1(1)$, $D_1(1)$, and $\mathcal Q_{J_{0,1}}(f)$ in $\sigma$ by $\gamma_1(1;\sigma)$, $D_1(1;\sigma)$, and  $\mathcal Q_{J_{0,1}}(f;\sigma)$. As $n\to \infty$,
\begin{equation*} \begin{aligned}\gamma_1(1;\widetilde \sigma),\, \gamma_1(1;\widehat \sigma)&=\gamma_1(1;\sigma)(1+\mathcal O(1/n)),\\
D_1(1;\widetilde \sigma), \, D_1(1;\widehat \sigma)&=D_1(1;\sigma)(1+\mathcal O(1/n)),\\
\mathcal Q_{J_{0,1}}(f;\widetilde\sigma),\, \mathcal Q_{J_{0,1}}(f;\widehat\sigma)&=\mathcal Q_{J_{0,1}}(f;\sigma)(1+\mathcal O(1/n)).
\end{aligned} \end{equation*}
Thus, by Lemma \ref{Unif} and \eqref{Hnk1d},
\begin{multline}\label{Hnk3b}
\frac{H_{nk+r}\left(e^{-(nk+r)V_{0,k}}\right)}{H_n\left(e^{-2\widetilde\sigma nx^2}\right)^{k-r}H_{n+1}\left(e^{-2\widehat \sigma (n+1)x^2}\right)^r}=\frac{ 1}{2^{n(k-1)(nk-1)+2nr(k-1)+(r-1)^2}}
\\ \times  \frac{\theta((nk+r)\Omega|\tau)}{\theta((nk+r) \Omega+\Upsilon|\tau)}\left(\frac{2}{1+\sqrt{1-\sigma^{-1}}}\right)^{\frac{k-1}{2}}\\ 
\times \left[ \gamma_1(1)+\gamma_1(1)^{-1}\right]^{k-1} \frac{1}{2^{k-1}}D_1(1)^{-k+1}e^{-(nk+r)\int_{J_{0,k}}f(x)d\mu_{V_{0,k}}(x)}  e^{-\mathcal Q_{J_{0,k}}(f)} \\ 
\times \exp \left(-\frac{\widehat \sigma (n+1)}{\pi} \int_{- \widehat\sigma^{-1/2}}^{\widehat\sigma^{-1/2}}\sqrt{1/\widehat \sigma-x^2}\log(1-x^2)dx+k\mathcal Q_{J_{0,1}}(f)\right)(1+\mathcal O(1/n)),
\end{multline} 
where all functions are evaluated at $\sigma$ and not $\widetilde \sigma$ or $\widehat \sigma$ (e.g. $\gamma_1(1)=\gamma_1(1;\sigma)$), except where dependence on $\widetilde \sigma$ or $\widehat \sigma$ is explicitly indicated.

\begin{lemma}\label{LemmaellCheby} For $\sigma>1$,
\begin{multline}\nonumber
\int_{J_{0,k}}\log(1-x^2)d\mu_{V_{0,k}}(x)=\frac{2\sigma}{\pi k}\int_{J_{0,k}}\log(1-x^2)|T_k'(x)|\left|T_k(x)^2-1/\sigma\right|^{1/2}dx
\\=-\frac{2}{k}\left(-\sigma+\frac{1}{2}+\frac{1}{2}\log \sigma+k\log 2+\sqrt{\sigma(\sigma-1)}-\log\left(\sqrt{\sigma-1}+\sqrt{\sigma}\right)\right).
\end{multline}

\end{lemma}

\begin{proof}

Observe that if $\gamma_0$ is as in Figure \ref{FigDeform} and $C_R$ a circle of (large) radius $R$ centered at $0$ with counter-clockwise orientation, then
\begin{multline}\label{eq:zeroint}
0=\frac{2\sigma}{\pi k i}\oint_{\gamma_0} \log (z+1) T_k'(z)\left(T_k(z)^2-1/\sigma\right)^{1/2} dz\\
=2\int_{J_{0,k}} \log(x+1)d\mu_{V_{0,k}}(x)+\frac{4\sigma}{k}\int_{-R}^{-1} T_k'(x)\left(T_k(x)^2-1/\sigma \right)^{1/2} dx\\+\frac{2\sigma}{\pi k i}\oint_{C_R} \log (z+1) T_k'(z)\left(T_k(z)^2-1/\sigma\right)^{1/2} dz,
\end{multline}
with $\log (z+1)$ analytic on $\mathbb C\setminus (-\infty,-1]$ and real for $z>-1$, with $\left(T_k(z)^2-1/\sigma\right)^{1/2}$ analytic on $\mathbb C\setminus J_{0,k}$ and positive for $z>b_k$. We start by evaluating the integral $\int_{-R}^{-1}$ on the right-hand side of \eqref{eq:zeroint}.
Since $T_k(-1)=(-1)^{k}$, and $T_k'(x)\left(T_k(x)^2-1/\sigma\right)^{1/2}$ is negative on $[-R,-1)$, it follows by the change of variable $y=(-1)^k\sqrt{\sigma}T_k(x)$ and the relation
\begin{equation}\frac{1}{2}\left(y\sqrt{y^2-1}-\log\left(\sqrt{y^2-1}+y\right)\right)'=\sqrt{y^2-1},\nonumber 
\end{equation}
that
\begin{multline}
\label{intTkrootR}\int_{-R}^{-1} T_k'(x)\left(T_k(x)^2-1/\sigma \right)^{1/2} dx=-\frac{1}{\sigma }\int_{\sqrt{\sigma}}^{\sqrt{\sigma}|T_k(-R)|}\sqrt{y^2-1}dy\\
=\frac{1}{2\sigma}\Bigg(-\sigma T_k(-R)^2+\frac{1}{2}+k\log R+\frac{1}{2}\log \sigma+k\log 2\\ +\sqrt{\sigma(\sigma-1)} - \log\left(\sqrt{\sigma-1}+\sqrt{\sigma}\right) \Bigg)+o(1),
\end{multline}
as $R\to \infty$.
We now evaluate the integral $\oint_{C_R}$ in \eqref{eq:zeroint}.
As $z\to \infty$ we have $T_k'(z)\sqrt{T_k(z)^2-1/\sigma}=T_k'(z)T_k(z)-\frac{k}{2\sigma z}+\mathcal O\left(z^{-2}\right)$, and it follows that as $R\to \infty$,
\begin{multline}\label{logTkrootint}
\oint_{C_R} \log (z+1) T_k'(z)\sqrt{T_k(z)^2-1/\sigma} dz\\=\oint_{C_R} \log (z+1) T_k'(z)T_k(z)dz-\oint_{C_R}\frac{k \log (z+1) }{2\sigma z}dz+o(1).
\end{multline}
By integration by parts we have
\begin{multline*} 
2\pi i T_{k}(-R)^2=\oint_{C_R}\left(\log(z+1)T_k(z)^2\right)'dz\\ =\oint_{C_R}\frac{T_k(z)^2}{z+1}dz+2\oint_{C_R}\log(z+1)T_k(z)T_k'(z)dz, 
\end{multline*}
and it follows that
\begin{equation}\label{logTkTk'int}
\oint_{C_R}\log (z+1)T_k(z)T_k'(z)dz=\pi i T_k(-R)^2- \pi i T_k(-1)^2. 
\end{equation}
Substituting \eqref{logTkTk'int} into \eqref{logTkrootint} and evaluating the second integral on the right-hand side of \eqref{logTkrootint} in the limit $R\to \infty$ we obtain that 
\begin{multline}\label{logTkrootinfty}
\oint_{C_R} \log (z+1) T_k'(z)\sqrt{T_k(z)^2-1/\sigma} dz \\
=\pi i T_k(-R)^2-\pi i T_k(-1)^2-\frac{k\pi i \log R}{\sigma}+o(1),
\end{multline}
as $R\to \infty$. Substituting \eqref{intTkrootR} and \eqref{logTkrootinfty} into \eqref{eq:zeroint}, taking the limit $R\to \infty$ we have proven the lemma.
\end{proof}

Substituting  the identity in Lemma \ref{LemmaellCheby}  into \eqref{Hnk3} we obtain
\begin{multline}\label{Hnk4}
\frac{H_{nk}\left(e^{-nkV_{0,k}}\right)}{H_n\left(e^{-2n\sigma x^2}\right)^k}=\frac{1}{2^{n^2k(k-1)}2^{k-1}}\frac{\theta(nk\Omega|\tau)}{\theta(nk \Omega+\Upsilon|\tau)}D_1(1)^{-k+1}\\ \times \left(\frac{2}{1+\sqrt{1-\sigma^{-1}}}\right)^{(k-1)/2}\left(\gamma_1(1)+\gamma_1(1)^{-1}\right)^{k-1} e^{-\mathcal Q_{J_{0,k}}(f)} e^{k\mathcal Q_{J_{0,1}}(f)}(1+\mathcal O(1/n)),
\end{multline} 
as $n\to \infty$.

Substituting the identity in Lemma \ref{LemmaellCheby} into \eqref{Hnk3b} we obtain
\begin{multline}\label{Hnk4b}
\frac{H_{nk+r}\left(e^{-(nk+r)V_{0,k}}\right)}{H_n\left(e^{-2\widetilde\sigma nx^2}\right)^{k-r}H_{n+1}\left(e^{-2\widehat \sigma (n+1)x^2}\right)^r}=\frac{ D_1(1)^{-k+1}}{2^{n^2k(k-1)+2nr(k-1)+(r-1)^2}2^{k-1}}
\\ \times  \exp \Bigg[\frac{1}{2}\left(1-\frac{r}{k} \right)\log \sigma +(1-r)\log 2
+\left(\frac{r}{k}-1\right)\log \left(\sqrt{\sigma}+\sqrt{\sigma-1}\right)\Bigg]  \\ 
\times \frac{\theta((nk+r)\Omega|\tau)}{\theta((nk+r) \Omega+\Upsilon|\tau)}\left(\frac{2}{1+\sqrt{1-\sigma^{-1}}}\right)^{\frac{k-1}{2}}
 \left[ \gamma_1(1)+\gamma_1(1)^{-1}\right]^{k-1}  \\
 \times e^{-\mathcal Q_{J_{0,k}}(f)+k\mathcal Q_{J_{0,1}}(f)}
(1+\mathcal O(1/n)),
\end{multline} 
as $n\to \infty$.

It is not a straightforward matter to further simplify formulas \eqref{Hnk4}-\eqref{Hnk4b} for fixed $\sigma>1$, however in the limit $\sigma \to 1$ it is rather simple, so we now proceed with the limit $\sigma \to 1$.

\subsection{Limit $\sigma \to 1$}\label{Sec:limsigma}
We now consider the limit $\sigma\to 1$. Let $\xi_1,\dots, \xi_{k-1}$ be the ordered zeros of $U_{k-1}$ so that $\xi_j=-\cos\frac{\pi j}{k}$.
 We observe that the parameters $\{a_j,b_j\}_{j=1}^k$  depend on $\sigma$. Recalling that these are the zeros of $T_k^2(x)-\sigma^{-1}$, it is easily verified by writing $T_k(\cos \theta)=\cos k\theta$ that
\begin{equation}\label{bklim}
b_k=-a_1=1-\frac{\sigma-1}{2k^2}+\mathcal O\left((\sigma-1)^2\right), 
\end{equation}
as $\sigma\to 1$, and for $j=1,\dots,k-1$,
\begin{equation}\label{bjlim} b_j=\xi_j-\frac{\sqrt{\sigma-1}}{k}\sin \frac{\pi j}{k}+\mathcal O\left(\sigma-1\right), \qquad a_{j+1}=\xi_j+\frac{\sqrt{\sigma-1}}{k}\sin \frac{\pi j}{k}+\mathcal O\left(\sigma-1\right).
\end{equation}
Thus we have
\begin{equation}\label{Rlim}
\mathcal R^{1/2}(z)\to 2^{-k+1}\left(z^2-1\right)^{1/2}U_{k-1}(z), \end{equation}
for any fixed $z\in \mathbb C\setminus\{\xi_1,\dots,\xi_{k-1}\}$, while on $ (b_j,a_{j+1})$, $j=1,\ldots,k-1$,  we have
\begin{equation}\label{Rxij}
\begin{aligned}\mathcal R^{1/2}(z)&=\sqrt{(a_{j+1}-z)(z-b_j)}\widetilde{\mathcal R}_j\left(1+o(1)\right),\\
\widetilde{\mathcal R}_j&=(-1)^{k-j}\bigg(\prod_{i\neq j}|\xi_j-\xi_i|\bigg)\sqrt{1-\xi_j^2}
\end{aligned}
\end{equation}
as $\sigma\to 1$, uniformly on $(b_j,a_{j+1})$, and \eqref{Rxij} is also valid on any shrinking neighbourhood of $\xi_j$ with branch cuts for $z<b_j$ and $z>a_{j+1}$.

In order to evaluate the $\theta$-function in \eqref{Hnk4} as $\sigma\to 1$, we need to evaluate $\tau$ as $\sigma\to 1$. Recall from Remark \ref{explicittau} that $\tau=-\boldsymbol B(\boldsymbol A^{-1})^{T}$.
By the definition of $\boldsymbol A$ in Remark \ref{explicittau}, the formula for the determinant of $\boldsymbol A$ in \eqref{eq:detA}, and \eqref{Rxij}, we obtain
\begin{equation} \label{limA}
\begin{aligned}\boldsymbol A_{rj}&\to \frac{ \pi \xi_j^{r-1}}{ \widetilde{\mathcal R}_j}, \\ 
\det(\boldsymbol A(\sigma)) & \to \frac{(-1)^{\frac{k(k-1)}{2}}\pi^{k-1}}{\prod_{i<j}(\xi_j-\xi_i)\sqrt{\prod_{j=1}^{k-1}(1-\xi_j^2)}},
\end{aligned}
\end{equation}
as $\sigma\to 1$. Above, $\xi_j^{r-1}$ should be interpreted as equal to $1$ if $r=1$ and $\xi_{j}=0$ (this happens for even values of $k$). In particular $\boldsymbol A$ is continuous as $\sigma \to 1$, and since $\boldsymbol A$ is bounded and $\det \boldsymbol A$ remains bounded away from $0$, it follows that $\boldsymbol A^{-1}$ is continuous as $\sigma\to 1$ as well.

Now consider the matrix $\boldsymbol B$ with $\boldsymbol B_{jr}=\sum_{i=1}^j\int_{a_i}^{b_i} \frac{x^{r-1}dx}{\mathcal R_+^{1/2}(x)}$. 
Uniformly for $x$ in a neighbourhood of $\xi_j$, for $j=1,\dots, k-1$, we have
\begin{equation} \label{Rlim2} \frac{x^{r-1}}{\mathcal R^{1/2}(x)}=\frac{\xi_j^{r-1}}{\widetilde {\mathcal R}_j \sqrt{(a_{j+1}-x)(x-b_j)}} +\mathcal O(1), \end{equation}
as $\sigma \to 1$, where $\sqrt{(a_{j+1}-x)(x-b_j)}$ is analytic on $\mathbb C\setminus ((-\infty,b_j]\cup [a_{j+1},\infty))$ and is positive on $(b_j,a_{j+1})$. 
By relying on \eqref{Rlim} and \eqref{Rlim2},  we obtain
\begin{equation*}\boldsymbol B_{jr}=-\frac{i\xi_j^{r-1}}{2\widetilde{\mathcal R}_j}\log\frac{1}{\sigma-1}+\mathcal O(1), \end{equation*}
as $\sigma \to 1$.
Thus, since $\tau=-\boldsymbol B(\boldsymbol A^{-1})^{T}$, it follows by \eqref{limA} that 
\begin{equation}\tau=I\, \frac{1}{2\pi i} \log (\sigma-1)+o\left(\log  \frac{1}{\sigma-1}\right), \label{limtau} \end{equation}
 as $\sigma \to 1$. It follows that $\theta(x|\tau)\to 1$ as $\sigma \to 1$ uniformly for $x\in \mathbb R^{k-1}$.

Now recall  the formula for $\mathcal Q_{J_{0,k}}$ from \eqref{Hnk1}. Since $f(z)=-\frac{1}{2}\log(1-z^2)$ and $W(z,\lambda)=\frac{d}{dz}w_z(\lambda)$, we find by integration by parts (recalling the definition of $D_k$ from \eqref{defDk})
\begin{equation}\label{formulaQJ0} \mathcal Q_{J_{0,k}}(f)=\frac{1}{4}\left(\log D_k(1)+\log D_k(-1)-2\log D_k(\infty)\right).\end{equation}

All that remains is to evaluate $D_k(1)$, $D_k(-1)$, and $D_k(\infty)$ as $\sigma \to 1$. We start with $D_k(\infty)$.
\begin{lemma}\label{LemDlim2}
As $\sigma\to 1$, 
\begin{equation*} \log D_k(\infty)\to \frac{1}{2}\log 2. \end{equation*}
\end{lemma}

\begin{proof}
By \eqref{Finalobs} with $f(x)=-\frac{1}{2}\log (1-x^2)$,
\begin{equation*} \log D_k(\infty)=\frac{1}{4\pi i} \int_{J_{0,k}}\frac{\widetilde{Q}(x) \log (1-x^2)}{\mathcal R^{1/2}_+(x)} dx, \end{equation*}
where $\widetilde{Q}(x)$ is the unique monic polynomial of degree $k-1$ satisfying $$\int_{b_k}^{a_{k+1}}\frac{\widetilde{Q}(x)}{\mathcal R^{1/2}(x)}dx=0$$ for $k=1,\dots,k-1$. Since $\mathcal R^{1/2}(x)=\frac{1}{2^{k-1}}\left(T_k(x)^2-1/\sigma\right)^{1/2}$, it is easily verified that $\widetilde{Q}(x)=\frac{1}{2^{k-1}}U_{k-1}(x)=\frac{1}{k2^{k-1}}T_k'(x)$. Since 
\begin{equation*}\frac{U_{k-1}(x)}{(T_k(x)^2-1/\sigma)_+^{1/2}}\to \frac{1}{\sqrt{x^2-1}_+}\end{equation*} as $\sigma\to 1$ pointwise on $(-1,1)\setminus\{\xi_j\}_{j=1}^{k-1}$, where $\xi_1,\dots,\xi_{k-1}$ are the ordered zeros of $U_{k-1}$, it follows that (we leave the details to the reader)
\begin{equation*} \log D_k(\infty)\to - \frac{1}{4\pi }\int_{-1}^1\frac{\log (1-x^2)dx}{\sqrt{1-x^2}}, \end{equation*}
as $\sigma \to 1$, and the right hand side is equal to $\frac{1}{2}\log 2$  (see e.g. \cite[formula (52)]{Krasovsky}).
\end{proof}
Recall that $\boldsymbol{\omega_{1}},\ldots,\boldsymbol{\omega_{k}}$, defined in \eqref{eq:Aint}, is the unique basis of holomorphic one-form satisfying $\oint_{A_{l}}\boldsymbol {\omega_{j}} = \delta_{l,j}$, and recall the form of $\omega_j(x)=\frac{\mathsf Q_j(x)}{\mathcal R^{1/2}(x)}dx$ from \eqref{formomega}. We also recall from below \eqref{eq:Aint2} that the coefficients of $\mathsf Q_j$ are real, and  since
\begin{equation*} \int_{b_i}^{a_{i+1}}\frac{\mathsf Q_j(x)}{\mathcal R_+^{1/2}(x)}dx=0,\end{equation*}
for $i\neq j$,
it follows that $\mathsf Q_j$ has a zero in each interval $(b_i,a_{i+1})$, again for $i\neq j$. Denote the zero of $\mathsf Q_j$  in the interval $(b_i,a_{i+1})$ by $y_{j,i}$. Since
\begin{equation*} \frac{(x-y_{j,i})}{\sqrt{(x-b_i)(x-a_{i+1})}}\to 1, \qquad i\neq j,\end{equation*} 
as $\sigma \to 1$, uniformly for $x$ bounded away from $\xi_i$, it follows that
\begin{equation*} u_j'(x)=\frac{\widehat C_j}{\sqrt{(x-b_k)(x-a_1)}(x-\xi_j)}(1+o(1)) \end{equation*}
as $\sigma \to 1$, uniformly for $x$ bounded away from $\xi_1,\dots,\xi_{k-1}$, for some constant $\widehat C_j$.
 Note that 
\begin{equation*}
\oint_{A_j} \frac{(\xi_l^2-1)_+^{1/2}}{2\pi i ((x-b_k)(x-a_1))^{1/2}(x-\xi_l)}dx\to \delta_{l,j},
\end{equation*}
as $\sigma\to 1$. Thus,
\begin{equation}\label{usigma1}
u_l'(x) = \frac{\boldsymbol{\omega_{l}}(x)}{dx} = \frac{(\xi_l^2-1)_+^{1/2}}{2\pi i ((x-b_k)(x-a_1))^{1/2}(x-\xi_l)}(1+o(1)),
\end{equation}
where the convergence is uniform for $x$  in compact subsets of $\mathbb{C}\setminus \{\xi_{1},\ldots,\xi_{k-1}\}$. 

\begin{lemma}\label{LemDlim}
As $\sigma\to 1$,
\begin{equation*} 
\log D_k(1),\, \log D_k(-1) = -\frac{1}{2} \log 2-\frac{1}{4}\log(\sigma-1)+\frac{1}{2}\log k+o(1).
\end{equation*}
\end{lemma}
\begin{proof}
By \eqref{2ndrepwlambda} and \eqref{defDk}, we have
\begin{multline*}
\log D_k(\pm 1)= \frac{\mathcal R^{1/2}(\pm 1)}{8\pi i} \\ \times \oint_\Gamma \log (1-\lambda^2)\left(\frac{1}{\mathcal R^{1/2}(\lambda)(\lambda \mp 1)}+ 2\sum_{j=1}^{k-1}u_j'(\lambda)\int_{b_j}^{a_{j+1}} \frac{dx}{\mathcal R^{1/2}(x)(x \mp 1)}\right) d\lambda,
\end{multline*}
where $\Gamma$ is a closed, counterclockwise oriented contour enclosing $J_{0,k}$, but not intersecting $(-\infty,-1]\cup[1,\infty)$. We can assume that $\Gamma $ is bounded away from $\xi_1,\dots, \xi_{k-1}$, and thus
$\oint_\Gamma \log (1-\lambda^2)u_j'(\lambda) d\lambda$ remains bounded as $\sigma \to 1$ by \eqref{usigma1}. By \eqref{Rxij},  $\int_{b_j}^{a_{j+1}}\frac{dx}{\mathcal R^{1/2}(x)(x-1)}$ remains bounded as $\sigma \to 1$. Since $\mathcal R^{1/2}(\pm 1)\to 0$ as $\sigma \to 1$, it follows that
\begin{equation}\label{Dintem1}\begin{aligned}
\log D_k(1)&=\frac{\mathcal R^{1/2}(1)}{8\pi i}\oint_\Gamma \frac{ \log(1-x^2) dx}{\mathcal R^{1/2}(x)(x-1)}+o(1),\\
\log D_k(-1)&=\frac{\mathcal R^{1/2}(-1)}{8\pi i}\oint_\Gamma \frac{ \log(1-x^2) dx}{\mathcal R^{1/2}(x)(x+1)}+o(1),
\end{aligned}
\end{equation}
where the branches of the logarithms are principal. Since $\frac{\mathcal R^{1/2}(1)}{\mathcal R^{1/2}_+(x)}=-\frac{\mathcal R^{1/2}(-1)}{\mathcal R^{1/2}_+(-x)}$ for all $x\in J$, it follows that
\begin{equation*} \frac{\mathcal R^{1/2}(1)}{8\pi i}\oint_\Gamma \frac{ \log(1-x^2) dx}{\mathcal R^{1/2}(x)(x-1)}=
\frac{\mathcal R^{1/2}(-1)}{8\pi i}\oint_\Gamma \frac{ \log(1-x^2) dx}{\mathcal R^{1/2}(x)(x+1)}. 
\end{equation*}
We evaluate the left hand side.
Deforming $\Gamma$ to the branch-cuts of the logarithm, we obtain
\begin{equation}\label{Dintem2}
 \begin{aligned}
\frac{1}{2\pi i}\oint_\Gamma \frac{ \log(1+x) dx}{\mathcal R^{1/2}(x)(x-1)}&=\int_{-\infty}^{-1}\frac{dx}{\mathcal R^{1/2}(x)(x-1)}
-\frac{\log 2}{\mathcal R^{1/2}(1)},\\
\frac{1}{2\pi i}\oint_\Gamma \frac{ \log(1-x) dx}{\mathcal R^{1/2}(x)(x-1)}&=\lim_{\epsilon\to 0}\left(-\frac{\log \epsilon}{\mathcal R^{1/2}(1)}-\int_{1+\epsilon}^\infty \frac{dx}{\mathcal R^{1/2}(x)(x-1)}\right).
\end{aligned} \end{equation}
Recall that $b_k$ is the largest zero of $\mathcal R(z)$,  and denote
\begin{equation}\nonumber \widehat {\mathcal R}^{1/2}(z)=\frac{\mathcal R^{1/2}(z)}{(z-b_k)^{1/2}}. \end{equation}
As $\epsilon \to 0$,
\begin{equation*}\int_{1+\epsilon}^\infty \frac{dx}{(x-1)(x-b_k)^{1/2}}\left(\frac{1}{\widehat{ \mathcal R}^{1/2}(x)}-\frac{1}{\widehat{ \mathcal R}^{1/2}(1)}\right) \end{equation*}
remains uniformly bounded for $\sigma>1$, and thus
\begin{equation*}
\int_{1+\epsilon}^\infty \frac{dx}{\mathcal R^{1/2}(x)(x-1)}=\frac{1}{\widehat{ \mathcal R}^{1/2}(1)}\int_{1+\epsilon}^\infty \frac{dx}{(x-1)(x-b_k)^{1/2}}+\mathcal O(1), \qquad \mbox{as } \epsilon \to 0,
\end{equation*}
uniformly for $\sigma>1$. We have
\begin{equation*}
\int_{1+\epsilon}^\infty \frac{dx}{(x-1)(x-b_k)^{1/2}}=\frac{1}{\sqrt{1-b_k}}\left(\log \big(4(1-b_k)\big)-\log \epsilon\right)+o(1),
\end{equation*} as $\epsilon \to 0$. Substituting this into \eqref{Dintem2}, we obtain
\begin{align*}
\frac{1}{2\pi i}\oint_\Gamma \frac{ \log(1-x) dx}{\mathcal R^{1/2}(x)(x-1)} = - \frac{\log \big( 4(1-b_{k}) \big)}{\mathcal{R}^{1/2}(1)} + \bigO(1), \qquad \mbox{as } \sigma \to 1,
\end{align*}
and recalling that $\mathcal R(1) \to 0$ as $\sigma \to 1$, from \eqref{Dintem1} we get
\begin{equation*}
\log D_k(1),\, \log D_k(-1) = -\frac{3}{4} \log 2-\frac{1}{4}\log(1-b_k)+o(1),\end{equation*}
as $\sigma \to 1$, and the lemma follows by \eqref{bklim}.
\end{proof}
By \eqref{eq:gamma} with $k=1$ and  \eqref{bklim}, we have
\begin{equation*} 
\gamma_1(1) = \bigg(\frac{\sigma-1}{4}\bigg)^{\frac{1}{4}}(1+o(1)), 
\end{equation*}
as $\sigma \to 1$. Thus, since $\theta(x|\tau)\to 1$ as $\sigma\to 1$ uniformly for $x\in \mathbb R^{k-1}$, we obtain by \eqref{formulaQJ0} and Lemmas \ref{LemDlim2} and \ref{LemDlim} that
\begin{multline*}
\frac{\theta(nk\Omega|\tau)}{\theta(nk \Omega+\Upsilon|\tau)} \left(\frac{2}{1+\sqrt{1-\sigma^{-1}}}\right)^{\frac{k-1}{2}}D_1(1)^{-k+1}\left(\gamma_1(1)+\gamma_1(1)^{-1}\right)^{k-1}\\ \times \frac{1}{2^{k-1}} e^{-\mathcal Q_{J_{0,k}}(f)} e^{k\mathcal Q_{J_{0,1}}(f)} 
 = k^{-1/4}(\sigma-1)^{-\frac{k-1}{8}}(1+o(1))
\end{multline*}
as $\sigma \to 1$.

Thus, from \eqref{Hnk4}, we obtain
\begin{multline}\label{Hnk5}\lim_{\sigma \downarrow 1 }\limsup_{n\to \infty}
\Bigg|\log H_{nk}\left(e^{-nkV_{0,k}}\right)-k\log H_n\left(e^{-2n\sigma x^2}\right) \\+\left[n^{2}k(k-1)\right]\log 2+\frac{1}{4}\log k
+\frac{k-1}{8}\log (\sigma-1)\Bigg|=0. 
\end{multline}

Similarly, from \eqref{Hnk4b}, 
\begin{multline}\label{Hnk5b}\lim_{\sigma \downarrow 1} \limsup_{n\to \infty}
\Bigg|\log H_{nk+r}\left(e^{-(nk+r)V_{0,k}}\right)-(k-r)\log H_n\left(e^{-2n\widetilde\sigma x^2}\right)\\-r\log H_{n+1}\left(e^{-2n\widetilde\sigma x^2}\right) +\left[n^{2}k(k-1)+2nr(k-1)+r(r-1)\right]\log 2 \\+\frac{1}{4}\log k
+\frac{k-1}{8}\log (\sigma-1)\Bigg|=0. 
\end{multline}
Substituting the asymptotics of \eqref{asymGUE} into \eqref{Hnk5} and \eqref{Hnk5b}, we obtain \eqref{LimCheby}.

\section{Partition function asymptotics}\label{sec:pasy}
In this section we compute the asymptotics of $\frac{d}{ds}\log H_N\left(e^{-NV_s}\right)$ where $V_s$ is the deformation introduced in Section \ref{sec:pfasy}. Our main goal is to prove \eqref{MainResSec11}.  Throughout the section, we are in the setting where $D(z)\equiv 1$ and $\Upsilon=0$, so that the main parametrix $M$ defined in \eqref{eq:Pdef} is given by $M(z)=N_\infty(z)=N_\infty(z;N\Omega)$.

Recall the notation $r(\nu_s)$ from Section \ref{sec:DI}.
 We start by proving that $r(e^{-NV_s})\to 0$ as $N\to \infty$. For $z$ outside the lenses,  we have by  \eqref{eq:Tdef}, \eqref{eq:Sdef}, and \eqref{eq:Rdef}
\[
Y(z)=e^{-N\frac{\ell}{2}\sigma_3} R(z)N_\infty(z)e^{N(g(z)+\frac{\ell}{2})\sigma_3},
\]
where we take $V=V_s$ in the definition of $N_\infty$, $R$, and $g$. Let $I$ be a single interval such that $[a_1,b_k]\subset I$.
By the uniform boundedness of  $R$, $R'$, $N_\infty$, and $N_\infty'$ (see \eqref{eq:expR} and \eqref{eq:Mlamdef} with $v=N\Omega$), it follows that 
 \begin{equation}\label{expsmall}\begin{aligned}
r(e^{-NV_s})=&-N\int_{\R\setminus I}[Y(x)^{-1} Y'(x)]_{21}e^{-NV_s(x)}\frac{\partial}{\partial s} V_s(x)\frac{dx}{2\pi i}\\
&=-N\int_{\R\setminus I}\mathcal O(1)\times \frac{\partial}{\partial s} V_s(x)e^{N(g_+(x)+g_-(x)-V_s(x) +\ell_s)}dx,
\end{aligned}
\end{equation}
where the implied constant is uniform in $s$ and $x$ as $N\to \infty$. We find from the $k$-cut regularity of $V_s(x)$, and in particular the assumption that \eqref{eq:EL2} is strict, that $g_+(x)+g_-(x)-V_s(x)+\ell_s<-\delta$ for some fixed $\delta>0$ and that $\frac{V_s(x)}{\log |x| }\to +\infty$ as $x\to \pm \infty$, so  it is straightforward to verify that \eqref{expsmall} is exponentially small as $N\to\infty$. Thus
\begin{multline}\label{eq:orders}
 \frac{d}{ds}\log D_N(H_N(e^{-NV_s}))= -N\oint_\Gamma [Y(z)^{-1} Y'(z)]_{11}\frac{\partial}{\partial s} V_s(z)\frac{dz}{2\pi i}+\mathcal O(N^{-1})\\
=-N^2\oint_\Gamma g'(z)\frac{\partial}{\partial s}V_s(z)\frac{dz}{2\pi i}-N\oint_\Gamma [N_\infty(z)^{-1}N_\infty'(z)]_{11}\frac{\partial}{\partial s}V_s(z)\frac{dz}{2\pi i}\\
 -\oint_\Gamma\left[N_\infty(z)^{-1}\left(\frac{d}{d z}R^{(1)}(z)\right) N_\infty(z)\right]_{11}
\frac{\partial}{\partial s}V_s(z)\frac{dz}{2\pi i}+\mathcal O(N^{-1})
\end{multline}
as $N\to \infty$,
where $R^{(1)}$, $N_\infty$, and $g$ were given in \eqref{formulaR1}, \eqref{eq:Mlamdef}, and \eqref{eq:g_fn} respectively.

\subsection{The order \texorpdfstring{$N^2$}{N squared} term}\label{sec:orderN2}
We prove that the leading order term has the following form.

\begin{lemma}\label{le:orderN2}
Let $V_s(x)$ denote either one of the family of potentials described in Section \ref{sec:pfasy}. Then
\begin{multline*}
\oint_\Gamma  g'(z)\frac{\partial}{\partial s} V_s(z)\frac{dz}{2\pi i}  \\
=\frac{\partial}{\partial s}\left(\iint \log|x-y|^{-1}d\mu_{V_s}(x)d\mu_{V_s}(y)+\int V_s(x)d\mu_{V_s}(x)\right),
\end{multline*}
where $d\mu_{V_s}(x)$ denotes the equilibrium measure associated to $V_s(x)$.
\end{lemma}
\begin{proof}
	For either interpolation, we saw in Section \ref{sec:deform} that the edge points of the support are smooth functions of $s$ (in fact, for one interpolation, they are affine functions of $s$ given by \eqref{eq:Js}, while for the other they are constant in $s$). We use the notation $d\mu_{V_s}(x)
	=\psi_{V_s}(x)dx$ (recall \eqref{eq:density})  for the equilibrium measure associated to $V_s$ and $J_s=\bigcup_{j=1}^k[a_j(s),b_j(s)]$ for its support. By definition \eqref{eq:g_fn},
	\[
	g'(z)=\int_{J_s}\frac{1}{z-x}d\mu_{V_s}(x)
	\]
	so by Cauchy's integral formula
	\begin{align}\label{eq:altrep}
	\oint_\Gamma  g'(z)\frac{\partial}{\partial s} V_s(z)\frac{dz}{2\pi i}&=\int_{J_s}\frac{\partial}{\partial s} V_s(x)d\mu_{V_s}(x)\\
&=\frac{\partial}{\partial s}\int_{J_s} V_s(x)d\mu_{V_s}(x)-\int_{J_s}V_s(x)\frac{\partial}{\partial s} \psi_{V_s}(x)dx.\label{eq:altrep2}
	\end{align}

	By the Euler-Lagrange equation \eqref{eq:EL1}
\begin{multline}\label{eq:altrep3}\int_{J_s}V_s(x)\frac{\partial}{\partial s} \psi_{V_s}(x)dx \\
=\int_{J_s}\int_{J_s}2\log|x-y|\, \psi_{V_s}(y)\frac{\partial}{\partial s}\psi_{V_s}(x)dydx+\ell_s\int_{J_s}\frac{\partial }{\partial s}\psi_{V_s}(x)dx, 
\end{multline}
and since $\int_{J_s}\frac{\partial }{\partial s}\psi_{V_s}(x)dx=\frac{\partial }{\partial s}\int_{J_s}\psi_{V_s}(x)dx=0$, the lemma follows upon substituting \eqref{eq:altrep3} into \eqref{eq:altrep2}.
\end{proof}

\subsection{The order \texorpdfstring{$N$}{N} term}\label{sec:orderN}

By \eqref{Ninftylambda},
\[
N_\infty(z)^{-1}N_\infty'(z)=N_\lambda^{-1}(z)N_\lambda'(z)
\]
for any $\lambda \in \mathbb C$.
Letting $\lambda \to z$, and relying on the fact that $N_z(z)=I$, we find that by the definition of $N_\lambda$ in \eqref{eq:Mlamdef} that
\begin{equation}\label{eq:PinvP'}
\left[N_\infty(z)^{-1}N_\infty'(z)\right]_{11}=\lim_{\lambda \to z}N_{\lambda,11}'(z)=\sum_{j=1}^{k-1} \frac{\partial_j \theta(N\Omega)}{\theta(N\Omega)}u_j'(z).
\end{equation}
We now prove the following proposition.
\begin{proposition}\label{pr:orderN}
    For either of the interpolations of Section \ref{sec:deform}, we have
    \[
	\frac{1}{2\pi i}\oint_\Gamma \partial_s V_s(z) u_j'(z)dz=- \frac{\partial}{\partial s} \Omega_j(s).
	\]
\end{proposition}
Substituting into \eqref{eq:PinvP'}, we obtain
\begin{multline}\label{eq:PinvP'2}
-N\oint_\Gamma [N_\infty(z)^{-1}N_\infty'(z)]_{11}\frac{\partial}{\partial s}V_s(z)\frac{dz}{2\pi i} \\
=\frac{\partial}{\partial s}\log \theta(N\Omega)-\sum_{j,l=1}^{k-1}\frac{\partial \tau_{j,l}}{\partial s} \frac{\partial }{\partial \tau_{j,l}}\log \theta(N\Omega).
\end{multline}
\begin{proof} By \eqref{ddzV} and integration by parts 
\begin{equation*}
\frac{1}{2\pi i}\oint_\Gamma \frac{\partial}{\partial s} V_s(z)  u_j'(z)dz=
\frac{1}{2\pi i} \oint_\gamma \left( \frac{\partial}{\partial s} \psi_{V_s}(w) \right)\oint_\Gamma \frac{u_j(z)}{z-w}dz dw,\end{equation*}
where $\gamma$ encloses $\Gamma$. We deform $\Gamma$ to $\infty$ (giving a residue at $\infty$ and at $z=w$), recalling that $u_j(z)=u_j(\infty)+\mathcal O(z^{-1})$ as $z\to \infty$, to obtain
\begin{equation*}
\frac{1}{2\pi i}\oint_\Gamma  \frac{\partial}{\partial s}V_s(z)  u_j'(z)dz=-\oint_\gamma  \frac{\partial}{\partial s} \left(\psi_{V_s}(w)\right)(u_j(w)-u_j(\infty))dw.
\end{equation*}
By \eqref{eq:density}
\begin{equation*}\oint_\gamma \psi_{V_s}(w)dw=-2\int_{J_s}d\mu_{V_s}(x)=-2,
\end{equation*}
it follows that $\oint_\gamma\frac{\partial}{\partial s} \psi_{V_s}(w)dw=0$, so
\begin{equation*}
\frac{1}{2\pi i}\oint_\Gamma  \frac{\partial}{\partial s}V_s(z)  u_j'(z)dz=-\oint_\gamma  \frac{\partial}{\partial s} \left(\psi_{V_s}(w)\right)u_j(w)dw.
\end{equation*}
By  \eqref{eq:ujump1} and  \eqref{eq:ujump2},
\begin{multline*}
\frac{1}{2\pi i}\oint_\Gamma  \frac{\partial}{\partial s} V_s(z) u_j'(z)dz\\ =
\sum_{l=1}^{k} \mathbf 1(l\leq j) \int_{a_l}^{b_l}  \frac{\partial}{\partial s} \psi_{V_s}(w)dw
+\sum_{l=1}^{k-1} \tau_{l,j}\int_{b_l}^{a_{l+1}}  \frac{\partial}{\partial s} \psi_{V_s}(w)dw, \end{multline*}
where $a_l=a_l(s)$ and $b_l=b_l(s)$.
Since $\psi_{V_s}$ vanishes at the endpoints,
\begin{equation*} \begin{aligned}
\int_{a_l}^{b_l} \frac{\partial}{\partial s} \psi_{V_s}(w)dw&= \frac{\partial}{\partial s}\int_{a_l}^{b_l} \psi_{V_s}(w)dw,\\
\int_{b_l}^{a_{l+1}} \frac{\partial}{\partial s} \psi_{V_s}(w)dw&= \frac{\partial}{\partial s}\int_{b_l}^{a_{l+1}} \psi_{V_s}(w)dw. \end{aligned}
\end{equation*}
The integrals from $b_l$ to $a_{l+1}$ are zero by \eqref{zerointpsiV},
 and thus we have proven the proposition by the definition of $\Omega_j$ in \eqref{eq:Omega}.
\end{proof}
\subsection{The constant term}

Let $U_{a_j}$ and $U_{b_j}$ be fixed but sufficiently small discs surrounding $a_j$ and $b_j$ as in Section \ref{sec:para} with boundaries oriented  \emph{clockwise}. We have
\begin{align}\label{eq:order1}
&\oint_\Gamma\left[N_\infty(z)^{-1}\left(\tfrac{d}{d z}R^{(1)}(z)\right) N_\infty(z)\right]_{11}\frac{\partial}{\partial s} V_s(z)
\frac{dz}{2\pi i}\\
\notag &=\sum_{j=1}^k \oint_\Gamma\left(\oint_{ \partial U_{a_j}}\frac{\Delta_{11}(\lambda,z)}{(\lambda-z)^2}\frac{d\lambda}{2\pi i}+\oint_{ \partial U_{b_j}}\frac{\Delta_{11}(\lambda,z)}{(\lambda-z)^2}\frac{d\lambda}{2\pi i}\right)\frac{\partial}{\partial s} V_s(z)\frac{dz}{2\pi i},
\end{align}
where $R^{(1)}$ was defined in \eqref{formulaR1}.
We obtain from \eqref{PM-1fine} and \eqref{Ninftylambda} that $\Delta(\lambda,z)$ equals
\begin{equation}\label{eq:Dzw}
\begin{cases}
\frac{1}{\zeta_{a_j}(\lambda)^{3/2}}N_\lambda(z)^{-1}e^{N\pi i \Omega_{j-1}\eta(\lambda)\sigma_3}\sigma_3\frac{1}{8}\begin{pmatrix}
\frac{1}{6} & i\\
i & -\frac{1}{6}
\end{pmatrix}\sigma_3 e^{-N\pi i \Omega_{j-1}\eta(\lambda)\sigma_3}N_\lambda(z), & \\
\frac{1}{\zeta_{b_j}(\lambda)^{3/2}}N_\lambda(z)^{-1}e^{N\pi i \Omega_{j}\eta(\lambda)\sigma_3}\frac{1}{8}\begin{pmatrix}
\frac{1}{6} & i\\
i & -\frac{1}{6}
\end{pmatrix}e^{-N\pi i \Omega_{j}\eta(\lambda)\sigma_3}N_\lambda(z), & 
\end{cases}
\end{equation}
where the first line reads for $\lambda\in  \partial U_{a_j}$ and the second line for $\lambda\in  \partial U_{b_j}$, and where $\eta$ was defined in \eqref{eq:etasign} by $\eta(z)={\rm sgn}(\Im z)$, $\zeta_{a_j}$ and $\zeta_{b_j}$ were defined in \eqref{eq:zetadef}, and $N_\lambda$ was defined in \eqref{eq:Mlamdef}.
We begin our study of $\Delta_{11}$ by simply writing out the quantity more explicitly. A straightforward calculation, making use of the fact that $\det N_\lambda(z)=1$
so that e.g. $(N_\lambda^{-1})_{11}=N_{\lambda,22}$ etc, shows that for $\lambda\in  \partial U_{a_j}$, with $j=1,\dots,k$,
\begin{equation}\label{eq:delata}\begin{aligned}
 \Delta_{11}(\lambda,z)&=\Delta_{11}^{(1)}(\lambda,z)+\Delta_{11}^{(2)}(\lambda,z),  \\
\Delta_{11}^{(1)}(\lambda,z)&=\frac{1}{48\zeta_{a_j}(\lambda)^{3/2}}\left(N_{\lambda,11}(z)N_{\lambda,22}(z)+N_{\lambda,12}(z)N_{\lambda,21}(z)\right),\\
\Delta_{11}^{(2)}(\lambda,z) &=-\frac{i}{8\zeta_{a_j}(\lambda)^{3/2}}\bigg(e^{2\pi i N\Omega_{j-1}\eta(\lambda)}N_{\lambda,22}(z)N_{\lambda,21}(z)\\ & \quad-e^{-2\pi i N\Omega_{j-1}\eta(\lambda)}N_{\lambda,12}(z)N_{\lambda,11}(z)\bigg),
\end{aligned}\end{equation}
(where we set $\Omega_0=\Omega_k=0$),
and for $\lambda\in  \partial U_{b_j}$, with $j=1,\dots,k$,
\begin{equation}\label{eq:delatb}\begin{aligned}
 \Delta_{11}(\lambda,z)&=\Delta_{11}^{(1)}(\lambda,z)+\Delta_{11}^{(2)}(\lambda,z),\\
\Delta_{11}^{(1)}(\lambda,z)&=\frac{1}{48\zeta_{b_j}(\lambda)^{3/2}}\left(N_{\lambda,11}(z)N_{\lambda,22}(z)+N_{\lambda,12}(z)N_{\lambda,21}(z)\right),\\
\Delta_{11}^{(2)}(\lambda,z)&= \frac{i}{8\zeta_{b_j}(\lambda)^{3/2}}\bigg(e^{2\pi i N\Omega_{j}\eta(\lambda)}N_{\lambda,22}(z)N_{\lambda,21}(z)\\ &\quad-e^{-2\pi i N\Omega_{j}\eta(\lambda)}N_{\lambda,12}(z)N_{\lambda,11}(z)\bigg).
\end{aligned}\end{equation}

\subsection{Application of $\theta$-function identities to evaluate $\Delta_{11}^{(1)}$ and $\Delta_{11}^{(2)}$}

We will shortly apply the $\theta$-function identities from Proposition \ref{PropForM} to obtain expressions for $\Delta_{11}^{(1)}$ and $\Delta_{11}^{(2)}$ where the oscillations (namely $\theta$ functions containing $N\Omega$) and the slowly varying terms separate, so we can average out the oscillations.

We now obtain an identity for the terms appearing in \eqref{eq:delata} and \eqref{eq:delatb}.

\begin{lemma}\label{le:paraprod} Let $N_\lambda$ be defined by \eqref{eq:Mlamdef} with $v=N\Omega$, and let $u$ be defined by \eqref{eq:Abel}. For $z,\lambda \in \mathbb C\setminus [a_1,b_k]$, we have
\begingroup \allowdisplaybreaks
\begin{align}
\label{eq:prod1} &N_{\lambda,11}(z)N_{\lambda,22}(z)+N_{\lambda,12}(z)N_{\lambda,21}(z)=\\ \nonumber \quad &(z-\lambda)^2\left(W(\lambda,z)+2\sum_{i,j=1}^{k-1}(\partial_i\partial_j\log\theta)(N\Omega)u_i'(z)u_j'(\lambda)\right)\\
\label{eq:prod2}&N_{\lambda,22}(z)N_{\lambda,21}(z)=i\frac{(\lambda-z)^2}{4}\frac{\theta(0)\theta(2u(\lambda)+N\Omega)}{\theta(2u(\lambda))\theta(N\Omega)}\sum_{j=1}^k \left(\frac{1}{\lambda-b_j}-\frac{1}{\lambda-a_j}\right)\\
&\quad \notag \times \left(w_\lambda(z) +\sum_{j=1}^{k-1}(\partial_j \log \theta(2u(\lambda)+N\Omega)-\partial_j \log \theta(N\Omega))u_j'(z)\right) \\
\label{eq:prod4}&N_{\lambda,12}(z)N_{\lambda,11}(z)=-i\frac{(\lambda-z)^2}{4}\frac{\theta(0)\theta(2u(\lambda)-N\Omega)}{\theta(2u(\lambda))\theta(N\Omega)}\sum_{j=1}^k \left(\frac{1}{\lambda-b_j}-\frac{1}{\lambda-a_j}\right)\\
&\quad \notag \times \left(w_\lambda(z)+\sum_{j=1}^{k-1}(\partial_j \log \theta(2u(\lambda)-N\Omega)+\partial_j \log \theta(N\Omega))u_j'(z)\right),
\end{align}
\endgroup
where $W$ is as in \eqref{def:W} and $w_{\lambda}(z)$ is as in \eqref{defwlambda}.
\end{lemma}
\begin{proof}
Since $\det N_\lambda=1$, 
the left hand side of \eqref{eq:prod1} is equal to $$2N_{\lambda,11}(z)N_{\lambda,22}(z)-1,$$ and substituting  \eqref{ForM11M22} into \eqref{eq:Mlamdef} (with $v=N\Omega$) yields the desired result for $N_{\lambda,11}(z)N_{\lambda,22}(z)$. Formulas \eqref{eq:prod2} and \eqref{eq:prod4} both follow in a straightforward manner from substituting \eqref{ForM11M12} into \eqref{eq:Mlamdef}.

 In all cases, it is useful to keep in mind that $\partial_j \log \theta$ and $\partial_j \log \theta \left[\substack{\alpha\\ \beta}\right]$ are odd functions.
\end{proof}

Let $v\in \mathbb C^{k-1}$. By \eqref{eq:ujump1}, \eqref{uofa}, \eqref{uofb}, and \eqref{quasitheta},
\begin{equation}\label{uintomega}
\begin{aligned}\frac{\theta(2u(\lambda)+v)}{\theta(2u(\lambda))}&=e^{- 2\pi i v_{j-1}\eta(\lambda)}\frac{\theta\left(2\int_{a_j}^\lambda \boldsymbol\omega+v\right)}{\theta \left(2\int_{a_j}^\lambda\boldsymbol \omega\right)}&&\textrm{for $\lambda \in \partial U_{a_j}$,}\\
\frac{\theta(2u(\lambda)+v)}{\theta(2u(\lambda))}&=e^{- 2\pi i v_{j}\eta(\lambda)}\frac{\theta\left(2\int_{b_j}^\lambda\boldsymbol \omega+v\right)}{\theta \left(2\int_{b_j}^\lambda \boldsymbol\omega\right)}&&\textrm{for $\lambda \in \partial U_{b_j}$,}
\end{aligned}
\end{equation}
where we set $v_0=v_k=0$.

Let 
\begin{equation}\label{def:wtilde}
\widetilde w_\lambda(z)=\begin{cases}
w_{\lambda}(z)-2\pi i\eta(\lambda) u_{j-1}'(z)&\textrm{for $\lambda\in U_{a_j}$, for $j=1,\dots,k$,}\\
w_{\lambda}(z)-2\pi i\eta(\lambda) u_{j}'(z)&\textrm{for $\lambda\in U_{b_j}$, for $j=1,\dots,k$,}
\end{cases} \end{equation}
(where $u_0'=u_k'=0$),
and observe that by \eqref{jumpswlambda2} and \eqref{jumpswlambda3}, $\frac{\widetilde w_\lambda(z)}{\mathcal R^{1/2}(\lambda)}$ is meromorphic as a function of $\lambda$ on each disc $U_x$ for $x\in \{a_j,b_j\}_{j=1}^k$, and since $w_\lambda(z)$ is bounded as $\lambda \to x$, $\frac{\widetilde w_\lambda(z)}{\mathcal R^{1/2}(\lambda)}$ is analytic on each disc $U_x$ as a function of $\lambda$.

Putting these remarks together, we find
\begin{multline}\label{Firstcutoff}
\oint_\Gamma 
\left[N_\infty(z)^{-1}\left(\tfrac{d}{d z}R^{(1)}(z)\right) N_\infty(z)\right]_{11}\frac{\partial}{\partial s} V_s(z)
\frac{dz}{2\pi i} =\hat r(s)\\+  \oint_\Gamma \frac{dz}{2\pi i}\frac{\partial}{\partial s}V_s(z)  \sum_{q\in\{a_j,b_j\}_{j=1}^k}
  \Bigg[ \oint_{\partial U_{q}}\frac{d\lambda}{2\pi i\zeta_{q}(\lambda)^{3/2}}\Bigg(\frac{1}{48}W(\lambda,z)\\ \pm \tfrac{\theta(0)}{32\theta(N\Omega)}
	\left(\tfrac{\theta\left(2\int_{q}^\lambda\boldsymbol\omega+N\Omega\right)+\theta\left(2\int_{q}^\lambda \boldsymbol\omega-N\Omega\right)}{\theta\left(2\int_q^\lambda \boldsymbol\omega\right)}\right)
\widetilde w_\lambda(z)
\sum_{l=1}^k
\left(
\tfrac{1}{\lambda-b_l}-\tfrac{1}{\lambda-a_l}
\right)\Bigg)\Bigg],
\end{multline}
where $\partial U_q$ is oriented \emph{clockwise}, and
where $\pm=+$ for $q\in \{a_j\}_{j=1}^k$ and $\pm=-$ for $q\in\{b_j\}_{j=1}^k$, and
\begin{equation} \label{def:hatr} \hat r(s)=\sum_{q\in\{a_j,b_j\}_{j=1}^k}\oint_{\partial U_{q}}\frac{d\lambda}{2\pi i\zeta_{q}(\lambda)^{3/2}} \sum_{i=1}^3\hat r_i(s;\lambda), \end{equation} 
where
\begin{equation}\label{BigOne}
\begin{aligned}
\hat r_1(s;\lambda)&=\frac{1}{24}\oint_\Gamma \frac{dz}{2\pi i}\frac{\partial}{\partial s}V_s(z)\sum_{i,l=1}^{k-1}(\partial_i\partial_l\log\theta)(N\Omega)u_i'(z)u_l'(\lambda),\\
\hat r_2(s;\lambda)&=\pm \frac{1}{32}\oint_\Gamma \frac{dz}{2\pi i}\frac{\partial}{\partial s}V_s(z)\theta(0)\frac{\theta\left(2\int_{q}^\lambda\boldsymbol\omega+N\Omega\right)}{\theta\left(2\int_q^\lambda\boldsymbol \omega \right)\theta(N\Omega)}\\& \times \left(w_\lambda(z)-\widetilde w_\lambda(z)+
	\sum_{l=1}^{k-1}u_l'(z)\partial_l \log \tfrac{\theta\left(2u(\lambda)+N\Omega\right)}{\theta(N\Omega)}\right)\sum_{l=1}^k
\left(
\tfrac{1}{\lambda-b_l}-\tfrac{1}{\lambda-a_l}
\right),\\
\hat r_3(s;\lambda)&=\pm \frac{1}{32}\oint_\Gamma \frac{dz}{2\pi i}\frac{\partial}{\partial s}V_s(z)\theta(0)\frac{\theta\left(2\int_{q}^\lambda \boldsymbol\omega-N\Omega\right)}{\theta\left(2\int_q^\lambda \boldsymbol\omega \right)\theta(N\Omega)}\\& \times \left(w_\lambda(z)-\widetilde w_\lambda(z)-
	\sum_{l=1}^{k-1}u_l'(z)\partial_l \log \tfrac{\theta\left(-2u(\lambda)+N\Omega\right)}{\theta(N\Omega)}\right)\sum_{l=1}^k
\left(
\tfrac{1}{\lambda-b_l}-\tfrac{1}{\lambda-a_l}
\right).
\end{aligned}
\end{equation}

If a function $g_N(s)$ satisfies $\int_0^2 g_N(s)ds=\mathcal O(1/N)$ as $N\to \infty$, then we denote $g_N(s)=\widetilde {\mathcal O}(1/N)$ as $N\to \infty$.
\begin{lemma} \label{LemmaOtilde}
As $N\to \infty$, we have $\hat r(s) 
=\widetilde{\mathcal O}(1/N)$.
\end{lemma}
Lemma \ref{LemmaOtilde} allows us to ignore terms which do not contribute to the final answer already before performing our analysis of residues, and thus we avoid performing calculations which would otherwise be rather involved. In this respect we were inspired by a similar procedure undertaken in \cite{FahsKra}, although the situation at hand is different to that of \cite{FahsKra}.
\begin{proof} 
We will use Fubini to change the order of integration in $s$ and $\lambda$. However, since the contour of integration $\partial U_q$ depends on $s$, we start by discussing a few technicalities to facilitate the use of Fubini. For $s\in[1,2]$, the support of the equilibrium measure $J$ is fixed, while 
for $s\in(0,1)$, the endpoints of the support $\{a_j,b_j\}_{j=1}^k$ depend on $s$. However, since the radius of $U_{x}$ is uniformly bounded from below for $x\in \{a_j,b_j\}_{j=1}^k$ for all $s\in[0,1)$, it follows that there are a finite number of points $0=s_0<s_1<\dots<s_K=2$ such that for any $s\in [s_{j-1},s_j]$ and $x(s)\in\{a_j(s),b_j(s)\}_{j=1}^{k}$, we have
$x(s)\in U_{x(s_{j-1})}\cap U_{x(s_j)}$. We will integrate over each such interval separately, and add the contributions at the end. 
To ease the notation, we simply say that we integrate from $s_1$ to $s_2$. Then for all $s_1\leq s \leq s_2$, we can assume that the contours $\partial U_q$ are fixed and independent of $s$ in the  integrals on the right-hand side of \eqref{eq:order1}, since we are evaluating only the residue.  By Fubini,
\begin{equation*}
\int_{s_1}^{s_2}\hat r(s)ds=\sum_{q\in\{a_j,b_j\}_{j=1}^k}\oint_{\partial U_{q}}\int_{s_1}^{s_2}\frac{1}{\zeta_{q}(\lambda)^{3/2}} \sum_{i=1}^3\hat r_i(s;\lambda)ds\frac{d\lambda}{2\pi i}.
\end{equation*}
We start by considering $\hat r_1$. By Proposition \ref{pr:orderN}
\begin{multline}\label{parts1}
24\frac{\hat r_1(s;\lambda)}{\zeta_q(\lambda)^{3/2}}=-\sum_{l=1}^{k-1}\frac{1}{N}\frac{\partial}{\partial s}\left(\frac{u_l'(\lambda)}{\zeta_q(\lambda)^{3/2}}(\partial_l\log\theta)(N\Omega)\right)\\+\sum_{l=1}^{k-1}\frac{u_l'(\lambda)}{N\zeta_q(\lambda)^{3/2}}\sum_{1\leq u\leq v \leq k-1}\frac{\partial \tau_{u,v}}{\partial s} \frac{\partial}{\partial \tau_{u,v}}(\partial_l\log\theta)(N\Omega)
\\ +\sum_{l=1}^{k-1}\frac{\partial}{\partial s}\left(\frac{u_l'(\lambda)}{N\zeta_q(\lambda)^{3/2}}\right)(\partial_l\log\theta)(N\Omega).
\end{multline}
In \eqref{parts1} we interpret $\tau$ to always be symmetric, so that when we take the derivative of a function $f$ depending on $\tau$ (e.g. $f(\tau)=\log \theta(N\Omega|\tau)$ above), it is given by
\begin{equation}\label{eqdifftausym}\begin{aligned}\frac{\partial}{\partial \tau_{j,l}}f(\tau)&=\lim_{\epsilon \to 0}\frac{f(\tau+\epsilon \chi_{j,l}+ \epsilon\chi_{l,j})-f(\tau)}{\epsilon}&& \textrm{for }1\leq j<  l\leq k-1,\\
\frac{\partial}{\partial \tau_{j,j}}f(\tau)&=\lim_{\epsilon \to 0}\frac{f(\tau+\epsilon \chi_{j,j})-f(\tau)}{\epsilon}&& \textrm{for }1\leq j \leq k-1,
 \end{aligned}\end{equation}
where $\chi_{j,l}$ is the matrix with zeros in all entries except the $j,l$ entry and $1$ in the $j,l$ entry.
We observe that $\partial_l \log \theta$ and  $\frac{\partial}{\partial \tau_{u,v}}\partial_l \log \theta$ in \eqref{parts1} are bounded functions on $\mathbb R$, and that $\frac{u_l'(\lambda)}{\zeta_q(\lambda)^{3/2}}$ and its derivative with respect to $s$ are both uniformly bounded for $\lambda\in \partial U_q$ for all $q\in \{a_j,b_j\}_{j=1}^k$. Thus the second and the third term on the right-hand side of \eqref{parts1} are both uniformly of order $\mathcal O(1/N)$ for $s\in[0,2]$. While the first term on the right-hand side is not in general small, the integral in $s$ from $s_1$ to $s_2$ is of order $\mathcal O(1/N)$ as $N\to \infty$. Thus it follows that 
\begin{equation*}\int_{s_1}^{s_2}\frac{\hat r_1(s;\lambda)}{\zeta_q(\lambda)^{3/2}}ds=\mathcal O(1/N),\end{equation*}
as $N\to \infty$, uniformly for $\lambda \in \partial U_q$, for $q\in \{a_j,b_j\}_{j=1}^k$. 
Now consider the term involving $\hat r_2$. By \eqref{uintomega}, \eqref{BigOne}, and Proposition \ref{pr:orderN}, for $q=b_j$
and $\Im \lambda>0$,
\begin{multline}\label{parts2} 32 \frac{\hat r_2(s;\lambda)}{\zeta_q(\lambda)^{3/2}}=-\frac{1}{N} \frac{\partial}{\partial s}\left(\frac{\theta(0)}{\zeta_q(\lambda)^{3/2}}\frac{\theta\left(2\int_{q}^\lambda \boldsymbol\omega+N\Omega\right)}{\theta\left(2\int_q^\lambda \boldsymbol\omega \right)\theta(N\Omega)}\sum_{l=1}^k
\left(
\tfrac{1}{\lambda-b_l}-\tfrac{1}{\lambda-a_l}
\right)\right)\\
+e^{2\pi i N\Omega_j}\frac{\theta\left(2u(\lambda)+N\Omega\right)}{N\theta(N\Omega)} \frac{\partial}{\partial s}\left(\frac{\theta(0)}{\zeta_q(\lambda)^{3/2}}\frac{1}{\theta\left(2u(\lambda) \right)}\sum_{l=1}^k
\left(
\tfrac{1}{\lambda-b_l}-\tfrac{1}{\lambda-a_l}
\right)\right)\\ +\frac{e^{2\pi i N\Omega_j}}{N\theta(N\Omega)} \frac{\theta(0)}{\zeta_q(\lambda)^{3/2}}\frac{1}{\theta\left(2u(\lambda) \right)}\left(\sum_{l=1}^k
\left(
\tfrac{1}{\lambda-b_l}-\tfrac{1}{\lambda-a_l}
\right)\right)\sum_{l=1}^{k-1}\frac{2\partial u_l(\lambda)}{\partial s} \partial_l\theta\left(2u(\lambda)+N\Omega\right)\\
+\frac{e^{2\pi i N\Omega_j} }{N}\left(\frac{\theta(0)}{\zeta_q(\lambda)^{3/2}}\frac{1}{\theta\left(2u(\lambda) \right)}\sum_{l=1}^k
\left(
\tfrac{1}{\lambda-b_l}-\tfrac{1}{\lambda-a_l}
\right)\right)\\ \times \sum_{1\leq u \leq v \leq k-1}\frac{\partial \tau_{u,v}}{\partial s}\frac{\partial}{\partial \tau_{u,v}}\frac{\theta\left(2u(\lambda)+N\Omega\right)}{\theta(N\Omega)},
\end{multline}
where we interpret $\frac{\partial}{\partial \tau_{u,v}}$ as in \eqref{eqdifftausym}.
It is easily verified that the second, third and fourth terms are all of order $\mathcal O(1/N)$ as $N\to \infty$, uniformly for $\Im \lambda>0$. The integral of the first term with respect to $s$ from $s_1$ to $s_2$ is of order $\mathcal O(1/N)$ as $N\to \infty$, again uniformly for $\lambda \in \partial U_{b_j}$ such that $\Im \lambda>0$. The situation $\Im \lambda<0$ is treated similarly, with $e^{2\pi i N\Omega_j}$ replaced by $e^{-2\pi i N\Omega_j}$. Likewise, the situation $q\in \{a_j\}_{j=1}^k$ is treated similarly, with $e^{2\pi i N\Omega_j}$ replaced by $e^{2\pi i N\Omega_{j-1}}$. It follows that
\begin{equation*} \int_{s_1}^{s_2}\frac{\hat r_2(s;\lambda)}{\zeta_q(\lambda)^{3/2}} ds=\mathcal O(1/N), \end{equation*}
as $N\to \infty$, uniformly for $\lambda \in \partial U_q$ for all $q\in \{a_j,b_j\}_{j=1}^k$. 

Finally, $\hat r_3$ is treated in a similar manner to $\hat r_2$, with $N\Omega$ exchanged for $-N\Omega$, and one obtains
\begin{equation*} \int_{s_1}^{s_2}\frac{\hat r_i(s;\lambda)}{\zeta_q(\lambda)^{3/2}} ds=\mathcal O(1/N), \end{equation*}
as $N\to \infty$, uniformly for $\lambda \in \partial U_q$ for all $q\in \{a_j,b_j\}_{j=1}^k$, where $i=1,2,3$. Thus the lemma is proven.
\end{proof}

\subsection{Evaluation of residues}
We now aim to evaluate the remaining terms in \eqref{Firstcutoff}.
As noted, on each disc $U_{a_j}$ and $U_{b_j}$, the integrand in the $\lambda$ variable is meromorphic with a pole of order $2$ at the center of each disc. We now calculate the residue at each pole, and to this end we introduce some notation. Recall, from e.g. \eqref{id:W},  that  $W(z,\lambda)\mathcal R^{1/2}(z)$ meromorphic as a function of $z$.
 We denote
\begin{equation}\label{eq:What}\begin{aligned}
\widetilde W(q,z)&=\lim_{\lambda\to q}\frac{1}{2}(\lambda-q)^{1/2}W(\lambda,z),\\
\widetilde W'(q,z)&=\lim_{\lambda\to q}\frac{1}{2}\frac{\partial}{\partial \lambda}\left[(\lambda-q)^{1/2}W(\lambda,z)\right],\\
\hat \psi(q)&=\lim_{\lambda\to q} \pi i \frac{\psi_{V_s}(\lambda)}{\left(\lambda-q\right)^{1/2}},\\
\hat \psi'(q)&=\lim_{\lambda \to q} \frac{d}{d\lambda}\left[\pi i\frac{\psi_{V_s}(\lambda)}{\left(\lambda-q\right)^{1/2}}\right],
\end{aligned}
\end{equation}
as $\lambda \to q\in \{a_j,b_j\}_{j=1}^{k}$, where all roots have branch cuts on $J$ and $(\lambda-a_j)$ has arguments in $[0,2\pi)$ and $(\lambda-b_j)$ has arguments in $[-\pi,\pi)$, $\hat \psi (b_j)>0$, and $ \arg \hat \psi(a_j)=\pi/2$. By the definition of $\zeta$ in \eqref{eq:zetadef},
\begin{equation*}
\frac{1}{\zeta_{q}(\lambda)^{3/2}}=\frac{1-\frac{3}{5}\frac{\hat \psi'(q)}{\hat \psi (q)}(\lambda-q)+\mathcal O\left((\lambda-q)^2\right)}{\hat \psi(q)(\lambda-q)^{3/2}},
\end{equation*}
as $\lambda \to q$.

As $\lambda\to q\in\{ a_j,b_j\}_{j=1}^k$ from the upper half plane,  
        \begin{equation}\label{eq:rec2}
       \widetilde w_\lambda(z)=4\widetilde W(q,z)(\lambda-q)^{1/2}+\frac{4}{3}\widetilde W'(q,z)(\lambda-q)^{3/2} +\mathcal O((\lambda-q)^{5/2}),
        \end{equation}
         and denoting $\hat u'_j(q)=\lim_{\lambda\to q}(\lambda-q)^{1/2}u'_j(\lambda)$,
        \begin{multline}\label{expantheta}
         \frac{\theta(0)}{\theta(N\Omega)}\frac{\theta \left(2\int_{q}^\lambda \boldsymbol\omega+N\Omega\right)+\theta \left(2\int_{q}^\lambda \boldsymbol\omega-N\Omega\right)}{\theta \left(2\int_{q}^\lambda \boldsymbol\omega\right)}\\=2+16\sum_{i,l=1}^{k-1} \left(\frac{\partial_i\partial_l  \theta(N\Omega)}{\theta(N\Omega)}-\frac{\partial_i\partial_l  \theta(0)}{\theta(0)}\right)\hat u_i'(q)\hat u_l'(q)(\lambda-q)
    +\mathcal O((\lambda-q)^{2}),
        \end{multline}

By the definition of the $\theta$-function, we have $\partial_j^2\theta(x|\tau)=4\pi i \frac{\partial}{\partial \tau_{j,j}}\theta(x|\tau)$ and $\partial_j \partial_l \theta(x|\tau)=2\pi i \frac{\partial}{\partial \tau_{j,l}}\theta(x|\tau)$, where we interpret the partial derivatives as in \eqref{eqdifftausym}. Thus,
substituting the expansions \eqref{eq:What}-\eqref{expantheta} into \eqref{Firstcutoff}, and relying on Lemma \ref{LemmaOtilde}, we obtain
\begin{multline} \label{Thirdcutoff}   \oint_\Gamma 
\left[N_\infty(z)^{-1}\left(\tfrac{d}{d z}R^{(1)}(z)\right) N_\infty(z)\right]_{11}\frac{\partial}{\partial s} V_s(z)
\frac{dz}{2\pi i}  \\
= \Bigg \{ \frac{1}{24}\sum_{q\in \{a_j,b_j\}_{j=1}^k}
\oint_\Gamma \frac{dz}{2\pi i} \frac{\partial}{\partial s} V_s(z)\left(\widetilde W'(q,z)-3\frac{\hat \psi'(q)}{\hat \psi(q)} \widetilde W(q,z)\right)\frac{1}{\hat \psi(q)}
\\ +4\sum_{1\leq  m\leq j \leq k-1}\frac{\partial}{\partial \tau_{m,j}} \log \frac{\theta(N\Omega)}{\theta(0)} \sum_{q\in\{a_l,b_l\}_{l=1}^k}  \frac{\hat u_j'(q)\hat u_m'(q)}{\hat{\psi}(q)} \oint_{\Gamma}dz \frac{\partial}{\partial s} V_s(z) \widetilde{W} (q, z)\\
+\frac{1}{4}\oint_\Gamma \frac{dz}{2\pi i}\frac{\partial}{\partial s} V_s(z)\Bigg[\sum_{l\neq j} \left(\frac{\widetilde W(a_j,z)}{(a_j-a_l)\hat \psi(a_j)}+\frac{\widetilde W(b_j,z)}{(b_j-b_l)\hat \psi(b_j)}\right)\\ -
	\sum_{j,l=1}^k\left(\frac{\widetilde W(a_j,z)}{(a_j-b_l)\hat \psi(a_j)}+\frac{\widetilde W(b_j,z)}{(b_j-a_l)\hat \psi(b_j)}\right)\Bigg]
	\Bigg \}+\widetilde{\mathcal O}(1/N),
\end{multline}
as $N\to \infty$.
	
        \subsection{Integration in $s$} We now aim  to evaluate the sum \eqref{Thirdcutoff}. To this end we need to develope several identities. 
First off, recall that by cutting and gluing $\mathcal S$, it may be represented as a $(4k-4)$-gon with edges $A_1,B_1,A_1^{-1},B_1^{-1},A_2,B_2,A_2^{-1},B_2^{-1},\dots, B_{k-1}^{-1}$. We denote 
\\ $\gamma=A_1B_1A_1^{-1}B_1^{-1}A_2B_2A_2^{-1}B_2^{-1}\dots B_{k-1}^{-1}$, so that $\gamma$ first traverses $A_1$, then traverses $B_1$, etc (see \cite[Chapter I.2.5]{FK}).
Recall Riemann's bilinear identity \cite[Chapter III.2]{FK}, namely given two meromorphic differentials $\lambda $ and $\nu$, and a base-point $z_0$,
\begin{equation*}\oint_\gamma \Lambda \nu=\sum_{j=1}^{k-1}\left(\oint_{A_j}\lambda\oint_{B_j}\nu-\oint_{B_j}\lambda \oint_{A_j}\nu\right), \qquad \Lambda(z)=\int_{z_0}^z\lambda. \end{equation*}
  Let $q\in \{a_j,b_j\}_{j=1}^k$,
and fix $\lambda=\frac{\partial}{\partial q}\boldsymbol{\omega_l}$ and $\nu=\boldsymbol{\omega_j}$. 
Then $\oint_{A_i}\nu=\oint_{A_i}\boldsymbol{\omega_j}=\delta_{ij}$ and  $0=\frac{\partial}{\partial q}\oint_{A_i}\boldsymbol{\omega_l}=\oint_{A_i}\frac{\partial}{\partial q}\boldsymbol{\omega_l}=\oint_{A_i}\lambda$, and thus
\begin{equation}\label{RBI} \oint_\gamma \left(\int_{z_0}^z\frac{\partial}{\partial q}\boldsymbol{\omega_l}\right)\boldsymbol{\omega_j}=
	-\oint_{A_j}\boldsymbol{\omega_j}\oint_{B_j}\frac{\partial}{\partial q}\boldsymbol{\omega_l}=-
	\frac{\partial}{\partial q}\tau_{lj}.
\end{equation}
Since $\gamma$ encloses $\mathcal S$, we can deform it to the poles of $\int_{z_0}^z \frac{\partial}{\partial q}\boldsymbol{\omega_l}$. 
$\int_{z_0}^z \frac{\partial}{\partial q}\boldsymbol{\omega_l}$ is analytic except for the pole at $z=q$, and as $z\to q$,
\begin{equation*}
\left(\int_{z_0}^z \frac{\partial}{\partial q}\boldsymbol{\omega_l}\right)\boldsymbol{\omega_j}(z)=- \frac{\hat u_j'(q)\hat u_l'(q)}{z-q}\left(1+\mathcal O\left(\sqrt{z-q}\right)\right)dz, \quad \hat u_j'(q)=\lim_{z\to q}\sqrt{z-q}u_j'(z).
\end{equation*}
The local coordinate at $q$ is $(z-q)^{1/2}=y$, and so
\begin{equation*} \left(\int_{z_0}^z \frac{\partial}{\partial q}{\boldsymbol{\omega_l}}\right)\boldsymbol{\omega_j}(z)=-2\frac{\hat u_j'(q)\hat u_l'(q)}{y}(1+O(y))dy
\end{equation*}
Thus, deforming $\gamma$ to surround $q$ on the left-hand side of \eqref{RBI}, we obtain
\begin{equation}\label{dtaudp}
\frac{\partial}{\partial q}\tau_{lj}=4\pi i \hat u_j'(q)\hat u_l'(q), \end{equation}
which is already a well-known formula, see e.g. \cite[Section 4]{Rauch}.

Next, we prove the following lemma.

\begin{lemma}Let $q\in \{a_j,b_j\}_{j=1}^k$. Then
\begin{equation}\label{partialp}
\frac{\partial}{\partial s} q= 	-\frac{2}{\hat \psi(q)}\oint_\Gamma \frac{\partial}{\partial s} V_s(z)\frac{\widetilde W(q,z)}{2\pi i}dz.
\end{equation}
\end{lemma}
We observe that by \eqref{partialp} and \eqref{dtaudp} we have
\begin{equation}\label{dtauds}
\partial_s \tau_{j,m}(s)=-4 \sum_{q\in\{a_l,b_l\}_{l=1}^k}  \frac{\hat u_j'(q)\hat u_m'(q)}{\hat{\psi}(q)} \oint_{\Gamma}dz \frac{\partial}{\partial s} V_s(z) \widetilde{W} (q, z).
\end{equation} 
\begin{proof}
Recall, e.g. from \eqref{id:W}, that $W(z,\lambda)\mathcal R^{1/2}(z)$ is meromorphic as a function of $z$, and that $\psi_{V_s}(z)\mathcal R^{1/2}(z)$ is analytic in a neighbourhood of $\mathbb R$.
Using \eqref{eq:What}, the right-hand side of \eqref{partialp} is given by
\begin{equation}\label{integralid1}
\frac{1}{\pi i}\oint_\Gamma \frac{dz}{2\pi i} \oint_{\partial U_q}\frac{d\lambda}{2\pi i} \frac{\frac{\partial}{\partial s} V_s(z) W(\lambda,z) }{\psi_{V_s}(\lambda)},\end{equation}
where as before the orientation of $\partial U_q$ is clockwise.
By integration by parts and the definition of $W(z,\lambda)$ in \eqref{def:W} and the definition of $w_z(\lambda)$ in \eqref{defwlambda},
\begin{equation}\label{integralid1b}
-\oint_\Gamma \frac{dz}{2\pi i}  \frac{\partial}{\partial s} V_s(z) W(\lambda,z) 
=\oint_\Gamma \frac{dz}{2\pi i}  \frac{\partial}{\partial s} V_s'(z) w_z(\lambda),\end{equation}
 and by \eqref{ddzV}, 
\begin{equation}\label{integralid4}\oint_\Gamma \frac{dz}{2\pi i}  \frac{\partial}{\partial s} V_s'(z) w_z(\lambda)=-\oint_{\Gamma'} d\xi \oint_\Gamma \frac{dz}{2\pi i} \frac{\frac{\partial}{\partial s} \psi_{V_s}(\xi)}{z-\xi} w_z(\lambda),\end{equation}
where $\Gamma'$ encloses $\Gamma$.
 We will now integrate in the $z$ variable, and observe that $\Gamma$ does not enclose $\xi$. We have that $w_z(\lambda)=w_\infty(\lambda)+\mathcal O\left(\frac{1}{z}\right)$ as $z\to \infty$. Thus, expanding the contour $\Gamma$ to $\infty$ in the right-hand side of \eqref{integralid4}, we only pick up a residue at $z=\xi$ and $z=\infty$:
\begin{equation}\label{integralid2}\oint_\Gamma \frac{dz}{2\pi i}\frac{w_z(\lambda)}{z-\xi}=-w_\xi(\lambda)+w_\infty(\lambda). \end{equation}

Next we integrate in $\xi$. Since 
\begin{equation*}\oint_{\Gamma'} \frac{\partial}{\partial s}\psi_{V_s}(\xi)d\xi=\frac{\partial}{\partial s}\oint_{\Gamma'} \psi_{V_s}(\xi)d\xi, \end{equation*}
which is zero since $\oint_{\Gamma'} \psi_{V_s}(\xi)d\xi=-2\int_J d\mu_{V_s}=-2$, the term including $w_\infty(\lambda)$ will be zero once substituted into \eqref{integralid4}.

We now  
evaluate $\oint_{\Gamma'} d\xi w_\xi(\lambda)\frac{\partial}{\partial s}\psi_{V_s}(\xi)$. We  
deform $\Gamma'$ to hug $[a_1,b_k]$. We recall from \eqref{polewlambda} that $w_\xi(\lambda)$ has  a pole of order 1 at $\xi=\lambda$. By \eqref{jumpswlambda2} and \eqref{jumpswlambda3}, and the fact that $\psi_{V_s}(\xi)_+=-\psi_{V_s}(\xi)_-$ on $J$, we obtain
\begin{equation}\label{integralid3}\oint_{\Gamma'} d\xi w_\xi(\lambda)\frac{\partial}{\partial s}\psi_{V_s}(\xi)=-2\pi i \frac{\partial}{\partial s}\psi_{V_s}(\lambda)-4\pi i   \sum_{j=1}^{k-1}u_{j}'(\lambda)\int_{b_j}^{a_{j+1}} d\xi\frac{\partial}{\partial s}\psi_{V_s}(\xi).
\end{equation}

By \eqref{zerointpsiV}, and since $\psi_{V_s}$ has square root vanishing at the endpoints of the support $a_j,b_j$, 
\begin{equation*}0=\frac{\partial}{\partial s}\int_{b_j}^{a_{j+1}} d\xi \, \psi_{V_s}(\xi)=\int_{b_j}^{a_{j+1}} d\xi \frac{\partial}{\partial s} \psi_{V_s}(\xi). \end{equation*}
Thus, by \eqref{integralid1b}--\eqref{integralid3}, we obtain that for $\lambda\in \partial U_q$,
\begin{equation}\label{integralid6}
-\oint_\Gamma \frac{dz}{2\pi i}  \frac{\partial}{\partial s} V_s(z) W(\lambda,z) =-2\pi i \frac{\partial}{\partial s} \psi_{V_s}(\lambda).
\end{equation}
Substituting this into \eqref{integralid1}, which is equal to the right-hand side of \eqref{partialp}, we obtain
\begin{equation}\label{integralid5}
-\frac{2}{\hat \psi(q)}\oint_\Gamma \frac{\partial}{\partial s} V_s(z)\frac{\widetilde W(q,z)}{2\pi i}dz=\frac{1}{\pi i}\oint_{\partial U_q} d\lambda \frac{\frac{\partial}{\partial s} \psi_{V_s}(\lambda)}{\psi_{V_s}(\lambda)}.
\end{equation}
where the orientation of $\partial U_q$ is clockwise.
As $\lambda \to q$, we have 
\begin{equation*}\frac{\frac{\partial}{\partial s} \psi_{V_s}(\lambda)}{\psi_{V_s}(\lambda)}=-\frac{1}{2(\lambda-q)}\frac{\partial}{\partial s}q+\mathcal O(1), \end{equation*}
which concludes the proof upon evaluating the residue at $q$ in \eqref{integralid5}.

\end{proof}

For a second identity, we observe that by \eqref{integralid6} we have
\begin{equation}\label{integralid7}
-\oint_\Gamma \frac{dz}{2\pi i} \oint_{\partial U_q}\frac{d\lambda}{2\pi i} \frac{\frac{\partial}{\partial s} V_s(z) W(z,\lambda)}{\psi_{V_s}(\lambda)(\lambda-q)}=-\oint_{\partial U_q} d\lambda \frac{\frac{\partial}{\partial s} \psi_{V_s}(\lambda)}{(\lambda-q)\psi_{V_s}(\lambda)}.
\end{equation}
Evaluating the residue at $\lambda=q$ on both sides and applying \eqref{partialp}, we obtain for $q\in\{a_j,b_j\}_{j=1}^k$
\begin{equation}\label{partialpsi}
\frac{\partial}{\partial s} \hat \psi(q)=\oint_\Gamma \frac{dz}{2\pi i} \frac{\partial}{\partial s} V_s(z)\left(\widetilde W'(q,z)-3\frac{\hat \psi'(q)}{\hat \psi(q)} \widetilde W(q,z)\right)  .\end{equation}

Substituting \eqref{partialpsi}, \eqref{partialp}, and \eqref{dtauds} into \eqref{Thirdcutoff}
\begin{multline} \label{Fourthcutoff} \oint_\Gamma 
\left[N_\infty (z)^{-1}\left(\tfrac{d}{d z}R^{(1)}(z)\right) N_\infty (z)\right]_{11}\frac{\partial}{\partial s} V_s(z)
\frac{dz}{2\pi i}  \\ =- \sum_{1\leq m \leq l \leq k-1}\frac{\partial \tau_{m,l}}{\partial s} \frac{\partial}{\partial \tau_{m,l}} \log \frac{\theta(N\Omega)}{\theta(0)}
+ \frac{\partial}{\partial s} \Bigg( \frac{1}{24}\sum_{q\in\{a_j,b_j\}_{j=1}^k} \log \hat \psi(q)\\-\frac{1}{8}\Bigg(\sum_{1\leq l<j \leq k} \left[\log(a_j-a_l)+\log(b_j-b_l)\right]-\sum_{l,j=1}^k\log|b_j-a_l|\Bigg)
\Bigg)+\widetilde{\mathcal O}(1/N),
\end{multline}
as $N\to \infty$, where $\frac{\partial}{\partial \tau_{m,l}}$ is interpreted as in \eqref{eqdifftausym}. Substituting \eqref{Fourthcutoff}, \eqref{eq:PinvP'2}, and the result of Lemma \ref{le:orderN2} into \eqref{eq:orders}, we obtain
\begin{multline} \label{enddiffid}\frac{d}{ds}\log H_N(e^{-NV_s})=\frac{\partial}{\partial s}\Bigg(-N^2 \iint \log|x-y|^{-1}d\mu_{V_s}(x)d\mu_{V_s}(y)\\ -N^2\int V_s(x)d\mu_{V_s}(x)+ \log \frac{\theta(N\Omega)}{\theta(0)}-\frac{1}{24}\sum_{q\in\{a_j,b_j\}_{j=1}^k} \log | \hat \psi(q)|\\+\frac{1}{8}\Bigg(\sum_{1\leq l<j \leq k} [\log(a_j-a_l)+\log(b_j-b_l)]-\sum_{l,j=1}^k\log|b_j-a_l|\Bigg) \Bigg)+\widetilde{\mathcal O}(1/N),
\end{multline}
as $N\to \infty$, since the argument of $\hat \psi (q)$ is independent of $s$ for $q\in \{a_l,b_l\}_{l=1}^k$.

\subsection{The Chebyshev polynomials}
Now consider the potential $V_{0}=\frac{2\sigma}{k}T_k(x)^2$ and take the limit $\sigma\downarrow 1$. Then, by \eqref{bklim} and \eqref{bjlim},
\begin{multline}\label{endlim1}\frac{1}{8}\Bigg(\sum_{1\leq l<j\leq k} [\log(a_j-a_l)+\log(b_j-b_l)]-\sum_{l,j=1}^k\log|b_j-a_l|\Bigg)\\=-\frac{k}{8}\log 2+\frac{k-1}{8}\log k-\frac{k-1}{16}\log (\sigma-1)-\frac{1}{8}\sum_{j=1}^{k-1}\log \left( \sin \frac{\pi j}{k}\right)+o(1),
\end{multline}
as $\sigma \downarrow 1$. Recalling the equilibrium measure for the Chebyshev polynomials in \eqref{eqlibCheby}, the definition of $\hat \psi$ from \eqref{eq:What}, and the fact that $|T_k(q)|=1/\sqrt \sigma$ for $q\in\{a_j,b_j\}_{j=1}^k$, it follows that
\begin{equation*}|\hat \psi(q)|=\frac{2^{3/2}\sigma^{3/4}}{k} |T_k'(q)|^{3/2}.
\end{equation*}
By \eqref{bklim} and \eqref{bjlim},
\begin{align*} 
& b_j=\cos \left(\pi-\frac{\pi j}{k}+\frac{\sqrt{\sigma-1}}{k}\right)+\mathcal O(\sigma-1), \\
& a_{j+1}=\cos \left(\pi-\frac{\pi j}{k}-\frac{\sqrt{\sigma-1}}{k}\right)+\mathcal O(\sigma-1), 
\end{align*}
 as $\sigma \downarrow 1$, for $j=1,2,\dots,k-1$. Thus, since $T_k'(x)=kU_{k-1}(x)$ and $U_{k-1}(\cos \theta)=\frac{\sin k\theta}{\sin \theta}$,
\begin{equation*} |\hat \psi(q)|=\frac{2^{3/2}\sqrt k (\sigma-1)^{3/4}}{\sin \left(\frac{\pi j}{k}\right)^{3/2}}(1+\mathcal O(\sqrt{\sigma-1})),
\end{equation*}
as $\sigma \downarrow 1$, for $q\in \{ b_j,a_{j+1}\}_{j=1}^{k-1}$. Combined with the fact that $|U_{k-1}(\pm 1)|=k$, we obtain
\begin{multline} \label{endlim3}
-\frac{1}{24}\sum_{q\in\{a_j,b_j\}_{j=1}^k} \log |\hat \psi(q)|=-\frac{1}{24}\Bigg(3k\log 2+(k+3)\log k\\ +\frac{3(k-1)}{2}\log (\sigma-1)-3\sum_{j=1}^{k-1}\log \sin \frac{\pi j}{k} \Bigg)+\mathcal O(\sqrt{\sigma-1}).
\end{multline}

By \eqref{LeadingPik}, writing $V_0(x)=\frac{2\nu}{k}\Pi_k(x)^2=\frac{2\sigma}{k} T_k(x)^2$ so that $\nu=\sigma 2^{2k-2}$,
\begin{equation}\label{endlim2}
 \iint \log|x-y|^{-1}d\mu_{V_0}(x)d\mu_{V_0}(y)+\int V_0(x)d\mu_{V_0}(x)=
 \frac{3}{4k}+\frac{1}{2k}\log \sigma +\log 2.
\end{equation}

Recalling from Section \ref{Sec:limsigma} that $\theta(x|\tau)\to 1$ as $\sigma \to 1$ for $x\in \mathbb R^{k-1}$, noting that $\tilde \psi(q)=|\hat \psi(q)|$, and substituting \eqref{endlim1}, \eqref{endlim3}, and \eqref{endlim2} into \eqref{enddiffid}, we obtain \eqref{MainResSec11}, which completes this section.

\end{document}